\documentclass[a4paper,11pt]{article} 
\usepackage[a4paper,hmargin=2.8cm,vmargin=3cm]{geometry}
\usepackage[utf8]{inputenc} 
\usepackage[english]{babel}

\usepackage{mathtools,amsthm, amsfonts, amssymb, mathrsfs, bbm, bm}
\usepackage{authblk}
\usepackage{lineno,csquotes}
\usepackage[dvipsnames]{xcolor}
\usepackage{hyperref}
\usepackage{comment}
\usepackage{upgreek}

\definecolor{BeauBlue}{rgb}{0, 0.2, .9}
\definecolor{BeauOrange}{rgb}{.8, .1, 0}
\usepackage{hyperref}
\hypersetup{
	colorlinks = true,
	linkcolor = BeauBlue,
	urlcolor = cyan,
	citecolor = BeauOrange,
}
 \usepackage[mathcal]{euscript}

\numberwithin{equation}{section}

\theoremstyle{plain}
\newtheorem{theorem}{Theorem}[section] % reset theorem numbering for each chapter
\newtheorem*{theorem*}{Theorem} 
\newtheorem{proposition}[theorem]{Proposition} 

\newtheorem{lemma}[theorem]{Lemma}

\theoremstyle{definition}
\newtheorem{definition}[theorem]{Definition} 
\newtheorem{assumption}{Assumption}
\theoremstyle{remark}% definition numbers are dependent on theorem numbers

\newtheorem{remark}[theorem]{Remark}

\DeclareMathOperator{\Tr}{Tr}
\newcommand{\mb}{\mathbf}
\newcommand{\mc}{\mathcal}
\newcommand{\mf}{\mathfrak}
\newcommand{\mr}{\mathrm}
\newcommand{\ms}{\mathsf}
\newcommand{\msc}{\mathscr}
\newcommand{\timord}{\mb T}

\newcommand{\FF}{^{\Delta F}}
\newcommand\pint{\int \mkern-2mu\mathbb P}

\newcommand\restr[2]{{% we make the whole thing an ordinary symbol
		\left.\kern-\nulldelimiterspace % automatically resize the bar with \right
		#1 % the function
		\vphantom{\Big|} % pretend it's a little taller at normal size
		\right|_{#2} % this is the delimiter
}}

\title{Edge transport in Haldane-like models with quasi-periodic disorder}

\author[1]{Fabrizio Caragiulo}
\author[2]{Vieri Mastropietro}
\author[1]{Marcello Porta}
\affil[1]{SISSA, Via Bonomea 265, 34136 Trieste, Italy}
\affil[2]{University of Rome Sapienza, P.le Aldo Moro 2, 00185 Roma, Italy}

\date{\today}

\begin{document}

\maketitle

\begin{abstract}
We consider Haldane-like $2d$ topological insulators on the cylinder, in the presence of weak quasi-periodic disorder. We prove that, at large distances, the boundary correlations agree with the correlations of a renormalized, translation-invariant, massless relativistic model in $1+1$ dimensions, multiplied by non-universal oscillatory factors, incommensurate with the lattice spacing. Furthermore, we compute the edge conductance and the edge susceptibility, starting from Kubo formula. We obtain explicit expressions for these response functions, completely determined by the renormalized Fermi velocity of the edge modes. In particular, we prove the quantization of the edge conductance, and the non-universality of the susceptibility. The proof relies on multiscale analysis and rigorous renormalization group methods for quasi-periodic systems, and on lattice Ward identities.
\end{abstract}

\tableofcontents

\section{Introduction}\label{sec:intro}
One of the hallmarks of topological phases in condensed matter physics is the presence of stable, metallic currents at the boundary of the material. These metallic states are called edge modes, and their very existence and stability is deeply related to the value of the bulk topological invariant labelling the phase of the system. Edge currents have first been predicted to arise for quantum Hall systems \cite{Halperin}. The bulk-edge duality for quantum Hall systems \cite{Hat, SKR, EG, EGS} relates the signed number of such chiral currents, where the sign takes into account the direction of propagation, to the value of the bulk Hall conductivity. In the last years, this remarkable duality has been extended to a large class of noninteracting topological materials; see {\it e.g.} \cite{PS} for a review.

However, the bulk-edge duality does not provide quantitative information about the properties of the edge modes. Spectral theoretic methods have been used to prove the stability of the continuous spectrum of the edge modes, with respect to a broad class of perturbations, and for models defined on a lattice or in the continuum, see {\it e.g.} \cite{DP, BW, Briet, FGW, HS, MMP}. Besides the existence and the stability of the edge modes, a natural question, addressed in this paper, is how to obtain rigorous predictions on physically relevant quantities, such as edge correlation functions, edge scaling exponents, or edge transport coefficients.

At a formal level, these objects can be investigated adopting an effective quantum field theory (QFT) viewpoint. At large scales, the edge modes are expected to be described by chiral, relativistic one dimensional fermions \cite{Wen, Frohlich}; the relevant effective QFT is the chiral Luttinger model, which is integrable via bosonization methods. This allows to obtain closed formulas for edge correlation functions, and for edge transport coefficients. Still, from a mathematical viewpoint the reduction step from a two-dimensional quantum system to a relativistic one-dimensional model is a nontrivial one; it is a natural question to understand in what sense this reduction holds, on which length scales, and for which class of lattice models.

In the last years there has been progress in the rigorous justification of the chiral Luttinger liquid as an effective one-dimensional theory for the boundary of weakly interacting two-dimensional quantum Hall systems. The works \cite{AMP, MP} allowed to put this reduction on rigorous grounds, for weakly interacting systems on a cylinder, via rigorous renormalization group (RG) methods. The technique allows to compute boundary correlation functions and to extract their scaling limit, which turns out to be described by the multi-channel Luttinger liquid. Furthermore, the combination of rigorous RG methods with lattice and with emergent conservation laws allows to compute the real-time transport coefficients relevant for linear response, without relying on bosonization, difficult to implement in rigorous arguments. In particular, the edge conductance turns out to be exactly quantized; as a corollary, if combined with the quantization of the many-body Hall conductivity \cite{HM, GMPhall, BBdRF} and in particular with its universality \cite{GMPhall}, the results of \cite{AMP, MP} allow to lift the bulk-edge duality to the many-body context, at least for weakly interacting systems. More recently, \cite{LMTW} proved the equality of bulk and edge transport coefficients for interacting fermions at positive temperatures, for systems satisfying local indistinguishability of the Gibbs state.

The main limitation of \cite{AMP, MP} is the restriction to systems that are translationally invariant in the direction of the edge. This symmetry allows to introduce the concept of quasi-momentum for the edge currents, and to rely on momentum conservation, crucial for the RG analysis. It is an important open problem to develop rigorous RG techniques that allow to study interacting and disordered systems, and to understand the scaling limit of the disordered correlation functions. 

In this work, we develop a rigorous RG method for the edge states of weakly disordered, non-interacting two-dimensional quantum systems. Our goal is to characterize precisely the large scale edge properties of the system, and to determine the linear response from a microscopic model. Disorder will be described by a quasi-periodic potential, whose frequency is incommensurate with the lattice spacing. We will actually suppose that the quasi-periodic potential is modulated via a frequency that is badly approximated by rationals: mathematically, we will assume that the frequency is a Diophantine number, which form a set of full measure in the reals. In the last five decades, there has been enormous work in the mathematics of quasi-periodic Schroedinger operators, mostly in one dimension, from the seminal works \cite{DS, MoPo, E} to more recent breakthroughs such as \cite{AJ1, AJ2, AYZ}. The standard example is the almost Mathieu model, which is the Hamiltonian of a quantum Hall system with no boundary, after a partial Bloch reduction. For this type of model, a very precise understanding of the spectrum has been achieved, at a non-perturbative level. In particular, the fractal nature of the spectrum as function of the magnetic field, related to the celebrated Hofstadter butterfly \cite{AOS}, has been understood \cite{AJ2}. 

The model considered in the present paper is motivated by the physics of topological insulators, and does not belong to the class of Schroedinger operators considered in the very large literature about quasi-periodic systems. In the absence of disorder, the Hamiltonian describes electrons hopping on a two-dimensional lattice wrapped around a cylinder. We will assume that the non-disordered Hamiltonian is finite-ranged, translationally invariant along the periodic direction of the cylinder, and that it satisfies Dirichlet boundary conditions on the edges. We will assume that the spectrum of the Hamiltonian supports edge modes: that is, generalized eigenfunctions that are exponentially localized on the two boundaries of the cylinder. Specifically, we will assume that every edge supports one edge mode. A model fitting the picture is the Haldane model \cite{Hal}, a paradigmatic example of topological insulator. We will then perturb the model adding a weak quasi-periodic potential, and we will address the problem of understanding the behavior of the edge modes, in a quantitative way. 

From the physics viewpoint, this analysis is motivated by the recent progress in condensed matter physics, where lattice distortions, Moiré superlattices and magnetic impurities are often described in terms of quasi-periodic potentials; see for instance \cite{HuangLiu, PixleyHuse, PixleyHuse2, Szabo, Okugawa, Karcher, Vongkovit}, as well as the experimental works \cite{Schreiber, Uri, He}. With respect to random systems, it is worth mentioning that quasiperiodic systems present additional challenges, due to the fact that many tools used in the random case (such as the replica trick in the physics literature) are no longer applicable.

Effectively, in our model each edge mode behaves as a one-dimensional system. The key difference with respect to the many one-dimensional cases that have been studied in the mathematics literature is that the edge modes are chiral: every boundary supports currents associated with a single Fermi point, instead of two counterpropagating currents as for the usual Laplacian in $1d$. For this model, one could address all the questions that have been investigated over the years for one-dimensional systems. Due to the bidimensionality of the lattice model, however, the existing theorems developed for $1d$ systems do not directly apply to the present case. In this work we adopt a statistical mechanics viewpoint, and we use KAM-like direct methods to rigorously characterize the Gibbs state of the system at weak quasi-periodic disorder. Furthermore, our method allows us to compute from first principles the edge conductance and the edge susceptibility of the system, defined respectively as the linear response of the edge current and edge density to an external potential that is supported in proximity of the boundary, and that is turned on adiabatially in time. 

Our first result is an explicit expression for the large-scale behavior of the edge two-point correlation function: we find that it is given by the two-point function of a dressed, translation invariant, massless relativistic model in $1+1$ dimensions, multiplied by non-universal oscillatory prefactors, that are incommensurate with the lattice spacing. Thus, informally, the scaling limit of the two-point function is given by an ``oscillatory''  QFT; to the best of our knowledge, this is the first time that such oscillatory behavior is obtained for the large scale properties of weakly disordered quantum systems on a lattice. Being the model non-interacting, the knowledge of the two-point function allows to fully characterize all edge correlation functions, via the fermionic Wick rule. Notice that, in contrast to the widely studied quasi-periodic models in 1d, where a Cantor set of gaps opens as soon as the disorder is turned on, here the system is always gapless for all values of the chemical potential within the bulk spectral gap.

Our second result is an explicit computation of the edge conductance and of the edge susceptibility. Their values turn out to be completely determined by the renormalized Fermi velocity, a quantity defined by the scaling limit of the two-point function. Both response functions are defined via Kubo formula, as response to an external perturbation supported in proximity of the edge, slowly varying in space and in time. We prove that the edge conductance is quantized, in agreement with the bulk Hall conductivity; instead, the edge susceptibility is non-universal, and it has a simple expression is terms of the effective Fermi velocity of the edge modes. To the best of our knowledge, this is the first time that edge transport coefficients are computed directly from Kubo formula, in the setting of weakly disordered systems.

The proofs are based on multiscale analysis for disordered systems, and on renormalization group techniques.  With respect to the many applications of fermionic RG of the last years, a simplification of the present case is that the fermions are non-interacting. Still, a major difficulty is due to the presence of the quasi-periodic disorder, that breaks translation invariance, and potentially introduces infinitely many new resonances in the multiscale analysis. Thanks to the Diophantine assumption on the modulation frequency of the potential, we can ultimately control all these singularities, and prove that the net effect of the disorder is to renormalize the Fermi velocity, to shift the Fermi points, and to dress the correlation functions by oscillatory factors. The use of renormalization group methods for quasi-periodic systems has been pioneered in  \cite{BGM}, where the Schwinger functions of one-dimensional systems in the presence of weak quasi-periodic disorder have been studied, and where the gap opening at special values of the chemical potential has been established in the sense of exponential decay of correlations. The analysis of the present paper is partially inspired by recent related work for the quasi-periodic Ising model \cite{GM}; however, the extension to the present case is nontrivial, due for instance to the presence of the bulk degrees of freedom, and to the necessity to control the scattering between different edges, ultimately suppressed by the exponential decay of the edge modes. Also, \cite{GM} studied the correlation functions for bilinear fermionic observables, while here we construct the two-point function, which allows to compute all correlation functions via the fermionic Wick's rule.

Similar methods have been applied in the past to study the localization and delocalization of one-dimensional interacting fermions with quasi-periodic disorder, for values of the density that rule out resonances between different Fermi points \cite{M0, M1, M2}.

With the result for the two-point function at hand, in the second part of the paper we address the problem of computing edge transport coefficients. We extend the strategy that has been introduced for translation-invariant edge modes in \cite{AMP, MP}; the key idea is to use conservation laws to impose strong constraints on all the contributions to transport that cannot be computed explicitly (that is, starting from the scaling limit of the two-point function). These contributions are crucial, and allow to obtain remarkably simple expressions for the response functions. The analysis of \cite{AMP, MP} relied strongly on translation invariance. Here, we extend the method to disordered systems; a major technical challenge faced in the present work is to prove that on large scales the transport coefficients can be approximated by those of a renormalized, translation-invariant model. To achieve this, we crucially rely on the number-theoretic properties of Diophantine frequencies. In particular, the method allows to establish nontrivial relations between the amplitudes of the modes of the oscillatory functions appearing in the scaling limit of the two-point function and the Fermi velocity, which seem to be novel, to the best of our knowledge.

In perspective, we believe that the methods of this paper open several interesting directions, such as the analysis of interacting, quasi-periodic edge currents, or the study of quasiperiodic models with several edge modes. The extension to several edge modes is far from being trivial, as the analysis of the almost-Mathieu model suggests: for suitable values of the chemical potential, counterpropagating chiral fermions are expected to annihilate, a phenomenon associated with mass generation from the QFT viewpoint. This of course cannot happen for all configurations of edge modes, since it would imply a violation of the bulk-edge duality. It would be interesting to have a RG perspective of the stability and on the disappearance of edge modes due to scattering induced by the disorder. Also, let us mention that the transport properties of interacting and disordered edge modes are not completely understood. For instance, we mention here the puzzle of the two-terminal conductance of the Hall bar: for interacting and translation-invariant systems, bosonization methods predict a non-universal result for materials with counterpropagating edge modes, in contrast with experiments. In order to recover the correct result, disorder should be taken into account; it has been suggested that disorder should drive the system to a different RG fixed point \cite{KF, KFP}, for which the quantization of the two-terminal conductance holds true. The rigorous formulation of this argument seems to be beyond present day mathematical methods.

The paper is organized as follows. In Section \ref{sec:themodel} we introduce the class of models we consider, we define the Gibbs state and we state the assumptions on the non-disordered edge spectrum. We then introduce the response functions, and the Euclidean two-point function. In Section \ref{sec:mainres} we state our main results, Theorem \ref{thm:2pt} about the two-point function, and Theorem \ref{thm:resp} about the response functions, and we give a sketch of the proof of Theorem \ref{thm:resp}. The rest of the paper is devoted to the proofs; the article is self-contained, and this contributes to its length. In Section \ref{sec:RG} we introduce the RG analysis, which is conveniently defined in the Grassmann representation of the model, and we prove Theorem \ref{thm:2pt}. In Section \ref{sec:transport} we use the result about the two-point function to compute the edge response function, after the rigorous Wick rotation, and to prove Theorem \ref{thm:resp}. In Appendix \ref{app:CT} we prove an estimate for the propagator associated with the energies away from the Fermi level; in Appendix \ref{app:2pt} we prove a recursion relation that allows to determine the two-point function from the effective potential; in Appendix \ref{app:bubble} we reproduce for completeness the computation of the relativistic bubble diagram, which plays a key role in the computation of the transport coefficients; while Appendices \ref{app:Rest}, \ref{app:contK} collect technical details of auxiliary results in our proofs.

\paragraph{Acknowledgements.} F. C. and M. P. acknowledge support by the European Research Council through the ERC-StG MaMBoQ, n. 802901. V. M and M. P. acknowledge support from the MUR, PRIN 2022 project MaIQuFi cod. 20223J85K3. This work has been carried out under the auspices of the GNFM of INdAM.

\section{The model}\label{sec:themodel}
\subsection{The Hamiltonian and the Gibbs state}
Let $L_{1}, L_{2}\in \mathbb{N}$, $L = (L_{1}, L_{2})$, and consider the lattice:
\begin{equation}
\Lambda_{L} = \Big\{ \vec x \in \mathbb{Z}^{2} \Big|\; 1 \le x_{i} \le L_{i},\quad i =1,2 \Big\}\;.
\end{equation}
We shall equip the lattice with periodic boundary conditions in the direction $1$, and Dirichlet boundary conditions in the direction $2$.  That is, we will consider functions $f(\vec x)$ satisfying {\it cylindric boundary conditions:}
\begin{equation}
f(x_{1} + n L_{1}, x_{2}) = f(x_{1}, x_{2})\;,\quad f(x_{1}, 1) = f(x_{1}, L_{2}) = 0\;,\quad \text{for all $\vec x\in \Lambda_{L}$.}
\end{equation}
We will use the following distance between points on $\Lambda_{L}$:
\begin{equation}\label{eq:distL}
\| \vec x - \vec y \|_{L} := \min_{n\in \mathbb{Z}} \| \vec x - \vec y + n L e_{1} \|\;,\qquad \text{for all $\vec x, \vec y \in \Lambda_{L}$},
\end{equation}
with $\|\cdot\|$ the usual Euclidean norm on $\mathbb{R}^{2}$ and $e_{1} = (1,0)$, $e_{2} = (0,1)$.
\begin{figure}
    \centering
    \includegraphics[scale=0.3]{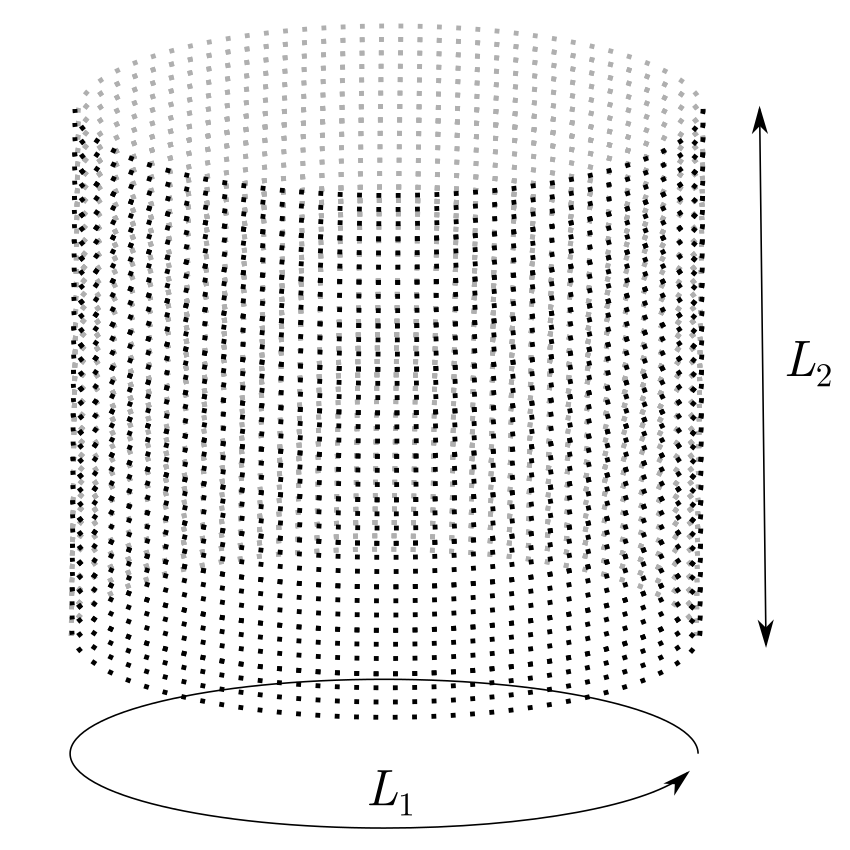}
    \caption{The lattice $\Lambda_L$.}
    \label{fig:cylinder}
\end{figure}

We are interested in describing a fermionic lattice gas on $\Lambda_{L}$, in a grand-canonical setting. The fermionic Fock space is:
\begin{equation}
\mathcal{F}_{L} = \bigoplus_{n\ge 0} \big(\ell^{2}(\Lambda_{L}; \mathbb{C}^{S})\big)^{\wedge n}\;;
\end{equation}
the number $S \in \mathbb{N}$ denotes the total number of internal degrees of freedom, {\it e.g.} the spin label or the sublattice label. The Hamiltonian of the model is a self-adjoint operator on $\mathcal{F}_{L}$, of the following form:
\begin{equation}\begin{split}\label{eq:Ham}
       \mathcal{H} &:=  {\sum_{\vec x,\vec y\in \Lambda_{L}}\sum_{\sigma,\zeta=1}^S  {a}^*_{\vec x,\sigma} H_{\sigma,\zeta}(\vec x,\vec y){a}_{\vec y,\zeta}} + \lambda \sum_{\vec x \in \Lambda_{L}} \sum_{\sigma = 1}^{S} \varphi(\vec x){a}^*_{\vec x,\sigma}  {a}_{\vec x,\sigma}    \; ,
\end{split}
\end{equation}
where $a, a^{*}$ are the usual fermionic creation and annihilation operators, $H$ is a self-adjoint operator on $\ell^{2}(\Lambda_{L}; \mathbb{C}^{S})$ (compatible with the cylindric boundary conditions), and $\lambda \varphi(\cdot )$ introduces a perturbation. The model describes an electron gas on a lattice in the presence of impurities, described by $\varphi(\cdot)$, and we will be interested in the equilibrium and non-equilibrium properties of the model. Before introducing the specific questions to be addressed, let us specify the type of Hamiltonians $H$ and perturbations $\varphi$ that we shall consider.

\begin{assumption}[On the free Hamiltonian]\label{ass:H1} The Hamiltonian $H$ is the periodization of a self-adjoint operator $H^{\infty}$ on $\ell^{2}\big(\mathbb{Z} \times ([1,L_2] \cap \mathbb{Z}); \mathbb{C}^{S}\big)$. That is, for all $\vec x, \vec y\in \Lambda_{L}$:
\begin{equation}\label{eq:period}
H(\vec x, \vec y) = \sum_{n\in \mathbb{Z}} H^{\infty}\big(( x_{1} + n L_{1}, x_{2} ), (y_{1}, y_{2})\big)\;.
\end{equation}
Furthermore, we assume that $H^{\infty}$ is finite-ranged and translation invariant:
\begin{equation}
H^{\infty}(\vec x, \vec y) = 0  \;,
\end{equation}
if $\|\vec x - \vec y\| > R\;,$ and
\begin{equation}
    H^{\infty}(\vec x, \vec y) \equiv H^{\infty}(x_{1} - y_{1}; x_{2}, y_{2})\;.
\end{equation}
We will assume that $R = \sqrt{2}$. This is not a loss of generality, up to increasing the number of internal degrees of freedom $S$. Thus, $H$ is also translation invariant, and finite ranged: $H(\vec x, \vec y) = 0$, if $\| \vec x - \vec y \|_{L} > \sqrt{2}$.
\end{assumption}
Later, we will make further assumptions on the spectrum of $H$. We now define the type of perturbations we will consider.
\begin{assumption}[On the perturbation]\label{ass:pert} The function $\varphi(\cdot)$ is real-valued, and it has the form:
\begin{equation}\label{eq:pot}
\varphi(x_{1},x_{2}) = \sum_{n\in \mathbb{Z}} e^{i n\alpha x_{1}} \hat \varphi_{n}(x_{2})\;,
\end{equation}
where:
\begin{equation}\label{eq:decFou}
| \hat \varphi_{n}(x_{2}) | \le Ce^{-c|n|}\;,\qquad \text{uniformly in $x_{2}$.}
\end{equation}
Also, $\hat \varphi_{n}(x_{2})$ satisfies the Dirichlet boundary conditions. Furthermore, the number $\alpha /2\pi=  m / L_{1}$ is the best rational approximant (in the sense of convergents of the continued fraction expansion, see {\it e.g.} \cite[Appendix 1]{BGM}), with denominator $L_{1}$, of a real number $\alpha_{\infty}/2\pi$ which is Diophantine. That is: there exists positive constants $\tilde c, \tau$ such that:
\begin{equation}\label{eq:bdiof}
| n\alpha_{\infty} |_{\mathbbm{T}} \ge \frac{\tilde c}{|n|^{\tau}}\;,\qquad \text{for all $n\neq 0$.}
\end{equation}
\end{assumption}
\begin{remark}\label{rem:lim}
\begin{itemize}
\item[(i)] The function $\varphi(\vec x)$ describes a quasi-periodic perturbation to the Hamiltonian $H$. Notice that we are not making any assumption about the $x_{2}$ dependence of the perturbation. In particular, translation-invariance might also be broken in the $x_{2}$ direction.
\item[(ii)] Thanks to the rational approximation of $\alpha$, the function $\varphi(\vec x)$ is compatible with the periodic boundary conditions. The thermodynamic limit is taken over sequences $\{L_{1}\}$ such that $\text{mcd}(m, L_{1}) = 1$.
\item[(iii)] The approximant $\alpha$ satisfies the following version of the bound (\ref{eq:bdiof}). There exists $c>0$ such that for every $n \neq 0$, $|n| \le L_{1} / 2$:
\begin{equation}\label{eq:diophantine}
| n \alpha |_{\mathbb{T}} \ge \frac{c}{|n|^{\tau}}\;.
\end{equation}
%
%Thanks to the constraint in the sum in (\ref{eq:pot}), all modes entering the definition of $\varphi(\cdot)$ satisfy the condition (\ref{eq:diophantine}). 
%\marginpar{}
%The constant $c$ can be taken as $\tilde c / 2$ for $n$ large enough (see \cite[Appendix 1]{BGM}). 
Furthermore, by the properties of the continued fractions expansion which is used to construct $\alpha$ (see {\it e.g.} \cite{BVDP}), we also have:
\begin{equation}
| \alpha - \alpha_{\infty} | \le \frac{1}{L_{1}^{2}}\;.
\end{equation}
\end{itemize}
\end{remark}
Next, let us introduce the Gibbs state associated with the Hamiltonian (\ref{eq:Ham}). Given an operator $\mathcal{O}$ on $\mathcal{F}_{L}$, we define:
\begin{equation}
\langle \mathcal{O} \rangle_{\beta, \mu, L} := \frac{\Tr \big( \mathcal{O}  \, e^{-\beta (\mathcal{H} - \mu \mathcal{N})} \big) }{\ms Z_{\beta, \mu, L}}\;,
\end{equation}
with $\beta>0$ the inverse temperature and $\mu \in \mathbb{R}$ the chemical potential. The operator $\mathcal{N}$ is the number operator, $\mathcal{N} = \sum_{\sigma, \vec x} a^{*}_{\vec x, \sigma} a_{\vec x, \sigma}$. Finally, $\ms Z_{\beta, \mu, L}$ is the partition function,
\begin{equation}
\ms Z_{\beta,\mu,L} = \Tr e^{-\beta (\mathcal{H} - \mu \mathcal{N})}\;.
\end{equation}
The Heisenberg evolution generated by $\mathcal{H}$ is denoted by $\tau_{t}(\mathcal{O}) = e^{i\mathcal{H} t} \mathcal{O} \,e^{-i\mathcal{H}t}$, and the Gibbs state is trivially invariant under such dynamics:
\begin{equation}
\langle \tau_{t}(\mathcal{O}) \rangle_{\beta, \mu, L} = \langle \mathcal{O} \rangle_{\beta, \mu, L}\;.
\end{equation}
Next, we will make an assumption on the form of the spectrum of the unperturbed Hamiltonian $H$ is proximity of the chemical potential $\mu$. Let:
\begin{equation}
\msc D_{L} := \frac{2\pi}{L_{1}} \big(\mathbb{Z} / L_{1}\mathbb{Z}\big)\;
\end{equation}
be the independent momenta associated with the periodic boundary conditions in $x_{1}$. As $L_{1}\to \infty$, these points fill densely the circle of length $2\pi$, $\mathbb{T}$. Recall the definition of partial Fourier transform for a function $f(x_{1}, x_{2})$ on $\Lambda_{L}$:
\begin{equation}
\hat f(k_{1}, x_{2}) = \sum_{x_{1} = 1}^{L_{1}} e^{-i k_{1} x_{1}} f(x_{1}, x_{2})\; ,\qquad \text{for $k_{1} \in \msc D_{L}\;.$}
\end{equation}
This identity can be inverted, as follows:
\begin{equation}
f(x_{1}, x_{2}) = \frac{1}{L_{1}} \sum_{k_{1} \in \msc D_{L}} e^{ik_{1} x_{1}} \hat f(k_{1}, x_{2})\;.
\end{equation}
Let us denote the partial Bloch decomposition $\hat H(k_{1})$ of $H$ as the operator such that:
\begin{equation}
\widehat{(H f)}(k_{1}, x_{2}) = \big(\hat H(k_{1}) \hat f(k_{1})\big)(x_{2})\;.
\end{equation}
Observe that $\hat H(k_{1})$ is an operator acting on functions on the discrete interval $[1, L_{2}]\cap \mathbb Z$. For every $k_{1}$, the operator $\hat H(k_{1})$ is self-adjoint. Explicitly,
\begin{equation}\label{eq:blochH}
\hat H_{\sigma,\zeta}(k_{1}; x_{2}, y_{2}) = \sum_{z_{1}=1}^{L_{1}} e^{-i k_{1} z_{1}} H_{\sigma,\zeta}\big((x_{1}+z_{1}, x_{2}), (x_{1}, y_{2})\big)\;.
\end{equation}
\begin{remark}
\begin{itemize}
\item[(i)] Being $H$ the periodization of $H^{\infty}$, see Eq. (\ref{eq:period}), it follows that:
\begin{equation}
\hat H(k_{1}) = \hat H^{\infty}(k_{1})\qquad \text{for all $k_{1} \in \msc D_{L}$.}
\end{equation}
\item[(ii)] Thanks to the finite range property of $H^{\infty}$, $\hat H^{\infty}(k_{1})$ is analytic in $k_{1}$.
\end{itemize}
\end{remark}
We will also adopt the following conventions for the Fourier transforms of the creation and annihilation operators:
\begin{equation}
\hat a_{k_{1}, x_{2}, \sigma} = \sum_{x_{1} = 1}^{L_{1}} e^{-ik_{1} x_{1}} a_{\vec x,\sigma}\;,\qquad \hat a^{*}_{k_{1}, x_{2}, \sigma} = \sum_{x_{1} = 1}^{L_{1}} e^{ik_{1} x_{1}} a^{*}_{\vec x,\sigma}\;,
\end{equation}
which can be inverted as:
\begin{equation}
a_{\vec x,\sigma} = \frac{1}{L_{1}} \sum_{k_{1} \in \msc D_{L}} e^{i k_{1} x_{1}} \hat a_{k_{1}, x_{2}, \sigma}\;,\qquad a^{*}_{\vec x,\sigma} = \frac{1}{L_{1}} \sum_{k_{1} \in \msc D_{L}} e^{-i k_{1} x_{1}} \hat a^{*}_{k_{1}, x_{2}, \sigma}\;.
\end{equation}
Observe that, with these definitions, the second-quantized unperturbed Hamiltonian can be written as:
\begin{equation}
\begin{split}
&{\sum_{\vec x,\vec y}\sum_{\sigma,\zeta}  {a}^*_{\vec x,\sigma} H_{\sigma,\zeta}(\vec x,\vec y){a}_{\vec y,\zeta}} \\&\qquad = \frac{1}{L_{1}} \sum_{k_{1}} \sum_{x_{2}, y_{2}} \sum_{\sigma,\zeta} \hat a^{*}_{k_{1}, x_{2}, \sigma} \hat H_{\sigma,\zeta}(k_{1}; x_{2}, y_{2}) \hat a_{k_{1}, y_{2}, \zeta}\;.
\end{split}
\end{equation}
The next assumption specifies properties on the spectrum of $H$. In particular, we will require that the model supports edge states, on the boundary of the cylinder $\Lambda_{L}$.
\begin{assumption}[On the spectrum of $H$]\label{ass:H3} Consider the eigenvalue equation for $\hat H^{\infty}(k_{1})$, for $k_{1} \in \mathbb{T}$:
\begin{equation}\label{eq:schro}
\hat H^{\infty}(k_{1}) \xi(k_{1}) = \varepsilon(k_{1}) \xi(k_{1})\;,
\end{equation}
with $\xi(k_{1}) \in \ell^{2}\big( [1, L_{2}]\cap \mathbb Z; \mathbb{C}^{S}\big)$ satisfying the Dirichlet boundary conditions. We suppose that there exists $\delta > 0$, independent of $L_{2}$, such that the following is true. Let $(\xi(k_{1}), \varepsilon(k_{1}))$ be a solution of (\ref{eq:schro}) with $\varepsilon(k_{1}) \in (\mu-\delta, \mu+\delta)$.
\begin{itemize}
\item[(i)] The eigenvalue $\varepsilon(k_{1})$ is a smooth function of $k_{1}$. We shall suppose that there are two and only two such functions in the energy range $(\mu - \delta, \mu + \delta)$, that we call $\varepsilon_{\omega}(k_{1})$ with $\omega = \pm$.
\item[(ii)] We call Fermi points $k_{F}^{\omega}$ the solutions of the equations:
\begin{equation}
\varepsilon_{\omega}(k_{F}^{\omega}) = \mu\;.
\end{equation}
%
%We will suppose that $k_{F}^{+} \neq k_{F}^{-}$. 
Let $v_{\omega}$ be the Fermi velocity associated with the Fermi point labelled by $\omega$,
\begin{equation}
v_{\omega} = \partial_{k_{1}} \varepsilon_{\omega}(k_{F}^{\omega})\;.
\end{equation}
We shall suppose that $v_{\omega} \neq 0$. Up to choosing a smaller $\delta$, without loss of generality we assume that the first derivatives of $\varepsilon_{\omega}(k_{1})$ are nonzero as long as $\varepsilon_{\omega}(k_{1})\in (\mu - \delta, \mu + \delta)$. 
\item[(iii)] Let $\xi^{\omega}(k_{1})$ be the eigenfunctions associated with $\varepsilon_{\omega}(k_{1})$. They are localized on opposite sides of the cylinder:
\begin{equation}\label{eq:decedge}
| \partial_{k_{1}}^{r} \xi_{\sigma}^{+}(k_{1};x_{2}) | \le C_{r}e^{-\tilde cx_{2}}\quad \text{and} \quad | \partial_{k_{1}}^{r} \xi_{\sigma}^{-}(k_{1}; x_{2}) | \le C_{r}e^{-\tilde c(L_{2} - x_{2})}\;.
\end{equation}
\end{itemize}
\end{assumption}
\begin{figure}
    \centering
    \includegraphics[scale=0.55]{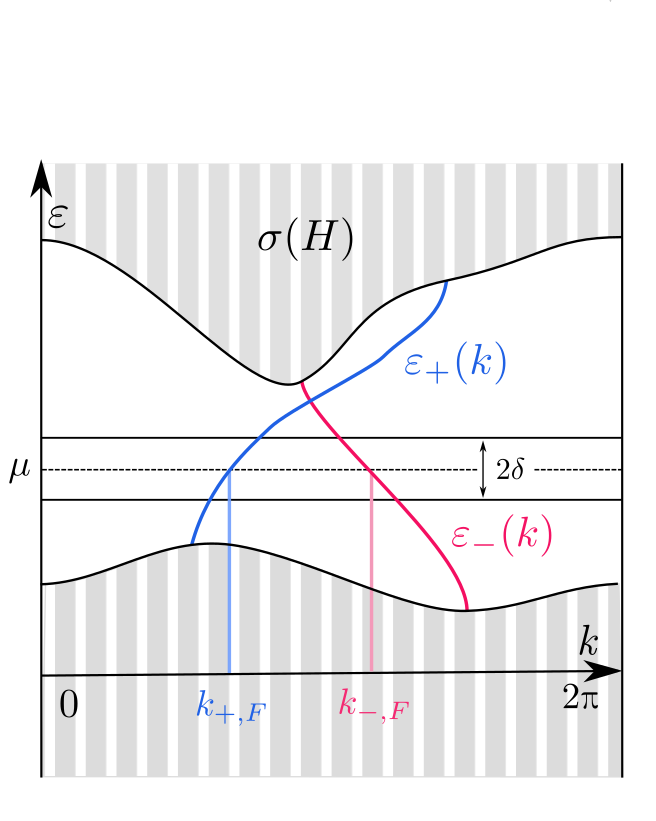}
    \caption{The colored lines corresponds to the edge modes, bulk spectrum sits inside the striped gray area.}
    \label{fig:spectrum}
\end{figure}
Thus, we are assuming the existence of two {\it chiral edge modes,} localized on the two sides of the cylinder. These edge modes correspond to (generalized) eigenfunctions for the original Hamiltonian $H$, of the form $e^{i k_{1} x_{1}} \xi^{\omega}(k_{1}; x_{2})$. An example of model that fulfills the existence of such edge modes is the {\it Haldane model}, a celebrated example of topological insulator \cite{Hal, KM}. Being delocalized in the $x_{1}$ direction and localized in proximity of one of the two boundaries, the edge modes support metallic currents. Of course, these edge modes are no longer eigenstates of the model after the quasi-periodic perturbation is turned on. One of the goals of the present paper will be to probe edge currents in the presence of weak quasi-periodic disorder, after exposing the system to a suitable class of time-dependent perturbations. This will be done introducing suitable transport coefficients, defined in the next section.
\subsection{Transport}
Let us consider the time-dependent Hamiltonian, for $t\le 0$ and $\eta, \theta > 0$:
\begin{equation}\label{eq:Ht}
\mathcal{H}(\eta t) = \mathcal{H} - e^{\eta t} \mc{P}\;,\qquad \mc{P} = \sum_{{\vec x} \in \Lambda_{L}} \mu(\theta \vec x) n_{{\vec x}}\;,
\end{equation}
where $n_{{\vec x}} = \sum_{\sigma }a^{*}_{{\vec x,\sigma}} a_{{\vec x,\sigma}}$ is the density operator, and $\mu(\theta \vec x)$ is the $x_{1}$-periodization of a smooth function $\mu_{\infty}(\vec x)$ on $\mathbb{R}^{2}$ rescaled by $\theta$:
\begin{equation}\label{eq:periodization}
\mu(\theta \vec x) = \sum_{n\in \mathbb{Z}} \mu_{\infty}\big(\theta (x_{1} + nL_{1}), \theta x_{2}\big)\;.
\end{equation}
We will suppose that $\mu_{\infty}(\vec x)$ decays faster than any power:
\begin{equation}\label{eq:decmu}
| \partial_{x_{0}}^{n_{0}} \partial_{x_{1}}^{n_{1}}\mu_{\infty}(\vec x) | \le  \frac{C_{r, n_{0},n_{1}}}{1+|\vec x|^{r}}\;,
\end{equation}
which implies decay of $\mu(\theta \vec x)$:
\begin{equation}
\begin{split}
|\mu(\theta \vec x)| &\le \sum_{n\in \mathbb{Z}} \frac{1 + \theta^{r} | x_{1} + nL_{1} |^{r}}{1 + \theta^{r} | x_{1} + nL_{1} |^{r}}  \big| \mu_{\infty}\big(\theta (x_{1} + nL_{1}), \theta x_{2}\big)\big| \\
&\le \frac{C_{r}}{1 + \theta^{r} | x_{1} |_{L_{1}}} \frac{1}{1 + \theta^{r}|x_{2}|^{r}}
\end{split}
\end{equation}
with $| x_{1} |_{L_{1}} = \min_{n\in \mathbb{Z}} | x_{1} + n L_{1} |$. Let us now consider the function restricted to the lattice $\Lambda_{L}$. In terms of the partial Fourier transform of $\mu_{\infty}(\theta \vec x)$, we have, for $\vec x \in \Lambda_{L}$:
\begin{equation}\label{eq:fourier}
\begin{split}
\mu(\theta \vec x) &= \frac{1}{L_{1}} \sum_{k \in \msc D_{L}}\mr e^{i kx_{1}} \hat \mu_{\theta}(k,  x_{2})\;,\\
\hat \mu_{\theta}(k,x_{2}) &= \frac{1}{\theta} \sum_{n\in \mathbb{Z}} \hat \mu_{\infty}\big((k + 2\pi n) / \theta,   \theta x_{2}\big)\;.
\end{split}
\end{equation}
Let $|\cdot|_{\mathbb{T}}=\min_{n\in \mathbb{Z}} | \cdot + 2\pi n |$ be the distance on the torus of length $2\pi$. By (\ref{eq:decmu}), the following bound holds:
\begin{equation}\label{eq:mutheta}
\begin{split}
|\hat \mu_{\theta}(k,x_{2})|& \le \frac{1}{\theta} \sum_{n\in \mathbb{Z}} \big|\hat \mu_{\infty}\big((k + 2\pi n) / \theta,   \theta x_{2}\big)\big| \\
&  \le \frac{1}{\theta} \sum_{n\in \mathbb{Z}} \frac{1 + |(k + 2\pi n) / \theta|^{r}}{1 + |(k + 2\pi n) / \theta|^{r}} \,  \hat \mu_{\infty}((k + 2\pi n) / \theta,   \theta x_{2}) \\
& \le \frac{C_{r}}{\theta} \frac{1}{1 + (| k |_{\mathbb{T}} / \theta)^{r}} \frac{1}{1 + \theta^{r} |x_{2}|^{r}}\;.
\end{split}
\end{equation}
 The dynamics of the equilibrium state of the system, generated by the Hamiltonian (\ref{eq:Ht}), is:
\begin{equation}\label{eq:dyn}
i\partial_{t} \rho(t) = [ \mathcal{H}(\eta t), \rho(t) ]\;,\qquad \rho(-\infty) = \rho_{\beta, \mu, L}\;,
\end{equation}
where $\rho_{\beta, \mu, L}$ is the Gibbs state of the Hamiltonian $\mathcal{H}$. We shall be interested in the linear response of the density and of the current operators, defined as follows. Consider the lattice continuity equation:
\begin{equation}\label{eq:cons}
\partial_{t} \tau_{t}(n_{\vec x}) + \sum_{i=1,2} \text{d}_{x_{i}} \tau_{t}( j_{i, \vec x} ) = 0\;,
\end{equation}
with $\tau_{t}$ the Heisenberg evolution generated by $\mathcal{H}$, and $\text{d}_{x_{i}}$ the discrete derivative: 
\begin{equation}
    \text{d}_{x_i} f(\vec y) = f(\vec y) - f(\vec y-\vec e_i)\;.
\end{equation} The operator $j_{i,\vec x}$ is called the current density operator, and  it can be explicitly determined, as follows. We compute: 
\begin{equation}\label{eq:densityevo}
\begin{split}
    \partial_t  \tau_t ( n_{\vec x}) &=i \tau_{t}\big(\big[\mathcal{H},  n_{\vec x}\big]\big) \\
    &= \sum_{{\vec y\in \Lambda_{L}  }} \tau_{t}(j_{\vec y, \vec x})\;,
    \end{split}
\end{equation}
where the bond current $j_{\vec y, \vec x}$ is:
\begin{equation}
j_{\vec y, \vec x} = i\sum_{\sigma,\zeta} \big( H_{\zeta,\sigma}(\vec y,\vec x)  a^*_{\vec y ,\zeta} a_{\vec x ,\sigma}-H_{\sigma,\zeta}(\vec x,\vec y) a^*_{\vec x,\sigma} a_{\vec y ,\zeta}\big)\;.
\end{equation}
Observe that the bond current is antisymmetric, $j_{\vec y, \vec x} = - j_{\vec x, \vec y}$. In terms of these operators, we can express the current density $j_{i,\vec x}$ as, for $j\neq i$:
\begin{equation}\label{eq:currdef}  
j_{i,\vec x}=  j_{\vec x,\vec x+\vec e_i}+\frac{1}{2}\big( j_{\vec x,\vec x+\vec e_i+\vec e_j}+  j_{\vec x,\vec x+\vec e_i-\vec e_j}+  j_{\vec x+\vec e_j,\vec x+\vec e_i}+  j_{\vec x-\vec e_j,\vec x+\vec e_i}\big)\;.\end{equation}
It is convenient to collect current and densities into a single $3$-current; to this end, we define $j_{0,\vec x}:=  n_{\vec x}$, and we shall collect density and currents into a single vector with components $j_{\nu,\vec x}$, $\nu = 0,1,2$. Also, we shall set:
\begin{equation}\label{eq:1dcurrent}
  j_{\nu,x_1} := \sum_{x_2 = 1}^{L_{2}}  j_{\nu,\vec x}\;.
\end{equation}
In the following, we shall consider smeared version of the current densities, against suitable test functions localized in proximity of the $x_{2} = 0$ boundary of the cylinder. 

Let $\ell > 1$, let $\phi_{\infty}(\vec x)$ be a smooth, compactly supported function on $\mathbb{R}^{2}$, and let $\phi_{\infty,\ell}(\vec x) = \phi_{\infty}(x_{1}, x_{2}/\ell)$. We denote by $\phi_{\ell}(\theta x_{1}, x_{2})$ the $x_{1}$-periodization of $\phi_{\infty,\ell}(\theta x_{1}, x_{2})$, as in (\ref{eq:periodization}), and we set $\phi_{\theta,\ell}(\vec x) = \theta \phi_{\ell}(\theta x_{1}, x_{2})$. We define:
\begin{equation}\label{eq:smearj}
j_{\nu}(\phi_{\theta,\ell}) := \theta \sum_{\vec x\in \Lambda_{L}}  \phi_{\ell}(\theta x_{1}, x_{2})\,j_{\nu,\vec x}\;.
\end{equation}
Thus, the smearing against the test function introduces an averaging of the current density, and a localization in proximity of the lower edge. Observe that the normalization of the test function implies that the operator $j_{\nu}(\phi_{\theta,\ell})$ is bounded uniformly in $\theta$ (but not in $\ell$). The parameter $\ell$ will be chosen independently of $\theta$ and $\eta$, and it defines the width of the strip where we will measure the edge current. We will be interested in the order of limits: $L_{i} \to \infty$; then $\beta \to \infty$; then $\eta,\theta \to 0$ (different relative orders of $\eta\to 0$ and $\theta\to 0$ will give different results); and finally $\ell \to \infty$.

\begin{figure}
    \centering
    \includegraphics[scale=0.5]{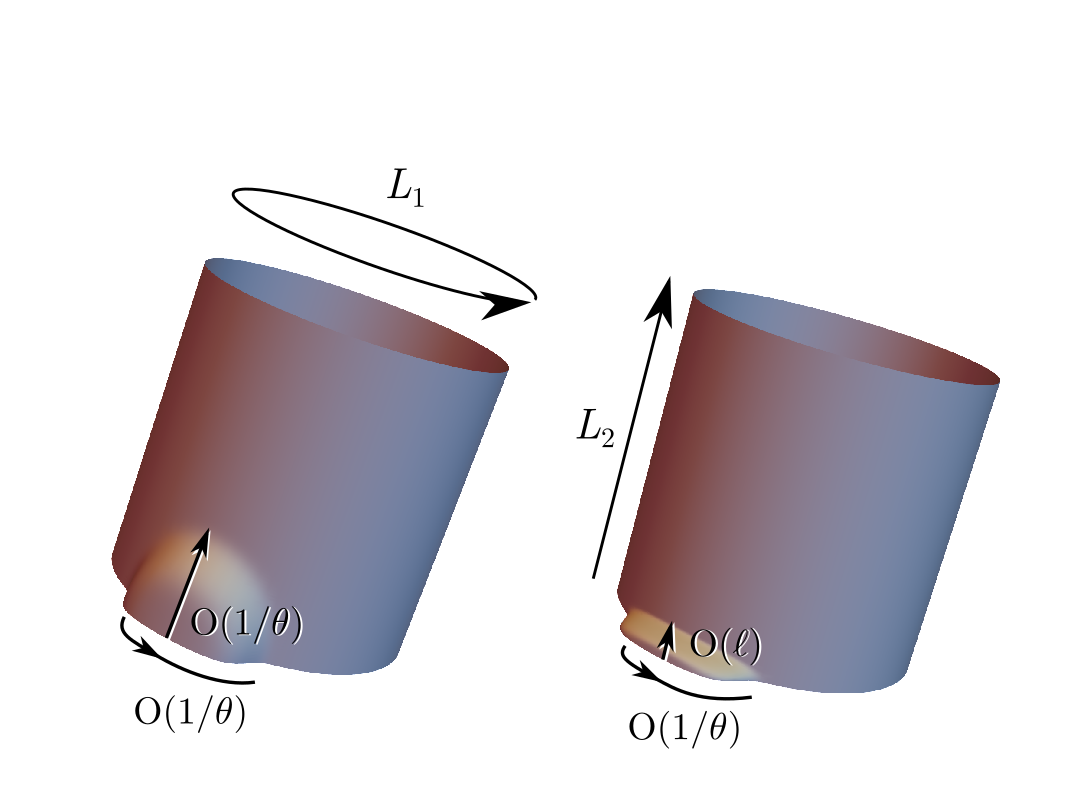}
    \caption{On the left  is depicted $\mu(\theta \vec x)$, on the right $\phi_{\ell}(\theta x_{1},x_{2})$.}
    \label{fig:testfunctions}
\end{figure}

We shall be interested in the variation of the average of $j_{\nu}(\phi_{\theta,\ell})$ over the solution of the Schr\"odinger equation (\ref{eq:dyn}), at first order in the time-dependent perturbation:
\begin{equation}\label{eq:linresp}
\Tr \big(j_{\nu}(\phi_{\theta,\ell}) \rho(0) \big)- \Tr \big(j_{\nu}(\phi_{\theta,\ell}) \rho_{\beta, \mu, L} \big)= \chi^{\beta,L}_{\nu}(\eta, \theta) + \text{h.o.t.}\;,
\end{equation}
where $\chi^{\beta, L}_{\nu}(\eta, \theta)$ is given by Kubo formula:
\begin{equation}\label{eq:kubo}
\chi^{\beta,L}_{\nu}(\eta, \theta) = i\int_{-\infty}^{0}  d t\, e^{\eta t} \Tr \big(\big[ j_{\nu}(\phi_{\theta,\ell}), \tau_{t}( \mc{P} )  \big] \rho_{\beta, \mu, L}\big)\;.
\end{equation}
In this work, we will not be concerned with proving the validity of linear response (\ref{eq:linresp}) in the zero temperature and infinite volume limit. Instead, we will start from Kubo formula (\ref{eq:kubo}), and study rigorously this transport coefficient. We shall be interested in computing the zero temperature and infinite volume limit of this quantity, 
\begin{equation}
\chi_{\nu}(\eta, \theta) = \lim_{\beta\to \infty} \lim_{L_{1}, L_{2} \to \infty} \chi^{\beta,L}_{\nu}(\eta, \theta)\;.
\end{equation}
Observe that already proving that the time-integral in (\ref{eq:kubo}) is finite uniformly in $\eta$ is non-trivial, in view of the lack of spectral gap at the Fermi level. Furthermore, the slow decay of correlations, also implied by the lack of spectral gap, gives rise to a non-trivial infrared problem in the zero temperature and infinite volume limit. Our main result will not only prove that these limits exist and are finite, but will also provide an explicit expression for $\chi_{\nu}(\eta, \theta)$, for small quasi-periodic perturbation. Furthermore, our main result will allow to characterize all equilibrium correlation functions of the disordered system. More precisely, we will be interested in the time-ordered, Euclidean correlation functions, defined as follows. Let us denote by $\gamma_{t}(\cdot)$ the imaginary-time evolution:
\begin{equation}
\gamma_{t}(\mathcal{O}) := e^{t (\mathcal{H} - \mu \mathcal{N})} \mathcal{O} \,e^{-t (\mathcal{H} - \mu \mathcal{N})}\;.
\end{equation}
The fermionic time-ordering $\timord$ acts on fermionic monomials, as follows. Let $t_{i}\in [0,\beta)$, for $i = 1,\ldots, n$. Then:
\begin{equation}\label{eq:T}
\begin{split}
&\timord \big(\gamma_{t_{1}}(a^{\sharp_{1}}_{x_{1},\sigma_{1}}) \cdots \gamma_{t_{n}}(a^{\sharp_{n}}_{x_{n},\sigma_{n}}) \big)\\
&= \text{sgn}(\pi) \gamma_{t_{\pi(1)}}(a^{\sharp_{\pi(1)}}_{x_{\pi(1)},\sigma_{\pi(1)}}) \cdots \gamma_{t_{\pi(n)}}(a^{\sharp_{\pi(n)}}_{x_{\pi(n)},\sigma_{\pi(n)}})\;,\quad \text{for $t_{i} \neq t_{j}$, if $i\neq j$,}
\end{split}
\end{equation}
where $a^{\sharp}$ can be either $a$ or $a^{*}$, and the permutation $\pi$ is such that $t_{\pi(1)} > t_{\pi(2)} > \cdots > t_{\pi(n)}$. The definition (\ref{eq:T}) is extended to coinciding times by normal ordering, that is by placing creation operators to the left of annihilation operators. Finally, the definition (\ref{eq:T}) is extended by linearity to the whole (finite dimensional) algebra of operators on $\mathcal{F}_{L}$. We shall be interested in the expectation value of time-ordered monomials. The simplest example is the two-point correlation function. Let us introduce the notations:
\begin{equation}
{\bm x} := (x_{0}, x_{1})\;,\qquad \vec {\bm x} := (x_{0}, \vec {x})\;.
\end{equation}
Let $0\le x_{0}, y_{0} < \beta$. We define the Euclidean two-point correlation function as:
\begin{equation}\label{eq:2pt}
S^{\beta, L}_{2; \sigma, \zeta}(\vec {\bm x}, \vec {\bm y}) := \langle \timord \gamma_{x_{0}}(a_{\vec x, \sigma}) \gamma_{y_{0}}(a^{*}_{\vec y,\zeta} )  \rangle_{\beta, \mu, L}\;.
\end{equation}
The two-point function is then extended to all values of $x_{0}, y_{0}$ in $\mathbb{R}$, by antiperiodicity:
\begin{equation}
S^{\beta, L}_{2; \sigma, \zeta}\big((x_{0} + n\beta, \vec x), (y_{0} + m\beta, \vec y)\big) := (-1)^{n+m} \langle \timord \gamma_{x_{0}}(a_{\vec x, \sigma}) \gamma_{y_{0}}(a^{*}_{\vec y,\zeta} )  \rangle_{\beta, \mu, L}\;.
\end{equation}
Being the Hamiltonian (\ref{eq:Ham}) quadratic in the fermionic operators, all Euclidean correlation functions can be computed starting from the two-point function (\ref{eq:2pt}), via the fermionic Wick rule. Furthermore, it turns out that also the response functions $\chi^{\beta, L}_{\mu}(\eta, \theta)$ can be expressed in terms of $S^{\beta, L}_{2}$, after Wick rotation, to be discussed later on.

For later use, we shall introduce the following distance, needed to quantify the decay properties of the two-point function, in a way compatible with its time anti-periodicity and space periodicity: 
\begin{equation}
\| \vec {\bm x} - \vec {\bm y} \|_{\beta, L}:= \min_{n\in \mathbb{Z}} | x_{0} - y_{0} + n\beta | + \| \vec x - \vec y \|_{L}\;,
\end{equation}
with $\| \vec x - \vec y \|_{L}$ as in (\ref{eq:distL}). We shall also set:
\begin{equation}
\| {\bm x} - {\bm y} \|_{\beta, L}:= \min_{n\in \mathbb{Z}} | x_{0} - y_{0} + n\beta | + \min_{m\in \mathbb{Z}} | x_{1} - y_{1} + mL|\;.
\end{equation}
Finally, we will use the short-hand notation:
\begin{equation}
\| \vec {\bm x} - \vec {\bm y} \| \equiv \| \vec {\bm x} - \vec {\bm y} \|_{\beta, L}\;,\qquad \| {\bm x} - {\bm y} \|\equiv \| {\bm x} - {\bm y} \|_{\beta, L}\;.
\end{equation}

\section{Main result}\label{sec:mainres}
We start by stating our result about the two-point function of the model, at weak disorder. To this end, let:
\begin{equation}
\msc D_{L,\beta} := \Big\{ {\bm k} = (k_{0}, k_{1})\, \Big|\, k_{0} \in \frac{2\pi}{\beta} \Big(\mathbb{N} + \frac{1}{2}\Big),\; k_{1} \in \msc D_{L} \Big\}\;.
\end{equation}
The numbers $\{k_{0}\}$ are the fermionic Matsubara frequencies, while $\{k_{1}\}$ are the independent momenta compatible with the periodic boundary conditions in the $x_{1}$ direction, on a finite lattice with $L_{1}$ sites. 

Let us denote by $\chi(\cdot)$ a smooth cutoff function: $\chi(x) \in C^{\infty}(\mathbb{R})$ such that $\chi(x) = \chi(-x)$ and, for $\delta>0$ and $\gamma>1$ to be chosen later:
\begin{equation}\label{eq:chidef}
\chi(x) = \begin{cases} 1, & |x|<\frac{\delta}{\upgamma} \\0, & |x|>\delta \end{cases}\;.
\end{equation}
The parameter $\delta$ is chosen as in Assumption \ref{ass:H3}. Also, let us define:
\begin{equation}\label{eq:x2omega}
|x_{2}|_{+} := x_{2}\;,\qquad |x_{2}|_{-} := L_{2} - x_{2}\;.
\end{equation}
We are now ready to state our first result.
\begin{theorem}[Two-point function]\label{thm:2pt} For any $\kappa>0$ there exists $\lambda_{0}, \beta_0 >0$ such that for $|\lambda| < \lambda_{0}$ the following holds. For $\beta_0\le \beta \le \kappa L_{1}$, $\beta \le \kappa L_{2}$ the two-point function can be written as:
\begin{equation}\label{eq:2ptmain}
S^{\beta,L}_{2; \sigma, \zeta}(\vec {\bm x}; \vec {\bm y}) = \sum_{\omega = \pm} e^{i k_{F,L}^{\omega}(\lambda)(x_{1} - y_{1})} Z_{\omega,\sigma}(\vec x) \check{g}_{\omega;\mr{s}}({\bm x} - {\bm y}) \overline{Z_{\omega,\zeta}(\vec y)} + R_{\sigma,\zeta}(\vec {\bm x}; \vec {\bm y})\;,
\end{equation}
where:
\begin{equation}\label{eq:2ptmain2}
\check{g}_{\omega;\mr{s}}({\bm x} - {\bm y}) = \frac{1}{\beta L_{1}} \sum_{{\bm q} \in \msc D_{L,\beta}} e^{i {\bm q}\cdot ( {\bm x} -  {\bm y})} \frac{\chi({\bm q})}{i v_{0,\omega}(\lambda) q_{0} + v_{1,\omega}(\lambda) q_{1}} 
\end{equation}
with $k^{\omega}_{F,L}(\lambda)$ is an approximant in $\msc D_{L}$ of a smooth function $k^{\omega}_{F}(\lambda)$ of $\lambda$ such that $k^{\omega}_{F}(0) = k_{F}^{\omega}$. The parameters $v_{0}(\lambda), v_{1}(\lambda)$ are smooth in $\lambda$ and satisfy:
\begin{equation}
v_{0,\omega}(0) = 1\;,\qquad v_{1,\omega}(0) = v_{\omega}\;;
\end{equation}
the functions $Z_{\omega,\sigma}(\vec x)$ are smooth in $\lambda$ and have the form:
\begin{equation}\label{eq:Zosc}
Z_{\omega,\sigma}(\vec x) = \sum_{n \in \mathbb{N}} Z_{n,\omega,\sigma}(x_{2}) e^{-i n \alpha x_{1}}\;.
\end{equation}
The Fourier coefficients satisfy:
\begin{equation}\label{eq:bdZm}
|Z_{n,\omega,\sigma}(x_{2})| \le C |\lambda|^{1-\delta_{n,0}} e^{-c |n|} e^{-c |x_{2}|_{\omega}}\;,\qquad Z_{0,\omega,\sigma}(\vec x)\Big|_{\lambda = 0} = \xi^{\omega}_{\sigma}(k^{+}_{F}; x_{2})\;.
\end{equation}
Also, the velocities $v_{\omega,0}(\lambda)$, $v_{\omega,1}(\lambda)$ satisfy the identities:
\begin{equation}\label{eq:relzeta}
\begin{split}
v_{0,\omega}(\lambda) &= \sum_{n,\sigma,x_{2}} Z_{n,\omega,\sigma}(x_{2}) \overline{Z_{n,\omega,\sigma}(x_{2})} \\
v_{1,\omega}(\lambda) &= \sum_{\substack{n,\sigma,\zeta \\ x_{2}, y_{2}}} \overline{Z_{n,\omega,\sigma}(x_{2})} \Big( \partial_{k_{1}} \hat H_{\sigma,\zeta}\big(k_{F,L}^{\omega}(\lambda) - n \alpha\;; x_{2}, y_{2}\big) \Big) {Z_{n,\omega,\zeta}(y_{2})}\;.
\end{split}
\end{equation}
Finally, the error term in (\ref{eq:2ptmain}) satisfies:
\begin{equation}\label{eq:Rest2pt}
\begin{split}
R_{\sigma, \zeta}(\vec {\bm x}; \vec {\bm y}) &= \sum_{n\in \mathbb{Z}} R_{n;\sigma, \zeta}({\bm x} - {\bm y}; x_{2}, y_{2}) e^{-in \alpha y_{1}} \\
| R_{n;\sigma, \zeta}({\bm x} - {\bm y}; x_{2}, y_{2}) | &\le C |\lambda|^{\delta_{n\neq 0}}e^{-c|n|} \frac{e^{-c| x_{2} -y_{2} |}}{1 + \|  {\bm x} -  {\bm y} \|^{1 + \xi}}\quad \text{for some $\xi>0$.}
\end{split}
\end{equation}
\end{theorem}
\begin{remark}
\begin{itemize}
\item[(i)] By Wick's rule, Theorem \ref{thm:2pt} allows to fully characterize the large scale behavior of all the edge correlation functions of the system.
\item[(ii)] All $\lambda$-dependent quantities in the above theorem are still (weakly) dependent on $\beta, L_{1}, L_{2}$. The method of the proof also allows to show that they converge to a limit as $\beta , L_{1}, L_{2} \to \infty$, where the $L_{1}\to \infty$ limit is always understood as taken over sequences specified in Remark \ref{rem:lim} item (ii).
\end{itemize}
\end{remark}
Our next result allows to compute the zero temperature and infinite volume edge transport coefficients, defined as:
\begin{equation}\label{eq:Gdefmain}
G_{0} := \lim_{\theta \to 0} \lim_{\eta \to 0^{+}} \frac{\chi_{0}(\eta, \theta)}{\langle  \mu,\phi \rangle_{\text{edge}}}\;,\qquad G_{1} := \lim_{\theta \to 0} \lim_{\eta \to 0^{+}} \frac{\chi_{1}(\eta,\theta)}{\langle \mu,\phi \rangle_{\text{edge}}}\;,
\end{equation}
where:
\begin{equation}\label{eq:edgescalar}
\langle\mu,\phi\rangle_{\text{edge}} := \int_{\mathbb{R}} dx\, \overline{\mu_{\infty}(x,0)} \phi_{\infty}(x,0)\;.
\end{equation}
The quantity $G_{0}$ is the edge susceptibility, while the quantity $G_{1}$ is the edge conductance.  The normalization in (\ref{eq:Gdefmain}) is natural, in view of the arbitrariness of the shape of the perturbation and of the test function in the observable.
\begin{theorem}[Transport coefficients]\label{thm:resp} Under the same assumptions of Theorem \ref{thm:2pt} the following is true. Let $v^{\infty}_{i,+}(\lambda) = \lim_{\beta\to \infty} \lim_{L_{1}, L_{2}\to \infty} v_{i,+}(\lambda)$ with $v_{i,+}(\lambda)$ as in Theorem \ref{thm:2pt}. Let $\mf{v}(\lambda) := v^{\infty}_{1,+}(\lambda) / v^{\infty}_{0,+}(\lambda)$. Then, the edge transport coefficients are given by: 
\begin{equation}\label{eq:respf}
G_{0} = \frac{1}{2\pi |\mf{v}(\lambda)|}\;,\qquad G_{1} = \frac{\text{sgn}(\mf{v}(\lambda))}{2\pi}\;.
\end{equation}
\end{theorem}
\begin{remark}
\begin{itemize}
\item[(i)] Theorem \ref{thm:resp} proves, in particular, the quantization of the edge conductance for the model in presence of weak disorder, starting from Kubo formula. The integer is the same as for the bulk Hall conductivity, in agreement with the bulk-edge duality.
\item[(ii)] The order of limits in the definitions (\ref{eq:Gdefmain}) is important. As the proof will show, taking the limits in the reverse order would give a vanishing result.\end{itemize}
\end{remark}
\subsection{Sketch of the proof}

Here we shall give an outline of the proof of Theorem \ref{thm:resp}. Theorem \ref{thm:2pt} is based on an extension of the multiscale analysis for quasi-periodic systems, see {\it e.g.} \cite{BGM, M0, M1, M2, GM}, here applied to edge modes of topological insulators for the first time. With respect to these works, our main new contribution is the determination of the asymptotic behavior of the two-point function, which allows to resolve its oscillatory behavior (\ref{eq:2ptmain}). Since the model is quasi-free, this allows to determine the oscillatory behavior of all correlation functions, via Wick's rule.

Theorem \ref{thm:resp} is based on the combination of several ingredients. The starting point is the rewriting of the response function (\ref{eq:kubo}) in imaginary times (Wick's rotation):
\begin{equation}\label{eq:sketch1}
\chi^{\beta,L}_{\nu}(\eta, \theta) = \theta\sum_{\vec x,\vec y } \mu(\theta\vec  x)  \phi_{\ell}(\theta y_{1}, y_{2}) \int_{-\beta/2}^{\beta/2}  d s\, e^{-i\eta_{\beta} s} \langle \timord \gamma_{s}(n_{\vec x})\;; j_{\nu,\vec y} \rangle_{\beta,\mu,L} + O(1 / (\beta \eta^{3}))\;.
\end{equation}
This way of rewriting response functions has been used in several cases, see {\it e.g.} \cite{GMPhall, GMPweyl, AMP, MP, greenlamp}. The advantage of the rewriting (\ref{eq:sketch1}) is that we got rid of real-time correlation functions, for which no effective decay estimates can be proved, in favor of imaginary-time correlations, which can be studied via multiscale analysis. In the case at hand, being the model quasi-free, the right-hand side can actually be expressed in terms of the Euclidean two-point function (\ref{eq:2pt}). Still, the presence of disorder makes the right-hand side very hard to compute directly.

The main idea, which is an extension to quasi-periodic systems of the method used in \cite{AMP, MP}, is to isolate the ``scaling limit'' contribution in the integrand in (\ref{eq:sketch1}):
\begin{equation}\label{eq:sketch1b}
\chi^{\beta,L}_{\nu}(\eta, \theta) = \chi^{\text{sing}}_{\nu}(\eta, \theta) + \chi^{\text{reg}}_{\nu}(\eta, \theta)\;,
\end{equation}
where: $ \chi^{\text{sing}}_{\nu}$, the singular part, is obtained by replacing the two-point function by its leading oscillatory contribution in (\ref{eq:2ptmain}); while the remainder $\chi^{\text{reg}}_{\nu}(\eta, \theta)$, the regular part, is expressed in terms of quantities that have improved decay estimates in configuration space. Concerning the regular part, to avoid the impossible task of directly computing it we use lattice conservation laws to express it in terms of the singular part. The key information is the lattice continuity equation (\ref{eq:cons}), which implies {\it Ward identities} for the Euclidean correlation functions. As discussed in Section \ref{sec:cucuWI}, {\it if} the regular part $\chi^{\text{reg}}_{\nu}(\eta, \theta)$ satisfies the estimate, for some $\alpha > 0$,
\begin{equation}\label{eq:sketch2}
\big|\chi^{\text{reg}}_{\nu}(\eta,\theta) - \chi^{\text{reg}}_{\nu}(\eta,\theta')\big|\le C_{\ell} (|\theta|^{\alpha} + |\theta'|^{\alpha})\;,
\end{equation}
then the lattice continuity equation (\ref{eq:cons}) implies:
\begin{equation}\label{eq:sketch2b}
\chi^{\text{reg}}_{\nu}(\eta, \theta) = - \chi^{\text{sing}}_{\nu}(\eta, \eta^{2}) + O_{\ell}(|\theta|^{\alpha} + |\eta|^{\alpha})\;.
\end{equation}
This is particularly useful, provided one has a more explicit control on $\chi^{\text{sing}}_{\nu}$. Before discussing this, let us comment on (\ref{eq:sketch2}). In translation-invariant cases, this property is a straightforward consequence of the improved decay estimates of the argument of the integral in the definition of $\chi^{\text{reg}}_{\nu}(\eta, \theta)$: it is a consequence of the fact that the (two-dimensional) Fourier transform of functions with decay $\|{\bm x}\|^{-2 - \alpha}$ have Holder regularity with exponent $\alpha'<\alpha$. In our case, we cannot easily rely on Fourier analysis, due to the lack of translation invariance. Instead, we use that $K_{0\nu}^{\text{reg}}(\eta, \theta)$ has the form:
\begin{equation}\label{eq:chireg}
\chi^{\text{reg}}_{\nu}(\eta, \theta) = \theta \sum_{{\vec x},{\vec y}} \mu(\theta\vec  x)\phi_{\ell}(\theta y_{1}, y_{2})
\int_{-\frac{\beta}{2}}^{\frac{\beta}{2}}  d s\,e^{-i\eta s} F((s,\vec x); (0, \vec y))
\end{equation}
where the function $F(\cdot)$ satisfies:
\begin{equation}\label{eq:Fm}
\begin{split}
F(\vec {\bm x}; \vec {\bm y}) &= \sum_{m \in \mathbb{Z}} F_{m}({\bm x} - {\bm y}; x_{2}, y_{2}) e^{-i m\alpha y_{1}} \\
\big|F_{m}({\bm x} - {\bm y}; x_{2}, y_{2})\big| &\leq \frac{C e^{-c|m|} e^{-c|x_{2} - y_{2}|}}{1 + \| {\bm x} - {\bm y} \|^{2+\xi}}\qquad \text{with $\xi>0$.}
\end{split}
\end{equation}
To prove Eq. (\ref{eq:sketch2}), we will use that the $m=0$ contribution to (\ref{eq:chireg}) coming from (\ref{eq:Fm}) admits a limit as $\theta \to 0$, which cancels in the left-hand side of (\ref{eq:sketch2}). Then, we show that, for $\theta$ small, all the $m\neq0$ contributions to (\ref{eq:chireg}) coming from (\ref{eq:Fm}) are subleading in $\theta$ uniformly in $\eta$. To prove this, we crucially rely on the Diophantine condition (\ref{eq:diophantine}), and on the regularity properties of the test function. See Lemma \ref{lem:contK} for the precise statement, and Appendix \ref{app:contK} for the proof.

Next, we turn to the analysis of the singular contribution (\ref{eq:sketch2}). For simplicity, we will outline the computation of the susceptibility, $\nu = 0$; the analysis for the conductance, $\nu = 1$, proceeds along the same lines.  Let us summarize the analysis of Section \ref{sec:susce}. We use that, thanks to the more explicit form of the leading contribution to the two-point function (\ref{eq:2ptmain}), the response function can be rewritten as, in Fourier space:
\begin{equation}\label{eq:sketch4}
\begin{split}
\chi^{\text{sing}}_{0}(\eta, \theta) &= - \frac{\theta}{L_{1}} \sum_{p} \sum_{n,m}  \sum_{x_{2},y_{2}} (Z_{+}\star\overline{Z_{+}})(m;y_{2}) (Z_{+}\star\overline{Z_{+}})(n;x_{2})\\
&\qquad \cdot \hat \mu_{\theta}(-p-n\alpha,x_2) \hat \phi_{\theta,\ell}( p-m\alpha,y_{2})   \mc{B}_{+}^{\infty}(\eta, p) + \mr{o}(1)\;,
\end{split}
\end{equation}
where: here and below, the $\mr{o}(1)$ error terms vanish in the order of limits: first $L_{1,2}\to \infty$ (irrespective of the order); then $\beta \to \infty$; then $\eta,\theta \to 0$ (irrespective of the order); then $\ell \to \infty$; the function $ \mc{B}_{+}^{\infty}(\eta, p)$ is the so-called anomalous bubble diagram, computed for a free relativistic $1+1$ dimensional model with propagator (\ref{eq:2ptmain2}) with $\omega = +$; its value is given by:
\begin{equation}\label{eq:sketch3}
\mc{B}_{+}^{\infty}(\eta, p) =  \frac{1}{4\pi \big|v_{1,+}(\lambda)  v_{0,+}(\lambda) \big|}\frac{iv_{0,+}(\lambda)\eta-v_{1,+}(\lambda) p}{i v_{0,+}(\lambda)\eta +v_{1,+}(\lambda) p} + r_{+}(\eta, p)\;,
\end{equation}
where $v_{\mu,+}(\lambda)$ are as in (\ref{eq:2ptmain2}) and $r_{+}(\eta, p)$ vanishes continuously as $(\eta, p) \to (0,0)$; and
\begin{equation}
(Z_{+}\star\overline{Z_{+}})(m;y_{2}) = \sum_{\sigma} \sum_{m\in \mathbb Z}Z_{m,+,\sigma}(x_2) \overline{Z_{m-n,+,\sigma}(x_2)}\;,
\end{equation}
with $Z_{m,+,\sigma}(x_2)$ the amplitudes of the oscillatory functions appearing in the leading term (\ref{eq:2ptmain}), see also (\ref{eq:Zosc}). Our goal at this point is to use (\ref{eq:sketch1b}), (\ref{eq:sketch2b}), (\ref{eq:sketch4}) to show that, approximating the sum by an integral and performing a change of variable in $p$:
\begin{equation}\label{eq:sketch5}
\begin{split}
\chi^{\beta,L}_{0}(\eta, \theta) &= - \int \frac{dp}{(2\pi)} \sum_{x_{2},y_{2}} (Z_{+}\star\overline{Z_{+}})(0;y_{2}) (Z_{+}\star\overline{Z_{+}})(0;x_{2})\\
&\qquad \cdot \hat \mu_{\infty}(-p,x_2) \hat \phi_{\infty,\ell}( p,y_{2}) \Big( \mc{B}_{+}^{\infty}(\eta, \theta p) - \mc{B}_{+}^{\infty}(\eta, \eta^{2} p)\Big) + \mr{o}(1)\;.
\end{split}
\end{equation}
The advantage of this rewriting is that is resembles the expression of the analogous response function computed in a translation-invariant setting, {\it but} with renormalized parameters, due to homogeneization effects on a macroscopic scale. For instance, the test functions are ``dressed'' by $(Z_{+}\star\overline{Z_{+}})(0;x_{2})$, $(Z_{+}\star\overline{Z_{+}})(0;y_{2})$, which are disorder-dependent quantities. Let us postpone for a moment the analysis of the leading term in (\ref{eq:sketch5}), and let us briefly discuss why (\ref{eq:sketch5}) holds. 

The subtraction of $- \mc{B}_{+}^{\infty}(\eta, \eta^{2} p)$ takes into account the presence of the regular part, thanks to (\ref{eq:sketch2b}), up to subleading terms in $\eta$ and $\theta$. Then, we observe that the smoothness of the test functions in configuration space morally impose, in (\ref{eq:sketch4}), that $|p + n \alpha|_{\mathbb{T}} \lesssim \theta$ and that $| p - m \alpha |_{\mathbb{T}} \lesssim \theta$. If $n\neq m$, this implies that $| (n + m)\alpha |_{\mathbb{T}} \lesssim \theta$; and by the Diophantine condition (\ref{eq:diophantine}), since $m+n\neq 0$, this means that $|n + m| \gtrsim \theta^{-\frac{1}{\tau}}$. Thus, the corresponding contribution to the transport coefficient is overwhelmingly small in $\theta$, thanks to the exponential decay of the amplitudes $Z_{k,+,\sigma}(z_2)$ appearing in the expression for the two-point function and hence in (\ref{eq:sketch5}), recall (\ref{eq:Zosc}), (\ref{eq:bdZm}).

Consider now the contribution of the terms $n= -m \neq 0$ in (\ref{eq:sketch4}). After a change of variables, we are left with:
\begin{equation}
- \frac{\theta}{L_{1}} \sum_{p} \sum_{0\neq n}  \sum_{x_{2},y_{2}} (Z_{+}\star\overline{Z_{+}})(-n;y_{2}) (Z_{+}\star\overline{Z_{+}})(n;x_{2}) \hat \mu_{\theta}(-p,x_2) \hat \phi_{\theta,\ell}( p,y_{2})   \mc{B}_{+}^{\infty}(\eta, p - n\alpha)\;.
\end{equation}
Here we rely again on the Diophantine condition. The properties of the test functions imply that $|p|_{\mathbb{T}} \lesssim \theta$; we shall distinguish two regimes, $|n\alpha|_{\mathbb{T}} \geq \sqrt{\theta}$ and $|n\alpha|_{\mathbb{T}} < \sqrt{\theta}$. In the first regime, and for $p$ in the range of the test functions, the function $\mc{B}_{+}^{\infty}(\eta, p - n\alpha)$ is actually {\it continuous} at $(\eta, p) = (0,0)$. This ultimately implies that this regime contributes to (\ref{eq:sketch4}) with a term that satisfies the bound (\ref{eq:sketch2}). Thus, the corresponding contribution can be reabsorbed in a redefinition of $\chi^{\text{reg}}_{0}(\eta, \theta)$. Concerning the second regime, since $n\neq 0$ we can use again the Diophantine condition, to prove that the corresponding contribution is overwhelmingly small in $\theta$.

Therefore, we are left with (\ref{eq:sketch5}). Up to subleading terms,
\begin{equation}
\begin{split}
\chi^{\beta,L}_{0}(\eta, \theta) &= - \int \frac{d p}{2\pi} \sum_{x_{2},y_{2}} (Z_{+}\star\overline{Z_{+}})(0;y_{2}) (Z_{+}\star\overline{Z_{+}})(0;x_{2}) \\
&\qquad \cdot \hat \mu_{\infty}(- p, 0) \hat \phi_{\infty}( p,0) \big(  \mc{B}_{+}^{\infty}(\eta, \theta p) - \mc{B}_{+}^{\infty}(\eta, 0)\big)  + \mr{o}(1)\;. 
\end{split}
\end{equation}
Here we localized the test functions at $x_{2} = y_{2} = 0$. This is possible, up to small errors, due to the fact that the test functions depend on $\theta x_{2}$ and $y_{2} /\ell$, recall (\ref{eq:periodization}) and the discussion before (\ref{eq:smearj}), and thanks to the decay estimates (\ref{eq:bdZm}), which ultimately follow from the decay of the edge modes away from the $x_{2} = 0$ boundary on the scale of the lattice spacing. Calling:
\begin{equation}\label{eq:sketch5b}
\zeta_{0,+} := \sum_{y_{2} = 0}^{L_{2}} (Z_{+}\star\overline{Z_{+}})(0;y_{2})\;,
\end{equation}
we then have, from the explicit expression of the anomalous bubble diagram (\ref{eq:sketch3}):
\begin{equation}
\chi^{\beta,L}_{0}(\eta, \theta) = \frac{\zeta_{0,+}^{2}}{v_{0,+}^{2}} \int \frac{dp}{(2\pi)^2} \hat \mu_{\infty}(- p, 0) \hat \phi_{\infty}( p,0)\frac{\text{sgn}(v_{1,+})\,\theta p}{i \eta + \mf{v}(\lambda) \theta p} + \mr{o}(1)\;,
\end{equation}
with $\frak{v}(\lambda) = v_{1}(\lambda) / v_{0}(\lambda)$. Hence, for $\eta \ll \theta$:
\begin{equation}\label{eq:sketch5c}
\chi^{\beta,L}_{0}(\eta, \theta) = \frac{\zeta_{0,+}^{2}}{v_{0,+}^{2}} \frac{\langle \mu, \phi \rangle_{\text{edge}}}{2\pi |\frak{v}(\lambda)|}  + \mr{o}(1)\;,
\end{equation}
recall the notation (\ref{eq:edgescalar}) for the edge scalar product. This implies our main result for the susceptibility {\it if} we can prove that $\zeta_{0,+}^{2} / v_{0,+}^{2} = 1$. 

To prove this, we rely once more on the lattice continuity equation, and on its implication called the vertex Ward identity, discussed in Section \ref{sec:vertWI}. Let $j_{\mu, x_{1}} = \sum_{x_{2} = 1}^{L_{2}} j_{\mu, x}$. Let us introduce the momentum-space vertex function as:
\begin{equation}
\begin{split}
&\hat S_{3; \mu,\sigma,\zeta}(\bm{p}, \bm{k}, \bm{h}; x_{2}, y_{2}) \\
&:= \frac{1}{\beta L_{1}} \int_{[0,\beta]^{3}}\mr  d w_0  d x_0  d y_0 \sum_{w_{1}, x_{1}, y_{1}} e^{-i\bm{p} \cdot \bm{w} - i\bm{k} \cdot \bm{x} + i\bm{h} \cdot \bm{y} } \big\langle \timord\, \gamma_{w_0}(j_{\mu; w_{1}})\;; \gamma_{x_{0}}(a_{{\vec x},\sigma}) \;; \gamma_{y_{0}}(a^{*}_{{\vec y},\zeta}) \big\rangle\;,
\end{split}
\end{equation}
and the momentum-space two-point function as:
\begin{equation}
\begin{split}
&\hat S_{2;\sigma,\zeta}(\bm{k}, \bm{h}; x_{2}, y_{2}) \\
&\quad := \frac{1}{\beta L_{1}} \int_{[0,\beta]^{2}}\mr  dx_{0}\mr  dy_{0} \sum_{x_{1}, y_{1}} e^{-i\bm{k} \cdot \bm{x} + i\bm{h} \cdot \bm{y} }\langle \timord \gamma_{x_{0}}(a_{{\vec x},\sigma}) \gamma_{y_{0}}(a^{*}_{{\vec y},\zeta}) \big\rangle\;.
\end{split}
\end{equation}
Notice that the momenta do not have to satisfy the canonical momentum conservation, due to lack of translation invariance.  The momentum-space vertex Ward identity reads, see Section \ref{sec:vertWI}, choosing ${\bm h} = {\bm k} + {\bm p}$:
\begin{equation}\label{eq:sketch6}
\begin{split}
&-p_{0} \hat S_{3; 0,\sigma,\zeta}(\bm{p}, \bm{k}, {\bm k} + {\bm p}; x_{2}, y_{2}) + (1 - e^{-ip_{1}}) \hat S_{3; 1,\sigma,\zeta}(\bm{p}, \bm{k}, {\bm k} + {\bm p}; x_{2}, y_{2}) \\
&\qquad =i \hat S_{2;\sigma,\zeta}(\bm{k}, \bm{k}; x_{2}, y_{2})- i \hat S_{2;\sigma,\zeta}(\bm{k}+  \bm{p}, {\bm k} + {\bm p}; x_{2}, y_{2})\;.
\end{split}
\end{equation}
This identity can be used to prove non-trivial relations for the renormalized parameters. In fact, the left-hand side can be computed via Wick's rule, in terms of the two-point function; thus, both left-hand side and right-hand side of the identity are determined by the two-point function. The goal is to show that, for suitable choices of the external momenta, both sides of the identity are dominated by a ``translation-invariant'' contribution, plus subleading terms. The translation-invariant contribution is renormalized, by homogeneization effects on the macroscopic scale; the analysis crucially relies on the Diophantine property, and it is morally similar to the one outlined before, to prove that the leading contribution to the susceptibility is (\ref{eq:sketch5}). Ultimately, in Section \ref{sec:vertexcons} we prove that the identity (\ref{eq:sketch6}) reduces to, for a suitable range of external momenta, setting ${\bm q} = {\bm k} - {\bm k}_{F}^{+}$ and ${\bm k}_F^{+} = (0, k_{F}^{+})$:
\begin{equation}\label{eq:sketch7}
(-p_{0} \zeta_{0,+} + ip_{1}\zeta_{1,+}) g_{+;\mr{s}}({\bm q} + {\bm p}) g_{+;\mr{s}}({\bm q}) = i(g_{+;\mr{s}}({\bm q}) - g_{+;\mr{s}}({\bm q} + {\bm p})) + \mathcal{E}({\bm q}, {\bm p})
\end{equation}
where: $g_{+;\mr{s}}({\bm q})$ is given by the momentum-space relativistic propagator in (\ref{eq:2ptmain2}); the error term satisfies
\begin{equation}
\mathcal{E}({\bm q}, {\bm p}) g_{+;\mr{s}}({\bm q} + {\bm p})^{-1} g_{+;\mr{s}}({\bm q})^{-1} = \mr{o}(\bm p)\;;
\end{equation}
the parameter $\zeta_{0,+}$ is given by (\ref{eq:sketch5b}) while the parameters $\zeta_{1,+}$ is
\begin{equation}
\zeta_{1,+} = \sum_{\substack{n,\sigma,\zeta \\ x_{2}, y_{2}}} \overline{Z_{n,\omega,\sigma}(x_{2})} \Big( \partial_{k_{1}} \hat H_{\sigma,\zeta}\big(k_{F}^{\omega}(\lambda) - n \alpha\;; x_{2}, y_{2}\big) \Big) {Z_{n,\omega,\zeta}(y_{2})}\;.
\end{equation}
The range of external momenta for which we can establish (\ref{eq:sketch7}) is such that $\|{\bm q}\|$ and $\|{\bm p}\|$ are sufficiently small, and such that $|q_{0}| \geq |q_{1}|^{M}$ for some $M>1$; this condition is  not required in the application of this strategy to translation-invariant problems, and it is ultimately needed to avoid new resonances in the analysis of the momentum-space correlation functions. Finally, the key observation is that:
\begin{equation}
g_{+;\mr{s}}({\bm q}) - g_{+;\mr{s}}({\bm q} + {\bm p}) = (i v_{0,+} p_{0} + v_{1,+} p_{1}) g_{+;\mr{s}}({\bm q}) g_{+;\mr{s}}({\bm q} + {\bm p})\;;
\end{equation}
plugging this identity in (\ref{eq:sketch7}), we obtain
\begin{equation}
\zeta_{0,+} = v_{0,+}\;,\qquad \zeta_{1,+} = v_{1,+}\;.
\end{equation}
The first identity allows to prove the desired cancellation in (\ref{eq:sketch5c}), while the second is used to compute the edge conductance. This concludes the sketch of the proof of Theorem \ref{thm:resp}.

\section{Renormalization group analysis}\label{sec:RG}
In this section we will introduce the renormalization group analysis that allows to prove Theorems \ref{thm:2pt}, \ref{thm:resp}. We will start from the analysis of the partition function, which will allow us to introduce the notion of effective potential, and its RG flow. Then, we will discuss how to adapt the construction to the case of correlation functions, which will be computed introducing an external field in the partition function, and taking derivatives.

\subsection{Integration of the massive modes}\label{sec:massive}
A convenient starting point for the renormalization group analysis is the representation of the partition function of the disordered system in terms of a Gaussian Grassmann integral. We have, denoting by $\ms{Z}^{0}$ the partition function of the non-disordered system and omitting all the other labels:
\begin{equation}\label{eq:ZZ}
\ms{Z} /\ms{Z}^{0}=\lim_{N\to +\infty} \int \mathbb{P}_{N}(d\psi) \,\exp\big({\mc V_{N}( \psi)}\big)\;,
\end{equation}
where:
\begin{itemize}
\item[(i)] $\mc V_{N}$ is a polynomial in the Grassmann variables $\psi^{\pm}_{{\bm k}, x_{2}, \sigma}$,
\begin{equation}
    \mc{V}_{N}( \psi):= - \frac{\lambda}{\beta L_{1}}\sum_{n}\, \sum_{\bm{k}, x_2, \sigma}\hat\varphi_{n}(x_2) \psi_{\bm{k}+n{\bm \alpha},x_2,\sigma}^+ \psi_{\bm{k},x_2,\sigma}^-  
\end{equation} 
defined on indices $\bm{k}=(k_0,  k_1)$ in:
\begin{equation}\label{eq:DNLB}
\msc D_{N, L,\beta} = \Big\{ (k_{0}, k_{1}) \in \msc D_{L,\beta}\, \Big|\, k_{0} = \frac{2\pi}{\beta} \Big(n_{0} + \frac{1}{2}\Big)\;,\; |n_{0}| \le N \Big\}\;;
\end{equation}
\item[(ii)] $\bm \alpha:=  (0,\alpha)$ with $\alpha/2\pi = m / L_{1}$ is the best rational approximant (with denominator $L_{1}$) of the Diophantine number $\alpha_{\infty}/2\pi$, recall Assumption \ref{ass:pert};
 \item[(iii)] $\int \mathbb{P}_{N}(d\psi)$ is the Grassmann Gaussian integration with covariance:
\begin{equation}\label{eq:cov}
    \begin{split}
&        \pint_{N}( d \psi)\psi_{\bm k,x_2;\sigma}^-\psi^{+}_{\bm h,y_2;\zeta} \\
&\qquad = \beta L_{1}\delta_{k_0,h_0}\delta_{k_1, h_1} \Big(\frac{1}{ik_{0} + \hat H(k_{1}) - \mu}\Big)_{\sigma, \zeta}(x_{2}, y_{2}) \\
        &\qquad =: \beta L_{1}\delta_{k_0,h_0}\delta_{k_1, h_1} G_{\sigma,\zeta}(\bm{k},x_2, y_2)\;.
    \end{split}
\end{equation}
\end{itemize}
The goal of this section will be to introduce an algorithm that allows to compute (\ref{eq:ZZ}) as a convergent series in $\lambda$. The difficulty in doing this is that, as Eq. (\ref{eq:cov}) shows, the Grassmann field is singular: in the limit $\beta, L_{1} \to \infty$, the momentum-space covariance is unbounded, since by assumption $\mu$ belongs to the spectrum of the unperturbed Hamiltonian. In order to deal with this infrared problem, we shall perform a multiscale analysis, which is made possible by the addition principle of Grassmann Gaussian integrations: we split the Grassmann field into a sum of independent Grassmann fields, whose covariances are supported at a given distance from the singularity. The fields will then be integrated iteratively, and this will give rise to a flow of effective potentials.

To begin, we isolate the potentially divergent contribution to the covariance $G$. To this end, recall the definition of the smooth cutoff function, Eq. (\ref{eq:chidef}). Let $\chi(\hat H(k_{1}) - \mu)$ be the smoothed spectral projection, defined via functional calculus, which only depends on eigenstates of $\hat H(k_{1})$ with energies $\varepsilon(k_{1})$ in the interval $|\mu - \varepsilon(k_{1})| \le \delta$. Let us choose $\delta$ as in Assumption \ref{ass:H3}. Then, all the eigenstates of $\hat H(k_{1})$ in this energy range correspond to edge modes. We define:
\begin{equation}
G^{(\mr e)}_{\sigma,\zeta}({\bm k}, x_{2}, y_{2}) := \chi(k_{0}) \Big(\chi\big(\hat H(k_{1}) - \mu\big) G(\bm{k})\Big)_{\sigma,\zeta}(x_2, y_2)\;,
\end{equation}
which we can rewrite more explicitly as:
\begin{equation}\label{eq:Ge}
G^{(\mr e)}_{\sigma,\zeta}({\bm k}, x_{2}, y_{2}) = \sum_{\omega=\pm} \chi(k_0)\chi\big(\varepsilon_\omega(k_1)-\mu\big)\frac{\hat\xi^{\omega}_{\sigma}(k_1;x_2) \overline{\hat\xi^{\omega}_{\zeta}(k_1; y_2)}}{i k_0+\varepsilon_{\omega}(k_1)-\mu}\;.
\end{equation}
The label $(\mr{e})$ stands for ``edge'': the covariance $G^{(\mr e)}$ is spatially supported in proximity of the two boundaries of the cylinder, by the exponential decay of $\hat\xi^{\omega}_{\sigma}(k_1;x_2)$, recall Eq. (\ref{eq:decedge}). Next, let us introduce:
\begin{equation}
G^{(\mr b)}_{\sigma,\zeta}({\bm k}, x_{2}, y_{2}) := G_{\sigma,\zeta}(\bm{k},x_2, y_2) - G^{(\mr e)}_{\sigma,\zeta}({\bm k}, x_{2}, y_{2})\;.
\end{equation}
where $(\mr{b})$ stands for ``bulk''.
By construction, this covariance is bounded uniformly in ${\bm k}$. Let us define the discrete derivatives:
\begin{equation}
\begin{split}
\text{d}_{k_{0}} f(k_{0}) &= \frac{\beta}{2\pi} \big( f(k_{0}) - f(k_{0} -  2\pi/\beta) \big) \\
\text{d}_{k_{1}} f(k_{1}) &= \frac{L_{1}}{2\pi} \big( f(k_{1}) - f(k_{1} - 2\pi/L_{1}) \big)\;.
\end{split}
\end{equation}
Then, we have:
\begin{equation}\label{eq:bdGb}
\Big| \text{d}_{k_{0}}^{n_{0}} \text{d}_{k_{1}}^{n_{1}} G^{(\mr b)}_{\sigma,\zeta}({\bm k}, x_{2}, y_{2}) \Big| \le \frac{C^{\delta}_{n_{0}, n_{1}} e^{-c|x_{2} - y_{2}|}}{1 + |k_{0}|}\; ,
\end{equation}
uniformly in ${\bm k}$. The bound (\ref{eq:bdGb}) can proved using the estimate for the edge modes (\ref{eq:decedge}), and a Combes-Thomas type argument. See Appendix \ref{app:CT} for details.

We now represent the Grassmann field $\psi^{\pm}$ as:
\begin{equation}
\psi^{\pm}_{{\bm k}, x_{2}, \sigma} = \psi^{(\mr{e})\pm}_{{\bm k}, x_{2}, \sigma} + \psi^{(\mr{b})\pm}_{{\bm k}, x_{2}, \sigma}\;,
\end{equation}
where $\psi^{(\mr{e})\pm}$, $\psi^{(\mr{b})\pm}$ have covariances given respectively by $G^{(\mr e)}$ and $G^{(\mr b)}$. Then, we write, using the addition principle of the Gaussian Grassmann integration:
\begin{equation}\label{eq:ve}
\begin{split}
\ms{Z} /\ms{Z}^{0} &= \lim_{N\to +\infty} \int \mathbb{P}_{\mr{e}}( d  \psi^{(\mr e)}) \int \mathbb{P}_{\mr{b}}( d  \psi^{(\mr b)})\,\exp\big({\mc V_{N}( \psi^{(\mr e)} + \psi^{(\mr b)})}\big) \\
&= z_{\mr{b}}  \int \mathbb{P}_{\mr{e}}( d  \psi^{(\mr e)}) \exp\big({\mc V^{(\mr{e})}( \psi^{(\mr e)})}\big)\;,
\end{split}
\end{equation}
where: $z_{\mr{b}}$ is a suitable constant, to be discussed; $\mathbb{P}_{\mr{e}}$, $\mathbb{P}_{\mr{b}}$ are the Grassmann Gaussian integrations with covariances given by $G^{(\mr e)}$ and $G^{(\mr b)}$; and $\mathcal{V}^{(\mr{e})}$ is the {\it edge effective potential.} The effective potential is defined as a sum of truncated expectations (or cumulants) of the $\psi^{(\mr{b})}$ field:
\begin{equation}
\log(z_{\mr{b} }) \mathcal{V}^{(\mr{e})}(\psi^{(\mr e)}) = \lim_{N\to \infty} \sum_{m\ge 1} \frac{1}{m!} \mathbb{E}^{\text{T}}_{\mr{b}} \big(\mathcal{V}_{N}^{(\mr{b})}(\psi^{(\mr e)} + \cdot)\;; m\big)\; ,
\end{equation}
where $\mathbb{E}_{\mr{b}}^{\text{T}}$ denotes the truncated expectation:
\begin{equation}
\begin{split}
   & \mathbb{E}^{\text{T}}_{\mr{b}} \big(\mathcal{V}_{N}^{(\mr{b})}(\psi^{(\mr e)} + \cdot)\;; m\big)\\
   & \quad := \frac{\partial^{m}}{\partial t^{m}} \bigg(\log \int \mathbb{P}_{\mr{b}}( d  \psi^{(\mr b)})\,\exp\big(t{\mc V_{N}( \psi^{(\mr e)} + \psi^{(\mr b)})}\big) \bigg)\bigg|_{t=0}\;.
\end{split}
\end{equation}
It is well-known that the truncated expectations of Gaussian fields can be conveniently expressed in terms of connected graphs. In this case, due to the fact that the initial potential $\mathcal{V}_{N}$ is quadratic, we are left with particularly simple graphs, called chains. We have:
\begin{equation}\label{eq:Ve}
\mathcal{V}^{(\mr{e})}(\psi^{(\mr e)}) = \frac{1}{\beta L_{1}} \sum_{\substack{n, {\bm k}, \\ x_{2}, y_{2}, \sigma, \zeta}} V_{n;\sigma,\zeta}^{(\textrm{$\mr e$})}(\bm k,x_2,y_2)\psi_{\bm{k},x_2,\sigma}^{(\textrm{$\mr e$})+} \psi_{\bm{k}+n\bm{\alpha},y_2,\zeta}^{(\textrm{$\mr e$})-} \;,
\end{equation}
where $V_{n;\sigma,\zeta}^{(\textrm{$\mr e$})}(\bm k,x_2,y_2)=-\lambda \hat \varphi_{-n}\delta_{x_2,y_2}\delta_{\sigma,\zeta}+\mr O(\lambda^2)$ is given by:
\begin{equation}\label{eq:Vedef}
\begin{split}
    V_{n;\sigma,\zeta}^{(\textrm{$\mr e$})}(\bm k,x_2,y_2) = \sum_{s\ge 1} \bigg(\frac{1}{s!} \sum_{\theta \in {\msc T}^{(\mr{b})}_{s;n}} \theta(\bm{k},x_2,y_2)\bigg)\;,
\end{split}
\end{equation}   
where the sum is over tree graphs associated with values $\theta(\bm{k},x_2,y_2)$, with the following properties: 
\begin{itemize}
\item[(i)] the trees $\theta$ are labelled, and have nodes $v = 1,\ldots, s$ with $s\ge 1$;
\item[(ii)] each node has an incoming line and an outgoing line. The tree is formed by joining lines with opposite orientation. The incoming line of the first node and the outgoing line of the last node are not paired;
\item[(iii)] to each node we associate an integer $n_{v}$. The integers satisfy the condition $\sum_{v = 1}^{s} n_{v} = n$;
\item[(iv)] each node is labelled by a coordinate $x_{2,v}$ and a spin index $\sigma_{v}$; the first and the last node satisfy the constraint $x_{2,1} = x_{2}$, $x_{2,v} = y_{2}$, $\sigma_{1} = \sigma$, $\sigma_{v} = \zeta$;
\item[(v)] the value of the tree is:
\begin{equation}\label{eq:theta}
\begin{split}
    &\theta(\bm k,x_2,y_2) \\
    &\quad:= \Big(\prod_{v=1}^s -\lambda \hat\varphi_{-n_v}(x_{2,v})\Big) \prod_{v=1}^{s-1}G^{(\textrm{$\mr b$})}_{\sigma_v,\sigma_{v+1}}\big(\bm{k} + n(v)\bm{\alpha}, x_{2,v}, x_{2,v+1}\big)\;,
\end{split}
\end{equation}
\end{itemize}
where $n(v) = \sum_{v'\le v} n_{v'}$.
\begin{figure}
    \centering
    \includegraphics[scale=0.7]{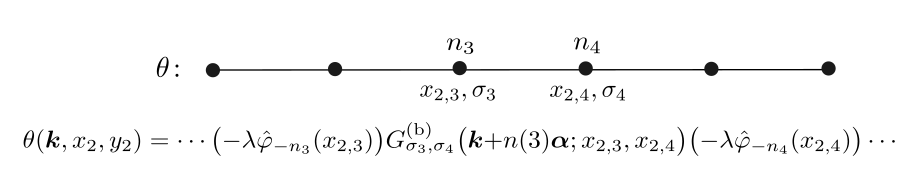}
    \caption{A couple of labels of a six nodes tree and the values to which they correspond.}
    \label{fig:bulktree}
\end{figure}
\begin{remark} 
\begin{itemize}
\item[(i)] We can visualize $\theta(\bm k,x_2,y_2)$, cfr. figure \ref{fig:bulktree}, as the value of a Feynman graph, obtained from the contraction of the lines exiting and incoming from the nodes. Each edge is associated with a propagator $G^{(\mr b)}$, while each node is associated with a perturbation $\lambda \hat\varphi_{n_{v}}$. The value of the graph is obtained by the product of all these objects, and the momentum flowing in each propagator is obtained from the Kirchhoff law, with the rule that each node is a ``source'' of momentum $n_{v} {\bm \alpha}$, which is transported along each propagator. The momentum ${\bm k}$ is the incoming momentum of the graph; the outgoing momentum is then ${\bm k}$ plus the net momentum shift due to the nodes, which is set to be $n{\bm \alpha}$.
\item[(ii)] The sum over $\theta$ takes into account the sum over all possible contractions between nodes, and the sum over all labels. If $s$ is the number of vertices, the number of possible connected graphs is $s!$.
\item[(iii)] The sum over ${\bm k} = (k_{0}, k_{1})$ in (\ref{eq:Ve}) involves all $k_{0} \in (2\pi / \beta) (\mathbb{Z} + 1/2)$, since the $N\to \infty$ limit has been taken. 
\end{itemize}
\end{remark}
Thus, in view of the bound (\ref{eq:bdGb}) and of the decay of the Fourier modes of the potential (\ref{eq:decFou}), we obtain that the effective potential is analytic in $\lambda$ for $|\lambda|$ small enough uniformly in $\beta, L$; moreover we get the following bound, for suitable constants $C_{n_{0},n_{1}}, c>0$:
\begin{equation}\label{eq:Vebd}
\Big| \text{d}_{k_{0}}^{n_{0}} \text{d}_{k_{1}}^{n_{1}} V_{n;\sigma,\zeta}^{(\textrm{$\mr e$})}(\bm k,x_2,y_2) \Big|\le C_{n_{0},n_{1}} |\lambda| e^{- c |n|} e^{-c|x_2-y_2|}\;.
\end{equation}
An important symmetry property of the effective potential is that:
\begin{equation}\label{eq:reale}
\overline{V_{n;\sigma,\zeta}^{(\mr e)}(\bm k,x_2,y_2)} = V_{-n;\zeta,\sigma}^{(\mr e)}\big((-k_{0}, k_{1} + n \alpha),y_2,x_2\big)\; ,
\end{equation}
which follows from (\ref{eq:theta}), from $\overline{\hat \varphi_{n}(x_{2})} = \hat \varphi_{-n}(x_{2})$ and from:
\begin{equation}
\overline{G^{(\mr b)}_{\sigma,\zeta}({\bm k}, x_{2}, y_{2})} = G^{(\mr b)}_{\zeta,\sigma}\big((-k_{0}, k_{1}), y_{2}, x_{2}\big)\; ,
\end{equation}
as a consequence of the self-adjointness of $\hat H(k_{1})$. 

Finally, the constant $z_{\mr{b}}$ in (\ref{eq:ve}) takes into account the exponential of the sum over ``vacuum diagrams'', that is with no external lines. It is:
\begin{equation}
\log(z_{\mr{b}}) = \sum_{{\bm k}} \sum_{n}\sum_{s\ge 1} \bigg(\frac{1}{s!} \sum_{\theta \in \widetilde{{\msc T}}^{(\mr{b})}_{s;n}} \theta(\bm{k})\bigg)\;,
\end{equation}
where the graphs in $\widetilde{{\msc T}}^{(\mr{b})}_{s;n}$ are obtained as before, with the only difference that all lines are now paired. The sum over $k_{0}$, which is unbounded as $N\to \infty$, is performed using the estimate (\ref{eq:bdGb}). This allows to control all graphs with at least two nodes, hence two propagators. In the case of a graph with a single node, the so-called tadpole graph, an explicit computation shows that the $k_{0}$ sum is bounded uniformly in $N$ (this is related to the fact that the propagator is aymptotically odd in $k_{0}$). Thus, $\log(z_{\mr{b}})$ is analytic in $\lambda$ for $|\lambda|$ small enough, and furthermore:
\begin{equation}\label{eq:zb}
 |\log(z_{\mr{b}}) | \le C \beta L_{1} L_{2}\;.
\end{equation}
%
%To conclude, we choose $\delta$ in the definition (\ref{eq:chidef}) of $\chi(\cdot)$ small enough, so that:
%
%\begin{equation}
%\text{supp}\Big(\chi\big(\varepsilon_{+}(k_{1}) - \mu\big)\Big) \cap \text{supp}\Big(\chi\big(\varepsilon_{-}(k_{1}) - \mu\big)\Big)  = \emptyset\;.
%\end{equation}
%
%In particular, the propagator $G^{(\mr e)}$ in Eq. (\ref{eq:Ge}) is given by the sum of two independent propagators, which depend on momenta supported on two disjoint neighbourhoods around the two Fermi points. 

\subsection{The flow of the effective potentials}
\subsubsection{One-dimensional reduction}\label{sec:1dred}
We are now left with the integration of the massless Gaussian field, with covariance given by $G^{(\mr e)}$, Eq. (\ref{eq:Ge}). As we will see, this can be reduced to the analysis of a massless, one-dimensional field. In order to see this, let us introduce the Grassmann field $\psi^{(\le 1)}_{\omega, {\bm k}}$, with measure $\mathbb{P}_{(\le 1)}$ and covariance:
\begin{equation}\label{eq:defcov}
\begin{split}
\int \mathbb{P}_{(\le 1)}(d\psi^{(\le 1)})\, \psi^{(\le 1)-}_{\omega,{\bm k}} \psi^{(\le 1)+}_{\omega', {\bm h}} &:= \beta L_{1}\delta_{\omega,\omega'} g^{(\le 1)}_{\omega}({\bm k})\;, \\
g^{(\le 1)}_{\omega}({\bm k}) &:= \frac{\chi(k_{0}) \chi(\varepsilon_\omega(k_1)-\mu)}{ik_{0} + \varepsilon_{\omega}(k_{1}) - \mu}\;.
\end{split}
\end{equation}
It is understood that the fields are only defined for momenta ${\bm k}$ in the support of the propagators. The one-dimensional reduction of the effective model obtained after integrating out $\psi^{(\mr{b})}$ is based on the observation that the random field $\psi^{(\mr e)}_{{\bm k}, x_{2}, \sigma}$ can be represented as:
\begin{equation}
\psi^{(\mr e)+}_{{\bm k}, x_{2}, \sigma} = \sum_{\omega = \pm} \psi^{(\le 1)+}_{\omega, {\bm k}} \overline{\xi^{\omega}_{\sigma}(k_{1}; x_{2})}\;,\qquad \psi^{(\mr e)-}_{{\bm k}, x_{2}, \sigma} = \sum_{\omega = \pm} \psi^{(\le 1)-}_{\omega, {\bm k}} \xi^{\omega}_{\sigma}(k_{1}; x_{2})\;.
\end{equation}
For short, we write $\psi^{(\mr e)} = \psi^{(\le 1)}\times \xi$. The $1d$ reduction of the model is a consequence of the fact that, given any monomial $M(\cdot)$ in the field $\psi^{(\mr e)}$:
\begin{equation}\label{eq:red1d}
\begin{split}
\int \mathbb{P}_{\mr{e}}( d  \psi^{(\mr e)}) M(\psi^{(\mr e)}) &= \int \mathbb{P}_{(\le 1)}(d\psi^{(\le 1)}) M(\psi^{(\le 1)}\times \xi) \\
&= \int \mathbb{P}_{(\le 1)}(d\psi^{(\le 1)}) \widetilde{M}(\psi^{(\le 1)})\;,
\end{split}
\end{equation}
where $\widetilde{M}(\psi^{(\le 1)})$ is a monomial in the field $\psi^{(\le 1)}$, with coefficients dependent on $\xi^{\omega}$, such that $\widetilde{M}(\psi^{(\le 1)}) = M(\psi^{(\le 1)}\times \xi)$. For
\begin{equation}
M(\psi) = \psi^{-}_{{\bm k}, x_{2}, \sigma} \psi^{+}_{{\bm h}, y_{2}, \zeta}\;,
\end{equation}
the identity in (\ref{eq:red1d}) is a trivial consequence of the definitions of the covariances of the two Gaussian fields. The general statement follows from the validity of the fermionic Wick rule, for both sides.

We can apply the one-dimensional reduction to rewrite the partition function in terms of the one-dimensional field $\psi^{(\le 1)}$. We get:
\begin{equation}\label{eq:Vred1d}
\begin{split}
\ms{Z} /\ms{Z}^{0} &= z_{\mr{b}}  \int \mathbb{P}_{\mr e}( d  \psi^{(\mr e)})\exp\big({\mc V^{(\mr{e})}( \psi^{(\mr e)})}\big) \\
&= z_{\mr{b}}  \int \mathbb{P}_{(\le 1)}(d\psi^{(\le 1)}) \exp\big( \mathcal{V}^{(1)}(\psi^{(\le 1)}) \big)\;,
\end{split}
\end{equation}
where $\mathcal{V}^{(1)}(\psi^{(\le 1)}) = \widetilde{\mathcal{V}}^{(\mr{e})}(\psi^{(\le 1)})$; explicitly,
\begin{equation}
\mathcal{V}^{(1)}(\psi^{(\le 1)}) = \frac{1}{\beta L_{1}} \sum_{\substack{n, {\bm k}, \\ \omega, \omega'}} V_{n;\omega,\omega'}^{(1)}(\bm k)\psi_{\bm{k} ,\omega}^{(\le 1)+} \psi_{\bm{k}+ n {\bm \alpha},\omega'}^{(\le 1)-}  \;;
\end{equation}
for ${\bm k}$ and ${\bm k} + n {\bm \alpha}$ in the support for the fields, the effective potential is:
\begin{equation}
V_{n;\omega,\omega'}^{(1)}(\bm k) = \sum_{\substack{x_{2}, y_{2} \\ \sigma, \zeta}} V_{n;\sigma,\zeta}^{(\textrm{$\mr e$})}(\bm k,x_2,y_2) \overline{\xi^{\omega}_{\sigma}(k_{1}; x_{2})} \xi^{\omega'}_{\zeta}(k_{1}+ n\alpha ; y_{2})\;,
\end{equation}
while it is set to zero for ${\bm k}$ or ${\bm k} + n {\bm \alpha}$ not in the support of the fields. 

Thus, in Eq. (\ref{eq:Vred1d}) we are reducing the computation of the partition function of the original $2d$ model to the analysis of the partition function of an effective $1d$ model; the effective potential of this $1d$ model correspond to the ``projection'' of the original kernels, on the two edges. The effective potential is analytic for $|\lambda|$ small enough; furthermore, from the bound (\ref{eq:Vebd}) and the assumptions on the edge modes (\ref{eq:decedge}), we have:
\begin{equation}\label{eq:est1}
\Big| \text{d}_{k_{0}}^{n_{0}} \text{d}_{k_{1}}^{n_{1}} V_{n;\omega,\omega'}^{(1)}(\bm k) \Big| \le C_{n_{0},n_{1}} |\lambda| e^{-c|n|} e^{-\tilde c\delta_{\omega,-\omega'}L_{2}}\;.
\end{equation}
In particular, we see that the scattering between the different edges is exponentially suppressed, for $L_{2}$ large. To conclude the paragraph, it is useful to write the first order contribution to the effective potential. It is:
\begin{equation}
V_{n;\omega,\omega'}^{(1)}(\bm k)= -\lambda \sum_{x_{2},\sigma} \overline{\xi^{\omega}_{\sigma}(k_{1}; x_{2})} \xi^{\omega'}_{\sigma}(k_{1} + n\alpha; y_{2}) \hat \varphi_{-n}(x_{2})+\mr O(\lambda^2)\;.
\end{equation}
Thus, we see that, already at first order in $\lambda$, the two edge modes are coupled by the disorder. Also, we conclude by observing that the analogue of (\ref{eq:reale}) holds true:
\begin{equation}
\overline{V_{n;\omega,\omega'}^{(1)}(\bm{k})} = V_{-n;\omega',\omega}^{(1)}(-k_{0}, k_{1} + n\alpha)\;.
\end{equation}
\subsubsection{Integration of the first scale}\label{sec:firstscale}
\paragraph{Shift of the Fermi points.} We now set up the multiscale integration of the one-dimensional field. In order to do this, it is convenient, for later purpose, to allow for a small shift of the Fermi points. Let:
\begin{equation}
\chi^{(\le 1)}_{\omega}({\bm k}) := \chi(k_{0}) \chi(\varepsilon_\omega(k_1)-\mu)\;.
\end{equation}
We rewrite:
\begin{equation}\label{eq:shift}
\begin{split}
g^{(\le 1)}_{\omega}({\bm k}) &= \frac{\chi^{(\le 1)}_{\omega}({\bm k})}{ik_{0} + \varepsilon_{\omega}(k_{1}) - \mu + \nu_{\omega}({\bm k}) - \nu_{\omega}({\bm k})}\;,
\end{split}
\end{equation}
where $\nu_{\omega}({\bm k}) = \nu_{\omega} \chi^{(\le 1)}_{\omega}({\bm k})$ with $\nu_{\omega} = O(\lambda)$ real-valued, to be determined. We rewrite the partition function as:
\begin{equation}
\ms{Z} /\ms{Z}^{0} = \tilde z_{\mr{b}}  \int \widetilde{\mathbb{P}}_{(\le 1)}(d\psi^{(\le 1)}) \exp\big( \widetilde{\mathcal{V}}^{(1)}(\psi^{(\le 1)}) \big)
\end{equation}
where: the new Grassmann Gaussian integration has covariance given by
\begin{equation}
\tilde g^{(\le 1)}_{\omega}({\bm k}) = \frac{\chi^{(\le 1)}_{\omega}({\bm k})}{ik_{0} + \varepsilon_{\omega}(k_{1}) - \mu + \nu_{\omega}({\bm k})}\;;
\end{equation}
the constant $\tilde z_{\mr{b}}$ takes into account the change in normalization of the Grassmann integration; the new effective potential is:
\begin{equation}\label{eq:tildev1}
\begin{split}
\widetilde{\mathcal{V}}^{(1)}(\psi^{(\le 1)}) &= \mathcal{V}^{(1)}(\psi^{(\le 1)}) + \frac{1}{\beta L_{1}} \sum_{{\bm k}, \omega} \nu_{\omega} \psi_{\bm{k},\omega}^{(\le 1)+}  \psi_{\bm{k},\omega}^{(\le 1)-} \\
&\equiv \frac{1}{\beta L_{1}} \sum_{\substack{n, {\bm k}, \\ \omega, \omega'}} \widetilde{V}_{n;\omega,\omega'}^{(1)}(\bm k)  \psi_{\bm{k} ,\omega}^{(\le 1)+}\psi_{\bm{k}+ n {\bm \alpha},\omega'}^{(\le 1)-} \;,
\end{split}
\end{equation}
with
\begin{equation}
\widetilde{V}_{n;\omega,\omega'}^{(1)}(\bm k)  = V_{n;\omega,\omega'}^{(1)}(\bm k) + \delta_{n,0} \delta_{\omega,\omega'} \nu_{\omega}\;.
\end{equation}
Thus, we now have that the denominator of the covariance is vanishing in correspondence with a slightly different point, to be denoted ${\bm k}_{F}^{\omega}(\lambda) = \big(0, k_{F}^{\omega}(\lambda)\big)$, where $k^{\omega}_{F}(\lambda)\equiv k^{\omega}_{F,L}(\lambda)$ is the the solution to:
\begin{equation}
\varepsilon_{\omega}\big(k^{\omega}_{F}(\lambda)\big) - \mu + \nu_{\omega} = 0\;.
\end{equation}
Clearly, since $|\nu_{\omega}|\le C|\lambda|$, we also have that $| k^{\omega}_{F}(\lambda) - k_{F}^{\omega} | \le K|\lambda|$. 

Next, we center coordinates around the two new Fermi points. We rewrite any ${\bm k}$ in the support of $\tilde g^{(\le 1)}_{\omega}({\bm k})$ as ${\bm k} = {\bm q} + {\bm k}_{F,L}^{\omega}(\lambda)$, with $\| {\bm q} \| \le C\delta$, and ${\bm k}_{F}^{\omega}(\lambda)$ a best lattice approximation of ${\bm k}_{F}^{\omega}(\lambda)$.  With a slight abuse of notation, we shall set:
\begin{equation}\label{eq:coordinatechange}
\psi^{(\le 1)\pm}_{{\bm q} + {\bm k_{F}^{\omega}}(\lambda), \omega} \equiv \psi^{(\le 1)\pm}_{{\bm q}, \omega}\;,
\end{equation}
and omit the $L$ subscript in the Fermi point.
\paragraph{Field integration.} Let us denote by $\tilde v_{\omega}$ the Fermi velocity of the shifted Fermi point:
\begin{equation}
\tilde v_{\omega} = \partial_{k_{1}} \varepsilon_{\omega}\big(k_{F}^{\omega}(\lambda)\big)\;.
\end{equation}
By the assumptions on the energy bands of the edge modes, $|v_{\omega} - \tilde v_{\omega}| \le C|\lambda|$. We define the new cutoff function:
\begin{equation}\label{eq:g0}
\chi^{(\le 0)}_{\omega}({\bm q}) := \chi\Big(\gamma \sqrt{ q_{0}^{2} + \tilde v_{\omega}^{2} q_{1}^{2}  }\Big)\;.
\end{equation}
For $\gamma$ large enough, the support of this cutoff function is strictly contained in the support of the cutoff function appearing the definition of the covariance (\ref{eq:defcov}). Let us define the new propagator:
\begin{equation}
g^{(\le 0)}_{\omega}({\bm q}) :=  \frac{\chi^{(\le 0)}_{\omega}({\bm q})}{iq_{0} + \varepsilon_{\omega}(q_{1}) - \mu + \nu_{\omega}}\;,
\end{equation}
where we used that, choosing $\gamma$ large enough, $\nu_{\omega}\big({\bm q} \big) = \nu_{\omega}$ for ${\bm q }$ in the support of $\chi^{(\le 0)}_{\omega}({\bm k})$. We rewrite the propagator $\tilde g^{(\le 1)}_{\omega}$ as:
\begin{equation}\label{eq:g1}
\tilde g^{(\le 1)}_{\omega}\big({\bm q} + {\bm k}_{F}^{\omega}(\lambda)\big) = g^{(\le 0)}_{\omega}({\bm q}) + g^{(1)}_{\omega}({\bm q})
\end{equation}
where $g^{(1)}_{\omega}({\bm q})$ is supported for momenta $(\delta/\gamma^2) \lesssim \| {\bm q} \| \lesssim \delta$, and it is bounded together with all its derivatives:
\begin{equation}\label{eq:dg1}
\big| \text{d}_{k_{0}}^{n_{0}} \text{d}_{k_{1}}^{n_{1}} g^{(1)}_{\omega}({\bm q})  \big| \le C_{n_{0}, n_{1}}\;,
\end{equation}
for a constant $C_{n_{0}, n_{1}}$ which depends on $\gamma$.

We then represent the propagator $\psi^{(\le 1)}$ as $\psi^{(\le 1)} = \psi^{(1)} + \psi^{(\le 0)}$, where the field $\psi_{\omega}^{(1)}$ has covariance $g^{(1)}_{\omega}$ while the field $\psi_{\omega}^{(\le 0)}$ has covariance $g^{(\le 0)}_{\omega}$. By the addition principle,
\begin{equation}\label{eq:scale0}
\begin{split}
\ms{Z} /\ms{Z}^{0} &= \tilde z_{\mr{b}}  \int \widetilde{\mathbb{P}}_{(\le 1)}(d\psi^{(\le 1)}) \exp\big( \widetilde{\mathcal{V}}^{(1)}(\psi^{(\le 1)}) \big) \\
&= \tilde z_{\mr{b}} \int \mathbb{P}_{(\le 0)}(d\psi^{(\le 0)}) \int \mathbb{P}_{(1)}(d\psi^{(1)}) \exp\big( \widetilde{\mathcal{V}}^{(1)}(\psi^{(\le 0)} + \psi^{(1)}) \big) \\
&\equiv \tilde z_{\mr{b}} z_{1} \int \mathbb{P}_{(\le 0)}(d\psi^{(\le 0)}) \exp\big(\mathcal{V}^{(0)}(\psi^{(\le 0)}) \big)\;,
\end{split}
\end{equation}
where the single-scale integration is performed via a cumulant expansion, as in Section \ref{sec:massive}. The new effective potential has the form:
\begin{equation}
\mathcal{V}^{(0)}(\psi^{(\le 0)}) = \frac{1}{\beta L_{1}} \sum_{\substack{n, {\bm q}, \\ \omega, \omega'}} V_{n;\omega,\omega'}^{(0)}({\bm q}) \psi_{{\bm q},\omega}^{(\le 0)+}  \psi_{{\bm q} + n {\bm \alpha},\omega'}^{(\le 0)-}\;,
\end{equation}
for new functions $V_{n;\omega,\omega'}^{(0)}({\bm q})$. Informally, these new functions are easily understood diagrammatically, by contracting the old $\widetilde{V}_{n;\omega,\omega'}^{(1)}(\cdot)$ through chains of single-scale propagators on scale $1$. Similarly, the constant $z_{1}$ is obtained by collecting the vacuum diagrams constructed using the contractions of the single-scale propagators.  In the next section we will give a precise definition of the general structure of these objects, in a more general context. What we can say at this stage is that the new constant $z_{1}$ and the new functions $V_{n;\omega,\omega'}^{(0)}({\bm q})$ are analytic in $\lambda$ for $|\lambda|$ small enough, and satisfy similar properties as their counterparts $z_{\mr{b}}$ and $V_{n;\omega,\omega'}^{(1)}(\bm k)$. In particular:
\begin{equation}\label{eq:V0}
\Big| \text{d}_{q_{0}}^{n_{0}} \text{d}_{q_{1}}^{n_{1}} V_{n;\omega,\omega'}^{(0)}(\bm q) \Big| \le C_{n_{0},n_{1}} |\lambda| e^{-c|n|} e^{-\tilde c\delta_{\omega,-\omega'}L_{2}}\;.
\end{equation}
The important remark at this point is that analyticity holds {\it for a smaller range of $\lambda$}, which so far depends on $\gamma$: this is due to the fact that the estimate (\ref{eq:dg1}) is sensitive to the smallest value that ${\bm q}$ can attain, which us now of order $1/\gamma^3$ instead of order $1/\gamma^2$.   In order to iterate this integration to smaller scales, and to ensure a positive radius of convergence in $\lambda$, we will have to introduce a suitable localization and renormalization procedure.

\subsubsection{Iterative integration}\label{sec:iterative}

The above procedure can be iterated for smaller values of the momenta. This will be done by grouping the momenta into scales, defined in terms of a given distance from the Fermi points. We will proceed by induction. 
\paragraph{Inductive assumptions.} Let $h_{\beta} \in \mathbb{Z}_{-}$, such that:
\begin{equation}\label{eq:hbeta}
\delta\gamma^{h_{\beta} - 1} \le \frac{\pi}{\beta}\le \delta \gamma^{h_{\beta}}\;.
\end{equation}
Let $h \in \mathbb{Z}_{-}$, such that $h_{\beta} \le h$. We assume that the partition function can be rewritten as, for $|\lambda| < \lambda_{0}$:
\begin{equation}\label{eq:ind}
\ms{Z} /\ms{Z}^{0} = z_{(> h)} \int \mathbb{P}_{(\le h)} (d\psi^{(\le h)}) \exp\big( \mathcal{V}^{(h)}(\psi^{(\le h)}) \big)\;,
\end{equation}
with the following meaning of the various objects in (\ref{eq:ind}).
\medskip

\noindent(i) The field $\psi^{(\le h)\pm}_{{\bm q}, \omega}$ has covariance given by:
\begin{equation}\label{eq:ind2}
\begin{split}
\int \mathbb{P}_{(\le h)} (d\psi^{(\le h)}) \psi^{-}_{{\bm q}, \omega} \psi^{+}_{{\bm h}, \omega'} &= \beta L_{1} \delta_{{\bm q}, {\bm h}} \delta_{\omega,\omega'} g^{(\le h)}_{\omega}({\bm q})\;, \\
g^{(\le h)}_{\omega}({\bm q}) &= \frac{\chi^{(\le h)}_{\omega}({\bm q})}{i v_{0,\omega,h+1} q_{0} + v_{1,\omega,h+1} q_{1} + r_{\omega}(q_{1})}\;,
\end{split}
\end{equation}
where: $\chi^{(\le h)}_{\omega}({\bm q}) = \chi\Big(\gamma^{-h-1} \sqrt{v_{0,\omega,h+1}^{2}q_{0}^{2} + v_{1,\omega,h+1}^{2}q_{1}^{2} }\Big)$ with $v_{0,\omega, h+1}$, $v_{1,\omega, h+1}$ real valued and such that:
\begin{equation}
| v_{0,\omega,h+1} - 1 |\le C|\lambda|\;,\qquad | v_{1,\omega,h+1} - v_{\omega}| \le C|\lambda|\;.
\end{equation}
The term $r_{\omega}(q_{1})$ is real-valued and it satisfies, for $\gamma^{h} \le |q_{1}|$:
\begin{equation}
\big| \text{d}_{q_{1}}^{n_{1}} r_{\omega}(q_{1}) \big| \le C_{ n_{1}} \gamma^{-h n_{1}}| q_{1} |^{1+\alpha}\;,\qquad \text{for $\alpha > 0$;}
\end{equation}
\noindent(ii) The effective potential has the form:
\begin{equation}\label{eq:Vh}
\mathcal{V}^{(h)}(\psi^{(\le h)}) = \frac{1}{\beta L_{1}} \sum_{\substack{n, {\bm q}, \\ \omega, \omega'}} V_{n;\omega,\omega'}^{(h)}({\bm q}) \psi_{{\bm q},\omega}^{(\le h)+}  \psi_{{\bm q} + n {\bm \alpha},\omega'}^{(\le h)-}\; ,
\end{equation}
where $V_{n;\omega,\omega'}^{(h)}({\bm q})$ satisfy:
\begin{equation}\label{eq:Vindh}
\Big| \text{d}_{q_{0}}^{n_{0}} \text{d}_{q_{1}}^{n_{1}} V^{(h)}_{n;\omega,\omega'}({\bm q}) \Big| \le K_{n_{0},n_{1}} \gamma^{h(1 - n_{0} - n_{1})} |\lambda| e^{-\frac{c}{4}|n|} e^{-\frac{\tilde c}{4}\delta_{\omega,-\omega'}L_{2}}\;,
\end{equation}
with $c, \tilde c >0$ is as in (\ref{eq:V0}). Furthermore,
\begin{equation}\label{eq:realh}
\overline{V_{n;\omega,\omega'}^{(h)}(\bm{q})} = V_{-n;\omega',\omega}^{(h)}\big(-q_{0}, q_{1} + n\alpha - \delta_{n\neq 0}(k_{F}^{\omega} - k_{F}^{\omega'})\big)\;.
\end{equation}
\noindent(iii) The overall normalization constant satisfies:
\begin{equation}
\big| \log (z_{(> h)}) \big| \le C \beta L_{1} L_{2}\;.
\end{equation}
We will discuss how these assumptions allow to integrate the scale $h$, and to obtain the effective potential on scale $h-1$. This will induce an iterative construction of the effective potentials, that allows to express the effective potential on scale $h-1$ in terms of the potential on scale $0$ and of a set of running coupling constants on scales $k\ge h$, to be introduced below. Under the assumption that the flow of these running coupling constant satisfies suitable estimates, which hold true on scale $0$, we will be able to prove the above assumptions on scale $h-1$. Then, we will show how to prove the assumption on the running coupling constant on scale $h-1$, and this will conclude the iterative argument. 
\begin{remark}
As we will see, the iterative construction provides much more precise information about the various objects entering in (\ref{eq:ind}) than just the stated bounds: it will introduce a convergent expansion for the effective potential on all scales.
\end{remark}

\paragraph{Localization and renormalization.} In order to integrate the scale $h$, we will renormalize the Gaussian integration, by reabsorbing into the covariance the contributions to the effective potential the dangerous contributions to the effective potentials. Let ${\bm 0}^{+}_{\beta} := (\pi/\beta, 0)$, ${\bm 0}^{-}_{\beta} := (-\pi/\beta, 0)$ and let:
\begin{equation}\label{eq:sym}
F({\bm 0}_{\beta}) := \frac{1}{2} \sum_{\alpha = \pm} F({\bm 0}^{\alpha}_{\beta})\;.
\end{equation}
We shall also use the notation:
\begin{equation}
 \mathbbm{1}\Big|_{{\bm q} = {\bm 0}_{\beta}} F({\bm q}) := F({\bm 0}_{\beta})\;.
\end{equation}
We are now ready to introduce the notion of localization.
\begin{definition}[Localization operator]\label{def:loc} Let $V^{(h)}_{n;\omega,\omega'}({\bm q})$ be a coefficient of the effective potential (\ref{eq:Vh}). We define:
\begin{equation}
\mf L V^{(h)}_{n;\omega,\omega'}({\bm q}) := \delta_{n,0}\big(V^{(h)}_{n;\omega,\omega'}({\bm 0}_{\beta}) + q_{0} \text{d}_{q_{0}}\! V^{(h)}_{n;\omega,\omega'}({\bm 0}_{\beta}) + q_{1} \text{d}_{q_{1}}\! V^{(h)}_{n;\omega,\omega'}({\bm 0}_{\beta})\big)\;.
\end{equation} 
The action of $\mf{L}$ is extended to $\mathcal{V}^{(h)}$ by linearity. Furthermore, we define:
\begin{equation}
\mf RV^{(h)}_{n;\omega,\omega'}({\bm q}) := V^{(h)}_{n;\omega,\omega'}({\bm q}) - \mf L V^{(h)}_{n;\omega,\omega'}({\bm q})\;.
\end{equation}
We also set $\mf L = \mf L_{0} + \mf L_{1;\text{d}} + \mf L_{1;\text{od}}$, where:
\begin{equation}\label{eq:Ldef2}
\begin{split}
\mf L_{0} V^{(h)}_{n;\omega,\omega'}({\bm q}) &:= \delta_{n,0} V^{(h)}_{0;\omega,\omega'}({\bm 0}_{\beta})\;,\\
\mf L_{1;\text{d}} V^{(h)}_{n;\omega,\omega'}({\bm q}) &:= \delta_{n,0}\delta_{\omega,\omega'}\big( q_{0} \text{d}_{q_{0}}\! V^{(h)}_{n;\omega,\omega'}({\bm 0}_{\beta}) + q_{1} \text{d}_{q_{1}}\! V^{(h)}_{n;\omega,\omega'}({\bm 0}_{\beta})\big)\; , \\
\mf L_{1;\text{od}} V^{(h)}_{n;\omega,\omega'}({\bm q}) &:= \delta_{n,0}\delta_{\omega,-\omega'}\big( q_{0} \text{d}_{q_{0}}\! V^{(h)}_{n;\omega,\omega'}({\bm 0}_{\beta}) + q_{1} \text{d}_{q_{1}}\! V^{(h)}_{n;\omega,\omega'}({\bm 0}_{\beta})\big)\;,
\end{split}
\end{equation}
where the subscript ``d'' stands for diagonal, and the subscript ``od'' stands for off-diagonal.
\end{definition}
\begin{remark}\label{rem:loc}
\begin{itemize}
\item[(i)] The operator $\mf L$ is called the localization operator, and its role is to extract from the effective potential $\mathcal{V}^{(h)}$ the relevant and marginal terms in the renormalization group sense; these terms require a separate discussion. Ultimately, they will be controlled by a suitable choice of the parameters $\nu_{\omega}$, entering the definition of shifted chemical potentials; recall Eq. (\ref{eq:shift}).
\item[(ii)] The operator $\mf L_{0}$ isolates {\it relevant terms} in the renormalization group sense, while the operators $\mf L_{1;\text{d}}$, $\mf{L}_{1;\text{od}}$ isolate {\it marginal terms} in the renormalization group sense. The naive dimensional estimate for the former diverges exponentially in $|h|$, while the naive dimensional estimate for the latter diverge linearly in $|h|$. As we will see, and as the reader might already expect, the off-diagonal terms (involving scattering processes between different edges) are actually exponentially small in $L_{2}$, recall the bound (\ref{eq:Vindh}); this allows to define a range of temperatures and of cylinder widths for which the edge-edge scatterings are subleading.
\item[(iii)] The reader might wonder why in the definition of localization we are only selecting the $n=0$ terms. As we will see, thanks to the Diophantine properties of $\alpha$, Eq. (\ref{eq:diophantine}), the contributions to the effective potential associated with $n \neq 0$ will be suppressed by an extra, scale dependent, small factor.
\item[(iv)] It is important to observe that the operator $\mf R$ acts as a projection: $\mf R^{2} V^{(h)}_{n}({\bm q})  = \mf R V^{(h)}_{n}({\bm q})$. If $n\neq 0$, this is trivial, since ${\mf R}$ acts as the identity in this case. If $n = 0$, we compute:
\begin{equation}
\begin{split}
&\mf R^{2} V^{(h)}_{n}({\bm q}) \\
&= \Big(1 - \mathbbm{1}\Big|_{{\bm q} = {\bm 0}_{\beta}} - q_{0} \text{d}_{q_{0}}\Big|_{{\bm q} = {\bm 0}_{\beta}} - q_{1} \text{d}_{q_{1}}\Big|_{{\bm q} = {\bm 0}_{\beta}}  \Big) {\mf R} V^{(h)}_{n}({\bm q}) \\
&= {\mf R} V^{(h)}_{n}({\bm q}) - \Big( \mathbbm{1}\Big|_{{\bm q} = {\bm 0}_{\beta}} + q_{0} \text{d}_{q_{0}}\Big|_{{\bm q} = {\bm 0}_{\beta}} + q_{1} \text{d}_{q_{1}}\Big|_{{\bm q} = {\bm 0}_{\beta}}\Big) \\
&\quad \circ \Big( 1 - \mathbbm{1}\Big|_{{\bm q} = {\bm 0}_{\beta}}- q_{0} \text{d}_{q_{0}}\Big|_{{\bm q} = {\bm 0}_{\beta}} - q_{1} \text{d}_{q_{1}}\Big|_{{\bm q} = {\bm 0}_{\beta}} \Big) V^{(h)}_{n}({\bm q}) \\
&= {\mf R} V^{(h)}_{n}({\bm q})\;,
\end{split}
\end{equation}
where in the last step we used that:
\begin{equation}
\text{d}_{q_{\mu}}\Big|_{{\bm q} = {\bm 0}_{\beta}} q_{\nu} = \delta_{\mu,\nu} \mathbbm{1}\Big|_{{\bm q} = {\bm 0}_{\beta}}\;.
\end{equation}
\item[(v)] Observe that, by (\ref{eq:sym}), we are defining the localization in a ``symmetrized'' way. This definition combined with (\ref{eq:realh}) implies:
\begin{equation}\label{eq:ccLR}
\begin{split}
\overline{{\mf L} V^{(h)}_{n;\omega,\omega'}({\bm q})} &= {\mf L} V^{(h)}_{-n;\omega',\omega}\Big(- q_{0}, q_{1} + n\alpha - \delta_{n\neq 0}\big(k_{F}^{\omega}(\lambda( - k_{F}^{\omega'}(\lambda)\big)\Big)\;, \\
\overline{{\mf R} V^{(h)}_{n;\omega,\omega'}({\bm q})} &= {\mf R} V^{(h)}_{-n;\omega',\omega}\Big(- q_{0}, q_{1} + n\alpha - \delta_{n\neq 0}\big(k_{F}^{\omega}(\lambda) - k_{F}^{\omega'}(\lambda)\big)\Big)\;.
\end{split}
\end{equation}
These properties will be important to ensure the reality properties of the velocities appearing in the propagator on smaller scales.
\end{itemize}
\end{remark}
\paragraph{The beta function.} Let us now use the localization operator to rewrite the argument of the Grassmann integral in a way which is more convenient fot setting up an iterative integration. We have:
\begin{equation}\label{eq:change}
\begin{split}
& \int \mathbb{P}_{(\le h)} (d\psi^{(\le h)}) \,e^{\mathcal{V}^{(h)}(\psi^{(\le h)})} \\
&\quad = \int \mathbb{P}_{(\le h)} (d\psi^{(\le h)})\, e^{\mf L _{1;\text{d}}\mathcal{V}^{(h)}(\psi^{(\le h)})} e^{(\mf L_{0} + \mf L_{1;\text{od}} + \mf R) \mathcal{V}^{(h)}(\psi^{(\le h)})} \\
&\quad = \tilde z_{h} \int \widetilde{\mathbb{P}}_{(\le h)} (d\psi^{(\le h)})\, e^{(\mf L_{0} + \mf L_{1;\text{od}} + \mf R) \mathcal{V}^{(h)}(\psi^{(\le h)})}\;,
 \end{split} 
\end{equation}
where $\widetilde{\mathbb{P}}_{(\le h)}$ is a new Grassmann Gaussian integration, whose covariance is {\it renormalized} by including $\mf L_{1;\text{d}} \mathcal{V}^{(h)}$ into its definition. The constant $\tilde z_{h}$ takes into account the change of normalization of the integral, and it is bounded as $|\tilde z_{h}| \leq \text{exp}(\beta L_{1} C|\lambda| \gamma^{h})$. Let us discuss the structure of the new covariance. To this end, we define:
\begin{equation}
\text{d}_{q_{0}}\! V^{(h)}_{0;\omega,\omega'}({\bm 0}_{\beta}) =: i \beta^{v}_{0,h+1,\omega,\omega'}\;,\qquad \text{d}_{q_{1}}\! V^{(h)}_{0;\omega,\omega}({\bm 0}_{\beta}) =: \beta^{v}_{1,h+1,\omega,\omega'}\;.
\end{equation}
Observe that, thanks to (\ref{eq:realh}), 
\begin{equation}\label{eq:realbetaa}
\overline{\beta^{v}_{\nu,\omega,\omega',h+1}} = \beta^{v}_{\nu,\omega',\omega,h+1} \qquad \nu = 0,1\;,
\end{equation}
and in particular $\beta^{v}_{\nu,\omega,\omega,h+1}$ is real. Also, the bound (\ref{eq:Vindh}) implies that:
\begin{equation}
|\beta^{v}_{\nu,\omega,\omega',h} | \le C|\lambda| e^{-\frac{\tilde c}{4}\delta_{\omega,-\omega'}L_{2}}\;.
\end{equation}
The propagator of the new Grassmann integration in (\ref{eq:change}) is:
\begin{equation}\label{eq:tildegh}
\widetilde{g}^{(\le h)}_{\omega}({\bm q}) := \frac{\chi^{(\le h)}_{\omega}({\bm q})}{i v_{0,\omega,h}({\bm q}) q_{0} + v_{1,\omega,h}({\bm q}) q_{1} + r_{\omega}(q_{1})}
\end{equation}
where:
\begin{equation}\label{eq:vflow}
v_{\nu,\omega,h}({\bm q}) := v_{\nu,\omega,h+1} + \beta^{v}_{\nu,\omega,\omega,h+1} \chi_{\omega}^{(\le h)}({\bm q})\;.
\end{equation}
By (\ref{eq:realbetaa}), the function $v_{\nu,\omega,h}({\bm q})$ is real-valued. For later purposes, we shall also define:
\begin{equation}\label{eq:rcc}
v_{\nu,\omega,h} := v_{\nu,\omega,h}({\bm 0}_{\beta})\;;
\end{equation}
by the properties of the function $\chi_{\omega}^{(\le h)}(\cdot)$, we have $v_{\nu,\omega,h}({\bm q}) = v_{\nu,\omega,h}$ for
\begin{equation}\label{eq:supph}
\sqrt{v_{0,\omega,h+1}^{2}q_{0}^{2} + v_{1,\omega,h+1}^{2}q_{1}^{2} } \le \delta\gamma^{h-1}\;.
\end{equation}
Thus, the new propagator has a form very similar to the previous one, up to a renormalization of order $\lambda$ of the parameters. Next, let us look at the effective interaction in the last step of (\ref{eq:change}). We set, recalling (\ref{eq:Ldef2}):
\begin{equation}\label{eq:Ldef3}
\mf L_{0} V^{(h)}_{0;\omega,\omega'}({\bm q}) =: \gamma^{h} \nu_{\omega,\omega',h}\;,\qquad  \mf L_{1;\text{od}} V^{(h)}_{0;\omega,-\omega}({\bm q}) =: iq_{0} \tilde v_{0,\omega,h} + q_{1} \tilde v_{1,\omega, h}\;.
\end{equation}
By the assumption (\ref{eq:Vindh}), we know that:
\begin{equation}
| \nu_{\omega,\omega',h} | \le C|\lambda| e^{-c\delta_{\omega,-\omega'} L_{2}}\;,\qquad |\tilde v_{\nu, \omega, h}| \le C|\lambda| e^{-c L_{2}}\;.
\end{equation}
Let us collect both $v_{\mu,\omega,h}$ and $\tilde v_{\mu,\omega,h}$ in a matrix $v_{\mu,\omega,\omega',h}$,
\begin{equation}
v_{\mu,\omega,\omega,h} := v_{\mu,\omega,h}\;,\qquad v_{\mu,\omega,-\omega,h} := \tilde v_{\mu,\omega,h}\;;
\end{equation}
the beta function for the parameters $v_{\mu,\omega,\omega',h}$, $v_{\mu,\omega,\omega',h}$ is:
\begin{equation}\label{eq:nuflow}
\beta^{\nu}_{\omega,\omega',h+1} := \gamma^{-1} \nu_{\omega,\omega',h} - \nu_{\omega,\omega',h+1}\;,\qquad \beta^{v}_{\nu,\omega,\omega',h+1} := v_{\nu,\omega,\omega',h} - v_{\nu,\omega,\omega',h+1}\;.
\end{equation}
Thus, we can write the {\it flow of the running coupling constants} as:
\begin{equation}\label{eq:beta}
\begin{split}
\nu_{\omega,\omega',h} &= \gamma \nu_{\omega,\omega',h+1} + \gamma \beta^{\nu}_{\omega,\omega',h+1} \\
 v_{\nu,\omega,\omega',h} &= v_{\nu,\omega,\omega',h+1} + \beta^{v}_{\nu,\omega,\omega',h+1}\;.
\end{split}
\end{equation}
We shall refer to the running coupling constants with $\omega \neq \omega'$ as the {\it off-diagonal} running coupling constants. These running coupling constants satisfy, as a consequence of (\ref{eq:realh}):
\begin{equation}
\overline{\nu_{\omega,\omega',h}} = \nu_{\omega',\omega,h}\;,\qquad \overline{v_{\nu,\omega,\omega',h}} = v_{\nu,\omega',\omega,h}\;.
\end{equation}
The initial datum of this discrete dynamical system is read off from (\ref{eq:g0}), (\ref{eq:scale0}). Thanks to (\ref{eq:V0}), it satisfies the estimates:
\begin{equation}
\begin{split}
&v_{0,\omega,\omega,0} = 1\;,\qquad |v_{1,\omega,\omega,0} - v_{\omega}| \le C|\lambda|\;,\qquad |v_{\nu,\omega,-\omega,0}| \le C|\lambda|e^{-\tilde c L_{2}}\;, \\
&|\nu_{\omega,\omega,0} - \nu_{\omega}|\le C|\lambda|\;,\qquad |\nu_{\omega,-\omega,0}| \le C|\lambda| e^{-\tilde cL_{2}}\;.
\end{split}
\end{equation}
The sequence $( \beta^{v}_{\nu,\omega,\omega',h}, \beta^{\nu}_{\omega,\omega',h} )_{h = h_{\beta} + 1}^{0}$ defines the {\it beta function} of the model. It is a nontrivial task to prove bounds for the beta function that are summable in the scale label, and that allow to prove that the solution of the recursion relation is bounded uniformly in $h$. The first line in (\ref{eq:beta}) describes the flow of the {\it relevant} running coupling constants, while the second line in (\ref{eq:beta}) describes the flow of the marginal running coupling constants. 

One of our key technical goals will be to show that, for a suitable choice of $\nu_{\omega} = O(\lambda)$, and for $L_{2} \ge \kappa\beta$:
\begin{equation}\label{eq:indrcc}
\begin{split}
|v_{\nu,\omega,\omega',h} - v_{\nu,\omega,\omega',0}| &\le C|\lambda|^{2} e^{-\frac{\tilde c}{16} L_{2}\delta_{\omega,-\omega'}} \\
| \nu_{\omega,-\omega,h} - \gamma^{-h} \nu_{\omega,-\omega,0} | &\le C|\lambda|^{2}\gamma^{-h}e^{-\frac{\tilde c}{16} L_{2}} \\
|\nu_{\omega,\omega,h}| &\leq C|\lambda| \gamma^{\xi h}\;\qquad \text{for $\xi > 0$.}
\end{split} 
\end{equation} 
\begin{remark} Thus, we can allow for a cylinder with finite width $L_{2}$, provided that the width grows at least linearly in the inverse temperature $\beta$. In this range of parameters, the two edge modes living on the two boundaries of $\Lambda_{L}$ are only weakly correlated.
\end{remark}

\paragraph{Single-scale integration.} We now integrate the scale $h$. This is done by writing the propagator (\ref{eq:tildegh}) as:
\begin{equation}
\widetilde{g}^{(\le h)}_{\omega}({\bm q}) = g^{(\le h-1)}_{\omega}({\bm q}) + g^{(h)}_{\omega}({\bm q})\;,
\end{equation}
where:
\begin{equation}\label{eq:proph}
\begin{split}
g^{(\le h-1)}_{\omega}({\bm q}) &:= \frac{\chi^{(\le h-1)}_{\omega}({\bm q})}{i v_{0,\omega,h} q_{0} + v_{1,\omega,h}q_{1} + r_{\omega}(q_{1})}\;, \\
g^{(h)}_{\omega}({\bm q}) &:= \frac{f^{(h)}_{\omega}({\bm q})}{i v_{0,\omega,h}({\bm q}) q_{0} + v_{1,\omega,h}({\bm q}) q_{1} + r_{\omega}(q_{1})}\;,
\end{split}
\end{equation}
with $\chi^{(\le h-1)}_{\omega}(\cdot)$ defined as after (\ref{eq:ind2}) with $h$ replaced by $h-1$, and 
\begin{equation}
f^{(h)}_{\omega}({\bm q}) = \chi^{(\le h)}_{\omega}({\bm q}) - \chi^{(\le h-1)}_{\omega}({\bm q})\;.
\end{equation}
Observe that, at the denominator of $g^{(\le h-1)}_{\omega}({\bm q})$, the functions $v_{\nu,\omega,h}({\bm q})$ have been replaced by the real constants $v_{\nu,\omega,h}$, recall Eq. (\ref{eq:rcc}); this is due to the fact that the condition (\ref{eq:supph}) is satisfied in the support of $\chi^{(\le h-1)}_{\omega}(\cdot)$, up to possibly choosing a larger value of $\gamma$ once and for all.

In constrast to $g^{(\le h-1)}$, the single-scale propagator $g^{(h)}$ is bounded; it satisfies the estimate:
\begin{equation}\label{eq:singlescale}
\Big| \text{d}_{q_{0}}^{n_{0}} \text{d}_{q_{1}}^{n_{1}} g^{(h)}_{\omega}({\bm q}) \Big| \le C_{n_{0}, n_{1}} \gamma^{-h (1 + n_{0} + n_{1})}\;.
\end{equation}
Also, the following property holds true:
\begin{equation}\label{eq:realg}
\overline{g^{(h)}_{\omega}({\bm q})} = g^{(h)}_{\omega}\big((-q_{0}, q_{1})\big)\;;
\end{equation}
this will be needed when checking (\ref{eq:realh}) on scale $h-1$.

We are now ready to integrate the scale $h$. We represent the Grassmann field as $\psi^{(\le h)} = \psi^{(\le h-1)} + \psi^{(h)}$, where $\psi^{(\le h-1)}$, $\psi^{(h)}$ are independent fields with propagators given by (\ref{eq:proph}). We then rewrite the partition function (\ref{eq:ind}) as:
\begin{equation}\label{eq:hint}
\begin{split}
\frac{\ms{Z}}{\ms{Z}^{0}} &= z_{(> h)} \tilde z_{h} \\&\quad \cdot \int \mathbb{P}_{(\le h-1)} (d\psi^{(\le h-1)}) \int \mathbb{P}_{(h)}(d\psi^{(h)}) e^{(\mf L_{0} + \mf L_{1;\text{od}} + \mf R) \mathcal{V}^{(h)}(\psi^{(\le h-1)} + \psi^{(h)})} \\
&= z_{(> h)} \tilde z_{h} z_{h} \int \mathbb{P}_{(\le h-1)} (d\psi^{(\le h-1)}) e^{\mathcal{V}^{(h-1)}(\psi^{(\le h-1)})}\;,
\end{split}
\end{equation}
for a new effective potential $\mathcal{V}^{(h-1)}(\psi^{(\le h-1)})$ which is obtained via a cumulant expansion associated with the Gaussian integration over the $\psi^{(h)}$ field, and where $z_{h}$ satisfies $|z_{h}| \leq \text{exp}(\beta L_{1} C|\lambda| \gamma^{h})$, as it will follow from the forthcoming analysis. The expression obtained in (\ref{eq:hint}) has the form of (\ref{eq:ind}) with $h$ replaced by $h-1$, after defining $z_{(>h-1)} := z_{(> h)} \tilde z_{h} z_{h}$.

In order to prove the inductive assumption stated at the beginning of this section, and to actually derive an explicit representation for the effective potential at all scales, we will introduce a graphical representation for the effective potential at a given scale in terms of a hierarchy of {\it clusters}. This will be the goal of the next section.

\subsubsection{Chain expansion}\label{sec:chain}

The goal of this section will be to introduce a graphical representation for the effective potential at all scales, which will be useful to prove estimates for the kernels of the potential, and in particular to prove Eq. (\ref{eq:Vindh}).

Being the Grassmann theory quadratic in the fermions, the graphical representation will involve chain graphs, as in Section \ref{sec:massive}. We will express the kernel of the effective potential on scale $h-1$ as:
\begin{equation}\label{eq:Vhtheta}
    V^{(h-1)}_{n;\omega,\omega'}({\bm q}) = \sum_{s \ge 1} \bigg(\frac{1}{s!} \sum_{\theta \in {\msc T}_{n;s, \omega,\omega'}^{(h)}}\theta({\bm q})\bigg)\;,
\end{equation}
where ${\msc T}_{n;s, \omega,\omega'}^{(h)}$ is the set of labelled chains on scale $h$, to be defined, and $\theta({\bm q})$ is the value associated to each chain with external quasi-momentum ${\bm q}$. Let us explain how the graphs in ${\msc T}_{n;s, \omega,\omega'}^{(h)}$ are defined. After the integration of the scale $h$, the effective potential is expressed as a sum of chain graphs similarly to what we discussed in (\ref{eq:Vedef}). The definition of the values of the chains proceeds as after (\ref{eq:Vedef}), with the following differences. To every node $v$, $v = 1,\ldots, s$, we attach a frequency label $n_{v}$, and two labels, $\omega_{v}$ and $\omega'_{v}$, associated respectively with the outgoing and the incoming line to the node; we impose the constraint $\omega'_{s} = \omega'$ and $\omega_{1} = \omega$. The edge outgoing from a node $v$ is associated with a single scale propagator,
\begin{equation}
g^{(h)}_{\omega_{v}}\big({\bm q} + {\bm n}(v)\big)\;,
\end{equation}
where:
\begin{equation}\label{eq:nv}
{\bm n}(v) := \sum_{v'\le v} \Big(n_{v'} {\bm \alpha} - \delta_{n_{v'}\neq 0}\big({\bm k}_{F}^{\omega_{v}}(\lambda) - {\bm k}_{F}^{\omega'_{v}}(\lambda)\big)\Big)
\end{equation}
has the interpretation of quasi-momentum outgoing a vertex $v$. The sum over all the quasi-momentum shifts satisfies ${\bm n}(v) = n {\bm \alpha} - \delta_{n\neq 0}({\bm k}_{F}^{\omega}(\lambda) - {\bm k}_{F}^{\omega'}(\lambda))$. To every node of the chain graph we associate:
\begin{equation}\label{eq:defNv}
N^{(h)}_{n_{v}; \omega_{v}, \omega'_{v}}\big({\bm q} + {\bm n}'(v)\big) := (\mf L_{0} + \mf L_{1;\text{od}} + \mf R) V^{(h)}_{n_{v};\omega_{v},\omega'_{v}}\big({\bm q} + {\bm n}'(v)\big)
\end{equation}
where ${\bm n}'(v) = {\bm n}(v-1)$: the momentum ${\bm q} + {\bm n}'(v)$ has the interpretation of incoming quasi-momentum in the node $v$. For later convenience, we shall also set ${\bm n}(v) = {\bm n}'(v) + \delta{\bm n}(v)$.

Thus, we can represent $V^{(h-1)}_{n;\omega,\omega'}$ as the sum of chain graphs with the following values (compare with (\ref{eq:theta})):
\begin{equation}\label{eq:newtrees}
\theta({\bm q}) = \Big(\prod_{v=1}^s N^{(h)}_{n_{v}; \omega_{v}, \omega'_{v}}\big({\bm q} + {\bm n}'(v)\big) \Big) \Big(\prod_{v=1}^{s-1}g^{(h)}_{\omega_{v}}\big({\bm q} + {\bm n}(v)\big)\Big)\;.
\end{equation}
This representation is not explicit enough to prove useful estimates for the effective potential. However, we can use this expression, together with the assumption (\ref{eq:realh}) on scales greater or equal than $h$, to check (\ref{eq:realh}) on scale $h-1$. We have, from (\ref{eq:ccLR}) and (\ref{eq:realg}):
\begin{equation}
\begin{split}
\overline{\theta({\bm q})} &= \bigg(\prod_{v=1}^s N^{(h)}_{-n_{v}; \omega'_{v}, \omega_{v}}\Big(-q_{0}, q_{1} + \big({\bm n}'(v) + \delta {\bm n}(v)\big)\Big) \bigg)\\
&\qquad\cdot \Big(\prod_{v=1}^{s-1}g^{(h)}_{\omega_{v}}\big(-q_{0}, q_{1} + {\bm n}(v)\big)\Big) \\
&= \Big(\prod_{v=1}^s N^{(h)}_{-n_{v}; \omega'_{v}, \omega_{v}}\big(-q_{0}, q_{1} + {\bm n}(v)\big) \Big) \Big(\prod_{v=2}^{s}g^{(h)}_{\omega'_{v}}\big(-q_{0}, q_{1} + {\bm n}'(v)\big)\Big)\;,
\end{split}
\end{equation}
where we used that the labels $\omega_{v}$ and $\omega'_{v+1}$ of two edges connected by a propagator must agree. We now relabel the vertices: $\tilde v = s +1 - v$, which reverses the order of the labelling of the chain. We have:
\begin{equation}
\begin{split}
\overline{\theta({\bm q})} &= \Big(\prod_{\tilde v=1}^{s} N^{(h)}_{-n_{\tilde v}; \omega'_{\tilde v}, \omega_{\tilde v}}\big(-q_{0}, q_{1} + {\bm n}(q) - {\bm n}'(\tilde v)\big) \Big) \\&\quad\cdot \Big( \prod_{v=2}^{s}g^{(h)}_{\omega'_{v}}\big(-q_{0}, q_{1} + {\bm n}(q) - {\bm n}(\tilde v)\big)\Big) \;.
\end{split}
\end{equation}
Next, we set $\tilde \omega_{\tilde v} := \omega'_{\tilde v}$, $\tilde \omega'_{\tilde v} := \omega_{\tilde v}$. Let us introduce the change of variables $\tilde n_{\tilde v} := - n_{\tilde v}$, and let us define:
\begin{equation}
\tilde{\bm n}(v) := \sum_{v'\le v} \Big(\tilde n_{v'} {\bm \alpha} - \delta_{n_{v'} \neq 0}\big({\bm k}_{F}^{\tilde\omega_{v}}(\lambda) - {\bm k}_{F}^{\tilde\omega'_{v}}(\lambda)\big)\Big)\;.
\end{equation}
We then get:
\begin{equation}\label{eq:realdim}
\begin{split}
\overline{\theta({\bm q})} &= \Big(\prod_{\tilde v=1}^s N^{(h)}_{\tilde n_{\tilde v}; \tilde \omega_{\tilde v}, \tilde \omega'_{\tilde v}}\big(-q_{0}, q_{1} + {\bm n}(q) + \tilde{{\bm n}}'(\tilde v)\big) \Big) \\&\quad\cdot \Big(\prod_{\tilde v = 1}^{s-1}g^{(h)}_{\tilde \omega'_{\tilde v}}\big(-q_{0}, q_{1} + {\bm n}(q) + \tilde{{\bm n}}(\tilde v)\big)\Big)\;,
\end{split}
\end{equation}
where the constraint now reads $\tilde {\bm n}(q) = -n{\bm \alpha} + \delta_{n\neq 0}({\bm k}_{F}^{\omega}(\lambda) - {\bm k}_{F}^{\omega'}(\lambda))$. Comparing with (\ref{eq:newtrees}), we see that the right-hand side of (\ref{eq:realdim}) is the value of a chain graph in ${\msc T}_{-n;s, \omega',\omega}^{(h)}$ with incoming momentum ${\bm q} + {\bm n}(q)$. Summing over all trees, this shows that:
\begin{equation}
\overline{V^{(h-1)}_{n;\omega,\omega'}({\bm q})} = V^{(h-1)}_{-n;\omega',\omega}\Big(-q_{0}, q_{1} + n\alpha - \delta_{n\neq 0}\big(k_{F}^{\omega}(\lambda) - k_{F}^{\omega'}(\lambda)\big)\Big)\;,
\end{equation}
which proves the validity of (\ref{eq:realh}) on scale $h-1$.

Let us now derive a more explicit expansion for the effective potential. Combining the definition (\ref{eq:defNv}) with the identity (\ref{eq:newtrees}), we can further express $V^{(h-1)}_{n;\omega,\omega'}({\bm q})$ in terms of new chain graphs, whose edges are associated with propagators on all scales from $h$ to $0$; this is done as follows. We write
\begin{equation}\label{eq:Nh+1}
\begin{split}
N^{(h)}_{n_{v}; \omega_{v}, \omega'_{v}}\big({\bm q} + {\bm n}'(v)\big) &= (\mf L_{0} + \mf L_{1;\text{od}}) V^{(h)}_{n_{v};\omega_{v},\omega'_{v}}\big({\bm q} + {\bm n}'(v)\big)\\
&\qquad + \mf R V^{(h)}_{n_{v};\omega_{v},\omega'_{v}}\big({\bm q} + {\bm n}'(v)\big)\;;
\end{split}
\end{equation}
for the first term in the right-hand side, we recall the definitions of Eqs. (\ref{eq:Ldef2}), (\ref{eq:Ldef3}):
\begin{equation}\label{eq:nodesh+1}
\begin{split}
\mf L_{0} V^{(h)}_{n_{v};\omega_{v},\omega'_{v}}\big({\bm q} + {\bm n}'(v)\big) &= \delta_{n_{v},0} \gamma^{h} \nu_{\omega_{v},\omega'_{v},h}\;, \\
\mf L_{1;\text{od}} V^{(h)}_{n_{v};\omega_{v},\omega'_{v}}\big({\bm q} + {\bm n}'(v)\big) &= \delta_{n_{v},0} \delta_{\omega_{v}, -\omega'_{v}} (q_{0} v_{0,\omega_{v},-\omega_{v},h} + q_{1} v_{1,\omega_{v},-\omega_{v},h})\;.
\end{split}
\end{equation}
The values of the running coupling constants are determined recursively, by solving (\ref{eq:beta}). Now, if $ \mf R V^{(h)}_{n_{v};\omega_{v},\omega'_{v}}(\cdot)$ was not present in (\ref{eq:Nh+1}), the expression (\ref{eq:newtrees}) combined with (\ref{eq:nodesh+1}) would give rise to chains whose nodes are labelled by running coupling constants and propagators on scale $h$. We shall call (\ref{eq:nodesh+1}) the {\it nodes on scale $h$}, and we shall denote by $\msc{N}_{h}(\theta)$ the set of nodes on scale $h$ of a tree $\theta$ in $ {\msc T}_{n;s, \omega,\omega'}^{(h)}$. The value of each such node will be denoted by $V_{v}$, and it corresponds to one of the two terms in (\ref{eq:nodesh+1}).

In order to take into account the presence of $\mf R V^{(h)}_{n_{v};\omega_{v},\omega'_{v}}(\cdot)$ in (\ref{eq:Nh+1}), we proceed as follows. Let $\widetilde{\msc{N}}_{h}(\theta)$ be the nodes that are associated with $\mf R V^{(h)}$. We express $V^{(h)}$ at the argument of $\mf R$ in terms of chains involving propagators on scale $g^{(h+1)}$, nodes on scale $h+1$, and kernels $\mf R V^{(h+1)}$. Thus, we can express the value of every chain $\theta\equiv \theta_{h}$ in ${\msc T}_{n;s, \omega,\omega'}^{(h)}$ as a linear combination of:
\begin{equation}\label{eq:thetah}
\Big( \prod_{v \in \msc{N}_{h}(\theta_{h})} V_{v}\Big)  \Big(\prod_{v=1}^{s-1}g^{(h)}_{\omega_{v}}\big({\bm q} + {\bm n}(v)\big)\Big) \bigg( \prod_{v\in \widetilde{\msc{N}}_{h}(\theta_{h})} \mf R\Big( \theta_{h+1}\big({\bm q} + {\bm n}'(v)\big)\Big) \bigg)\;,
\end{equation}
where $\theta_{h+1}$ at the argument of the last product belongs to ${\msc T}_{n_{v};w_{v}, \omega_{v},\omega_{v}'}^{(h+1)}$. Recall the action of $\mf R$, see Definition \ref{def:loc}:
\begin{equation}\label{eq:Rtheta1}
\mf R\big( \theta({\bm q})\big) = \theta({\bm q})\;,\qquad \text{if $\theta \in {\msc T}_{n_{v};s_{v}, \omega_{v},\omega_{v}'}^{(h+1)}$, with $n_{v}\neq 0$},
\end{equation}
and:
\begin{equation}\label{eq:Rtheta2}
\mf R\big( \theta({\bm q})\big) = \Big( 1 - \mathbbm{1}\Big|_{{\bm q} = {\bm 0}_{\beta}} - q_{0} \text{d}_{q_{0}}\Big|_{{\bm q} = {\bm 0}_{\beta}} - q_{1} \text{d}_{q_{1}}\Big|_{{\bm q} = {\bm 0}_{\beta}} \Big) \theta({\bm q})\;,\qquad \text{otherwise.}
\end{equation}
Also, by Remark \ref{rem:loc}, it is important to recall that $\mf R$ acts as a projection: $\mf R^{2} = \mf R$, which avoids the (potentially dangerous) accumulation of derivatives.

Iterating the above argument, we end up with chain graphs involving nodes and propagators from scales $h$ to $0$. Let us introduce the necessary concepts in order to describe the resulting expansion.
\begin{itemize}
\item[(i)] The chain is formed by a sequence of nodes connected by edges. Each node $v$ is labelled by a scale label $h_{v} \ge h$, by a frequency label $n_{v}$, by quasi-particle labels $(\omega_{v}, \omega'_{v})$. The nodes are decorated by an incoming and an outgoing line; the incoming line is decorated by the quasi-particle label $\omega'_{v}$, the outgoing line is decorated by the quasi-particle label $\omega_{v}$. The chain is formed joining lines attached to nodes, with opposite orientations and with coinciding quasi-particle labels.
\item[(ii)] The nodes with scale $h_{v} < 0$ correspond to $n_{v} = 0$. They correspond to the running coupling constants (\ref{eq:nodesh+1}). The nodes with $h_{v} = 0$ might have $n_{v} = 0$ or $n_{v} \neq 0$. If they are labelled by $n_{v} \neq 0$, they correspond to ${\mf R}V^{(0)}_{n_{v}; \omega_{v}, \omega'_{v}}(\cdot )$. If they are labelled by $n_{v} = 0$, they correspond to either the running coupling constant (\ref{eq:nodesh+1}), or to ${\mf R} V^{(0)}_{n_{v}; \omega_{v}, \omega'_{v}}(\cdot)$. 
\item[(iii)] The nodes associated with $n_{v} = 0$ are called {\it resonant}, while the nodes associated with $n_{v}\neq 0$ are called {\it non-resonant}. We denote by $\msc{N}(\theta)$ the set of nodes, by $\msc{N}_{\text{r}}(\theta)$ the set of resonant nodes, and by $\msc{N}_{\text{nr}}(\theta)$ the set of non-resonant nodes.
\item[(iv)] Two consecutive nodes $v, v+1$ are connected by an edge $e$, corresponding to a propagator $g^{(h_{e})}_{\omega_{v}}$. We denote by $\msc{E}(\theta)$ the set of all edges of $\theta$. The scale $h_{e}$ is equal to the minimum of $h_{v}$, $h_{v+1}$. The incoming momentum of the whole chain is ${\bm q}$, and the outgoing momentum of the chain is ${\bm q} + {\bm n}(|\msc{N}(\theta)|)$ with ${\bm n}(v)$ as in (\ref{eq:nv}). Each propagator supports a momentum which is determined by the incoming momentum plus the momentum shift determined by the frequency labels of all vertices preceding the edge (Kirchhoff rule): if the edge $e$ connects $v$ to $v+1$, then the corresponding propagator supports the momentum ${\bm q}_{e} = {\bm q} + {\bm n}(v)$.
\item[(v)] A {\it cluster} $T$ is a connected segment contained in the chain $\theta$, so that the incoming line and the outgoing line of $T$ have scale labels strictly smaller than the scales contained in the cluster. We use the convention that clusters must contain at least one edge: that is, nodes are not clusters. We denote by $\msc{C}(\theta)$ the set of clusters of $\theta$, and we include $\theta$ in the set of clusters. We denote by $h^{\text{ext}}_{T}$ is the maximal scale of the external lines of $T$ (the {\it external scale} of $T$), and by $h_{T}$ the smallest scale contained in $T$ (the {\it scale} of $T$); as a consequence of the above definitions, $h^{\text{ext}}_{T} < h_{T}$.
\item[(vi)] A cluster is {\it resonant} if the incoming and the outgoing momenta of $T$ are equal. Otherwise, the cluster is called {\it non-resonant.} We denote by $\msc{C}_{\text{r}}(\theta)$ the set of resonant clusters of $\theta$, and by $\msc{C}_{\text{nr}}(\theta)$ the set of non-resonant clusters of $\theta$. A cluster $T'\subset T$ is {\it maximal} if it is not contained into any other subcluster of $T$. Similarly, a node $v\in T$ is called maximal if it is not contained into any other subcluster of $T$.  We call $s^{\text{c}}_{T}$ the number of maximal clusters contained in $T$, $s^{\text{v}}_{T}$ the number of maximal nodes contained in $T$, and we set $s_{T} = s^{\text{c}}_{T} + s^{\text{v}}_{T}$. We denote by $s_{T}^{\text{nr}}$ the number of maximal non-resonant clusters or nodes contained in $T$, and by $s^{\text{r}}_{T}$ the number of maximal resonant clusters or nodes contained in $T$.
%
%\item[(vii)] We denote by $e(T)$ the number of edges contained in the cluster $T$ and not contained into any smaller subcluster of $T$. By item (v), the scale of such edges is $h_{T}$.
\end{itemize}
\begin{figure}
    \centering
    \includegraphics[scale=0.65]{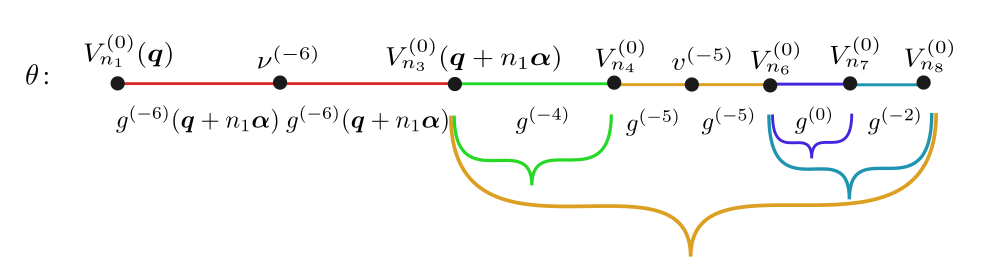}
    \caption{An example of a tree with $8$ nodes and its clusters structure. For simplicity we omitted the $\omega$ labels and all arguments excepts a few ones. Each clusters is represented by a curly bracket. }
    \label{fig:clusters}
\end{figure}

The next proposition is the main technical result of the section.
\begin{proposition}[Bound for chain graphs.]\label{prp:chainbd} Let $\theta\in {\msc T}_{n;s, \omega,\omega'}^{(h)}$. Then, assuming the validity of (\ref{eq:indrcc}) the following estimate holds true:
\begin{equation}\label{eq:chainbd}
\begin{split}
| \text{d}_{0}^{n_{0}} \text{d}_{1}^{n_{1}} \theta({\bm q}) | &\le C_{n_{0},n_{1}}\gamma^{h(1 - n_{0} - n_{1})} K^{|\msc{N}(\theta)|} |\lambda|^{|\msc{N}_{\text{nr}}(\theta)|} \Big( \prod_{v\in \msc{N}(\theta)} e^{-\frac{c}{2}|n_{v}|} \Big) \\
&\quad\cdot \Big( \prod_{v\in \msc{N}_{\text{r}}(\theta)} \zeta_{h_{v}} \Big)\Big(\prod_{v\in \msc{N}_{\text{nr}}(\theta)} e^{- \frac{\tilde c}{4} \delta_{\omega_{v},-\omega_{v}} L_{2}}  \Big)  \Big(\prod_{T\in \msc{C}(\theta)} \gamma^{ (h^{\text{ext}}_{T} - h_{T})}\Big)\;,
\end{split}
\end{equation}
where $\zeta_{h_{v}}$ denotes the absolute of $\nu_{\omega_{v},\omega'_{v},h_{v}}$ or of $v_{\nu,\omega_{v},-\omega_{v},,h_{v}}\delta_{\omega'_{v},-\omega_{v}}$.
\end{proposition}
The bound (\ref{eq:chainbd}) is summable over the scale labels, and over the frequency labels. In fact:
\begin{equation}\label{eq:chain2V}
\sum_{\{n_{v}\}: \sum_{v} n_{v} = n}\Big( \prod_{v\in \msc{N}(\theta)} e^{-\frac{c}{2}|n_{v}|}\Big) \le C^{|\msc{N}(\theta)|} e^{-\frac{c}{4} |n|}\;,
\end{equation}
and:
\begin{equation}
\sum_{\{h_{T}\}: h^{\text{ext}}_{\theta} = h} \Big(\prod_{T\in \msc{C}(\theta)} \gamma^{ (h^{\text{ext}}_{T} - h_{T})}\Big) \le 1\;.
\end{equation}
Concerning the combinatorics for a given cluster $T$, the number of chain graphs obtained connecting $s_{T}$ maximal clusters or nodes is $s_{T}!$. On the other hand, following the iterative construction of the chain graphs discussed starting from (\ref{eq:Vhtheta}), we see that the value of the resulting graphs is accompanied by a prefactor:
\begin{equation}
\prod_{T\subseteq \theta} \frac{1}{s_{T}!}\;.
\end{equation}
Thus, the overall combinatorial factor is $1$. In conclusion, from (\ref{eq:Vhtheta}) we get, up to a redefinition of the overall constant and for $|\lambda|$ small enough:
\begin{equation}
| \text{d}_{0}^{n_{0}} \text{d}_{1}^{n_{1}} V^{(h-1)}_{n;\omega,\omega'}({\bm q})| \le C_{n_{0},n_{1}}|\lambda| \gamma^{(h-1)(1 - n_{0} - n_{1})} e^{-\frac{c}{4} |n|} e^{- \frac{\tilde c}{4} \delta_{\omega,-\omega'} L_{2}}\;,
\end{equation}
which is the claimed bound (\ref{eq:Vindh}) on scale $h-1$.
\begin{proof}(of Proposition \ref{prp:chainbd}) To begin, it is useful to first introduce the {\it unrenormalized value} $\theta_{\text{nr}}$, as follows:
\begin{equation}
\begin{split}
\theta_{\text{nr}}({\bm q}) &= \Big(\prod_{v \in \msc{N}_{\text{r}}(\theta)} V_{v}\Big) \Big(\prod_{v\in \msc{N}_{\text{nr}}(\theta)} {\mf R}V^{(0)}_{n_{v}; \omega_{v},\omega'_{v}}\big({\bm q} + {\bm n}'(v)\big) \Big) \\&\qquad \cdot \Big(\prod_{e\in \msc{E}(\theta)} g^{(h_{e})}_{\omega_{e}}\big({\bm q} + \tilde {\bm n}(e)\big)\Big)\;,
\end{split}
\end{equation}
where if $e = (v, v+1)$, $\tilde {\bm n}(e) = {\bm n}(v)$. This expression is generated iterating (\ref{eq:thetah}) neglecting the ${\mf R}$ operators acting on $\theta_{h+1}$. Let us first derive an estimate for the unrenormalized value, and then discuss the actual estimate for $\theta({\bm q})$, taking into account the action of ${\mf R}$. 

Assume that the estimates for the running coupling constants (\ref{eq:indrcc}) hold true on scales greater or equal than $h$, and recall the bound (\ref{eq:thetah}) for the single-scale propagator and the bound (\ref{eq:V0}) for the kernels on scale $0$. Then, we have:
\begin{equation}\label{eq:thetanr}
\begin{split}
| \theta_{\text{nr}}({\bm q}) | &\le \Big(\prod_{v \in \msc{N}_{\text{r}}(\theta)} \gamma^{h_{v}} \zeta_{h_{v}}\Big) \Big(\prod_{v\in \msc{N}_{\text{nr}}(\theta)} C|\lambda| e^{-c|n_{v}|} e^{-\frac{\tilde c}{2} \delta_{\omega_{v},-\omega_{v}} L_{2}}  \Big) \\
&\qquad \cdot \Big(\prod_{e\in \msc{E}(\theta)} C\gamma^{-h_{e}}\Big)\;.
\end{split}
\end{equation}
The sum over all trees implies a sum over all scale labels. Under the sole constraint that if $e$ joins $v$ to $v+1$ then $h_{e} = \min(h_{v}, h_{v+1})$, it is clear that the bound (\ref{eq:thetanr}) cannot be used to prove the bound (\ref{eq:Vindh}) for the effective potential on scale $h$. 

In order to improve the situation, we have to take into account the action of the ${\mf R}$ operator in (\ref{eq:thetah}), on a general $\theta_{h+1}({\bm q})$. The action is nontrivial only if $\theta_{h+1}({\bm q})$ is associated with a resonant cluster. It is convenient to define the remainder operator:
\begin{equation}\label{eq:Rexp}
\text{Rem}_{{\bm p}, {\bm q}; 2} F({\bm q}) := \Big( 1 - \mathbbm{1}\Big|_{{\bm q} = {\bm 0}_{\beta}} - p_{0} \text{d}_{q_{0}}\Big|_{{\bm q} = {\bm 0}_{\beta}} - p_{1} \text{d}_{q_{1}}\Big|_{{\bm q} = {\bm 0}_{\beta}} \Big) F({\bm q}) \;;
\end{equation}
we recall that evaluating a function at ${\bm 0}_{\beta}$ means (\ref{eq:sym}). Thus, 
\begin{equation}
{\mf R} \theta_{h+1}({\bm q}) = \text{Rem}_{{\bm q}, {\bm q}; 2} \theta_{h+1}({\bm q})\;.
\end{equation}
The operator $\text{Rem}_{{\bm q}, {\bm q}; 2}$ could either act on $g^{(h+1)}({\bm q} + {\bm r})$, or on ${\mf R}(\theta_{h+2}({\bm q} + {\bm r}))$, with ${\bm r}$ a generic momentum determined by Kirchhoff rule. Suppose it acts on $g^{(h+1)}$; then, we are left with estimating:
\begin{equation}
\text{Rem}_{{\bm q}, {\bm q}; 2}\, g^{(k)}_{\omega}({\bm q} + {\bm r})\;,
\end{equation}
with ${\bm q}$ on scale $h$ and ${\bm r}$ on scale $k = h+1$. By interpolation, we get:
\begin{equation}\label{eq:Rg}
\Big| \text{Rem}_{{\bm q}, {\bm q}; 2}\, g^{(k)}_{\omega}({\bm q} + {\bm r})\Big|_{{\bm q} = {\bm 0}_{\beta}} \Big| \le C \gamma^{2h} \gamma^{-3k} = C\gamma^{2(h-k)} \gamma^{-k}\;;
\end{equation}
that is, we gain a factor $\gamma^{2(h-k)}$ with respect to the dimensional bound (\ref{eq:singlescale}) for the single-scale propagator on scale $k = h+1$. In this bound, observe that $\gamma^{h}$ is the external scale of the cluster $T$ associated with $\theta_{h+1}$, while $\gamma^{k}$ with $k=h+1$ is the scale of the cluster; thus, the factor $\gamma^{2(h - k)}$ in (\ref{eq:Rg}) is equivalent to a factor $\gamma^{2(h^{\text{ext}}_{T} - h_{T})}$, recall item (v) in the previous definitions.

The other possibility is that the operator (\ref{eq:Rexp}) acts on ${\mf R}\theta_{h+2}({\bm q} + {\bm r})$. If $\theta_{h+2}$ is associated with a non-resonant cluster, then the action is trivial. If it is resonant, ${\mf R}\theta_{h+2}({\bm q} + {\bm r})$ is:
\begin{equation}\label{eq:Rh+1}
\text{Rem}_{{\bm q} + {\bm r}, {\bm q} + {\bm r}; 2}\, \theta_{h+2}({\bm q} + {\bm r})\;,
\end{equation}
and the action of $\text{Rem}_{{\bm q}, {\bm q}; 2}$ on (\ref{eq:Rh+1}) gives:
\begin{equation}\label{eq:comp}
\text{Rem}_{{\bm q}, {\bm q}; 2}\, \theta_{h+2}({\bm q} + {\bm r})\;.
\end{equation}
That is, derivatives do not accumulate. We now express $\theta_{h+2}$ in terms of propagators on scale $h+2$ and ${\mf R}\theta_{h+3}$. If the remainder operator acts on $g^{(h+2)}$, we use (\ref{eq:Rg}), to gain a factor $\gamma^{(2(h - (h-2)))}$ with respect to the dimensional bound of the single-scale propagator; this bound can be viewed as a factor $\gamma^{2(h^{\text{ext}}_{T} - h_{T})} \cdot \gamma^{2(h^{\text{ext}}_{T'} - h_{T'})}$ where $T$ is the cluster associated with $\theta_{h+1}$ and $T' \subset T$ is the cluster associated with $\theta_{h+2}$. If instead the derivatives act on ${\mf R}\theta_{h+3}$ we iterate the argument, until scale zero. 

It is now clear how to quantify the action of the ${\mf R}$ operators on the resonant clusters. Observe that the estimate (\ref{eq:thetanr}) can be rewritten as:
\begin{equation}
\begin{split}
| \theta_{\text{nr}}({\bm q}) | &\le \Big(\prod_{v \in \msc{N}_{\text{r}}(\theta)}  \gamma^{h_{v}} \zeta_{h_{v}}\Big) \Big(\prod_{v\in \msc{N}_{\text{nr}}(\theta)} C|\lambda| e^{-c|n_{v}|} e^{-\frac{\tilde c}{2} \delta_{\omega_{v},-\omega_{v}} L_{2}}  \Big) \\
&\qquad \cdot \Big(\prod_{T\in \msc{C}(\theta)} \gamma^{-h_{T} (s_{T} - 1)}\Big)\;.
\end{split}
\end{equation}
From the above discussion, we see that the action of the ${\mf R}$ introduces a factor 
\begin{equation}
\prod_{T \in \msc{C}_{\text{r}}(\theta)} \gamma^{2(h^{\text{ext}}_{T} - h_{T})}\;,
\end{equation}
in the bound for the chain (recall that the action of ${\mf R}$ is nontrivial only when the cluster is resonant, hence the constraint in the product). Thus, the bound for the renormalized chain is:
\begin{equation}\label{eq:esttheta}
\begin{split}
| \theta({\bm q}) | &\le C^{|\msc{N}(\theta)|} \Big(\prod_{v \in \msc{N}_{\text{r}}(\theta)}  \gamma^{h_{v}} \zeta_{h_{v}}\Big) \Big(\prod_{v\in \msc{N}_{\text{nr}}(\theta)} C|\lambda| e^{-c|n_{v}|} e^{-\frac{\tilde c}{2} \delta_{\omega_{v},-\omega_{v}} L_{2}}  \Big) \\
&\qquad \cdot \Big(\prod_{T\in \msc{C}(\theta)} \gamma^{-h_{T} (s_{T} - 1)}\Big) \Big(\prod_{T\in \msc{C}_{\text{r}}(\theta)} \gamma^{2 (h^{\text{ext}}_{T} - h_{T})}\Big)\;,
\end{split}
\end{equation}
where the factor $C^{|\msc{N}(\theta)|}$ bounds the number of ways to distribute the derivatives coming from the ${\mf R}$ operators.\footnote{Whenever ${\mf R}$ acts on a non-resonant clusters, we have a combinatorial factor produced by the action of ${\mf R}$ on the propagators, plus a combinatorial factor produced by the action of ${\mf R}$ on the resonant maximal subclusters. For each cluster, we have $s_{T} - 1$ propagators and at most $s_{T}$ resonant maximal subclusters, on which an operator ${\mf R}$ is already applied (we call them renormalized clusters). Thus, the number of ways to distribute the derivatives in a cluster $T$ is $2s_{T} - 1$. Using the property (\ref{eq:comp}), whenever ${\mf R}$ acts on the renormalized maximal subclusters, we get that the overall combinatorial factor produced by the ${\mf R}$ operators is bounded as:
\begin{equation}
\prod_{T \subseteq \theta} (2s_{T} -  1) = e^{\sum_{T\subseteq \theta} \log(2s_{T} - 1)} \le e^{\sum_{T\subseteq \theta} (2s_{T} - 1)} \le e^{3\sum_{T\subseteq \theta} (s_{T} - 1)}\;.
\end{equation}
Then, since $\sum_{T\subseteq \theta} (s_{T} - 1) = | \msc{N}(\theta)| - 1$, we obtain that the combinatorial factor produced by the action of the ${\mf R}$ operators is bounded by $C^{| \msc{N}(\theta)| }$, as claimed after (\ref{eq:esttheta}).
}
To better understand the right-hand side of (\ref{eq:esttheta}), let us write:
\begin{equation}\label{eq:split}
\prod_{T\in \msc{C}(\theta)} \gamma^{-h_{T} (s_{T} - 1)} = \Big(\prod_{T\in \msc{C}(\theta)} \gamma^{-h_{T} (s^{\text{r}}_{T} - 1)}\Big) \Big(\prod_{T\in \msc{C}(\theta)} \gamma^{-h_{T} s^{\text{nr}}_{T}}\Big)\;,
\end{equation}
where we recall that $s_{T}^{\text{r}}$ is the number of maximal resonant clusters or vertices contained in $T$, and $s_{T}^{\text{nr}}$ is the number of non-resonant maximal clusters or vertices contained in $T$.  We will control the first factor in the right-hand side of (\ref{eq:split}) using part of the last factor in (\ref{eq:esttheta}). In fact, we have, calling $s^{\text{c,r}}_{T}$ the number of maximal resonant clusters contained in $T$:
\begin{equation}
\begin{split}
&\Big(\prod_{T\in \msc{C}(\theta)} \gamma^{-h_{T} (s^{\text{r}}_{T} - 1)}\Big) \Big(\prod_{T\in \msc{C}_{\text{r}}(\theta)} \gamma^{(h^{\text{ext}}_{T} - h_{T})}\Big) \\
&\quad \le \Big(\prod_{T\in \msc{C}(\theta)} \gamma^{-h_{T} s^{\text{r}}_{T}}\Big) \Big(\prod_{T\in \msc{C}_{\text{r}}(\theta)} \gamma^{h^{\text{ext}}_{T}}\Big) \\
&\quad = \gamma^{\varepsilon(\theta) h} \Big(\prod_{T\in \msc{C}(\theta)} \gamma^{-h_{T} s^{\text{r}}_{T}}\Big) \Big(\prod_{T\in \msc{C}(\theta)} \gamma^{h_{T} s^{\text{c,r}}_{T}}\Big) \\
&\quad = \gamma^{\varepsilon(\theta) h} \prod_{v\in \msc{N}_{\text{r}}(\theta)} \gamma^{-h_{v}}\;,
\end{split}
\end{equation}
where $\varepsilon(\theta) = 1$ if $\theta$ is resonant and $\varepsilon(\theta) = 0$ otherwise. Thus, plugging this bound in (\ref{eq:esttheta}) we obtain:
\begin{equation}\label{eq:thetabd3}
\begin{split}
&| \theta({\bm q}) | \le \gamma^{\varepsilon(\theta) h} K^{|\msc{N}(\theta)|} |\lambda|^{|\msc{N}_{\text{nr}}(\theta)|}  \Big(\prod_{v \in \msc{N}_{\text{r}}(\theta)} \zeta_{h_{v}}\Big) \\
&\cdot \Big(\prod_{v\in \msc{N}_{\text{nr}}(\theta)} e^{-c|n_{v}|} e^{-\frac{\tilde c}{2} \delta_{\omega_{v},-\omega_{v}} L_{2}}  \Big) \Big(\prod_{T\in \msc{C}(\theta)} \gamma^{-h_{T} s^{\text{nr}}_{T}}\Big) \Big(\prod_{T\in \msc{C}_{\text{r}}(\theta)} \gamma^{ (h^{\text{ext}}_{T} - h_{T})}\Big)\;.
\end{split}
\end{equation}
Hence, thanks to the renormalization procedure, the sum over the scale labels of the resonant clusters can be controlled by the last factor in (\ref{eq:thetabd3}). 
\begin{remark}
The above estimate makes it clear why the localization operation is defined up to order $1$ in the Taylor expansion in the momenta, recall Definition \ref{def:loc}: it is essential that the prefactor of $(h^{\text{ext}}_{T} - h_{T})$ in $\gamma^{2 (h^{\text{ext}}_{T} - h_{T})}$ in the bound (\ref{eq:esttheta}) is bigger than $1$, in order to control the sum over the scale labels of the resonant clusters, as we did above. 
\end{remark}
We are left with discussing the dependence on the scale labels associated with non-resonant clusters. Here, a key role will be played by the approximate Diophantine property (\ref{eq:diophantine}); it will allow us to use part of the second factor in the second line of (\ref{eq:thetabd3}) to overcompensate the bad third factor in (\ref{eq:thetabd3}). Let us define $\msc{N}_{\text{nr}}^{\text{d}}(\theta)$ as the set of diagonal non-resonant nodes, that is with the same $\omega_{v}, \omega'_{v}$ nodes, and $\msc{N}_{\text{nr}}^{\text{od}}(\theta)$ as the set of off-diagonal non-resonant nodes, that is with different $\omega_{v}, \omega'_{v}$ labels. We start by writing:
\begin{equation}
\Big(\prod_{v\in \msc{N}_{\text{nr}}^{\text{od}}(\theta)}e^{-\frac{\tilde c}{4} L_{2}}\Big) \Big(\prod_{v\in \msc{N}^{\text{d}}_{\text{nr}}(\theta)} e^{-\frac{c}{2}|n_{v}|}\Big) = \prod_{v\in \msc{N}_{\text{nr}}(\theta)} e^{-a_{v}}
\end{equation}
with 
\begin{equation}
a_{v} := (c/2) |n_{v}| \quad \text{if $v\in \msc{N}_{\text{nr}}^{\text{d}}(\theta)$} \quad \text{and}\quad  a_{v} = (\tilde c /4) L_{2}\quad  \text{if $v\in \msc{N}_{\text{nr}}^{\text{od}}(\theta)$.} 
\end{equation}
Then,
\begin{equation}\label{eq:diof1}
\begin{split}
\prod_{v\in \msc{N}_{\text{nr}}(\theta)} e^{-a_{v}} &= \prod_{v\in \msc{N}_{\text{nr}}(\theta)} \prod_{k = -\infty}^{0} e^{-2^{k} a_{v}} \\
&\le \prod_{k = -\infty}^{0} \prod_{\substack{v\in \msc{N}_{\text{nr}}(\theta) \\ v\in T\, \text{s.t.}\, h^{\text{ext}}_{T} = k}} e^{-2^{h^{\text{ext}}_{T}} a_{v} } \\
&\le \prod_{T \in \msc{C}_{\text{nr}}(\theta) \cup \msc{N}_{\text{nr}}(\theta)} e^{-2^{h^{\text{ext}}_{T}}a_{T}}
\end{split}
\end{equation}
where we set $a_{T} := \sum_{v\in T} a_{v}$. In the last inequality of (\ref{eq:diof1}), the product runs over non-resonant clusters and non-resonant endpoints, and with a slight abuse of notation we use the symbol $T$ to denote the elements of both sets. The presence of the non-resonant endpoints in the last product is due to the contributions of the clusters $T$ with external scale zero, which are by definition endpoints. Given a cluster $T$ in the last product in (\ref{eq:diof1}), suppose that $T$ contains at least one non-resonant off-diagonal node. Then:
\begin{equation}\label{eq:aT1}
\begin{split}
e^{-2^{h^{\text{ext}}_{T}}a_{T}} &\le e^{-2^{h^{\text{ext}}_{T}} \frac{\tilde c}{4} L_{2}} \\
&\le e^{-2^{h^{\text{ext}}_{T}} \tilde C \gamma^{-h^{\text{ext}}_{T}}}
\end{split}
\end{equation}
where we used that $\kappa L_{2} \ge \beta \ge C\gamma^{-h^{\text{ext}}_{T}}$. 

Next, suppose that $T$ does not contain any non-resonant off-diagonal  node. The difference between the incoming and the outgoing momenta of the cluster/node $T$ is bounded by $C \gamma^{h^{\text{ext}}_{T}}$. Since no off-diagonal node is present, we derive the following constraint for the sum over the momentum shifts associated with the diagonal non-resonant nodes:
\begin{equation}
\bigg| \sum_{v\in T \cap \msc{N}^{\text{d}}_{\text{nr}}(\theta)} n_{v} \alpha \bigg| \le C\gamma^{h^{\text{ext}}_{T}}\;.
\end{equation}
Let $n_{T}^{\text{d}} := \sum_{v\in T \cap \msc{N}_{\text{nr}}^{\text{d}}} n_{v}$. Being the cluster $T$ non-resonant, we must have $n_{T}^{\text{d}} \neq 0$. Suppose that $| n^{\text{d}}_{T} | \le L_{1} / 2$. Then, by the Diophantine condition (\ref{eq:diophantine}) we have:
\begin{equation}
\frac{c}{| n^{\text{d}}_{T} |^{\tau}} \le \bigg| \sum_{v\in T \cap \msc{N}^{\text{d}}_{\text{nr}}(\theta)} n_{v} \alpha \bigg|_{\mathbb{T}} \le C \gamma^{h^{\text{ext}}_{T}}\;,
\end{equation}
or equivalently:
\begin{equation}
|n^{\text{d}}_{T}| \ge K\gamma^{-h^{\text{ext}}_{T} / \tau}\;.
\end{equation}
Therefore, we get:
\begin{equation}\label{eq:aT2}
e^{-2^{h^{\text{ext}}_{T}}a_{T}} \le e^{-2^{h^{\text{ext}}_{T}} K \gamma^{-h^{\text{ext}}_{T} / \tau}}\;.
\end{equation}
Finally, suppose that $| n^{\text{d}}_{T} | > L_{1} / 2$. In this case, we cannot use the approximate Diophantine condition. We estimate:
\begin{equation}\label{eq:aT3}
\begin{split}
e^{-2^{h^{\text{ext}}_{T}} a_{T}} &\le e^{-2^{h^{\text{ext}}_{T}}\frac{c}{8} L_{1}} \\
&\le  e^{-2^{h^{\text{ext}}_{T}}C \gamma^{-h_{\beta}}} \\
&\le e^{-2^{h^{\text{ext}}_{T}}C \gamma^{-h^{\text{ext}}_{T}}}
\end{split}
\end{equation}
where in the second inequality we used that $\kappa L_{1} \ge \beta$, and in the last we used that $h^{\text{ext}}_{T} \ge h_{\beta}$. Putting together (\ref{eq:aT1}), (\ref{eq:aT2}), (\ref{eq:aT3}), we have:
\begin{equation}
e^{-2^{h^{\text{ext}}_{T}} a_{T}} \le e^{-2^{h^{\text{ext}}_{T}} C \gamma^{-h^{\text{ext}}_{T} / \tau}}\;.
\end{equation}
Using this estimate in (\ref{eq:diof1}), we find:
\begin{equation}
\begin{split}
\prod_{v\in \msc{N}_{\text{nr}}(\theta)} e^{-a_{v}} &\le \prod_{T \in \msc{C}_{\text{nr}}(\theta) \cup \msc{N}_{\text{nr}}(\theta)} e^{-2^{h^{\text{ext}}_{T}} C \gamma^{-h^{\text{ext}}_{T} / \tau}} \\
&\le e^{-(1 - \varepsilon(\theta))2^{h}\widetilde{C} \gamma^{-h / \tau}} \prod_{T\in \msc{C}(\theta)} e^{- s^{\text{nr}}_{T} 2^{h_{T}}\widetilde{C} \gamma^{-h_{T} / \tau}}\;.
\end{split}
\end{equation}
We will now use this bound in (\ref{eq:thetabd3}). We have:
\begin{equation}\label{eq:diof4}
\begin{split}
| \theta({\bm q}) | &\le \gamma^{h} K^{|\msc{N}(\theta)|}  |\lambda|^{|\msc{N}_{\text{nr}}(\theta)|}   \Big(\prod_{v \in \msc{N}_{\text{r}}(\theta)}  \zeta_{h_{v}}\Big)  \Big(\prod_{v\in \msc{N}_{\text{nr}}(\theta)} e^{-\frac{c}{2} |n_{v}|}e^{-\frac{\tilde c}{4} \delta_{\omega_{v},-\omega_{v}} L_{2}}  \Big) \\& \quad\cdot \Big(\prod_{T\in \msc{C}(\theta)} \gamma^{-h_{T} s^{\text{nr}}_{T}} e^{- s^{\text{nr}}_{T} 2^{h_{T}}\widetilde{C} \gamma^{-h_{T} / \tau}}\Big) \Big(\prod_{T\in \msc{C}_{\text{r}}(\theta)} \gamma^{ (h^{\text{ext}}_{T} - h_{T})}\Big)\;.
\end{split}
\end{equation}
To obtain this estimate, we used that:
\begin{equation}
e^{-(1 - \varepsilon(\theta))2^{h}\widetilde{C} \gamma^{-h / \tau}} \le C\gamma^{(1 - \varepsilon(\theta))h}\;;
\end{equation}
the factor $\gamma^{(1 - \varepsilon(\theta))h}$, combined with the factor $\gamma^{\varepsilon(\theta)}$ in (\ref{eq:thetabd3}), yields the overall $\gamma^{h}$ factor in (\ref{eq:diof4}). Next, the factor in the third product can be estimated as $C \gamma^{h_{T} s_{T}^{\text{nr}}}$, for a constant $C$ that only depends on $\gamma$. Therefore, since the number of clusters of $\msc{C}(\theta)$ is bounded by $|\msc{N}(\theta)|$, we have:
\begin{equation}
\begin{split}
| \theta({\bm q}) | &\le \gamma^{h} \widetilde{K}^{|\msc{N}(\theta)|}  |\lambda|^{|\msc{N}_{\text{nr}}(\theta)|}   \Big(\prod_{v \in \msc{N}_{\text{r}}(\theta)} \zeta_{h_{v}}\Big)  \Big(\prod_{v\in \msc{N}_{\text{nr}}(\theta)} e^{-\frac{c}{2} |n_{v}|}e^{-\frac{\tilde c}{4} \delta_{\omega_{v},-\omega_{v}} L_{2}}  \Big) \\& \quad\cdot  \Big(\prod_{T\in \msc{C}(\theta)} \gamma^{ (h^{\text{ext}}_{T} - h_{T})}\Big)\;,
\end{split}
\end{equation}
where we used that: 
\begin{equation}\label{eq:4116}
\prod_{T\in \msc{C}(\theta)} \gamma^{h_{T} s^{\text{nr}}_{T}} = \prod_{T\in \msc{C}_{\text{nr}}(\theta) \cup \msc{N}_{\text{nr}}(\theta)} \gamma^{h^{\text{ext}}_{T}} \le \prod_{T\in \msc{C}_{\text{nr}}(\theta) } \gamma^{h^{\text{ext}}_{T} - h_{T}}\;,
\end{equation}
recall that the scale of the non-resonant nodes is zero. This concludes the proof of (\ref{eq:chainbd}), for $n_{0} = n_{1} = 0$. The bound for the derivatives of $\theta({\bm q})$ can be proved in the same way, using the estimates (\ref{eq:singlescale}) for the derivatives of the propagator. We omit the details.
\end{proof}
\subsubsection{The flow of the running coupling constants}
To conclude the analysis of the effective potential, in this section we will control the flow of the running coupling constants. In particular, our goal will be to show the bounds (\ref{eq:indrcc}) on all scales greater or equal than $h$. This is the content of the next proposition.
\begin{proposition}[The flow of the beta function - Part 1]\label{prp:rcc1} Suppose that the estimates (\ref{eq:indrcc}) hold true for all scales greater or equal than $h+1$. Then:
\begin{equation}\label{eq:bdrcc}
\begin{split}
|v_{\nu,\omega,\omega',h} - v_{\nu,\omega,\omega',0}| &\le C|\lambda|^{2} e^{-\frac{\tilde c}{16} L_{2}\delta_{\omega,-\omega'}} \\
| \nu_{\omega,-\omega,h} - \gamma^{-h} \nu_{\omega,-\omega,0} | &\le C|\lambda|^{2}\gamma^{-h}e^{-\frac{\tilde c}{16} L_{2}}\;.
\end{split}
\end{equation}
\end{proposition}
\begin{remark} Thus, to conclude the inductive check we are left with proving the bound for $\nu_{\omega,\omega,h}$. This will be done in Proposition \ref{prp:rcc2}, for a suitable choice of $\nu_{\omega}$.
\end{remark}
\begin{proof} Let us start from the bounds for the marginal terms, $v_{\nu,\omega,\omega',h}$. Recall the flow equation (\ref{eq:vflow})-(\ref{eq:beta}); the beta function can be represented in terms of chain graphs, as follows:
\begin{equation}\label{eq:betachain}
\beta^{v}_{\nu,\omega,\omega',h+1} = i^{\delta_{\nu},0}\sum_{s \ge 1} \frac{1}{s!} \sum_{\theta \in {\msc T}_{n;s, \omega,\omega'}^{(h+1)}} \text{d}_{\nu} \theta({\bm q})\Big|_{{\bm q} = {\bm 0}_{\beta}}\;.
\end{equation}
The argument of the sum could be estimated with (\ref{eq:chainbd}). This however would produce an estimate for $\beta^{v}_{\nu,\omega,\omega',h+1}$ that is not summable in $h$, and it would not allow to check the desired estimates on scale $h$. To obtain a better estimate, we observe that if all nodes of the tree $\theta$ contributing to the beta function are resonant, {\it i.e.} associated with $n_{v} = 0$, then the quasi-momentum flowing on every propagator on the chain is either ${\bm 0}_{\beta}^{+}$ or ${\bm 0}_{\beta}^{-}$. Since these quasi-momenta do not belong to the support of the propagators on scales $h+1$ with $h\ge h_{\beta}$, the corresponding chain graph is vanishing.

Thus, all the chain graphs contributing to the beta function must have at least one non-resonant node. In fact, since the outgoing and the incoming momentum of the whole chain have to be the same, there must be at least two non-resonant nodes. Also, the chain graphs contributing to $\beta^{v}_{\nu,\omega,\omega',h+1}$ must have at least one propagator on scale $h+1$; let us denote by $(v_{*}, v_{*}+1)$ the two nodes connected by $g^{(h+1)}_{\omega_{v_{*}}}$. By the Kirchhoff rule, the momentum flowing on this propagator is ${\bm n}(v_{*})$, and ${\bm n}(v_{*}) \neq 0$ otherwise the propagator is vanishing, by the support of the single-scale cutoff in the definition of $g^{(h+1)}$. 

Let us start by considering $\beta^{v}_{\nu,\omega,\omega,h+1}$. Let us rewrite:
\begin{equation}
\begin{split}
\sum_{\theta \in {\msc T}_{n;s, \omega,\omega}^{(h+1)}} \text{d}_{\nu} \theta({\bm q})\Big|_{{\bm q} = {\bm 0}_{\beta}} = \sum_{j=2}^{s-1}\sum_{\substack{\theta \in {\msc T}_{n;s, \omega,\omega}^{(h+1)} \\ v_{*} = j}} \text{d}_{\nu}\theta({\bm q})\Big|_{{\bm q} = {\bm 0}_{\beta}}\;,
 \end{split}
\end{equation}
and let us further split:
\begin{equation}\label{eq:AB}
\sum_{\substack{\theta \in {\msc T}_{n;s, \omega,\omega}^{(h+1)} \\ v_{*} = j}} \text{d}_{\nu}\theta({\bm q})\Big|_{{\bm q} = {\mb 0}_{\beta}} = \text{A} + \text{B}\;,
\end{equation}
where $\text{A}$ takes into account chains with no off-diagonal non-resonant vertex before $v_{*}$, and $\text{B}$ takes into account the contributions where we have at least  one off-diagonal non-resonant vertex before $v_{*}$. Let us estimate $\text{B}$. Using the bound (\ref{eq:chainbd}), we easily get:
\begin{equation}\label{eq:Best}
| \text{B} | \le C s! |\lambda|^{2} e^{-\frac{\tilde c}{4} L_{2}}\;;
\end{equation}
the factor $|\lambda|^{2}$ comes from the fact that we must have at least two non-resonant nodes. Consider now the term ${\text A}$. Here, there are no off-diagonal non-resonant vertices before $v_{*}$. Since all non-resonant vertices before $v_{*}$ are diagonal, by the support of the propagator on scale $h+1$,
\begin{equation}
\Big| \sum_{v\le v_{*}} n_{v} \alpha \Big| \le C\gamma^{h}\;.
\end{equation}
Suppose that $\big|\sum_{v\le v_{*}} n_{v}\big| \le L_{1} / 2$. Since $\sum_{v\le v_{*}} n_{v} \neq 0$ (otherwise the propagator $g^{(h+1)}$ connecting $v_{*}$ to $v_{*} + 1$ is vanishing), by the Diophantine condition (\ref{eq:diophantine}) we have:
\begin{equation}\label{eq:dioff}
\Big| \sum_{v\le v_{*}} n_{v} \Big| \ge K\gamma^{-h / \tau}\;.
\end{equation}
This will introduce an extra small factor in the estimate for the beta function. To see this, we estimate the first product in Eq. (\ref{eq:chainbd}) as:
\begin{equation}\label{eq:gainbeta}
\begin{split}
\Big( \prod_{v\in \msc{N}(\theta)} e^{-\frac{c}{2}|n_{v}|} \Big) &\le  \Big( \prod_{\substack{v\in \msc{N}(\theta) \\ v\le v_{*}}} e^{-\frac{c}{4}|n_{v}|} \Big) \Big( \prod_{v\in \msc{N}(\theta)} e^{-\frac{c}{4}|n_{v}|} \Big) \\
&\le e^{-\frac{c}{4} K \gamma^{-h/\tau}} \Big( \prod_{v\in \msc{N}(\theta)} e^{-\frac{c}{4}|n_{v}|} \Big)\;,
\end{split}
\end{equation}
where the second bound follows from (\ref{eq:dioff}). If $\Big|\sum_{v\le v_{*}} n_{v}\Big| > L_{1} / 2$ we proceed as after (\ref{eq:4116}), and also in this case we end up with the estimate (\ref{eq:gainbeta}). We then get:
\begin{equation}\label{eq:Aest}
|\text{A}| \le Cs! |\lambda|^{2}e^{-\frac{c}{4} K \gamma^{-h/\tau}}
\end{equation}
We are now ready to prove the estimate for the beta function. Putting together (\ref{eq:AB}), (\ref{eq:Best}), (\ref{eq:Aest}) we obtain, recalling that $L_{2} \ge C\beta$:
\begin{equation}\label{eq:betanu}
| \beta^{v}_{\nu, \omega,\omega,h} | \le K C_{m}|\lambda|^{2} \gamma^{m h}\;;
\end{equation}
The constant $K$ takes into account the bound over all the other running coupling constant, for $|\lambda|$ small enough, and the factor $|\lambda|^{2}$ comes from the fact that we must have at least two non-resonant nodes.

Consider now the case $\omega \neq \omega'$. Here, we must have at least one off-diagonal node, either resonant or non-resonant. If there is at least one non-resonant off-diagonal node, we extract a factor $e^{-\frac{\tilde c}{8} L_{2}}$ from the third product in (\ref{eq:chainbd}), which we further estimate as $e^{-\frac{\tilde c}{8} L_{2}} \le C_{m} \gamma^{m h} e^{-\frac{\tilde c}{16} L_{2}}$. Otherwise, if there are only resonant off-diagonal nodes, we extract the exponential factor from one of the running coupling constants, using the Diophantine gain discussed above. In fact:
\begin{equation}
 e^{-\frac{c}{8} K \gamma^{-h/\tau}} ( |v_{\nu, \omega_{v},-\omega_{v},h_{v}}| + | \nu_{\omega_{v},-\omega_{v},h_{v}} | ) \le C_{m} e^{-\frac{\tilde c}{16} L_{2}} |\lambda| \gamma^{(m-1) h}\;.
\end{equation}
Therefore, we obtained:
\begin{equation}
|\beta^{v}_{\nu,\omega,\omega',h+1}| \le KC_{m} |\lambda|^{2} \gamma^{m h} e^{-\frac{\tilde c}{16} L_{2}\delta_{\omega,-\omega'}}\;,
\end{equation}
which gives:
\begin{equation}
\begin{split}
| v_{\nu, \omega,\omega',h} - v_{\nu, \omega,\omega',0}  | &\le \sum_{k=h+1}^{0} KC_{m} |\lambda|^{2} \gamma^{m k} e^{-\frac{\tilde c}{16} L_{2}\delta_{\omega,-\omega'}} \\
&\le C|\lambda|^{2} e^{-\frac{\tilde c}{16} L_{2}\delta_{\omega,-\omega'}}\;.
\end{split}
\end{equation}
This proves the first of (\ref{eq:bdrcc}). Let us now prove the second of (\ref{eq:bdrcc}). Proceeding exactly as for the marginal terms, we estimate the beta function $\beta^{\nu}_{\omega,-\omega,h}$ as:
\begin{equation}\label{eq:betanuest}
| \beta^{\nu}_{\omega,-\omega,h} | \le KC_{m} |\lambda|^{2} \gamma^{m h} e^{-\frac{\tilde c}{16} L_{2}}\;.
\end{equation}
From the flow equation (\ref{eq:beta}) we obtain:
\begin{equation}
\nu_{\omega,-\omega,h} = \gamma^{-h} \nu_{\omega,-\omega,0} + \sum_{k=h+1}^{0} \gamma^{k - h} \beta^{\nu}_{\omega,-\omega,k}\;,
\end{equation}
which gives, thanks to (\ref{eq:betanuest}):
\begin{equation}
| \nu_{\omega,-\omega,h} - \gamma^{-h} \nu_{\omega,-\omega,0} | \le C|\lambda|^{2}\gamma^{-h}e^{-\frac{\tilde c}{16} L_{2}}\;.
\end{equation}
This proves the second of (\ref{eq:bdrcc}) and concludes the proof of Proposition \ref{prp:rcc1}.
\end{proof}
With the next proposition we conclude the control of the flow of the running coupling constants on scale $h$, by discussing the flow of the remaining relevant term.
\begin{proposition}[The flow of the beta function - Part 2]\label{prp:rcc2} There exists a choice of $\nu_{\omega} = O(\lambda)$ in $\mathbb{R}$ such that:
\begin{equation}\label{eq:nuest}
|\nu_{\omega,\omega,k}|\le C|\lambda| \gamma^{\theta h}\;,\qquad \text{for all $k \ge h$.}
\end{equation}
for $0<\theta < 1$.
\end{proposition}
\begin{proof} The proof is based on a fixed point argument. We rewrite the first of (\ref{eq:beta}) as:
\begin{equation}\label{eq:nu}
\begin{split}
\nu_{\omega,\omega,h+1} &= \gamma^{-1} \nu_{\omega,\omega,h} - \beta^{\nu}_{\omega,\omega,h+1} \\
&= \gamma^{h_{\beta} - h - 1} \nu_{h_{\beta}} - \sum_{k=h_{\beta} + 1}^{h+1} \gamma^{k - h - 1} \beta^{\nu}_{\omega,\omega,k}\;.
\end{split}
\end{equation}
We look at this equation as an equation for the sequence $\underline{\nu} = (\nu_{\omega,\omega,h})$ with $\omega = \pm$ and $h = h_{\beta} + 1, \ldots, 0$ in the space of sequences:
\begin{equation}\label{eq:defS}
\mc{S} = \big\{ \underline{\nu} \in \mathbb{R}^{2(|h_{\beta}|+1)} \, \big| \, |\nu_{\omega,\omega,k}| \le C|\lambda| \gamma^{\theta k} \big\}\;.
\end{equation}
Our goal will be to show that Eq. (\ref{eq:nu}) can be viewed as a fixed point equation for a contraction in the space $\mc{S}$. Then, the existence and the uniqueness of the solution will follow from Banach fixed point theorem. Observe that the beta function in (\ref{eq:nu}) depends on all running coupling constants; we will consider these as fixed except for $(\nu_{\omega,\omega,k})$,  and we suppose that they satisfy the bounds (\ref{eq:bdrcc}) on all scales between $h_{\beta}$ and $0$.

Given $\underline{\nu} \in \mc{S}$, let us define the map $\bm T: \mathbb{R}^{2(|h_{\beta}|+1)} \to \mathbb{R}^{2(|h_{\beta}|+1)}$:
\begin{equation}
\begin{split}
    \big(\bm T(\underline{\nu})\big)_{\omega,j} &:= \gamma^{h_{\beta} - j} \nu_{\omega,\omega,h_{\beta}} - \sum_{k=h_{\beta} + 1}^{j} \gamma^{k - j} \beta^{\nu}_{\omega,\omega,k}(\underline{\nu})\;,\quad j = h_{\beta} + 1, \ldots, 0\;,\\
    \big(\bm T(\underline{\nu})\big)_{\omega,h_\beta} &:= \nu_{\omega,\omega,h_\beta}\;.
\end{split}
\end{equation}
The map has to be though as parametrized by all the other running coupling constants, satisfying the estimates (\ref{eq:bdrcc}) on all scales. These estimates, combined with the estimates defining $\mc{S}$, imply the validity of the bound (\ref{eq:betanu}) for $\beta^{\nu}_{\omega,\omega,k}$.

We claim that the map $\bm{T}$ leaves $\mc{S}$ invariant. First of all, observe that, thanks to the symmetry (\ref{eq:realh}), the beta function is real, if all $(\nu_{\omega,\omega,k})$ are real. Then, we estimate:
\begin{equation}
\begin{split}
\big| \big(\bm{T}(\underline{\nu})\big)_{\omega,j} \big| &\le C\gamma^{h_{\beta} - j} |\lambda| \gamma^{\theta j} + \sum_{k = h_{\beta} + 1}^{j} \gamma^{k-j} C_{m} \gamma^{m k} |\lambda|^{2} \\
&\le C\gamma^{\theta h_{\beta}}\gamma^{(1 - \theta)(h_{\beta} - j)} |\lambda| + K_{m} \gamma^{m j} |\lambda|^{2}\;,
\end{split}
\end{equation}
where we used the bound (\ref{eq:betanuest}) to bound the beta function in the sum. Choosing $\gamma$ so that $\gamma^{-(1 - \theta)} \le 1/2$, and $|\lambda|$ small enough, we see that the right-hand side is bounded by $C|\lambda| \gamma^{\theta j}$ with the same $C$ as in (\ref{eq:defS}). This proves that $\bm{T}$ leaves $\mc{S}$ invariant. 

Let us now prove that $\bm{T}$ defines a contraction on subspace of $\mc{S}$ defined by $\nu_{\omega,\omega,h_\beta}=0$. Let $\underline{\nu}$, $\underline{\tilde\nu}$ be two sequences such subspace. We have:
\begin{equation}
\begin{split}
\big|\big(\bm{T}(\underline{\nu})\big)_{\omega,j} - \big(\bm{T}(\underline{\tilde\nu})\big)_{\omega,j}\big| \le  \sum_{k = h_{\beta} + 1}^{j} \gamma^{k-j} \big| \beta^{\nu}_{\omega,\omega,k}(\underline{\nu}) - \beta^{\nu}_{\omega,\omega,k}(\underline{\tilde\nu}) \big|\;.
\end{split}
\end{equation}
Observe that $\beta^{\nu}_{\omega,\omega,k}$ only depends on the running coupling constants labelled by $j \ge k$. The difference of the beta functions can be estimated as:
\begin{equation}
\big|\beta^{\nu}_{\omega,\omega,k}(\underline{\nu}) - \beta^{\nu}_{\omega,\omega,k}(\underline{\tilde\nu})\big| \le \sum_{\ell = k+1}^{0}\sum_{\omega = \pm} | \nu_{\omega,\omega,\ell} - \tilde \nu_{\omega,\omega,\ell} | \sup^{*}_{\underline{\nu}} \big|\partial_{\nu_{\ell}} \beta^{\nu}_{\omega,\omega,k}(\underline{\nu})\big|
\end{equation}
where the sup is taken over all the sequences such that $|\nu_{\omega,\omega,\ell}| \le 2C|\lambda| \gamma^{\theta \ell}$. Since the beta function is given by a convergent series in the $(\nu_{\omega,\omega,\ell})$ running coupling constants, the argument leading to the estimate (\ref{eq:betanuest}) can be easily adapted to prove:
\begin{equation}
\sup^{*}_{\underline{\nu}}\Big|\partial_{\nu_{\ell}} \beta^{\nu}_{\omega,\omega,k}(\underline{\nu})\Big| \le C_{m} |\lambda|^{2} \gamma^{m k}\;;
\end{equation}
the factor $|\lambda|^{2}$ follows from the fact that the chain graphs contributing to the beta function contain at least two non-resonant vertices. Therefore, we obtain:
\begin{equation}\label{eq:cont}
\begin{split}
&\big|\big(\bm{T}(\underline{\nu})\big)_{\omega,j} - \big(\bm{T}(\underline{\tilde\nu})\big)_{\omega,j}\big|\\
& \quad\le \sum_{k = h_{\beta} + 1}^{j} \gamma^{k-j}  \sum_{\ell = k+1}^{0} \sum_{\omega = \pm}\gamma^{\theta \ell} | \nu_{\omega,\omega,\ell} - \tilde \nu_{\omega,\omega,\ell} | C_{m} |\lambda|^{2} \gamma^{m (k - \ell)}
\end{split}
\end{equation}
where in the last step we used that $1\le \gamma^{\theta \ell} \gamma^{-m\ell}$ choosing $m \ge \theta$. Let us define the distance, for $0<\theta<1$:
\begin{equation}
\| \underline{\nu} - \tilde{\underline{\nu}} \|_{\infty} := \max_{\omega = \pm} \max_{h_{\beta} \le k\le 0} | \nu_{\omega,\omega,k} - \tilde \nu_{\omega,\omega,k} |\;.
\end{equation}
Then, Eq. (\ref{eq:cont}) implies, recalling that $j > h_{\beta}$:
\begin{equation}
\begin{split}
\big\| \bm{T}(\underline{\nu}) - \bm{T}(\underline{\tilde\nu})  \big\|_{\infty} &\le  C|\lambda|^{2} \| \underline{\nu} - \underline{\tilde \nu} \|_{\infty} \\
&< \| \underline{\nu} - \underline{\tilde \nu} \|_{\infty}\;,
\end{split}
\end{equation}
for $|\lambda|$ small enough. Thus, $\bm{T}(\cdot)$ defines a contraction in the subspace of $\mc{S}$ such that $\nu_{h_{\beta},\omega,\omega} = 0$. By Banach fixed point theorem, there exists a unique fixed point $\underline{\nu}^{*} \in \mc{S}$ of the map $\bm{T}$ on this subspace. By construction, the fixed point solves the recursion relation (\ref{eq:nu}), with initial datum $\nu^{*}_{\omega,\omega,0}$. Since $\nu_{\omega,\omega,0}$, defined in (\ref{eq:Ldef3}) with $h=0$, is an analytic function of $\nu_{\omega}$ such that $\partial_{\nu_{\omega}} \nu_{\omega,\omega,0} = 1 + O(\lambda^{2})$, by the implicit function theorem we can choose $\nu_{\omega} = O(\lambda)$ so that $\nu_{\omega,\omega,0} = \nu^{*}_{\omega,\omega,0}$. Hence, for this choice of $\nu_{\omega}$, the solution of the recursion equation (\ref{eq:nu}) satisfies the bound (\ref{eq:nuest}). This concludes the proof of Proposition \ref{prp:rcc2}.
\end{proof}
\begin{remark} In the previous proof, we could have replaced the choice $\nu_{\omega,\omega,h_\beta}=0$ by a more general choice such that $|\nu_{\omega,\omega,h_\beta}|\le C|\lambda|\gamma^{\theta h_\beta}$ with $0<\theta <1$; this would correspond to a slightly different value of $\nu_{\omega}$, and hence of the Fermi points, which however converge to the previous values as $\beta \to \infty$. We omit the details.
\end{remark}
The control of the flow of the running coupling constants concludes the construction of the effective potential on all scales. Next, we observe that the effective potential admits improved estimates.
\begin{proposition}[Improved bound for the effective potential] Let $\nu_{\omega}$ as in Proposition \ref{prp:rcc2}. Then, for $1>\theta>0$:
\begin{equation}\label{eq:Vimpro}
| \text{d}_{0}^{n_{0}} \text{d}_{1}^{n_{1}} V^{(h)}_{n;\omega,\omega'}({\bm q})| \le C_{n_{0},n_{1}}|\lambda| \gamma^{\theta h}\gamma^{h(1 - n_{0} - n_{1})} e^{-\frac{c}{4} |n|} e^{- \frac{\tilde c}{4} \delta_{\omega,-\omega'} L_{2}}\;.
\end{equation}
\end{proposition}
\begin{remark} That is, the bound (\ref{eq:Vimpro}) improves the estimate (\ref{eq:Vindh}) by a gain factor $\gamma^{\theta h}$.
\end{remark}
\begin{proof} The proof is based on a simple adaptation of the chain expansion discussed above; let us sketch it. Consider the graphs with only resonant nodes. For $\kappa L_{2} \ge \beta$, all running coupling constants can be bounded by $C|\lambda| \gamma^{\theta h}$; recall (\ref{eq:bdrcc}), (\ref{eq:nuest}). The factor $\gamma^{\theta h}$ in the estimate for the running coupling constants allows to obtain the desired dimensional gain in the estimate for the fully resonant graphs. Consider now the graphs with at least one non-resonant node (there has to be at least one non-resonant node if $n\neq 0$). Non-resonant nodes are on scale $0$: thus, the cluster stucture of the graphs contributing to the effective potential reaches scale $0$. For this reason, we can rewrite the last factor in (\ref{eq:chainbd}) as, for $0<\theta<1$:
\begin{equation}
\begin{split}
\prod_{T\in \msc{C}(\theta)} \gamma^{ (h^{\text{ext}}_{T} - h_{T})} &= \Big(\prod_{T\in \msc{C}(\theta)} \gamma^{ \theta (h^{\text{ext}}_{T} - h_{T})}\Big) \Big( \prod_{T\in \msc{C}(\theta)} \gamma^{ (1-\theta)(h^{\text{ext}}_{T} - h_{T})} \Big) \\
&\le \gamma^{\theta h} \Big( \prod_{T\in \msc{C}(\theta)} \gamma^{ (1-\theta)(h^{\text{ext}}_{T} - h_{T})} \Big)\;,
\end{split}
\end{equation}
where we used that the external scale of the deepest cluster is $h$ and the internal scale of the highest cluster is $0$, and we estimated the first factor as a telescopic product. This allows to extract a factor $\gamma^{\theta h}$ also for the graphs containing at least one non-resonant node; the remaining product allows to control the sum over the scale labels, since $1-\theta >0$. This concludes the proof.
\end{proof}
 Our next task will be to compute the two-point correlation function of the model; this will be discussed in the next section, via an adaptation of the multi-scale analysis used to construct the effective potential.

\subsection{The two-point function}\label{sec:2pt}
In this section we shall discuss how to adapt the previous multiscale analysis, to compute the two point correlation function (\ref{eq:2pt}). Being the state quasi-free, the computation of the two-point correlation function allows to determine all correlation functions of the model via Wick's rule. 
\subsubsection{The generating functional of correlations}
The generating functional of the correlation functions of the model is defined similarly to the partition function, after introducing an external Grassmann field. Let us define the configuration space Grassmann field as:
 \begin{equation}\label{eq:defrealspace}
       \psi^{\pm}_{\vec{\bm x},\sigma}:= \frac{1}{L\beta} \sum_{{\bm k} \in \msc D_{N, L,\beta}} e^{\mp i{\bm x}\cdot {\bm k}} \psi^{\pm}_{{\bm k},x_2,\sigma}\;,
 \end{equation}
with $\msc D_{N, L,\beta}$ as in Eq. (\ref{eq:DNLB}). The Euclidean two-point correlation function of the model is, for $0\le x_{0}, y_{0} < \beta$, $x_{0}\ne y_{0}$: 
\begin{equation}
\big\langle \timord \, \gamma_{x_0}(a_{{\vec x},\sigma}) ; \gamma_{y_0}(a_{\vec y,\zeta}^{*})\big\rangle_{\beta,\mu,L} = \lim_{N\to+\infty}\ms Z_{N}^{-1}\int \mathbb{P}_{N}(d\psi) \,\exp\big({\mc V_{N}( \psi)}\big)\psi^{-}_{\vec{\bm x},\sigma}\psi^{+}_{\vec{\bm y},\zeta}\;;
\end{equation}
for $x_{0} = y_{0}$, the identity holds adding $- (1/2)\delta_{\vec x,\vec y} \delta_{\sigma,\zeta}$ to the right-hand side. 
It is convenient to express the two-point correlation function in terms of the derivatives of the generating functional of correlations,
\begin{equation}
\exp\big(\mathcal{W}_{N}(\phi)\big) := \int \mathbb{P}_{N}(d\psi) \,\exp\big({\mc V_{N}( \psi) + (\psi, \phi)}\big)\;,
\end{equation}
with:
\begin{equation}\label{eq:source}
(\psi, \phi) := \sum_{\sigma = 1}^{M} \sum_{\vec{ x} \in \Lambda_{L}} \int_{0}^{\beta} dx_0\, (\phi^+_{\vec{\bm x},\sigma} \psi_{\vec{\bm x},\sigma}^- +  \psi^+_{\vec{\bm x},\sigma}\phi^-_{\vec{\bm x},\sigma})\;,
\end{equation}
and where $\phi^{\pm}$ is an external Grassmann field, which can also be written as in (\ref{eq:defrealspace}) in terms of suitable momentum space fields $\phi^{\pm}_{{\bm k}, x_{2}, \sigma}$. Then, the two-point correlation function (at non-coinciding space-time points) is:
\begin{equation}\label{eq:2ptgrassman}
\big\langle \timord \, \gamma_{x_0}(a_{{\vec x},\sigma}) ; \gamma_{y_0}(a_{\vec y,\zeta}^{*})\big\rangle_{\beta,\mu,L} = \lim_{N\to \infty} \frac{\partial^{2}}{\partial \phi^{+}_{\vec {\bm x}, \sigma} \partial \phi^{-}_{\vec {\bm y}, \zeta}} \mathcal{W}_{N}(\phi)\Big|_{\phi = 0}\;.
\end{equation}
Our goal will be to compute the generating functional of correlations, adapting the multiscale analysis developed for the partition function. Proceeding as in Section \ref{sec:massive}, we end up with:
\begin{equation}
\exp\big(\mathcal{W}_{N}(\phi)\big) = z_{\mr{b}}(\phi) \int \mathbb{P}_{\mr{e}}(d\psi^{(\mr{e})}) \exp\big(\mathcal{V}^{(\mr{e})}(\psi^{(\mr{e})}) + \mathcal{W}^{(\mr{e})}(\psi^{(\mr{e})}, \phi)\big)\;,
\end{equation}
with the following differences with respect to the analysis of the effective potential.
\begin{itemize}
\item[(i)] The prefactor $z_{\mr{b}}(\phi)$ is:
\begin{equation}
\begin{split}
z_{\mr{b}}(\phi) &= z_{\mr{b}} \exp\big( \phi^{+}, W^{(\mr{b})} \phi^{-} \big) \\
\big( \phi^{+}, W^{(\mr{b})} \phi^{-} \big) &= \int_{0}^{\beta} dx_{0} dy_{0}\, \sum_{\vec x, \vec y \in \Lambda_{L}} \sum_{\sigma, \zeta} \phi^{+}_{\vec {\bm x}, \sigma} W^{(\mr{b})}_{\sigma,\zeta}(\vec {\bm x}, \vec {\bm y}) \phi^{-}_{\vec {\bm y}, \zeta}\;,
\end{split}
\end{equation}
where $W^{(\mr{b})}$ is analytic in $\lambda$ for $|\lambda|$ small enough and it satisfies:
\begin{equation}\label{eq:2ptWb}
| W^{(\mr{b})}_{\sigma, \zeta}(\vec {\bm x}, \vec {\bm y}) | \le Ce^{-c\| \vec {\bm x} - \vec {\bm y} \|}\;;
\end{equation}
\item[(ii)] $\mathcal{W}^{(\mr{e})}(\psi, \phi)$ has the form:
\begin{equation}
\begin{split}
\mc W^{(\mr e)}( \psi^{(\mr e)},\phi ) &= \int_{0}^{\beta} d x_0\, d  y_0  \sum_{\vec x, \vec y \in \Lambda_{L}} \sum_{\sigma, \zeta} \Big( \phi^+_{\vec{\bm x},\sigma}W^{(\mr e)}_{\phi\psi; \sigma,\zeta}(\vec{\bm x},\vec{\bm y}) \psi^{(\mr e)-}_{\vec{\bm y},\zeta}\\
&\qquad +\psi^{(\mr e)+}_{\vec{\bm x},\sigma}W^{(\mr e)}_{\psi\phi; \sigma,\zeta}(\vec{\bm x},\vec{\bm y})\phi^-_{\vec{\bm y},\zeta}\Big)\;,
\end{split}
\end{equation}
where the kernels $W^{(\mr e)}$ are analytic in $\lambda$ for $|\lambda|$ small enough and decay exponentially, as in (\ref{eq:2ptWb}).
\end{itemize}
The analyticity and the exponential decay of the kernels is proved via a chain expansion produced by the integration of the massive field, as discussed in Section \ref{sec:massive}. Next, we parametrize the field $\psi^{(\mr{e})}$ in terms of the boundary fields, as discussed in Section \ref{sec:1dred}, and we shift the Fermi momentum and integrate the first scale as in Section \ref{sec:firstscale}. In these steps, the presence of the external field does not play any significant role: it simply introduces a different type of external line in the chain expansions, associated with $\phi^{\pm}$, which cannot be contracted on smaller scales.

Thus, we end up with the following representation of the generating functional of correlations:
\begin{equation}
\exp\big(\mathcal{W}_{N}(\phi)\big) = z_{>0}(\phi) \int \mathbb{P}_{(\le 0)}(d\psi^{(\le 0)}) \exp\big(\mathcal{V}^{(0)}(\psi^{(\le 0)}) + \mathcal{W}^{(0)}(\psi^{(\le 0)}, \phi)\big)\;,
\end{equation}
where:
\begin{itemize}
\item[(i)] the overall prefactor is:
\begin{equation}\label{eq:2ptz>}
z_{>0}(\phi) = \tilde z_{\mr{b}} \exp\big( \phi^{+} , W^{(>0)} \phi^{-} \big)\;,
\end{equation}
with $W^{(>0)} = W^{(\mr{b})} + W^{(1)}$, and
\begin{equation}
\big( \phi^{+} , W^{(1)} \phi^{-} \big) = \int_{0}^{\beta} dx_{0} dy_{0}\, \sum_{\vec x, \vec y \in \Lambda_{L}} \sum_{\sigma, \zeta} \phi^{+}_{\vec {\bm x}, \sigma} W^{(1)}_{\sigma,\zeta}(\vec {\bm x}, \vec {\bm y}) \phi^{-}_{\vec {\bm y}, \zeta}\;.
\end{equation}
The kernels satisfy the bound:
\begin{equation}
\big| W^{(1)}_{\sigma,\zeta}(\vec {\bm x}, \vec {\bm y}) \big| \le C e^{-c\| {\bm x} - {\bm y} \|} e^{-\tilde c |x_{2}|_{L}}\;,
\end{equation}
where we introduced the notation:
\begin{equation}
|x_{2}|_{L} := \min(x_{2}, L_{2}-x_{2})\;.
\end{equation}
The kernel $W^{(1)}$ is expressed in terms of chain diagrams constructed using the propagator $g^{(1)}$ in (\ref{eq:g1}).
\item[(ii)] The effective potential $\mathcal{W}^{(\le 0)}(\psi^{(\le 0)}, \phi)$ has the form:
\begin{equation}\label{eq:2ptW0}
\begin{split}
\mc W^{(0)}( \psi^{(\le 0)},\phi ) &= \frac{1}{\beta L_{1}} \sum_{\substack{n, {\bm k} \\ x_{2},\omega,\sigma}} \Big( \phi^+_{{\bm k}, x_{2}, \sigma}  W^{(0)}_{\phi\psi;n,\omega,\sigma} ( {\bm k} , x_{2}) \psi^{(\le 0)-}_{ {\bm k}   + n{\bm \alpha}, \omega}\\
&\qquad + \psi^{(\le 0)+}_{{\bm k}  , \omega}  W^{(0)}_{\psi\phi;n,\omega,\sigma} ( {\bm k}  , x_{2}) \phi^{-}_{ {\bm k} + n{\bm \alpha}, x_{2},\sigma}\Big)\\
&\equiv \frac{1}{\beta L_{1}} \sum_{\substack{n, {\bm k} \\ x_{2},\omega,\sigma}} \Big( \phi^+_{{\bm k}, x_{2}, \sigma}  W^{(0)}_{\phi\psi;n,\omega,\sigma} \big(\bm{q}(\bm{k}), x_{2}\big) \psi^{(\le 0)-}_{ \bm{q}(\bm{k}) + n{\bm \alpha}, \omega}\\
&\qquad + \psi^{(\le 0)+}_{\bm{q}(\bm{k}) , \omega}  W^{(0)}_{\psi\phi;n,\omega,\sigma} \big(\bm{q}(\bm{k}) , x_{2}\big) \phi^{-}_{ {\bm k} + n{\bm \alpha}, x_{2},\sigma}\Big) \;;
    \end{split}
\end{equation}
 we recall $\bm{q}\equiv \bm q(\bm{k})=\bm{k}-\bm{k}_F^\omega$ as in (\ref{eq:coordinatechange}). The two kernels are related by:
\begin{equation}
\overline{W^{(0)}_{\phi\psi;n,\omega,\sigma} ( {\bm q}, x_{2})} = W^{(0)}_{\psi\phi;-n,\omega,\sigma} \big( (-q_{0}, q_{1} + n\alpha), x_{2}\big)\;,
\end{equation} 
and they satisfy the estimate:
\begin{equation}\label{eq:w0phipsi}
\Big| \text{d}_{q_{0}}^{n_{0}} \text{d}_{q_{1}}^{n_{1}} W^{(0)}_{\phi\psi;n,\omega,\sigma} ( {\bm q}, x_{2}) \Big| \le C_{n_{0}, n_{1}} |\lambda|^{\delta_{n\neq 0}} e^{-c|n|} e^{-\tilde c |x_{2}|_{\omega}}\;,
\end{equation}
with $|x_{2}|_{\omega}$ defined in (\ref{eq:x2omega}).
\end{itemize}
The plan will be to integrate scale-by-scale the generating functional of the correlations, as we did for the effective potential, and to understand the recursion relation satisfied by the new kernels involving $\phi^{\pm}$ external lines.

To see this, suppose that the generating functional of correlations can be written as:
\begin{equation}\label{eq:indgen}
\begin{split}
&\exp\big(\mathcal{W}_{N}(\phi)\big) \\&\quad = z_{(>h)}(\phi) \int \mathbb{P}_{(\le h)}(d\psi^{(\le h)}) \exp\big(\mathcal{V}^{(h)}(\psi^{(\le h)}) + \mathcal{W}^{(h)}(\psi^{(\le h)}, \phi)\big)\;,
\end{split}
\end{equation}
for suitable $z_{(>h)}(\cdot)$, $ \mathcal{W}^{(h)}(\cdot)$, similarly to (\ref{eq:2ptz>}), (\ref{eq:2ptW0}). In particular, we assume that:
\begin{equation}
z_{(>h)}(\phi) = z_{(>h)} \exp(\phi^{+}, W^{(>h)} \phi^{-})\;,
\end{equation}
with $z_{(>h)}$ as in Section \ref{sec:iterative}, and:
\begin{equation}\label{eq:indgen3}
\begin{split}
(\phi^{+}, W^{(>h)} \phi^{-}) &= \sum_{j=h+1}^{0}\frac{1}{\beta L_{1}} \sum_{n, \bm{k}} \phi^{+}_{\bm k, x_{2},\sigma} W_{\phi\phi;n,\sigma,\zeta}^{(j)}({\bm k}; x_{2}, y_{2}) \phi^{-}_{\bm k + n{\bm \alpha}, y_{2},\zeta}\\&\quad + (\phi^{+}, W^{(>0)} \phi^{-})\;,
\end{split}
\end{equation}
for suitable kernels $W^{(k)}$ to be determined inductively. To set up the integration of the single-scale, we proceed as in Section \ref{sec:iterative}. We have:
\begin{equation}\label{eq:tildeP2pt}
\begin{split}
&\exp\big(\mathcal{W}_{N}(\phi)\big) \\
&= z_{(>h)}(\phi) \tilde z_{h} \int \widetilde P_{(\le h)}(d\psi^{(\le h)}) \exp \big(\widetilde{ \mathcal{V}}^{(h)}(\psi^{(\le h)}) + \mathcal{W}^{(h)}(\psi^{(\le h)}, \phi)\big)\;,
\end{split}
\end{equation}
where:
\begin{equation}\label{eq:tildeVh}
\begin{split}
\widetilde{\mathcal{V}}^{(h)}(\psi^{(\le h)}) &= (\mf L_{0} + \mf L_{1;\text{od}} + \mf R) \mathcal{V}^{(h)}(\psi^{(\le h)}) \\
&= \frac{1}{\beta L_{1}} \sum_{\substack{n, {\bm q}, \\ \omega, \omega'}} \widetilde{V}_{n;\omega,\omega'}^{(h)}({\bm q}) \psi_{{\bm q},\omega}^{(\le h)+}  \psi_{{\bm q} + n {\bm \alpha},\omega'}^{(\le h)-}\;.
\end{split}
\end{equation}
Let us discuss the structure of $z_{(>h-1)}$ and of $\mathcal{W}^{(\le h-1)}$, which are obtained after integrating the scale $h$. These objects collect respectively all chain graphs with both external lines given by $\phi$ fields, and all chain graphs having one external line given by a $\phi$ field and a $\psi$ field. We have:
\begin{equation}
z_{(>h-1)}(\phi) = z_{(>h)}(\phi) \tilde z_{h} z_{h} \exp{(\phi^{+}, W^{(h)} \phi^{-})}
\end{equation}
where
\begin{equation}
(\phi^{+}, W^{(h)} \phi^{-}) = \frac{1}{\beta L_{1}} \sum_{n, {\bm k}} \phi^{+}_{\bm k, x_{2},\sigma} W_{\phi\phi;n,\sigma,\zeta}^{(h)}({\bm k}; x_{2}, y_{2}) \phi^{-}_{\bm k + n{\bm \alpha}, y_{2},\zeta}\;;
\end{equation}
the kernels satisfy the following recursion relation, again ${\bm q} = {\bm k} - {\bm k}_{F}^{\omega}$:
\begin{equation}\label{eq:recphi2}
\begin{split}
&W_{\phi\phi;n,\sigma,\zeta}^{(h)}({\bm k}; x_{2}, y_{2}) \\
&= \sum_{m,\omega} W^{(h)}_{\phi\psi;m,\omega,\sigma} ( {\bm q}, x_{2}) g^{(h)}_{\omega}({\bm q} + m{\bm \alpha}) W^{(h)}_{\psi\phi;-m + n,\omega,\zeta} ( {\bm q} + m{\bm \alpha}, y_{2}) \\
& \quad + \sum_{\substack{m,m' \\ \omega,\omega'}}W^{(h)}_{\phi\psi;m,\omega,\sigma} ( {\bm q}, x_{2}) g^{(h)}_{\omega}({\bm q} + m{\bm \alpha}) V^{(h-1)}_{m';\omega,\omega'}({\bm q} + m{\bm \alpha} ) \\
&\qquad\quad \cdot g^{(h)}_{\omega'}\big({\bm q} + (m + m'){\bm \alpha}\big) W^{(h)}_{\psi\phi;-m - m' + n,\omega',\zeta} \big( {\bm q} + (m + m'){\bm \alpha}, y_{2}\big)\;,
\end{split}
\end{equation}
and:
\begin{equation}\label{eq:rec2}
\begin{split}
&W^{(h-1)}_{\phi\psi;n,\omega,\sigma} ( {\bm q}, x_{2}) \\
&\quad = W^{(h)}_{\phi\psi;n,\omega,\sigma} ( {\bm q}, x_{2})\\
&\qquad + \sum_{m,\omega'} W^{(h)}_{\phi\psi;m,\omega',\sigma} ( {\bm q}, x_{2}) g^{(h)}_{\omega'}({\bm q} + m{\bm \alpha}) V^{(h-1)}_{n-m;\omega',\omega}({\bm q} + m {\bm \alpha})\;,
\end{split}
\end{equation}
with $V^{(h-1)}$ the effective potential on scale $h-1$. See Appendix \ref{app:2pt} for the proof of (\ref{eq:recphi2}), (\ref{eq:rec2}).
\begin{remark}
In the generating functional, the kernels $W^{(h)}_{\phi\psi;n,\omega}({\bm q})$, respectively $W^{(h)}_{\phi\psi;-m+n,\omega}({\bm q} + m{\bm \alpha})$, are multiplied by fields $\psi^{(\le h)-}_{{\bm q} + m{\bm \alpha}.\omega}$, respectively $\psi^{(\le h)+}_{{\bm q} + m{\bm \alpha},\omega}$. These are not defined for ${\bm q} + m{\bm \alpha}$ outside of the support of the corresponding propagator. Nevertheless, the chain expansion for the kernels and the recursion relation (\ref{eq:rec2}) make sense for all ${\bm q}$ and for all $n$. In what follows, we will consider the $W^{(h)}_{\phi\psi;n,\omega}({\bm q})$ as defined for all ${\bm q}$ and for all $n$.
\end{remark}
The next proposition allows to bound the solution of the recursion relation (\ref{eq:rec2}).
\begin{proposition}[Flow of the $\phi\psi$-kernels]\label{prp:phipsi} For $|\lambda|$ small enough, the following is true:
\begin{equation}\label{eq:nnw}
\Big| W^{(h-1)}_{\phi\psi;n,\omega,\sigma} ( {\bm q}, x_{2}) \Big| \le C |\lambda|^{\delta_{n\neq 0}} e^{-\frac{c}{8}|n|} e^{- \frac{\tilde c}{32}|x_{2}|_{\omega}}\;.
\end{equation}
Also, for $k\ge h$:
\begin{equation}\label{eq:diffW}
\big| W^{(h-1)}_{\phi\psi;n,\omega,\sigma} ( {\bm q}, x_{2})  - W^{(k)}_{\phi\psi;n,\omega,\sigma} ( {\bm q}, x_{2})  \big| \le K |\lambda|^{\delta_{n\neq 0}} \gamma^{\theta k} e^{-\frac{c}{8}|n|} e^{- \frac{\tilde c}{32}|x_{2}|_{\omega}}\;.
\end{equation}
Furthermore, for $n_{0} + n_{1} =1$:
\begin{equation}\label{eq:derW}
\big| \text{d}_{q_{0}}^{n_{0}} \text{d}_{q_{1}}^{n_{1}}W^{(h-1)}_{\phi\psi;n,\omega,\sigma} ( {\bm q}, x_{2}) \big| \le C \gamma^{(\theta - 1)h }  |\lambda|^{\delta_{n\neq 0}} e^{-\frac{c}{8}|n|} e^{- \frac{\tilde c}{32}|x_{2}|_{\omega}}\;.
\end{equation}
\end{proposition}
\begin{remark} These bounds also allow to control the kernels $W^{(h-1)}_{\psi\phi}$, using that:
\begin{equation}\label{eq:symW}
W^{(h-1)}_{\psi\phi;n,\omega,\sigma} ( {\bm q}, x_{2}) = \overline{W^{(h-1)}_{\phi\psi;-n,\omega,\sigma} \big( (-q_{0}, q_{1} + n\alpha), x_{2}\big)}\;.
\end{equation}
\end{remark}
This identity is proved as (\ref{eq:realh}), see Section \ref{sec:chain}.
\begin{proof} Let us start by proving (\ref{eq:nnw}). To begin, recall the bound (\ref{eq:Vimpro}), for $L_{2} \ge C\beta$:
%for $\|{\bm q} + m{\bm \alpha}\| \le \gamma^{h+1}$ and for $L_{2} \ge C\beta$:
%
%\begin{equation}\label{eq:bdrec}
%\begin{split}
%\big| V^{(h)}_{n-m;\omega',\omega}({\bm q} + m {\bm \alpha})\big| &\le C\delta_{n,m} ( \gamma^{h(1 + \theta)} |\lambda|\delta_{\omega,\omega'} + |\lambda|\gamma^{-h} e^{-\frac{\tilde c}{16} L_{2}}\delta_{\omega,-\omega'} ) \\
%&\quad + \delta_{n\neq m} C \gamma^{(1+\theta)h} e^{-\frac{c}{4} |n-m|} e^{-\frac{\tilde c}{16} L_{2} \delta_{\omega,-\omega'}} |\lambda|\;.
%\end{split}
%\end{equation}
%
%The reason for the momentum constraint is that in (\ref{eq:rec2}) the momentum at the argument of $\widetilde{V}^{(h)}$ also appears at the argument of the single-scale propagator $g^{(h)}$. The first term in the right-hand side of (\ref{eq:bdrec}) comes from the local terms in (\ref{eq:tildeVh}), diagonal and off-diagonal in the $\omega$ labels, while the second term comes from ${\mf R} V^{(h)}$ in (\ref{eq:tildeVh}); it is proved using the bound (\ref{eq:Vimpro}), and $\| {\bm q} + m {\bm \alpha} \|\le \gamma^{h+1}$. For $L_{2} \ge C\beta$, the bound (\ref{eq:bdrec}) implies:
%
\begin{equation}\label{eq:tildebd}
\big| V^{(h-1)}_{n-m;\omega',\omega}({\bm q} + m {\bm \alpha})\big| \le C \gamma^{h(1+\theta)} e^{-\frac{c}{4} |n-m|} |\lambda| e^{-\frac{\tilde c}{4} L_{2} \delta_{\omega,-\omega'} }\;.
\end{equation}
%
%The reason for the momentum constraint is that in (\ref{eq:rec2}) the momentum at the argument of $\widetilde{V}^{(h)}$ also appears at the argument of the single-scale propagator $g^{(h)}$.
%
We will proceed by induction. Let $h<0$, and assume that:
\begin{equation}\label{eq:ind2pth}
\big| W^{(h)}_{\phi\psi;n,\omega,\sigma} ( {\bm q}, x_{2}) \big| \le \sum^{0}_{j > h} \gamma^{\theta j} C|\lambda|^{\delta_{n\neq 0}} e^{-\frac{c}{8}|n|} e^{- \frac{\tilde c}{32}|x_{2}|_{\omega}}\;.
\end{equation}
Clearly, Eq. (\ref{eq:ind2pth}) also implies:
\begin{equation}\label{eq:ind2pth2}
\big| W^{(h)}_{\phi\psi;n,\omega,\sigma} ( {\bm q}, x_{2}) \big| \le K|\lambda|^{\delta_{n\neq 0}} e^{-\frac{c}{8}|n|} e^{- \frac{\tilde c}{32}|x_{2}|_{\omega}}\;,
\end{equation}
which is (\ref{eq:nnw}). The assumption (\ref{eq:ind2pth}) is true for $h=0$, by (\ref{eq:w0phipsi}). Our plan will be to check (\ref{eq:ind2pth}), using the recursion relation (\ref{eq:rec2}). From (\ref{eq:rec2}), we have:
\begin{equation}\label{eq:recbd}
\begin{split}
&\big| W^{(h-1)}_{\phi\psi;n,\omega,\sigma} ( {\bm q}, x_{2})  \big| \le \big| W^{(h)}_{\phi\psi;n,\omega,\sigma} ( {\bm q}, x_{2})  \big| \\
&\qquad + \sum_{m,\omega'} \big| W^{(h)}_{\phi\psi;m,\omega',\sigma} ( {\bm q}, x_{2})\big| \big| g^{(h)}_{\omega'}({\bm q} + m{\bm \alpha})  \big| \big|V^{(h-1)}_{n-m;\omega',\omega}({\bm q} + m {\bm \alpha})\big|\;.
\end{split}
\end{equation}
Consider the second term in the right-hand side of (\ref{eq:recbd}). We have, using the bounds (\ref{eq:tildebd}), (\ref{eq:ind2pth2}):
\begin{equation}\label{eq:a1}
\begin{split}
&\sum_{m,\omega'}\big| W^{(h)}_{\phi\psi;n,\omega',\sigma} ( {\bm q}, x_{2})\big| \big| g^{(h)}_{\omega'}({\bm q} + n{\bm \alpha})  \big| \big| V^{(h-1)}_{n-m;\omega',\omega}({\bm q} + n {\bm \alpha})\big| \\
&\quad \le \sum_{m,\omega'} K |\lambda|^{1+\delta_{n\neq 0}} e^{-\frac{c}{8}|m|} e^{- \frac{\tilde c}{32}|x_{2}|_{\omega'}} \gamma^{\theta h} e^{-\frac{c}{4} |n-m|} e^{-\frac{\tilde c}{32} L_{2} \delta_{\omega,-\omega'} } \\
&\quad \le  C |\lambda|^{\delta_{n\neq 0}} e^{-\frac{c}{8}|n|}  e^{- \frac{\tilde c}{32}|x_{2}|_{\omega}} \gamma^{\theta h}
\end{split}
\end{equation}
for $|\lambda|$ small enough, with $C$ as in (\ref{eq:tildebd}). We used that $|x_{2}|_{\omega'} - |x_{2}|_{\omega} + 2L_{2} \ge L_{2}$, recall (\ref{eq:x2omega}). Plugging the bounds (\ref{eq:ind2pth}), (\ref{eq:a1}) in (\ref{eq:recbd}) we get:
\begin{equation}
\begin{split}
\big| W^{(h-1)}_{\phi\psi;n,\omega,\sigma} ( {\bm q}, x_{2})  \big| &\le \sum^{0}_{j > h} \gamma^{\theta j} C|\lambda|^{\delta_{n\neq 0}} e^{-\frac{c}{8}|n|} e^{- \frac{\tilde c}{32}|x_{2}|_{\omega}} \\
&\quad +  \gamma^{\theta h} C |\lambda|^{\delta_{n\neq 0}} e^{-\frac{c}{8}|n|} e^{- \frac{\tilde c}{32}|x_{2}|_{\omega}}\;,
\end{split}
\end{equation}
which concludes the check of (\ref{eq:ind2pth}) with $h$ replaced by $h-1$. Let us now prove (\ref{eq:diffW}). From the recursion relation (\ref{eq:rec2}), we have:
\begin{equation}
\begin{split}
&\big| W^{(h-1)}_{\phi\psi;n,\omega,\sigma} ( {\bm q}, x_{2}) - W^{(h)}_{\phi\psi;n,\omega,\sigma} ( {\bm q}, x_{2})\big| \\
&\quad \le \sum_{m,\omega'} \big| W^{(h)}_{\phi\psi;m,\omega',\sigma} ( {\bm q}, x_{2})\big| \big| g^{(h)}_{\omega'}({\bm q} + m{\bm \alpha})\big|  \big| V^{(h-1)}_{n-m;\omega',\omega}({\bm q} + m {\bm \alpha})\big|\;.
\end{split}
\end{equation}
By (\ref{eq:a1}), (\ref{eq:tildebd}) we get:
\begin{equation}\label{eq:diffWh}
\big| W^{(h-1)}_{\phi\psi;n,\omega,\sigma} ( {\bm q}, x_{2}) - W^{(h)}_{\phi\psi;n,\omega,\sigma} ( {\bm q}, x_{2})\big| \le C \gamma^{\theta h} |\lambda|^{1+\delta_{n\neq 0}} e^{-\frac{c}{8}|n|} e^{- \frac{\tilde c}{32}|x_{2}|_{\omega}}\;.
\end{equation}
Therefore, by a telescopic sum argument:
\begin{equation}
\begin{split}
& \big| W^{(h-1)}_{\phi\psi;n,\omega,\sigma} ( {\bm q}, x_{2})  - W^{(k)}_{\phi\psi;n,\omega,\sigma} ( {\bm q}, x_{2})  \big| \\
&\quad \le \sum_{j=h}^{k}  \big|W^{(j-1)}_{\phi\psi;n,\omega,\sigma} ( {\bm q}, x_{2})  - W^{(j)}_{\phi\psi;n,\omega,\sigma} ( {\bm q}, x_{2})  \big| \\
&\quad \le \sum_{j=h}^{k} C \gamma^{\theta j} |\lambda|^{1+\delta_{n\neq 0}} e^{-\frac{c}{8}|n|} e^{- \frac{\tilde c}{32}|x_{2}|_{\omega}} \\
&\quad \le K |\lambda|^{1+\delta_{n\neq 0}} \gamma^{\theta k} e^{-\frac{c}{8}|n|} e^{- \frac{\tilde c}{32}|x_{2}|_{\omega}}
\end{split}
\end{equation}
where the second inequality follows from (\ref{eq:diffWh}). This proves (\ref{eq:diffW}). To conclude, let us now prove (\ref{eq:derW}). Assume that it holds on scales $k\ge h$:
\begin{equation}\label{eq:derWk}
\big| \text{d}_{q_{0}}^{n_{0}} \text{d}_{q_{1}}^{n_{1}}W^{(k)}_{\phi\psi;n,\omega,\sigma} ( {\bm q}, x_{2}) \big| \le C \gamma^{(\theta - 1)(k+1)}  |\lambda|^{\delta_{n\neq 0}} e^{-\frac{c}{8}|n|} e^{- \frac{\tilde c}{32}|x_{2}|_{\omega}}\;.
\end{equation}
As for (\ref{eq:ind2pth2}), the bound (\ref{eq:derWk}) is true on scale $0$. Let us check it on smaller scales. By the recursion relation:
\begin{equation}
\begin{split}
\text{d}_{q_{0}}^{n_{0}} &\text{d}_{q_{1}}^{n_{1}} W^{(h-1)}_{\phi\psi;n,\omega,\sigma} ( {\bm q}, x_{2}) = \text{d}_{q_{0}}^{n_{0}} \text{d}_{q_{1}}^{n_{1}} W^{(h)}_{\phi\psi;n,\omega,\sigma} ( {\bm q}, x_{2})\\&\quad + \sum_{m,\omega'} \text{d}_{q_{0}}^{n_{0}} \text{d}_{q_{1}}^{n_{1}} \Big(W^{(h)}_{\phi\psi;m,\omega',\sigma} ( {\bm q}, x_{2}) g^{(h)}_{\omega'}({\bm q} + m{\bm \alpha}) V^{(h-1)}_{n-m;\omega',\omega}({\bm q} + m {\bm \alpha})\Big)\;.
\end{split}
\end{equation}
Using the improved estimate (\ref{eq:Vimpro}) on the derivatives of the effective potential together with (\ref{eq:derW}) on scale $h$, we get:
\begin{equation}\label{eq:derder}
\begin{split}
\Big| \text{d}_{q_{0}}^{n_{0}} &\text{d}_{q_{1}}^{n_{1}} W^{(h-1)}_{\phi\psi;n,\omega,\sigma} ( {\bm q}, x_{2}) \Big| \le C \gamma^{(\theta - 1)(h+1) }  |\lambda|^{\delta_{n\neq 0}} e^{-\frac{c}{8}|n|} e^{- \frac{\tilde c}{32}|x_{2}|_{\omega}} \\
&+ \sum_{m,\omega'} K |\lambda|^{1+\delta_{n\neq 0}} e^{-\frac{c}{8}|m|} e^{- \frac{\tilde c}{32}|x_{2}|_{\omega'}} \gamma^{(\theta - 1) h} e^{-\frac{c}{4} |n-m|} e^{-\frac{\tilde c}{32} L_{2} \delta_{\omega,-\omega'} }\;;
\end{split}
\end{equation}
choosing $\gamma$ such that $\gamma^{\theta - 1} \le 1/2$ and $|\lambda|$ small enough, we get that both terms in the right-hand side of (\ref{eq:derder}) are bounded by:
\begin{equation}
\frac{C}{2} \gamma^{(\theta - 1)h}  |\lambda|^{\delta_{n\neq 0}} e^{-\frac{c}{8}|n|} e^{- \frac{\tilde c}{32}|x_{2}|_{\omega}}
\end{equation}
with $C$ as in (\ref{eq:derWk}). This concludes the proof of (\ref{eq:derW}).
\end{proof}
\begin{remark} The bound (\ref{eq:derW}) allows to prove that, for $\|{\bm q} + n{\bm \alpha}\| \le \gamma^{h}$:
\begin{equation}\label{eq:interp}
\begin{split}
&\big| W^{(h-1)}_{\phi\psi;n,\omega,\sigma} ( {\bm q}, x_{2}) - W^{(h-1)}_{\phi\psi;n,\omega,\sigma} ( -n{\bm \alpha}, x_{2}) \big| \\
&\qquad \le \int_{0}^{1} dt\, \Big|\frac{d}{dt} W^{(h-1)}_{\phi\psi;n,\omega,\sigma} \big( -n{\bm \alpha} + t ({\bm q} + n{\bm \alpha}), x_{2}\big)\Big| \\
&\qquad \le C \gamma^{\theta h}  |\lambda|^{\delta_{n\neq 0}} e^{-\frac{c}{8}|n|} e^{- \frac{\tilde c}{32}|x_{2}|_{\omega}}\;.
\end{split}
\end{equation}
This will be particularly useful in the analysis of the two-point function.
\end{remark}
\subsubsection{Proof of Theorem \ref{thm:2pt}}\label{sec:thm2ptproof}
In this section we will show how to use Proposition \ref{prp:phipsi} to obtain the asymptotic behavior of the two-point function, Eq. (\ref{eq:2ptmain}). We start by writing:
\begin{equation}
\begin{split}
    &\big\langle \timord \, \gamma_{x_0}(a_{{\vec x},\sigma}) ; \gamma_{y_0}(a_{\vec y,\zeta}^{*})\big\rangle_{\beta,\mu,L}\\
    &\quad = \frac{1}{\beta L_{1}} \sum_{n} \sum_{{\bm k}} e^{i{\bm k}\cdot ({\bm x} - {\bm y})} e^{-i n \alpha y_{1}} \hat S_{2;n,\sigma,\zeta}({\bm k}; x_{2}, y_{2})\;,
\end{split}
\end{equation}
where, using the Grassmann representation of the two-point function:
\begin{equation}
\hat S_{2;n,\sigma,\zeta}({\bm k}; x_{2}, y_{2}) = \beta L_{1} \lim_{N\to \infty} \frac{\partial^{2}}{\partial \hat \phi^{+}_{{\bm k}, x_{2}, \sigma} \partial \hat \phi^{-}_{{\bm k} + n{\bm \alpha}, y_{2}, \zeta}} \mathcal{W}_{N}(\phi)\Big|_{\phi = 0}\;.
\end{equation}
By (\ref{eq:indgen})-(\ref{eq:indgen3}), we have:
\begin{equation}\label{eq:2ptdec}
\hat S_{2;n,\sigma,\zeta}({\bm k}; x_{2}, y_{2}) = \sum_{h = h_{\beta}}^{0} W^{(h)}_{\phi\phi; n, \sigma,\zeta}({\bm k}; x_{2}, y_{2}) + \hat R_{2;n,\sigma,\zeta}({\bm k}; x_{2}, y_{2})\;,
\end{equation}
where:
\begin{itemize}
\item[(i)] the kernels $W^{(h)}_{\phi\phi; n, \sigma,\zeta}({\bm k}; x_{2}, y_{2})$ satisfy the recursion (\ref{eq:recphi2});
\item[(ii)] $\hat R_{2} = W^{(>0)}$, recall (\ref{eq:2ptz>}), and it satisfies:
\begin{equation}\label{eq:R2bd}
\Big| \text{d}_{k_{0}}^{n_{0}} \text{d}_{k_{1}}^{n_{1}} \hat R_{2;n,\sigma,\zeta}({\bm k}; x_{2}, y_{2})  \Big| \le C_{n_{0},n_{1}} |\lambda|^{\delta_{n\neq 0}} e^{-c|n|}e^{- c|x_{2} - y_{2}|}\;.
\end{equation}
\end{itemize}
Let us now focus on the main term in (\ref{eq:2ptdec}), due to the integration of the single-scale fields. Recall Eq. (\ref{eq:recphi2}); we rewrite it as:
\begin{equation}
W_{\phi\phi;n,\sigma,\zeta}^{(h)}({\bm k}; x_{2}, y_{2}) = \text{A}^{(h)}_{n,\sigma,\zeta}({\bm k}; x_{2}, y_{2}) + \mr{B}^{(h)}_{n,\sigma,\zeta}({\bm k}; x_{2}, y_{2})
\end{equation}
with:
\begin{equation}\label{eq:Ah0}
\begin{split}
&\text{A}^{(h)}_{n,\sigma,\zeta}({\bm k}; x_{2}, y_{2}) \\&\qquad := \sum_{m,\omega} W^{(h)}_{\phi\psi;m,\omega,\sigma} ( {\bm q}, x_{2}) g^{(h)}_{\omega}({\bm q} + m{\bm \alpha}) W^{(h)}_{\psi\phi;-m + n,\omega,\zeta} ( {\bm q} + m{\bm \alpha}, y_{2}) \\
&\qquad\equiv \sum_{\omega} \text{A}^{(h)}_{n,\sigma,\zeta,\omega}({\bm q}; x_{2}, y_{2})\;,
\end{split}
\end{equation}
where we recall that, at the argument of the sum over $\omega$, ${\bm q} \equiv {\bm q}({\bm k}) = {\bm k} - {\bm k_{F}^{\omega}}$ and:
\begin{equation}\label{eq:Bh}
\begin{split}
&\mr{B}^{(h)}_{n,\sigma,\zeta}({\bm k}; x_{2}, y_{2}) \\
&\qquad := \sum_{\substack{m,m' \\ \omega,\omega'}}W^{(h)}_{\phi\psi;m,\omega,\sigma} ( {\bm q}, x_{2}) g^{(h)}_{\omega}({\bm q} + m{\bm \alpha}) V^{(h-1)}_{m';\omega,\omega'}({\bm q} + m{\bm \alpha} ) \\ 
&\qquad\qquad \cdot g^{(h)}_{\omega'}\big({\bm q} + (m + m'){\bm \alpha}\big) W^{(h)}_{\psi\phi;-m - m' + n,\omega',\zeta} \big( {\bm q} + (m + m'){\bm \alpha}, y_{2}\big)\\
&\qquad \equiv \sum_{\omega,\omega'} \mr{B}^{(h)}_{n,\sigma,\zeta,\omega,\omega'}({\bm q}; x_{2}, y_{2})\;.
\end{split}
\end{equation}
We define:
\begin{equation}
\begin{split}
Z^{(h)}_{\phi\psi;n,\omega,\sigma}(x_{2}) &:= W^{(h)}_{\phi\psi;n,\omega,\sigma} ({\bm 0}_{\beta} - n{\bm \alpha}, x_{2})\\
Z^{(h)}_{\psi\phi;n,\omega,\sigma}(x_{2}) &:= W^{(h)}_{\psi\phi;n,\omega,\sigma}({\bm 0}_{\beta}, x_{2})\;.
\end{split}
\end{equation}
By the identity (\ref{eq:symW}), these two quantities are related by:
\begin{equation}\label{eq:symZ}
Z^{(h)}_{\psi\phi;n,\omega,\sigma}(x_{2}) = \overline{Z^{(h)}_{\phi\psi;-n,\omega,\sigma}(x_{2})}\;.
\end{equation}
Furthermore, we define:
\begin{equation}
\begin{split}
{\mf R} W^{(h)}_{\phi\psi;n,\omega,\sigma} ( {\bm q}, x_{2}) &:= W^{(h)}_{\phi\psi;n,\omega,\sigma} ( {\bm q}, x_{2}) - Z^{(h)}_{\phi\psi;n,\omega,\sigma}(x_{2}) \\
{\mf R} W^{(h)}_{\psi\phi;-m + n,\omega,\zeta} ( {\bm q} + m{\bm \alpha}, y_{2}) &:= W^{(h)}_{\psi\phi;-m + n,\omega,\zeta} ( {\bm q} + m{\bm \alpha}, y_{2})\\
&\qquad - Z^{(h)}_{\psi\phi;-m+n,\omega,\zeta}(y_{2})\;.
\end{split}
\end{equation}
\paragraph{Discussion of $\text{A}^{(h)}$.} We rewrite:
\begin{equation}\label{eq:Ah}
\begin{split}
\text{A}^{(h)}_{n,\sigma,\zeta,\omega}({\bm q}; x_{2}, y_{2}) &= \sum_{m} Z^{(h)}_{\phi\psi;m,\omega,\sigma} ( x_{2}) g^{(h)}_{\omega}({\bm q} + m{\bm \alpha}) Z^{(h)}_{\psi\phi;-m + n,\omega,\zeta} (y_{2})\\
&\qquad + \text{A}^{(h)}_{1; n,\sigma,\zeta,\omega}({\bm q}; x_{2}, y_{2})
\end{split}
\end{equation}
where $\text{A}^{(h)}_{1}$ contains at least one between
\begin{equation}\label{eq:RW}
 {\mf R} W^{(h)}_{\phi\psi;m,\omega,\sigma} ( {\bm q}, x_{2})\;,\qquad  {\mf R} W^{(h)}_{\psi\phi;-m + n,\omega,\zeta} ( {\bm q} + m{\bm \alpha}, y_{2})\;.
\end{equation}
For $\| {\bm q} + m\alpha \| \le \gamma^{h+1}$, which is imposed by the support of the single-scale propagator in $\text{A}^{(h)}$, both kernels in (\ref{eq:RW}) can be estimated as in Eq. (\ref{eq:interp}). We have:
\begin{equation}
\big|  {\mf R} W^{(h)}_{\phi\psi;m,\omega,\sigma} ( {\bm q}, x_{2}) \big| \le C \gamma^{\theta (h+1)}  |\lambda|^{\delta_{m\neq 0}} e^{-\frac{c}{8}|m|} e^{- \frac{\tilde c}{32}|x_{2}|_{\omega}}\;,\
\end{equation}
and the same bound holds true for the second kernel in (\ref{eq:RW}). Therefore, using that every ${\bm q}$ derivative introduces a factor $\gamma^{-h}$ in the estimate for the kernel, we find:
\begin{equation}\label{eq:errA1}
\begin{split}
&\big| \text{d}_{{ q}_{0}}^{n_{0}}\text{d}_{{ q}_{1}}^{n_{1}} \text{A}^{(h)}_{1; n,\sigma,\zeta,\omega}({\bm q}; x_{2}, y_{2}) \big| \\
&\qquad \le C_{n_{0},n_{1}} \gamma^{-h(1 + n_{0} + n_{1})} \gamma^{\theta h} |\lambda|^{\delta_{n\neq 0}} e^{-\frac{c}{16}|n|} e^{- \frac{\tilde c}{32}|x_{2}|_{\omega}} e^{- \frac{\tilde c}{32}|y_{2}|_{\omega}}\;.
\end{split}
\end{equation}
Let us now consider the main term in (\ref{eq:Ah}). By using (\ref{eq:diffW}), we have:
\begin{equation}
\big|Z^{(h_{\beta})}_{\phi\psi;n,\omega,\sigma}(x_{2}) - Z^{(h)}_{\phi\psi;n,\omega,\sigma}(x_{2})\big|\le K |\lambda|^{\delta_{n\neq 0}} \gamma^{\theta h} e^{-\frac{c}{8}|n|} e^{- \frac{\tilde c}{32}|x_{2}|_{\omega}}\;.
\end{equation}
Thus, we have:
\begin{equation}\label{eq:A2err}
\begin{split}
&\sum_{m} Z^{(h)}_{\phi\psi;m,\omega,\sigma} ( x_{2}) g^{(h)}_{\omega}({\bm q} + m{\bm \alpha}) Z^{(h)}_{\psi\phi;-m + n,\omega,\zeta} (y_{2}) \\
&\quad = \sum_{m} Z^{(h_{\beta})}_{\phi\psi;m,\omega,\sigma} ( x_{2}) g^{(h)}_{\omega}({\bm q} + m{\bm \alpha}) Z^{(h_{\beta})}_{\psi\phi;-m + n,\omega,\zeta} (y_{2})\\
&\quad\qquad + \text{A}^{(h)}_{2; n,\sigma,\zeta,\omega}({\bm q}; x_{2}, y_{2})\;,
\end{split}
\end{equation}
where $\text{A}^{(h)}_{2}$ satisfies the bound (\ref{eq:errA1}). Let us now focus on the main term in (\ref{eq:A2err}). We introduce the short-hand notation:
\begin{equation}
Z_{m,\omega,\sigma} ( x_{2}) := Z^{(h_{\beta})}_{\phi\psi;m,\omega,\sigma} ( x_{2})\;;
\end{equation}
by (\ref{eq:symZ}) we have 
\begin{equation}
Z^{(h_{\beta})}_{\psi\phi;-m+n,\omega,\zeta} ( y_{2}) = \overline{Z_{-n+m,\omega,\zeta}(y_{2})}\;.
\end{equation}
Thus, we can rewrite (\ref{eq:Ah}) as:
\begin{equation}
\begin{split}
\text{A}^{(h)}_{n,\sigma,\zeta}({\bm k}; x_{2}, y_{2}) &= \sum_{\omega,m} Z_{m,\omega,\sigma}(x_{2}) g^{(h)}_{\omega}({\bm q} + m{\bm \alpha})  \overline{Z_{-n+m,\omega,\zeta}(y_{2})} \\
&\quad + \sum_{i=1,2} \text{A}^{(h)}_{i;n,\sigma,\zeta}({\bm k}; x_{2}, y_{2})\;.
\end{split}
\end{equation}
Now, recall the expression of the single-scale propagator, Eq. (\ref{eq:proph}). By the estimate (\ref{eq:betanu}) on the beta function, we know that the velocities $v_{\nu,\omega,h}({\bm q})$ approach their infrared limit $v_{\nu,\omega} := v_{\nu,\omega,h_{\beta}}$ with the following rate:
\begin{equation}
|v_{\nu,\omega,h}({\bm q}) - v_{\nu,\omega}| \le C_{r} \gamma^{r h} |\lambda|^{2}\qquad \text{for all $r\in \mathbb{N}$.}
\end{equation}
Let us denote by $g^{(h)}_{\omega,\mr{s}}$ the propagator obtained after this replacement:
\begin{equation}
g^{(h)}_{\omega;\mr{s}}({\bm q}) := \frac{f^{(h)}_{\omega;\mr{s}}({\bm q})}{i v_{0,\omega} q_{0} + v_{1,\omega} q_{1}}
\end{equation}
where we denote by $f^{(h)}_{\omega;\mr{s}}({\bm q})$ the new single-scale cutoff function, after the above replacement of the velocities:
\begin{equation}
f^{(h)}_{\omega;\mr{s}}({\bm q}) := \chi\Big(2^{-h} \sqrt{v_{0,\omega}^{2}q_{0}^{2} + v_{1,\omega}^{2}q_{1}^{2} }\Big) - \chi\Big(2^{-(h-1)} \sqrt{v_{0,\omega}^{2}q_{0}^{2} + v_{1,\omega}^{2}q_{1}^{2} }\Big)\;.
\end{equation} 
Hence,
\begin{equation}\label{eq:AAA}
\begin{split}
\text{A}^{(h)}_{n,\sigma,\zeta}({\bm k}; x_{2}, y_{2}) &= \sum_{\omega,m} Z_{m,\omega,\sigma}(x_{2}) g^{(h)}_{\omega;\mr{s}}({\bm q} + m{\bm \alpha})  \overline{Z_{-n+m,\omega,\zeta}(y_{2})} \\
&\quad + \sum_{i=1,2,3} \text{A}^{(h)}_{i;n,\sigma,\zeta}({\bm k}; x_{2}, y_{2})\;,
\end{split}
\end{equation}
where the new error term $\text{A}^{(h)}_{3}$ takes into account the replacement of the velocities, and it is not difficult to see that it satisfies the bound (\ref{eq:errA1}). The main term in (\ref{eq:AAA}) will contribute to the scaling limit of the two-point function, while all the other terms will contribute with subleading corrections, that decay faster than the main term in configuration space.

\paragraph{Discussion of $\mr{B}^{(h)}$.} Consider now the term $\mr{B}^{(h)}$, Eq. (\ref{eq:Bh}). This term is subleading with respect to the main term in (\ref{eq:AAA}): informally, with respect to (\ref{eq:AAA}), this term contains an extra single-scale propagator, and an extra potential term. Since the single scale propagator is bounded dimensionally as $\gamma^{-h}$, while the potential term is bounded proportionally to $\gamma^{h(1+\theta)}$, we have a net dimensional gain of $\gamma^{\theta h}$ with respect to the leading term in (\ref{eq:AAA}). Ultimately, the term $\mr{B}^{(h)}_{n,\sigma,\zeta,\omega,\omega'}({\bm q}; x_{2}, y_{2})$ satisfies a bound similar to (\ref{eq:errA1}): 
\begin{equation}\label{eq:errB1}
\begin{split}
&\big| \text{d}_{{q}_{0}}^{n_{0}}\text{d}_{{  k}_{1}}^{n_{1}} \mr{B}^{(h)}_{n,\sigma,\zeta,\omega}({\bm q}; x_{2}, y_{2}) \big| \\
&\quad \le C_{n_{0},n_{1}} \gamma^{-h(1 + n_{0} + n_{1})} \gamma^{\theta h} |\lambda|^{1+\delta_{n\neq 0}} e^{-\frac{c}{16}|n|} e^{- \frac{\tilde c}{32}|x_{2}|_{\omega}} e^{- \frac{\tilde c}{32}|y_{2}|_{\omega'}}\;.
\end{split}
\end{equation}
We omit the details. 

\paragraph{Bounds for the remainder term.} All in all, we have:
\begin{equation}\label{eq:2ptres}
\begin{split}
&\big\langle \timord \, \gamma_{x_0}(a_{{\vec x},\sigma}) ; \gamma_{y_0}(a_{\vec y,\zeta}^{*})\big\rangle_{\beta,\mu,L} \\
&\quad = \frac{1}{\beta L_{1}} \sum_{n,m,\omega} \sum_{{\bm k}} e^{i{\bm k}\cdot ({\bm x} - {\bm y})} e^{-i n \alpha y_{1}} Z_{m,\omega,\sigma}(x_{2}) g_{\omega;\mr{s}}\big({\bm q(\bm k)} + m{\bm \alpha}\big)  \overline{Z_{-n+m,\omega,\zeta}(y_{2})} \\
&\qquad + R_{\sigma,\zeta}(\vec {\bm x}, \vec {\bm y})\;,
\end{split}
\end{equation}
where we set $g_{\omega,\mr{s}} := \sum_{h=h_{\beta}}^{0} g^{(h)}_{\omega,\mr{s}}$,
\begin{equation}\label{eq:grel}
g_{\omega;\mr{s}}({\bm q}) = \frac{\chi\Big(\sqrt{v_{0,\omega}^{2}q_{0}^{2} + v_{1,\omega}^{2}q_{1}^{2} }\Big)}{i v_{0,\omega} q_{0} + v_{1,\omega} q_{1}}\;.
\end{equation}
Eq. (\ref{eq:grel}) is an effective relativistic propagator, with renormalized velocities, and an ultraviolet cutoff; recall the definition of the cutoff function, Eq. (\ref{eq:chidef}). The remainder term in (\ref{eq:2ptres}) is:
\begin{equation}
R_{\sigma,\zeta}(\vec {\bm x}, \vec {\bm y}) = R_{1;\sigma,\zeta}(\vec {\bm x}, \vec {\bm y}) + R_{2;\sigma,\zeta}(\vec {\bm x}, \vec {\bm y})\;,
\end{equation}
with:
\begin{equation}\label{eq:rrn}
\begin{split}
R_{i;\sigma,\zeta}(\vec {\bm x}, \vec {\bm y}) &= \frac{1}{\beta L_{1}} \sum_{n,{\bm k}} e^{i{\bm k}\cdot ({\bm x} - {\bm y})} e^{-i n \alpha y_{1}} \hat R_{i;n,\sigma,\zeta}({\bm k}; x_{2}, y_{2}) \\
&\equiv \sum_{n} e^{-i n \alpha y_{1}} R_{i;n,\sigma,\zeta}({\bm x - {\bm y}}; x_{2}, y_{2})
\end{split}
\end{equation}
where $\hat R_{2;n,\sigma,\zeta}({\bm k}; x_{2}, y_{2})$ satisfies the estimate (\ref{eq:R2bd}) while $\hat R_{1;n,\sigma,\zeta}({\bm k}; x_{2}, y_{2})$ is given by:
\begin{equation}\label{eq:R1n}
\hat R_{1;n,\sigma,\zeta}({\bm k}; x_{2}, y_{2}) = \sum_{h=h_{\beta}}^{0} \Big(\sum_{i=1,2,3} \text{A}^{(h)}_{i;n,\sigma,\zeta}({\bm k}; x_{2}, y_{2}) + \mr{B}^{(h)}_{n,\sigma,\zeta}({\bm k}; x_{2}, y_{2})\Big)\;.
\end{equation}
To prove a decay estimate for $R_{i;n,\sigma,\zeta}({\bm x - {\bm y}}; x_{2}, y_{2})$, we proceed as follows. Let:
\begin{equation}
\begin{split}
    (x_{0} - y_{0})_{L} &:= \frac{\beta}{\pi} \sin\big(\pi (x_{0} - y_{0}) / \beta\big) \;,\\
    (x_{1} - y_{1})_{L} &:= \frac{L_{1}}{\pi} \sin\big(\pi (x_{0} - y_{0}) / L_{1}\big)\;.
\end{split}
\end{equation}
Observe that, for $\beta, L_{1}\to \infty$, these objects converge pointwise to $(x_{0} - y_{0})$ and to $(x_{1} - y_{1})$. Then, observe that, for $\nu=0,1$:
\begin{equation}
\begin{split}
&\Big|(x_{\nu} - y_{\nu})_{L}^{m}R^{(h)}_{i;n,\sigma,\zeta}({\bm x - {\bm y}}; x_{2}, y_{2})\Big| \\
&\qquad = \Big|\frac{1}{\beta L_{1}} \sum_{{\bm k}} \big(\text{d}_{k_\nu}^{m} e^{i{\bm k}\cdot ({\bm x} - {\bm y})}\big) e^{-i n \alpha y_{1}} \hat R^{(h)}_{i;n,\sigma,\zeta}({\bm k}; x_{2}, y_{2})\Big| \\
&\qquad = \Big|\frac{1}{\beta L_{1}} \sum_{{\bm k}} e^{i{\bm k}\cdot ({\bm x} - {\bm y})} e^{-i n \alpha y_{1}} \big(\text{d}_{\nu}^{m} \hat R^{(h)}_{i;n,\sigma,\zeta}({\bm k}; x_{2}, y_{2})\big)\Big|\;,
\end{split}
\end{equation}
where the last step follows from summation by parts, using that $\hat R^{(h)}_{i}$ vanishes as $k_{0}\to \infty$ and it is periodic in $k_{1}$ in the (discrete) Brillouin zone. Consider the error term $R_{1}$. From the estimate
\begin{equation}\label{eq:R1hest}
\begin{split}
&\big|\text{d}_{\nu}^{m} \hat R^{(h)}_{1;n,\sigma,\zeta}({\bm k}; x_{2}, y_{2})\big| \\
&\quad \le C_{m} \gamma^{-h(1 + m)} \gamma^{\theta h} |\lambda|^{\delta_{n\neq 0}} e^{-\frac{c}{16}|n|} e^{- \frac{\tilde c}{32}|x_{2}|_{L}} e^{- \frac{\tilde c}{32}|y_{2}|_{L}}\;,
\end{split}
\end{equation}
we find:
\begin{equation}\label{eq:bdRx}
\begin{split}
&\big|\gamma^{mh}(x_{\nu} - y_{\nu})_{L}^{m}R^{(h)}_{1;n,\sigma,\zeta}({\bm x - {\bm y}}; x_{2}, y_{2})\big| \\&\qquad \le K_{m} \gamma^{h(1+\theta)}  |\lambda|^{\delta_{n\neq 0}} e^{-\frac{c}{16}|n|} e^{- \frac{\tilde c}{32}|x_{2}|_{L}} e^{- \frac{\tilde c}{32}|y_{2}|_{L}}\;.
\end{split}
\end{equation}
From (\ref{eq:bdRx}) we easily get:
\begin{equation}
\big|R^{(h)}_{1;n,\sigma,\zeta}({\bm x - {\bm y}}; x_{2}, y_{2})\big| \le \frac{K_{m} \gamma^{h(1+\theta)}  |\lambda|^{\delta_{n\neq 0}} e^{-\frac{c}{16}|n|} e^{- \frac{\tilde c}{32}|x_{2}|_{L}} e^{- \frac{\tilde c}{32}|y_{2}|_{L}}}{1 + (\gamma^{h} \| {\bm x} - {\bm y} \|)^{m}}\;.
\end{equation}
Thus, writing, for $0<\theta'<\theta$:
\begin{equation}
\begin{split}
\gamma^{h(1 + \theta)} &=  \gamma^{h(\theta - \theta')} \gamma^{h(1 + \theta')} \\
&=  \gamma^{h(\theta - \theta')} \frac{1+\|{\bm x} - {\bm y}\|^{1 + \theta'}}{1+\|{\bm x} - {\bm y}\|^{1 + \theta'}}   \gamma^{h(1 + \theta')}\;,
\end{split}
\end{equation}
and observing that, uniformly in $\| {\bm x} - {\bm y} \|$,
\begin{equation}
\frac{(1+\|{\bm x} - {\bm y}\|^{1 + \theta'}) \gamma^{h(1 + \theta')}}{ 1 + (\gamma^{h} \| {\bm x} - {\bm y} \|)^{m} } \le C\;,
\end{equation}
we obtain:
\begin{equation}\label{eq:R1nest0}
\begin{split}
\big|R_{1;n,\sigma,\zeta}({\bm x - {\bm y}}; x_{2}, y_{2})\big| &\le \sum_{h=h_{\beta}}^{0} \big|R^{(h)}_{1;n,\sigma,\zeta}({\bm x - {\bm y}}; x_{2}, y_{2})\big| \\
&\le \frac{K |\lambda|^{\delta_{n\neq 0}} e^{-\frac{c}{16}|n|} e^{- \frac{\tilde c}{32}|x_{2}|_{L}} e^{- \frac{\tilde c}{32}|y_{2}|_{L}}}{1 + \|{\bm x} - {\bm y}\|^{1 + \theta'}} \sum_{h=h_{\beta}}^{0} \gamma^{h(\theta - \theta')}  \\
&\le  \frac{C |\lambda|^{\delta_{n\neq 0}} e^{-\frac{c}{16}|n|} e^{- \frac{\tilde c}{32}|x_{2}|_{L}} e^{- \frac{\tilde c}{32}|y_{2}|_{L}}}{1 + \|{\bm x} - {\bm y}\|^{1 + \theta'}}\;.
\end{split}
\end{equation}
Similarly, using the bound (\ref{eq:R2bd}) we obtain:
\begin{equation}\label{eq:R2est0}
\big|R_{2;n,\sigma,\zeta}(\vec {\bm x}, \vec {\bm y})\big| \le C |\lambda|^{\delta_{n\neq 0}} e^{-\frac{c}{16}|n|} e^{-c\|\vec {\bm x} - \vec {\bm y}\|}\;.
\end{equation}
Eqs. (\ref{eq:rrn}), (\ref{eq:R1nest0}), (\ref{eq:R2est0}) prove (\ref{eq:Rest2pt}).

\paragraph{Asymptotics of the leading term.} Consider now the main term in (\ref{eq:2ptres}). It is:
\begin{equation}\label{eq:Ss}
\begin{split}
&S_{2; \sigma, \zeta}^{\mr{s}}(\vec {\bm x}, \vec {\bm y}) \\
&:= \frac{1}{\beta L_{1}} \sum_{\substack{n,m \\ \omega, {\bm k}}} e^{i{\bm k}\cdot ({\bm x} - {\bm y})} e^{-i n \alpha y_{1}} Z_{m,\omega,\sigma}(x_{2}) g_{\omega;\mr{s}}\big(\bm q(\bm k )+ m{\bm \alpha}\big)  \overline{Z_{-n+m,\omega,\zeta}(y_{2})}\;.
\end{split}
\end{equation}
Let:
\begin{equation}\label{eq:Zfou}
Z_{\omega,\sigma}(\vec x) := \sum_{m} e^{-im\alpha x_{1}} Z_{m,\omega,\sigma}(x_{2})\;.
\end{equation}
Using that:
\begin{equation}
\begin{split}
&\frac{1}{\beta L_{1}}\sum_{{\bm k}} e^{i{\bm k}\cdot ({\bm x} - {\bm y})} g_{\omega;\mr{s}}\big({\bm q}(\bm k) + m{\bm \alpha}\big) \\
&\qquad = e^{i k_{F}^{\omega}(\lambda) (x_{1} - y_{1}) - i m \alpha (x_{1} - y_{1})} \frac{1}{\beta L_{1}}\sum_{{\bm q}} e^{i({\bm q} + m {\bm \alpha})\cdot ({\bm x} - {\bm y})} g_{\omega;\mr{s}}({\bm q} + m{\bm \alpha}) \\
&\qquad = e^{i k_{F}^{\omega}(\lambda) (x_{1} - y_{1}) - i m \alpha (x_{1} - y_{1})} \check{g}_{\omega;\mr{s}}({\bm x} - {\bm y})\;,
\end{split}
\end{equation}
we rewrite:
\begin{equation}\label{eq:S2s}
\begin{split}
&S_{2; \sigma, \zeta}^{\mr{s}}(\vec {\bm x}, \vec {\bm y}) \\
&= \sum_{\omega} e^{i k_{F}^{\omega}(\lambda) (x_{1} - y_{1})} \sum_{n,m}  e^{-i m \alpha (x_{1} - y_{1})} e^{-i n \alpha y_{1}}  Z_{m,\omega,\sigma}(x_{2}) \check{g}_{\omega;\mr{s}}({\bm x} - {\bm y}) \overline{Z_{-n+m,\omega,\zeta}(y_{2})} \\
&= \sum_{\omega} e^{i k_{F}^{\omega}(\lambda)(x_{1} - y_{1})} Z_{\omega,\sigma}(\vec x) \check{g}_{\omega;\mr{s}}({\bm x} - {\bm y}) \overline{Z_{\omega,\zeta}(\vec y)}\;.
\end{split}
\end{equation}
 This concludes the proof of (\ref{eq:2ptmain}). The only ingredient missing to conclude the proof of Theorem \ref{thm:2pt} are the relations (\ref{eq:relzeta}), which will be proved in Section \ref{sec:vertexcons}.
\section{Transport coefficients}\label{sec:transport}
In this section we shall use the construction of the two-point function to compute the edge conductivity and the edge susceptibility of the model. We will start from the linear response ansatz, as given from Kubo formula, Eq. (\ref{eq:kubo}). As we will see, the precise asymptotics of the two-point function derived in the previous section will allow us to compute these transport coefficients, and to prove Theorem \ref{thm:resp}.
 
\subsection{Response functions in imaginary times} 
A key role in the analysis is played by the lattice continuity equation (\ref{eq:cons}), which implies Ward identities for the Euclidean correlation functions. Specifically, we will need the current-current Ward identity and the vertex Ward identity; they will be discussed in this section. Later, we will discuss how these Ward identities can be used to study transport in real time; this is possible thanks to a complex deformation argument, that allows to relate the real-time transport coefficients to imaginary-time correlation functions. In fact, it turns out that the real-time response functions (\ref{eq:kubo}) can be rewritten as, see {\it e.g.} \cite{AMP, MP, greenlamp},
\begin{equation}\label{eq:wick}
\begin{split}
\chi_{\nu}^{\beta,L}(\eta,\theta) &= \theta\sum_{\vec x,\vec y } \mu(\theta\vec  x)  \phi_{\ell}(\theta y_{1}, y_{2}) \int_{-\beta/2}^{\beta/2}  d s\, e^{-i\eta_{\beta} s} \langle \timord \gamma_{s}(n_{\vec x})\;; j_{\nu,\vec y} \rangle_{\beta,\mu,L} \\
&\quad + \varepsilon_{\nu}^{\beta,L}(\eta,\theta)
\end{split}
\end{equation}
where: $\eta_{\beta}$ is the best approximation of $\eta$ in $(2\pi/\beta)\mathbb{N}$ such that $\eta_{\beta} \ge \eta$; $\varepsilon_{\nu}^{\beta,L}(\eta,\theta) = 0$ if $\eta \in \frac{2\pi}{\beta} \mathbb{N}$, and it satisfies the bound:
\begin{equation}\label{eq:epsest}
| \varepsilon_{\nu}^{\beta,L}(\eta,\theta) | \le \frac{C}{\beta \eta^{3}}\;
\end{equation}
uniformly in $L$ and $\theta$. The rewriting (\ref{eq:wick}) is particularly convenient: the Euclidean correlation function in the right-hand side can be studied via our renormalization group construction. Furthermore, as we shall see in the next section, the lattice continuity equation introduces nontrivial constraints on the structure of this correlation function, that will play an important role in the computation of the response function.
\subsection{Current-current Ward identity}\label{sec:cucuWI}
To begin, let us rewrite the conservation law (\ref{eq:cons}) in imaginary time. Using that $\tau_{t}(\mc{O}) = \gamma_{it}(\mc{O})$ if $[\mc{O}, \mc{N}] = 0$, we get the Euclidean continuity equation:
\begin{equation}\label{eq:cons2}
i\partial_{t} \gamma_{t}(n_{x}) + \sum_{k=1,2} \text{d}_{x_{k}} \gamma_{t}(j_{k,x}) = 0\;.
\end{equation}
We will be interested in the consequences of this identity at the level of the Euclidean correlation functions. Let $\eta \in \frac{2\pi}{\beta}\cdot \mathbb{N}$, and recall the definition (\ref{eq:smearj}) of the smeared current operator. We define:
\begin{equation}
j_{\nu}(\eta, \phi_{\theta,\ell}) := \int_{0}^{\beta} dt\,e^{-i\eta t} \gamma_{t}\big( j_{\nu}(\phi_{\theta,\ell}) \big)\;.
\end{equation}
By time-translation invariance and $\beta$-periodicity of the Gibbs state, we have:
\begin{equation}\label{eq:heavi}
\begin{split}
&\langle  {\timord} j_{0}(\eta, \mu_{\theta})\;; j_{\nu}(-\eta, \phi_{\theta,\ell})  \rangle_{\beta, \mu, L} \\
&\quad = \beta \int_{-\frac{\beta}{2}}^{\frac{\beta}{2}}  d s\,e^{-is\eta} \big[ \Theta(s\ge 0) \big\langle \gamma_{s}\big( j_{0}(\mu_{\theta}) \big)\;; j_{\nu}(\phi_{\theta,\ell})  \big\rangle_{\beta,\mu,L} \\
&\quad\qquad  + \Theta(s\le 0) \big\langle j_{\nu}(\phi_{\theta,\ell})\;; \gamma_{s}\big( j_{0}(\mu_{\theta}) \big)  \big\rangle_{\beta,\mu,L} \big]\;,
 \end{split}
\end{equation}
with $\Theta(\cdot)$ the Heaviside function\footnote{The $s=0$ ambiguity in (\ref{eq:heavi}) has no relevance, since all operators are bounded, and it involves a zero measure set of times.}. Hence,
\begin{equation}
\begin{split}
&\eta\langle  \timord j_{0}(\eta, \mu_{\theta})\;; j_{\nu}(-\eta, \phi_{\theta,\ell})  \rangle_{\beta, \mu, L} \\
&= \beta \int_{-\frac{\beta}{2}}^{\frac{\beta}{2}}\mr  ds\, (i\partial_{s}\mr e^{-is\eta})  \Big[ \Theta(s\ge 0) \big\langle \gamma_{s}\big( j_{0}(\mu_{\theta}) \big)\;; j_{\nu}(\phi_{\theta,\ell})  \big\rangle_{\beta,\mu,L}\\
&\qquad + \Theta(s\le 0) \big\langle j_{\nu}(\phi_{\theta,\ell})\;; \gamma_{s}\big( j_{0}(\mu_{\theta}) \big)  \big\rangle_{\beta,\mu,L} \Big]\;;
\end{split}
\end{equation}
integrating by parts, using the KMS identity $\langle A \gamma_{s}(B) \rangle = \langle \gamma_{s+\beta}(B) A  \rangle$ to show the cancellation of the boundary terms, and using the continuity equation (\ref{eq:cons2}):
\begin{equation}
\begin{split}
&\eta \langle  \timord j_{0}(\eta, \mu_{\theta})\;; j_{\nu}(-\eta, \phi_{\theta,\ell})  \rangle_{\beta, \mu, L} \\
&= \sum_{k=1,2} \langle  \timord (\text{d}_{x_{k}}j_{k})(\eta, \mu_{\theta})\;; j_{\nu}(-\eta, \phi_{\theta,\ell})  \rangle_{\beta, \mu, L} -i\beta \langle [ j_{0}(\mu_{\theta}),  j_{\nu}(\phi_{\theta,\ell}) ] \rangle_{\beta,\mu,L}\;.
\end{split}
\end{equation}
Next, observe that:
\begin{equation}
\begin{split}
\sum_{k=1,2} (\text{d}_{x_{k}}j_{k})(\eta, \mu_{\theta}) &= \theta\int_{-\frac{\beta}{2}}^{\frac{\beta}{2}}  d s\, e^{-i\eta s} \sum_{\vec x \in \Lambda_{L} } \sum_{k=1,2}\mu(\theta\vec x) \text{d}_{x_{k}}\gamma_{s}(j_{k,\vec x}) \\
&= - \theta\int_{-\frac{\beta}{2}}^{\frac{\beta}{2}} d s\, e^{-i\eta s}\sum_{\vec x \in \Lambda_{L}} \sum_{k=1,2}  \big(\text{d}_{x_{k}}\mu(\theta \vec x)\big) \gamma_{s}(j_{k,\vec x}) \\
& = -\theta^{2} \int_{-\frac{\beta}{2}}^{\frac{\beta}{2}} d s\, e^{-i\eta s}\sum_{\vec x \in \Lambda_{L}} \sum_{k=1,2} \mu_{\theta; k}(\theta \vec x) \gamma_{s}(j_{k;\vec x})
\end{split}
\end{equation}
where in the second equality we summed by parts and used the cylindric boundary conditions, while in the last equality we introduced the function:
\begin{equation}
\mu_{\theta;k}(\theta\vec  x) := \frac{ \mu(\theta \vec x) - \mu\big(\theta(\vec x - \vec{e}_{k})\big)}{\theta}\;,
\end{equation}
which has similar qualitative properties as $\mu(\theta \vec x)$. Thus, we obtained the following identity:
\begin{equation}\label{eq:WI}
\begin{split}
&\eta \langle  \timord j_{0}(\eta, \mu_{\theta})\;; j_{\nu}(-\eta, \phi_{\theta,\ell})  \rangle_{\beta, \mu, L} \\
&\quad = - \theta \sum_{k=1,2} \langle  \timord j_{k}(\eta, \mu_{\theta;k})\;; j_{\nu}(-\eta, \phi_{\theta,\ell})  \rangle_{\beta, \mu, L} - i \beta \langle [ j_{0}(\mu_{\theta}), j_{\nu}( \phi_{\theta,\ell}) ] \rangle_{\beta,\mu,L}\;.
\end{split}
\end{equation}
Eq. (\ref{eq:WI}) is the Ward identity for the current-current correlation function.
\begin{remark} The Ward identity (\ref{eq:WI}) immediately implies that the response functions are vanishing if one first takes the limit $\theta \to 0$ and then $\eta \to 0$.
\end{remark}

\subsubsection{Implications of the current-current Ward identity}\label{sec:JJ}
The Ward identity (\ref{eq:WI}) allows to derive important constraints on the current-current correlation functions. Let:
\begin{equation}
K^{\beta,L}_{\mu\nu}(\eta,\theta) :=  \frac{1}{\theta \beta}\langle  \timord j_{\mu}(\eta, \mu_{\theta;\mu})\;; j_{\nu}(-\eta, \phi_{\theta,\ell})  \rangle_{\beta, \mu, L}\;,
\end{equation}
with $\mu_{\theta;0} = \mu_{\theta}$. Observe that, for $\eta \in \frac{2\pi}{\beta} \mathbb{N}$:
\begin{equation}\label{eq:chiK}
\chi_{\nu}^{\beta,L}(\eta,\theta) = K^{\beta,L}_{0\nu}(\eta,\theta)\;;
\end{equation}
for general $\eta>0$, the identity holds with the additive error $\varepsilon_{\nu}^{\beta,L}(\eta,\theta)$, as in Eq. (\ref{eq:wick}). Suppose that $K^{\beta,L}_{\mu\nu}(\eta,\theta)$ satisfies the estimate:
\begin{equation}\label{eq:KKlog}
| K^{\beta,L}_{\mu\nu}(\eta,\theta) | \le C_{\ell} |\log \theta|\;,
\end{equation}
uniformly in $\eta, \beta, L_{1}, L_{2}$. This bound will be proven later, as a consequence of the result on the two-point function, Theorem \ref{thm:2pt}, and of the support properties of the test functions. Also, suppose that the following decomposition holds:
\begin{equation}\label{eq:singreg}
K^{\beta,L}_{\mu\nu}(\eta,\theta) = K^{\text{sing}}_{\mu\nu}(\eta,\theta) + K^{\text{reg}}_{\mu\nu}(\eta,\theta)\;,
\end{equation}
where, for $\alpha>0$:
\begin{equation}\label{eq:contK}
\big|K^{\text{reg}}_{\mu\nu}(\eta,\theta) - K^{\text{reg}}_{\mu\nu}(\eta,\theta')\big|\le C_{\ell} (|\theta|^{\alpha} + |\theta'|^{\alpha})\;.
\end{equation}
The constant $C_{\ell}$ might depend on $\ell$, entering the definition of the test function $\phi$, but it is independent from all the other parameters. We shall call $K^{\text{reg}}_{\mu\nu}$ the regular part of $K^{\beta,L}_{\mu\nu}$, while $K^{\text{sing}}_{\mu\nu}$, which may not satisfy (\ref{eq:contK}), is called the singular part. As we will see later, Theorem \ref{thm:2pt} will imply the decomposition (\ref{eq:singreg}), with a relatively explicit form for $K^{\text{sing}}_{\mu\nu}(\eta,\theta)$. 

Let $\nu = 0$. By (\ref{eq:WI}), we get:
\begin{equation}\label{eq:WI2}
\eta K_{00}^{\text{sing}}(\eta, \theta) + \eta K_{00}^{\text{reg}}(\eta, \theta) = - \theta \sum_{k=1,2} K^{\beta,L}_{k0}(\eta,\theta)\;.
\end{equation}
This identity implies that:
\begin{equation}\label{eq:WIKcons}
K_{00}^{\text{reg}}(\eta, \theta) = -K_{00}^{\text{sing}}(\eta,\theta) - \frac{\theta}{\eta} \sum_{k=1,2} K^{\beta,L}_{k0}(\eta,\theta)\;.
\end{equation}
By (\ref{eq:contK}), we have:
\begin{equation}
\begin{split}
K_{00}^{\text{reg}}(\eta, \theta) &= K_{00}^{\text{reg}}(\eta, \eta^{2}) + \mf{e}_{1}(\eta,\theta) \\
&= - K_{00}^{\text{sing}}(\eta, \eta^{2}) + \mf{e}_{2}(\eta,\theta)\;,
\end{split}
\end{equation}
where:
\begin{equation}\label{eq:eiest}
| \mf{e}_{i}(\eta,\theta) | \le C_{\ell}(|\theta|^{\alpha} + |\eta|^{\alpha})\;,\qquad i =1,2\;;
\end{equation}
the bound on the error terms follows from (\ref{eq:KKlog}), (\ref{eq:contK}) and from (\ref{eq:WIKcons}). Therefore, we obtained the following rewriting of the susceptibility, recall Eqs. (\ref{eq:wick}), (\ref{eq:chiK}):
\begin{equation}\label{eq:chi0WI}
\begin{split}
 K^{\beta,L}_{00}(\eta, \theta)= \big( K^{\text{sing}}_{00}(\eta, \theta) - K^{\text{sing}}_{00}(\eta, \eta^{2})\big) + E^{\beta,L}_{00}(\eta,\theta)\;,
\end{split}
\end{equation}
where $E^{\beta,L}_{00}(\eta,\theta)$ is bounded as in (\ref{eq:eiest}). This rewriting is particularly useful, 
due to the relatively explicit form of $K^{\text{sing}}_{00}$ that we will obtain.

A similar rewriting can be obtained for $\chi_{1}(\eta,\theta)$. Using (\ref{eq:WI}), we have:
\begin{equation}
\eta K_{01}^{\text{sing}}(\eta, \theta) + \eta K_{01}^{\text{reg}}(\eta, \theta) = - \theta \sum_{j=1,2} K^{\beta,L}_{j1}(\eta,\theta) + \Delta_{1}(\theta)\;,
\end{equation}
where:
\begin{equation}
\Delta_{1}(\theta) = -\frac{i}{\theta} \langle [ j_{0}(\mu_{\theta}), j_{\nu}( \phi_{\theta,\ell}) ] \rangle_{\beta,\mu,L}\;.
\end{equation} 
We claim that:
\begin{equation}\label{eq:bdS}
|\Delta_{1}(\theta)| \le C |\theta| \ell\;,
\end{equation}
where $\ell$ is the width of the strip in proximity of the edge, entering the definition of $\phi_{\theta,\ell}$. In fact:
\begin{equation}
[ j_{0}(\mu_{\theta}), j_{\nu}( \phi_{\theta,\ell}) ] = \theta^{2} \sum_{\vec x,\vec y } \mu(\theta \vec x) \phi(\theta y_{1}, y_{2}) [ n_{\vec x}, j_{1;\vec y} ]\;.
\end{equation}
Using that $j_{1,x}$ commutes with the number operator, we can rewrite:
\begin{equation}
\begin{split}
[ j_{0}(\mu_{\theta}), j_{\nu}( \phi_{\theta,\ell}) ] &= \theta^{2} \sum_{\vec x,\vec{y} } \big(\mu(\theta \vec x) - \mu(\theta \vec y)\big) \phi(\theta y_{1}, y_{2}) [ n_{\vec x}, j_{1;\vec y} ] \\
&= \theta^{2} \sum_{\substack{\vec x,\vec y  \\ \|\vec x-\vec y\|\le \sqrt{2}}} \big(\mu(\theta \vec x) - \mu(\theta y)\big) \phi(\theta y_{1}, y_{2}) [ n_{\vec x}, j_{1;\vec y} ]\;,
\end{split}
\end{equation}
where in the last step we used that the commutator between the density and the current operator is zero if $\|\vec x-\vec y\|>\sqrt{2}$. Also, from the regularity properties of $\mu(\cdot)$ we have $|\mu(\theta \vec x) - \mu(\theta  \vec y)| \le C \theta \|\vec x-\vec y\|$. Thus, we have:
\begin{equation}
\begin{split}
\big\| \big[ j_{0}(\mu_{\theta}), j_{\nu}( \phi_{\theta,\ell}) \big] \big\| &\le K|\theta|^{3} \sum_{\vec y } \varphi(\theta y_{1}, y_{2}) \\
&\le C\theta^{2} \ell\;,
\end{split}
\end{equation}
 by the support properties of $\varphi$. This implies (\ref{eq:bdS}). Therefore, we can repeat the strategy followed for $\chi_{0}(\eta,\theta)$, and we obtain: 

\begin{equation}\label{eq:chi1WI}
\begin{split}
K^{\beta,L}_{01}(\eta, \theta) = \big( K^{\text{sing}}_{01}(\eta, \theta) - K^{\text{sing}}_{01}(\eta, \eta^{2})\big) + E^{\beta,L}_{01}(\eta,\theta)\;,
\end{split}
\end{equation}
where $E^{\beta,L}_{01}(\eta,\theta)$ satisfies the bound (\ref{eq:eiest}). Eqs. (\ref{eq:chi0WI}), (\ref{eq:chi1WI}) are the main result of this section, and will allow to explicitly compute the response functions.

\subsection{Vertex Ward identity}\label{sec:vertWI}

Let us introduce the vertex function:
\begin{equation}
S_{3;\mu,\sigma,\zeta}(\vec{\bm w},\vec{\bm x},\vec{\bm y}) := \big\langle \timord\, \gamma_{w_0}(j_{\mu, \vec w})\;; \gamma_{x_{0}}(a_{{\vec x},\sigma}) \gamma_{y_{0}}(a^{*}_{{\vec y},\zeta}) \big\rangle_{\beta, \mu, L}\;.
\end{equation}
We will derive an identity for the vertex function, starting from the lattice continuity equation (\ref{eq:cons2}) in imaginary times. We compute:
\begin{equation}\label{eq:vertexWI}
\begin{split}
&i\partial_{w_0} \big\langle \timord\, \gamma_{w_0}(j_{0, \vec w})\;; \gamma_{x_{0}}(a_{{\vec x},\sigma}) \gamma_{y_{0}}(a^{*}_{{\vec y},\zeta}) \big\rangle_{\beta, \mu, L} \\ 
&\quad = -\sum_{\nu=1,2} \text{d}_{w_{\nu}} \big\langle \timord\, \gamma_{w_0}(j_{\nu, \vec w})\;; \gamma_{x_{0}}(a_{{\vec x},\sigma}) \gamma_{y_{0}}(a^{*}_{{\vec y},\zeta}) \big\rangle_{\beta, \mu, L}\\
&\qquad -i \big( \delta({\vec {\bm x}} - {\vec {\bm w}})- \delta({\vec {\bm y}} - {\vec {\bm w}})    \big) \langle \timord \gamma_{x_0}(a_{{\vec x},\sigma}) \gamma_{y_0}(a^{*}_{{\vec y},\zeta}) \rangle_{\beta, \mu, L}\;.
\end{split}
\end{equation}
Eq. (\ref{eq:vertexWI}) is the vertex Ward identity. Summing over $w_{2}$ from $1$ to $L_{2}$ and using the Dirichlet boundary conditions, omitting the $\beta,\mu,L$ labels:
\begin{equation}
\begin{split}
&i\partial_{w_0} \big\langle \timord\, \gamma_{w_0}(j_{0, w_{1}})\;; \gamma_{x_{0}}(a_{{\vec x},\sigma}) \gamma_{y_{0}}(a^{*}_{{\vec y},\zeta}) \big\rangle \\
&= -\text{d}_{w_{1}}\big\langle \timord\, \gamma_{w_0}(j_{1; w_{1}})\;; \gamma_{x_{0}}(a_{{\vec x},\sigma}) \gamma_{y_{0}}(a^{*}_{{\vec y},\zeta}) \big\rangle \\
& \mkern30mu { -} i\big(\delta({\vec {\bm x}} - {\vec {\bm w}})-  \delta(\vec{\bm y} - \vec{\bm w})   \big) \big\langle \timord \gamma_{x_0}(a_{x_{1}, w_{2}; \sigma}) \gamma_{y_0}(a^{*}_{y_{1}, w_{2},\zeta}) \big\rangle\;,
\end{split}
\end{equation}
where we recall:
\begin{equation}\label{eq:cursum}
j_{\nu,x_{1}} := \sum_{x_{2} = 1}^{L_{2}} j_{\nu,\vec x}\;.
\end{equation}
Next, let us introduce the partial Fourier transform of the vertex function:
\begin{equation}
\begin{split}
&\hat S_{3; \mu,\sigma,\zeta}(\bm{p}, \bm{k}, \bm{h}; x_{2}, y_{2}) \\
&\quad:= \frac{1}{\beta L_{1}} \int_{[0,\beta]^{3}}\mr  d w_0  d x_0  d y_0\Big(\sum_{w_{1}, x_{1}, y_{1}} e^{-i\bm{p} \cdot \bm{w} - i\bm{k} \cdot \bm{x} + i\bm{h} \cdot \bm{y} } \\
&\quad\qquad \cdot \big\langle \timord\, \gamma_{w_0}(j_{\mu; w_{1}})\;; \gamma_{x_{0}}(a_{{\vec x},\sigma}) \;; \gamma_{y_{0}}(a^{*}_{{\vec y},\zeta}) \big\rangle\Big)\; ,
\end{split}
\end{equation}
and the partial Fourier transform of the two-point function:
\begin{equation}
\begin{split}
&\hat S_{2;\sigma,\zeta}(\bm{k}, \bm{h}; x_{2}, y_{2}) \\
&\quad := \frac{1}{\beta L_{1}} \int_{[0,\beta]^{2}}\mr  dx_{0}\mr  dy_{0} \sum_{x_{1}, y_{1}} e^{-i\bm{k} \cdot \bm{x} + i\bm{h} \cdot \bm{y} }\langle \timord \gamma_{x_{0}}(a_{{\vec x},\sigma}) \gamma_{y_{0}}(a^{*}_{{\vec y},\zeta}) \big\rangle\;.
\end{split}
\end{equation}
Then, the vertex Ward identity (\ref{eq:vertexWI}) can be rewritten as
\begin{equation}\label{eq:WIvert}
\begin{split}
&-p_{0} \hat S_{3; 0,\sigma,\zeta}(\bm{p}, \bm{k}, \bm{h}; x_{2}, y_{2}) + (1 - e^{-ip_{1}}) \hat S_{3; 1,\sigma,\zeta}(\bm{p}, \bm{k}, \bm{h}; x_{2}, y_{2}) \\
&\qquad =i \hat S_{2;\sigma,\zeta}(\bm{k}, \bm{h} -\bm{p}; x_{2}, y_{2})- i \hat S_{2;\sigma,\zeta}(\bm{k}+  \bm{p}, \bm{h}; x_{2}, y_{2})\;.
\end{split}
\end{equation}
In order to extract useful information from this identity, it will be convenient to rewrite the current operator $j_{1;z_{1}}$ as, recalling (\ref{eq:currdef}) and (\ref{eq:cursum}), and using the Dirichlet boundary conditions:
\begin{equation}\label{eq:jsumx2}
j_{1;z_{1}} = \sum_{z_{2} = 1}^{L_{2}} \big( j_{\vec z,\vec z+\vec e_{1}} + j_{\vec z, \vec z+ \vec e_{1} - \vec e_{2}} + j_{\vec z, \vec z+\vec e_{1} + \vec e_{2}} \big)\;.
\end{equation}
Taking the Fourier transform:
\begin{equation}\label{eq:currentkernel}
\begin{split}
\hat j_{1;p_{1}} &= \sum_{z_{1} = 1}^{L_{1}} e^{-ip_1  z_{1}} j_{1;z_{1}} \\
& = \frac{i}{L_{1}} \sum_{k_{1} } \Big( e^{ik_{1}} \big( \hat a^{*}_{k_1-p_{1}}, H(-\vec e_{1}) \hat a_{k_1} \big) - e^{-i (k_{1} - p_{1})} \big( \hat a^{*}_{k_{1} - p_{1}}, H(\vec e_{1}) \hat a_{k_{1}} \big)\Big)\;,
\end{split}
\end{equation}
where $H(\vec e_{1})$ is the operator with kernel $H_{\sigma\zeta}(\vec e_{1}; x_{2}, y_{2})$. Recall that the Hamiltonian is translation invariant in the $x_{1}$ direction, which means that $H_{\sigma\zeta}(\vec e_{1}; x_{2}, y_{2}) \equiv H_{\sigma\zeta}\big((x_{1} + 1, x_{2}); (x_{1}, y_{2})\big)$, and we used the notation
\begin{equation}
( f, A g ) := \sum_{\sigma,\zeta} \sum_{x_{2}, y_{2}} f_{\sigma,x_{2}} A_{\sigma\zeta}(x_{2},y_{2}) g_{\zeta,y_{2}}\;.
\end{equation}
Observe that, for $p_{1} = 0$:
\begin{equation}\label{eq:currentkernelzero}
\hat j_{1;0} = \frac{1}{L_{1}} \sum_{k_{1}  } \Big( \hat a^{*}_{k_{1}}, \big(\partial_{k_{1}} \hat H(k_{1})\big) \hat a_{k_{1}} \Big)\;,
\end{equation}
where:
\begin{equation}
\partial_{k_{1}} \hat H(k_{1}; x_{2}, y_{2}) = \partial_{k_{1}} \sum_{|x_{1}|\le 1} e^{-ik_{1} x_{1}} H(x_{1}, 0; x_{2}, y_{2})\;.
\end{equation}
\subsubsection{Implications of the vertex Ward identity}\label{sec:vertexcons}
Here we shall use the vertex Ward identity (\ref{eq:WIvert}) in order to prove the relations (\ref{eq:relzeta}) between the renormalized parameters appearing in the leading contribution to the two-point function, Eqs. (\ref{eq:2ptmain}), (\ref{eq:2ptmain2}). These relations will be important in the computations of the response functions. To do this, we will resolve the main contribution to the vertex function for small external momenta. We will then compare the resulting expression with the two-point function at low momenta, via the Ward identity (\ref{eq:WIvert}).

Let us rewrite the operators $\hat j_{0;p_{2}}$, $\hat j_{1;p_{1}}$ as:
\begin{equation}\label{eq:currentkernels}
\begin{split}
\hat j_{\mu;p_{1}} &= \frac{1}{L_{1}} \sum_{k_{1}} \big( \hat a^{*}_{k_{1} - p_{1}}, \hat J_{\mu}(k_{1},p_{1}) \hat a_{k_{1}} \big)\;, \\
\hat J_{0;\sigma,\zeta}(k_{1},p_{1};x_{2}, y_{2}) &= \delta_{\sigma,\zeta}\delta_{x_{2},y_{2}}\;,\\
\hat J_{1;\sigma,\zeta}(k_{1},p_{1};x_{2},y_{2}) &= i e^{ik_{1}} H_{\sigma,\zeta}(-\vec e_{1}; x_{2},y_{2}) - i e^{-i(k_{1}-p_{1})} H_{\sigma,\zeta}(\vec e_{1}; x_{2},y_{2})\;.
\end{split}
\end{equation}
In what follows, we will choose external momenta $\bm{k}, \bm{p}$ such that, for some $M> 1$:
\begin{equation}\label{eq:qqkk}
\big|k_{1} - k_{F}^{+}(\lambda)\big|^{M} \le |k_{0}|\le |k_{1} - k_{F}^{+}(\lambda)|\;,\quad |p_{1}| \leq \frac{1}{2}| k_{1} - k_{F}^{+}(\lambda) |\;,\qquad |p_{0}| \leq \frac{1}{2}|k_{0}|\;. 
\end{equation}
We recall that $k_{F}^{+}$ is the Fermi momentum of the edge mode localized in proximity of the boundary at $x_{2} = 0$. %Also, we assume that:
%
%\begin{equation}\label{eq:pcond}
%\| {\bm p} \| \le \frac{1}{2}\| {\bm k} - {\bm k}_{F}^{+}(\lambda) \|\;,\qquad |p_{0}| \le \frac{1}{2}|k_{0}|\;.
%\end{equation}
%
\begin{lemma}[Implication of the vertex Ward identity]\label{lem:vert} Under the same assumptions of Theorem \ref{thm:2pt}, the following holds. Let:
\begin{equation}
\zeta_{\mu,+} := \sum_{n} \langle Z_{-n,+}, J_{\mu}(k_{F}^{+}(\lambda) + n\alpha,0) Z_{-n,+} \rangle\;,
\end{equation}
and let $v_{\mu,+}$ be the parameters appearing in the asymptotics of the two-point function (\ref{eq:2ptmain2}). Then:
\begin{equation}\label{eq:vertrel}
\zeta_{0,+} = v_{0,+}\;,\qquad \zeta_{1,+} = v_{1,+}\;.
\end{equation}
\end{lemma}
Lemma \ref{lem:vert} proves the relations (\ref{eq:relzeta}), and concludes the proof of Theorem \ref{thm:2pt} about the two-point function. These relations will play an important role in the computation of the transport coefficients, discussed in the next section. The analysis can be adapted to prove the relations $\zeta_{0,-} = v_{0,-}$, $\zeta_{1,-} = v_{1,-}$, which however will not be needed in the following.

The proof of Lemma \ref{lem:vert} is based on the two following propositions. In what follows, we shall use the notation ${\bm q} = {\bm k} - {\bm k}_{F}^{+}$.
\begin{proposition}[Asymptotics of the vertex function]\label{prop:vertasy} Assume (\ref{eq:qqkk}). Then, for $\| {\bm k} - {\bm k}_{F}^{+}(\lambda) \| \ll 1$, $\| {\bm p} \| \ll 1$ and $L_{2}$ large enough:
\begin{equation}
\begin{split}
\hat S_{3; \mu,\sigma,\zeta}(\bm{p}, \bm{k}, \bm{k} + {\bm p}; x_{2}, y_{2}) &= Z_{0,+,\sigma}(x_{2}) \overline{Z_{0,+,\zeta}(y_{2})} g_{+;\mr{s}}({\bm q}) g_{+;\mr{s}}({\bm q} + {\bm p}) \zeta_{\mu,+} \\
&\quad + R^{\text{tot}}_{3; \mu,\sigma,\zeta}(\bm{p}, \bm{k}, \bm{k} + {\bm p}; x_{2}, y_{2})\;,
\end{split}
\end{equation}
where:
\begin{equation}\label{eq:zetamu}
\zeta_{\mu,+} := \sum_{n} \langle Z_{-n,+}, J_{\mu}(k_{F}^{+}(\lambda) + n\alpha,0) Z_{-n,+} \rangle\;,
\end{equation}
and:
\begin{equation}\label{eq:estvert}
\begin{split}
&\sum_{x_{2}, y_{2}} \big| Z_{0,+,\sigma}(x_{2})\big| \big| Z_{0,+,\zeta}(y_{2})\big| \Big|g^{-1}_{+;\mr{s}}({\bm q})\Big| \Big|g^{-1}_{+;\mr{s}}({\bm q} + {\bm p})\Big|\\&\quad \cdot \Big|R^{\text{tot}}_{3; \mu,\sigma,\zeta}(\bm{p}, \bm{k}, \bm{k} + {\bm p}; x_{2}, y_{2})\Big| = \mr{o}(1)\;.
\end{split}
\end{equation}
\end{proposition}
\begin{proposition}[Asymptotics of the two-point function]\label{prp:asy2pt} Assume (\ref{eq:qqkk}). Then, for $\| {\bm k} - {\bm k}_{F}^{+}(\lambda) \| \ll 1$ and $L_{2}$ large enough:
\begin{equation}\label{eq:Rn00}
\hat S_{2;\sigma,\zeta}(\bm{k},{\bm k} ; x_{2}, y_{2}) = Z_{0,+,\sigma}(x_{2}) \overline{Z_{0,+,\zeta}(y_{2})} g_{+;\mr{s}}({\bm q}) + \hat R^{\text{tot}}_{\sigma,\zeta}({\bm k}; x_{2}, y_{2})\;,
\end{equation}
with:
\begin{equation}\label{eq:Rn0}
\Big| \text{d}_{k_{0}}^{n_{0}} \text{d}_{k_{1}}^{n_{1}} \hat R^{\text{tot}}_{\sigma,\zeta}({\bm k}; x_{2}, y_{2}) \Big| \leq C_{n_{0},n_{1}}\frac{e^{-\tilde c|x_{2} - y_{2}|}}{\| {\bm k} - {\bm k}_{F}^{+} \|^{n_{0} + n_{1} + 1-\theta}} + \sum_{\substack{\omega_{1},\omega_{2} \\ (\omega_{1},\omega_{2}) \neq (+,+)}} \!\!\! C_{n_{0},n_{1}}\frac{e^{-c(|x_{2}|_{\omega_{1}} + |y_{2}|_{\omega_{2}})}}{|k_{0}|^{n_{0} + n_{1} + 1}}\;.
\end{equation}
%
%\begin{equation}\label{eq:diffR2pt}
%\begin{split}
%&\sum_{x_{2}, y_{2}} \big| Z_{0,+,\sigma}(x_{2})\big| \big| Z_{0,+,\zeta}(y_{2})\big| \Big|g^{-1}_{+;\mr{s}}({\bm q})\Big| \Big|g^{-1}_{+;\mr{s}}({\bm q} + {\bm p})\Big| \\&\quad \cdot \Big| \hat R^{\text{tot}}_{\sigma,\zeta}({\bm q} + {\bm k}_{F}^{+}, {\bm q} + {\bm k}_{F}^{+}; x_{2}, y_{2}) - \hat R^{\text{tot}}_{\sigma,\zeta}({\bm q} + {\bm p} + {\bm k}_{F}^{+}, {\bm q} + {\bm p} +{\bm k}_{F}^{+}; x_{2}, y_{2})\Big| = \mr{o}(\|{\bm p}\|)\;.
%\end{split}
%\end{equation}
%
\end{proposition}
Before discussing the proofs of these propositions, let us see how they imply Lemma \ref{lem:vert}.
\begin{proof} (of Lemma \ref{lem:vert}.) Suppose that the assumptions of Propositions \ref{prop:vertasy}, \ref{prp:asy2pt} hold true. Consider the left-hand side of the Ward identity (\ref{eq:WIvert}), for external momenta satisfying (\ref{eq:qqkk}). We have, from Proposition \ref{lem:vert},
\begin{equation}\label{eq:WIvertcons}
\begin{split}
&\sum_{\substack{x_{2}, y_{2} \\ \sigma, \zeta}}  \overline{Z_{0,+,\sigma}(x_{2})} Z_{0,+,\zeta}(y_{2}) \\&\quad \cdot \Big(-p_{0}\hat S_{3; 0,\sigma,\zeta}(\bm{p}, \bm{k}, \bm{k} + {\bm p}; x_{2}, y_{2}) + (1 - e^{-ip_{1}}) \hat S_{3; 1,\sigma,\zeta}(\bm{p}, \bm{k}, \bm{k} + {\bm p}; x_{2}, y_{2})\Big)\\
&= \| Z_{0,+} \|^{4} (-p_{0} \zeta_{0,+} + ip_{1}\zeta_{1,+}) g_{+;\mr{s}}({\bm q}) g_{+;\mr{s}}({\bm q} + {\bm p}) + \mathcal{E}_{1}({\bm k},{\bm p})
\end{split}
\end{equation}
where the error term satisfies:
\begin{equation}
\big| g^{-1}_{+;\mr{s}}({\bm q})\big| \big|g^{-1}_{+;\mr{s}}({\bm q} + {\bm p})\big| \big|\mathcal{E}_{1}({\bm k},{\bm p})\big| = o(\|{\bf p}\|)\;.
\end{equation}
Next, consider the right-hand side of the Ward identity (\ref{eq:WIvert}). We have, by Proposition \ref{prp:asy2pt}
\begin{equation}\label{eq:WIvertcons2}
\begin{split}
&\sum_{\substack{x_{2}, y_{2} \\ \sigma, \zeta}}  \overline{Z_{0,+,\sigma}(x_{2})} Z_{0,+,\zeta}(y_{2}) \\
&\quad \cdot \Big(i \hat S_{2;\sigma,\zeta}(\bm{k}, \bm{k}; x_{2}, y_{2})- i \hat S_{2;\sigma,\zeta}(\bm{k}+  \bm{p}, \bm{k}+  \bm{p}; x_{2}, y_{2})\Big) \\
&= \| Z_{0,+} \|^{4} \big( g_{+;\mr{s}}({\bm q}) - g_{+;\mr{s}}({\bm q} + {\bm p}) \big) + \mathcal{E}_{2}({\bm k},{\bm p})\;,
\end{split}
\end{equation}
where the error term satisfies:
\begin{equation}
\big| g^{-1}_{+;\mr{s}}({\bm q})\big| \big|g^{-1}_{+;\mr{s}}({\bm q} + {\bm p})\big| \big|\mathcal{E}_{2}({\bm k},{\bm p})\big| = o(\|{\bf p}\|)\;.
\end{equation}
Here we used that the difference of the error terms in (\ref{eq:Rn00}) can be estimated as $\|{\bm p}\|$ times the derivative of the error term, which is estimated in (\ref{eq:Rn0}). Then, observing that, for momenta small enough so that the cutoff function in the definition of $g_{+;\mr{s}}$ is $1$,
\begin{equation}
g_{+;\mr{s}}({\bm q}) - g_{+;\mr{s}}({\bm q} + {\bm p}) = (i v_{0,+} p_{0} + v_{1,+} p_{1}) g_{+;\mr{s}}({\bm q}) g_{+;\mr{s}}({\bm q} + {\bm p})\;,
\end{equation}
and using that $\| Z_{0,+} \|^{2} = 1 + O(\lambda)$, we have, equating (\ref{eq:WIvertcons}), (\ref{eq:WIvertcons2}):
\begin{equation}\label{eq:WIvertcons3}
\begin{split}
&(-p_{0} \zeta_{0,+} + ip_{1}\zeta_{1,+})  \\
&\quad = (-v_{0,+} p_{0} + i v_{1,+} p_{1})  + g_{+;\mr{s}}({\bm q})^{-1} g_{+;\mr{s}}({\bm q} + {\bm p})^{-1}\Big(\mathcal{E}_{2}({\bm k},{\bm p}) -  \mathcal{E}_{1}({\bm k},{\bm p})\Big)\;.
\end{split}
\end{equation}
The relations (\ref{eq:vertrel}) follows from (\ref{eq:WIvertcons3}) by choosing $|p_{1}| \ll |p_{0}|$ resp. $|p_{0}| \ll |p_{1}|$.
\end{proof}
We are left with proving Propositions \ref{prop:vertasy}, \ref{prp:asy2pt}.
\begin{proof} (of Proposition \ref{prop:vertasy}.) By the Wick rule, the vertex function in the left-hand side of (\ref{eq:WIvert}) can be written in terms of the two-point function as:
\begin{equation}
\begin{split}\label{eq:V3}
\hat  S_{3; \mu,\sigma,\zeta}(\bm{p}, \bm{k}, \bm{h}; x_{2}, y_{2}) &= \sum_{\bm{r}} \sum_{\rho,\varrho}  \sum_{z_{2},w_{2}}\big(\hat J_{\mu;\rho,\varrho}(r_{1},p_{1};z_{2},w_{2}) \\
&\qquad \cdot \hat S_{2;\sigma,\rho} (\bm{k},\bm{r} - \bm{p}, x_{2},z_{2})  \hat S_{2;\varrho,\zeta} ( \bm{r}, \bm{h},  w_{2}, y_{2} )\big)\;.
\end{split}
\end{equation}
Recall the decomposition (\ref{eq:2ptmain}) of the two-point function. In Fourier space, it reads:
\begin{equation}\label{eq:S2Fou}
\hat S_{2;\sigma,\zeta} (\bm{k}, {\bm h}; x_{2}, y_{2}) = \hat S^{\mr{s}}_{2;\sigma,\zeta} (\bm{k}, {\bm h}; x_{2}, y_{2}) + \hat R_{\sigma,\zeta}({\bm k}, {\bm h}; x_{2}, y_{2})\;;
\end{equation}
both sides can be non-zero only if ${\bm h} = {\bm k} + n{\bm \alpha}$ for some integer $n$, and we set $\hat S^{\mr{s}}_{2;\sigma,\zeta} (\bm{k}, {\bm k} + n{\bm \alpha}; x_{2}, y_{2}) \equiv \hat S^{\mr{s}}_{2;n,\sigma,\zeta} (\bm{k}; x_{2}, y_{2})$. The main term in (\ref{eq:S2Fou}) can be written as, see Eq. (\ref{eq:Ss}):
\begin{equation}\label{eq:ZgZ}
\hat S^{\mr{s}}_{2;n,\sigma,\zeta} \big(\bm{k}; x_{2}, y_{2}\big) = \sum_{m,\omega} Z_{m,\omega,\sigma}(x_{2}) g_{\omega;\mr{s}}({\bm q} + m{\bm \alpha})  \overline{Z_{-n+m,\omega,\zeta}(y_{2})}\;,
\end{equation}
where ${\bm q} = {\bm k} - {\bm k}_{F}^{\omega}(\lambda)$. Concerning the term $\hat R$ in (\ref{eq:S2Fou}), it is not explicit; the only information we will actually need is the following estimate in momentum space, proved in Appendix \ref{app:Rest}:
\begin{equation}\label{eq:RbdFou}
\begin{split}
\big| \text{d}_{k_{0}}^{n_{0}} \text{d}_{k_{1}}^{n_{1}} \hat R_{\sigma,\zeta}({\bm k}, {\bm k} + n{\bm \alpha}; x_{2}, y_{2}) \big| &\le \frac{C_{n_{0},n_{1}} |\lambda|^{\delta_{n\neq 0}} e^{-\frac{c}{16}|n|} e^{-\tilde c|x_{2} - y_{2}|}}{\| {\bm k} - {\bm k}_{F}^{+} \|^{n_{0} + n_{1} + 1-\theta}}\\&\quad + \sum_{\substack{\omega_{1},\omega_{2} \\ (\omega_{1},\omega_{2}) \neq (+,+)}} \!\!\! C_{n_{0},n_{1}}\frac{e^{-\frac{c}{16}|n|} e^{-c(|x_{2}|_{\omega_{1}} + |y_{2}|_{\omega_{2}})}}{|k_{0}|^{n_{0} + n_{1} + 1-\theta}}\;,
\end{split}
\end{equation}
for values of the momentum ${\bm k}$ safisfying the assumptions of the proposition, Eq. (\ref{eq:qqkk}). We rewrite (\ref{eq:V3}) as, using that the only nonzero terms in the ${\bm r}$ sum are those such that ${\bm r} - {\bm p}= {\bm k} + n{\bm \alpha}$, and choosing ${\bm h} = {\bm k} + {\bm p}$:
\begin{equation}\label{eq:S3S2}
\begin{split}
&\hat  S_{3; \mu,\sigma,\zeta}(\bm{p}, \bm{k}, \bm{k} + {\bm p}; x_{2}, y_{2})\\
& \quad =  \sum_{n} \sum_{\rho,\varrho}  \sum_{z_{2},w_{2}}\big(\hat J_{\mu;\rho,\varrho}(k_{1} + p_{1}+ n\alpha , p_{1};z_{2},w_{2}) \\
&\quad\qquad \cdot \hat S_{2;n,\sigma,\rho} (\bm{k}, x_{2}, z_{2})  \hat S_{2;-n,\varrho,\zeta} ( \bm{k} + {\bm p} + n{\bm \alpha} ,  w_{2},  y_{2} )\big)\;. 
\end{split}
\end{equation}
Next, we write:
\begin{equation}\label{eq:S3dec}
\begin{split}
\hat  S_{3; \mu,\sigma,\zeta}(\bm{p}, \bm{k}, \bm{k} + {\bm p}; x_{2}, y_{2}) &= \hat  S^{\mr{s}}_{3; \mu,\sigma,\zeta}(\bm{p}, \bm{k}, \bm{k} + {\bm p}; x_{2}, y_{2})\\
&\qquad + \hat  R_{3; \mu,\sigma,\zeta}\big(\bm{p}, \bm{k}, \bm{k} + {\bm p}; x_{2}, y_{2}\big)\;,
\end{split}
\end{equation}
where $\hat  S^{\mr{s}}_{3}$ is obtained from $\hat  S_{3}$ replacing all $\hat S_{2}$ with $\hat S^{\mr{s}}_{2}$, while $\hat R_{3}$ collects all the error terms, where there is at least one $\hat  R$. These last terms are subleading with respect to $\hat  S^{\mr{s}}_{3}$, due to the bound (\ref{eq:RbdFou}). In particular, the term $\hat  R_{3; \mu,\sigma,\zeta}$ satisfies the estimate (\ref{eq:estvert}).

Using (\ref{eq:ZgZ}), we have:
\begin{equation}\label{eq:SS}
\begin{split}
&\sum_{n}\hat S^{\text{s}}_{2;n,\sigma,\rho} \big(\bm{k}, x_{2}, z_{2}\big)  \hat S^{\text{s}}_{2;-n,\varrho,\zeta} \big( \bm{k} + {\bm p} + n{\bm \alpha} ,  w_{2},  y_{2} \big) \\
&\quad = \sum_{n} \sum_{m_{1}, m_{2}} \sum_{\omega_{1}, \omega_{2}}Z_{m_{1},\omega_{1},\sigma}(x_{2}) g_{\omega_{1};\mr{s}}({\bm q} + m_{1}{\bm \alpha})  \overline{Z_{-n+m_{1},\omega_{1},\rho}(z_{2})} \\
&\qquad \cdot Z_{m_{2},\omega_{2},\varrho}(w_{2}) g_{\omega_{2};\mr{s}}\big({\bm q} + {\bm p} + (m_{2} + n){\bm \alpha}\big)  \overline{Z_{n+m_{2},\omega_{2},\zeta}(y_{2})}\;.
\end{split}
\end{equation}
We rewrite:
\begin{equation}
(\ref{eq:SS}) = \text{I} + \text{II} + \text{III}\;,
\end{equation}
where: the terms $\text{I}$, $\text{II}$ collect the contributions to (\ref{eq:SS}) with $\omega_{1} = \omega_{2} = +$, and $\text{III}$ all the others. Specifically, the term $\text{I}$ collects the contribution of the indices $n,m_{1},m_{2}$ such that $m_{1} = 0$ and $n+m_{2} = 0$:
\begin{equation}
\begin{split}
\text{I} &= \sum_{n} Z_{0,+,\sigma}(x_{2}) g_{+;\mr{s}}({\bm q})  \overline{Z_{-n,+,\rho}(z_{2})} \\
&\qquad \cdot Z_{-n,+,\varrho}(w_{2}) g_{+;\mr{s}}({\bm q} + {\bm p})  \overline{Z_{0,+,\zeta}(y_{2})}\;;
\end{split}
\end{equation}
instead, the term $\text{II}$ collects the combinations $m_{1}\neq 0$ and/or $n+m_{2} \neq 0$. Before discussing them, let us consider $\text{III}$. From (\ref{eq:RbdFou}) we easily get the following bound:
\begin{equation}\label{eq:calE1bd}
\big|\text{III}\big|\leq \sum_{\substack{\omega_{1}, \omega_{2} \\ (\omega_{1},\omega_{2}) \neq (+,+)}}\frac{e^{-c(|x_{2}|_{\omega_{1}} + |z_{2}|_{\omega_{1}} + |w_{2}|_{\omega_{2}} +|y_{2}|_{\omega_{2}})}}{|k_{0}| |k_{0} + p_{0}|}\;.
\end{equation}
In particular, the term $\text{III}$ satisfies the estimate (\ref{eq:estvert}).
\paragraph{Analysis of $\text{II}$.} Let us consider the contribution coming from $m_{2} + n\neq 0$:
\begin{equation}\label{eq:IIA}
\begin{split}
&\text{II}^{\text{A}} := \sum_{\substack{n,m_{1},m_{2} \\ n+m_{2} \neq 0}} Z_{m_{1},+,\sigma}(x_{2}) g_{+;\mr{s}}({\bm q} + m_{1}{\bm \alpha})  \overline{Z_{-n+m_{1},+,\rho}(z_{2})} \\
&\quad\qquad \cdot Z_{m_{2},+,\varrho}(w_{2}) g_{+;\mr{s}}\big({\bm q} + {\bm p} + (m_{2} + n){\bm \alpha}\big)  \overline{Z_{n+m_{2},+,\zeta}(y_{2})}\;,
\end{split}
\end{equation}
we rewrite the sum over $m_{1}, m_{2}, n$ as:
\begin{equation}
\sum_{m_{1}, m_{2}, n:\, n+m_{2} \neq0} (\cdots) =  \sum_{\substack{m_{1}, m_{2}, n:\, n+m_{2} \neq0 \\ m_{1} = 0}} (\cdots) + \sum_{\substack{m_{1}, m_{2}, n:\, n+m_{2} \neq0 \\ m_{1} \neq 0}} (\cdots)\;.
\end{equation}
Consider the first sum. We further rewrite it as:
\begin{equation}
\begin{split}
& \sum_{\substack{m_{1}, m_{2}, n:\, n+m_{2} \neq0 \\ m_{1} = 0}} (\cdots) =\\&\qquad  \sum_{\substack{m_{1}, m_{2}, n:\, n+m_{2} \neq0 \\ m_{1} = 0,\; \| (n + m_{2}) {\bm \alpha} \| \le \sqrt{\|{\bm q} + {\bm p}\|}}} (\cdots)\; + \sum_{\substack{m_{1}, m_{2}, n:\, n+m_{2} \neq0 \\ m_{1} = 0,\; \| (n + m_{2}) {\bm \alpha} \| > \sqrt{\|{\bm q} + {\bm p}\|}}} (\cdots)\;;
 \end{split}
\end{equation}
Let us denote by $\text{II}^{\text{A}}_{1}$ and $\text{II}^{\text{A}}_{2}$ the two corresponding contributions to $\text{II}^{\text{A}}$. Consider $\text{II}^{\text{A}}_{1}$:
\begin{equation}
\begin{split}
&\text{II}^{\text{A}}_{1} = \\
&\sum_{\substack{m_{2}, n:\, n+m_{2} \neq0 \\ \| (n + m_{2}) {\bm \alpha} \| \le \sqrt{\|{\bm q} + {\bm p}\|}}} Z_{0,+,\sigma}(x_{2}) g_{+;\mr{s}}({\bm q})  \overline{Z_{-n,+,\rho}(z_{2})} \\
&\qquad \cdot Z_{m_{2},+,\varrho}(w_{2}) g_{+;\mr{s}}\big({\bm q} + {\bm p} + (m_{2} + n) {\bm \alpha}\big)  \overline{Z_{n + m_{2},+,\zeta}(y_{2})}\;.
\end{split}
\end{equation}
By the Diophantine condition (\ref{eq:diophantine}), since $n+m_{2} \neq 0$ we have:
\begin{equation}\label{eq:nm2diof}
\| (n + m_{2}) {\bm \alpha} \| \le \sqrt{\|{\bm q} + {\bm p}\|} \Rightarrow | n + m_{2} | \ge \frac{C}{\| {\bm q} + {\bm p} \|^{\tau/2}}\;,
\end{equation}
therefore, by the estimate (\ref{eq:bdZm}) for $Z_{n+m_{2}}$:
\begin{equation}\label{eq:decZZ}
| Z_{n+m_{2},+,\zeta}(y_{2}) | \le C |\lambda|^{1-\delta_{n+m_{2},0}} e^{-\frac{c}{2} |n+m_{2}|} e^{-c y_{2}} \exp\Big(-\frac{c}{2\| {\bm q} + {\bm p}\|^{\tau/2}}\Big)\;.
\end{equation}
Strictly speaking, the bound (\ref{eq:nm2diof}) only holds for $|n + m_{2}| \le L_{1}/2$, by the approximate Diophantine condition. In the case $|n + m_{2}| > L_{1}$, we use the exponential decay in $|n+m_{2}|$ of $Z_{n+m_{2}}$ to make sure that the bound (\ref{eq:decZZ}) holds, possibly with a different $c>0$ (this is possible, since the external momenta are fixed). 

This bound allows to show that $\text{II}^{\text{A}}_{1}$ is small. In fact:
\begin{equation}
\begin{split}
&|g_{+;\mr{s}}({\bm q})| \big| g_{+;\mr{s}}({\bm q} + {\bm p} + (m_{2} + n) {\bm \alpha})\big| \Big| \overline{Z_{n + m_{2},+,\zeta}(y_{2})} \Big| \\
&\qquad \le K |\lambda|^{1-\delta_{n + m_{2},0}} e^{-\frac{c}{2} |n + m_{2}|} e^{-c y_{2}} \frac{1}{|k_{0}| |k_{0} - p_{0}|}  \exp\Big(-\frac{c}{2\| {\bm q} + {\bm p}\|^{\tau/2}}\Big)\;.
\end{split}
\end{equation}
Recalling the conditions (\ref{eq:qqkk}), we get:
\begin{equation}
\begin{split}
&|g_{+;\mr{s}}({\bm q})| \big| g_{+;\mr{s}}\big({\bm q} + {\bm p} + (m_{2} + n) {\bm \alpha}\big)\big| \Big| \overline{Z_{n + m_{2},+,\zeta}(y_{2})} \Big| \\
&\qquad \le K_{M} |\lambda|^{1-\delta_{n+m_{2},0}} e^{-\frac{c}{2} |n + m_{2}|} e^{-c y_{2}} \frac{1}{\|{\bm q}\|^{2M}} \exp\Big(-\frac{c}{2\| {\bm q} \|^{\tau/2}}\Big) \\
&\qquad \le C_{r}  |\lambda|^{1-\delta_{n+m_{2},0}} e^{-\frac{c}{2} |n + m_{2}|} e^{-c y_{2}}  \|{\bm q}\|^{r}\;,\qquad \text{for any $r\in \mathbb{N}$.}
\end{split}
\end{equation}
In particular, the term $\text{II}^{\text{A}}_{1}$ is bounded as:
\begin{equation}
\big| \text{II}^{\text{A}}_{1} \big| \le C_{r} e^{-c (x_{2} + y_{2} + w_{2} + z_{2})} \|{\bm q}\|^{r}\;.
\end{equation}
Consider now the term $\text{II}^{\text{A}}_{2}$,
\begin{equation}
\begin{split}
\text{II}^{\text{A}}_{2} &=  \sum_{\substack{m_{2}, n:\, n+m_{2} \neq0 \\ \| (n + m_{2}) {\bm \alpha} \| > \sqrt{\|{\bm q} + {\bm p}\|}}} \Big( Z_{0,+,\sigma}(x_{2}) g_{+;\mr{s}}({\bm q})  \overline{Z_{-n,+,\rho}(z_{2})} \\
&\qquad \cdot Z_{m_{2},+,\varrho}(w_{2}) g_{+;\mr{s}}\big({\bm q} + {\bm p} + (m_{2} + n) {\bm \alpha}\big)  \overline{Z_{n + m_{2},+,\zeta}(y_{2})}\Big)\;.
\end{split}
\end{equation}
Here we use that the second propagator is partially regularized by $(n+m_{2}){\bm \alpha}$. In fact, we have:
\begin{equation}
\big| g_{+;\mr{s}}\big({\bm q} + {\bm p} + (m_{2} + n){\bm \alpha}\big) \big| \le \frac{C}{\sqrt{\| {\bm q} + {\bm p} \|}}\;.
\end{equation}
Therefore, again by the estimate (\ref{eq:bdZm}):
\begin{equation}
\big| \text{II}^{\text{A}}_{2} \big| \le e^{-c (x_{2} + y_{2} + w_{2} + z_{2})}\frac{1}{\| {\bm q} \|}  \frac{1}{\sqrt{\|{\bm q} + {\bm p}\|}}\;.
\end{equation}
This concludes the analysis of the $m_{1} = 0$ contribution to $\text{II}^{\text{A}}$. Let us now discuss the contribution of the modes $m_{1} \neq 0$ to $\text{II}^{\text{A}}$; we denote by $\text{II}^{\text{A}}_{3}$  this contribution. Here we separate both $n+m_{2}$ and $m_{1}$ in $\| (n + m_{2}) {\bm \alpha} \| \le \sqrt{\|{\bm q} + {\bm p}\|}$, $\| (n + m_{2}) {\bm \alpha} \| > \sqrt{\|{\bm q} + {\bm p}\|}$ and $\| m_{1}{\bm \alpha} \| \le \sqrt{\| {\bm q} \|}$, $\| m_{1}{\bm \alpha} \| > \sqrt{\| {\bm q} \|}$. As we did before for $n + m_{2}$, for the modes with small norm we use twice the Diophantine condition, while for the modes with larger norm, we use that the propagator is effectively regularized; we omit the details. The resulting estimate is:
\begin{equation}
|\text{II}^{\text{A}}_{3}| \le e^{-c (x_{2} + y_{2} + w_{2} + z_{2})} \frac{1}{\sqrt{\| {\bm q} \|}}  \frac{1}{\sqrt{\|{\bm q} + {\bm p}\|}}\;.
\end{equation} 
All in all, the term $\text{II}^{\text{A}}$ in (\ref{eq:IIA}) is estimated as:
\begin{equation}
|\text{II}^{\text{A}}| \le e^{-c (x_{2} + y_{2} + w_{2} + z_{2})} \frac{1}{\| {\bm q} \|}  \frac{1}{\sqrt{\|{\bm q} + {\bm p}\|}}\;.
\end{equation}
In order to conclude the discussion of $\text{II}$, we are left with discussing the contributions associated with $m_{1}\neq 0$ and $n + m_{2} = 0$. These are estimated exactly as we did for $n + m_{2} \neq 0$, and we omit the details. All in all, we obtained:
\begin{equation}\label{eq:IIest}
|\text{II}| \le e^{-c (x_{2} + y_{2} + w_{2} + z_{2})} \Big( \frac{1}{\sqrt{\| {\bm q} \|}}  \frac{1}{\|{\bm q} + {\bm p}\|} + \frac{1}{\| {\bm q} \|}  \frac{1}{\sqrt{\|{\bm q} + {\bm p}\|}}\Big)\;.
\end{equation}
In particular, the term $\text{II}$ satisfies the bound (\ref{eq:estvert}).
%
%\paragraph{Analysis of $\text{I}$.} The term $\text{I}$ is:
%
%\begin{equation}
%\begin{split}
%\text{I} &= \sum_{n} \big(Z_{0,+,\sigma}(x_{2}) g_{+;\mr{s}}({\bm q})  \overline{Z_{-n,+,\rho}(z_{2})} \\
%&\qquad \cdot Z_{-n,+,\varrho}(w_{2}) g_{+;\mr{s}}({\bm q} + {\bm p})  \overline{Z_{0,+,\zeta}(y_{2})}\,\big)\;. 
%\end{split}
%\end{equation}
%
%The explicit form of $g_{+;\mr{s}}$, Eq. (\ref{eq:grel}), combined with the estimate (\ref{eq:IIest}), shows that the term $\text{II}$ is subleading with respect to $\text{I}$ as $\bm{q}\to \bm{0}$ (i.e. ${\bm k} \to {\bm k}_{F}^{+}(\lambda)$) and ${\bm p} \to {\bm 0}$.
%
\paragraph{Conclusion of the proof of Proposition \ref{prop:vertasy}.} We are now ready to determine the low momentum behavior of the vertex function, starting from (\ref{eq:S3dec}) and from the above estimates. We have, recall Eq. (\ref{eq:S3S2}):
\begin{equation}\label{eq:S3s}
\begin{split}
&\hat S^{\mr{s}}_{3; \mu,\sigma,\zeta}(\bm{p}, \bm{k}, \bm{k} + {\bm p}; x_{2}, y_{2})\\
&\quad= \sum_{n} \sum_{\rho,\varrho}  \sum_{z_{2},w_{2}}\Big(\hat J_{\mu;\rho,\varrho}(k_{1} + n\alpha + p_{1},p_{1};z_{2},w_{2}) \\
&\qquad\quad \cdot  Z_{0,+,\sigma}(x_{2}) g_{+;\mr{s}}\big({\bm q}(\bm{k})\big)  \overline{Z_{-n,+,\rho}(z_{2})} \\
&\qquad\quad \cdot Z_{-n,+,\varrho}(w_{2}) g_{+;\mr{s}}\big({\bm q}(\bm{k}) + {\bm p}\big)  \overline{Z_{0,+,\zeta}(y_{2})}\Big) \\
&\qquad + \widetilde R_{3; \mu,\sigma,\zeta}(\bm{p}, \bm{k}, \bm{k} + {\bm p}; x_{2}, y_{2})\;,
\end{split}
\end{equation}
where the error term $\widetilde R_{3}$ takes into account the contribution obtained plugging $\text{II}$ (bounded in (\ref{eq:IIest})) and $\text{III}$ (bounded in (\ref{eq:calE1bd})) into (\ref{eq:S3S2}), while the main term is obtained plugging $\text{I}$ in (\ref{eq:S3S2}). Consider the main term. Choosing ${\bm k}$ and ${\bm k} + {\bm p}$ close enough to ${\bm k}_{F}^{+}(\lambda)$, the main contribution is:% up to subleading terms for $\| {\bm k} - {\bm k}_{F}^{+}(\lambda) \| \ll 1$, $\|{\bm p}\| \ll 1$:
\begin{equation}\label{eq:S3smain}
\begin{split}
&\hat S^{\mr{s}; \text{main}}_{3; \mu,\sigma,\zeta}(\bm{p}, \bm{k}, \bm{k} + {\bm p}; x_{2}, y_{2}) \\
&= Z_{0,+,\sigma}(x_{2}) \overline{Z_{0,+,\zeta}(y_{2})} g_{+;\mr{s}}({\bm q}) g_{+;\mr{s}}({\bm q} + {\bm p}) \sum_{n} \langle Z_{-n,+}, J_{\mu}(k_{F}^{+}(\lambda) + n\alpha,0) Z_{-n,+} \rangle\;,
\end{split}
\end{equation}
where we used the notation
\begin{equation}
\langle f,g \rangle := \sum_{x_{2}, \rho} \overline{f_{\rho}(x_{2})} g_{\rho}(x_{2})\;.
%J_{\mu, \rho,\rho'}^{+}(n\alpha)(z_{2}, w_{2}) &:= \hat J_{\mu;\rho,\varrho} \big(k_{F}^{+}(\lambda) + n\alpha,0;z_{2},w_{2}\big)\;.
\end{equation}
All in all, putting together (\ref{eq:S3S2}), (\ref{eq:S3dec}), (\ref{eq:S3s}), (\ref{eq:S3smain}):
\begin{equation}\label{eq:vert1}
\begin{split}
\hat S_{3; \mu,\sigma,\zeta}(\bm{p}, \bm{k}, \bm{k} + {\bm p}; x_{2}, y_{2}) &= Z_{0,+,\sigma}(x_{2}) \overline{Z_{0,+,\zeta}(y_{2})} g_{+;\mr{s}}({\bm q}) g_{+;\mr{s}}({\bm q} + {\bm p}) \zeta_{\mu,+} \\
&\quad + R^{\text{tot}}_{3; \mu,\sigma,\zeta}(\bm{p}, \bm{k}, \bm{k} + {\bm p}; x_{2}, y_{2})\;,
\end{split}
\end{equation}
where
\begin{equation}
\zeta_{\mu,+} = \sum_{n} \langle Z_{-n,+}, J_{\mu}(k_{F}^{+}(\lambda) + n\alpha,0) Z_{-n,+} \rangle\;,
\end{equation}
and where, for $\| {\bm k} - {\bm k}_{F}^{\omega}(\lambda) \| \ll 1$, $\| {\bm p} \| \ll 1$ and $L_{2}$ large enough:
\begin{equation}\label{eq:vert3}
\begin{split}
&\sum_{x_{2}, y_{2}} \big| Z_{0,+,\sigma}(x_{2})\big| \big| Z_{0,+,\zeta}(y_{2})\big| \Big|g^{-1}_{+;\mr{s}}({\bm q})\Big| \Big|g^{-1}_{+;\mr{s}}({\bm q} + {\bm p})\Big|\\&\quad \cdot \Big|R^{\text{tot}}_{3; \mu,\sigma,\zeta}(\bm{p}, \bm{k}, \bm{k} + {\bm p}; x_{2}, y_{2})\Big| = \mr{o}(1)\;.
\end{split}
\end{equation}
This concludes the proof of Proposition \ref{prop:vertasy}
\end{proof}
\begin{proof}(of Proposition \ref{prp:asy2pt}.) As in the proof of Proposition \ref{prop:vertasy}, we start by writing:
\begin{equation}\label{eq:S2diff0}
\begin{split}
\hat S_{2;\sigma,\zeta}(\bm{k},{\bm k} ; x_{2}, y_{2}) &= \sum_{m,\omega} Z_{m,\omega,\sigma}(x_{2}) g_{\omega;\mr{s}}({\bm q} + m{\bm \alpha})  \overline{Z_{-n+m,\omega,\zeta}(y_{2})} \\
&\quad + \hat R_{\sigma,\zeta}({\bm k}, {\bm k}; x_{2}, y_{2})\;,
\end{split}
\end{equation}
where $\hat R_{\sigma,\zeta}({\bm k}, {\bm k}; x_{2}, y_{2})$ is bounded as, recall (\ref{eq:RbdFou}):
\begin{equation}
\Big| \text{d}_{k_{0}}^{n_{0}} \text{d}_{k_{1}}^{n_{1}} \hat R_{\sigma,\zeta}({\bm k}, {\bm k}; x_{2}, y_{2})\Big| \leq \frac{C_{n_{0},n_{1}} e^{-\tilde c|x_{2} - y_{2}|}}{\| {\bm k} - {\bm k}_{F}^{+} \|^{n_{0} + n_{1} + 1-\theta}} + \sum_{\substack{\omega_{1},\omega_{2} \\ (\omega_{1},\omega_{2}) \neq (+,+)}} \!\!\!  C_{n_{0},n_{1}}\frac{ e^{-c(|x_{2}|_{\omega_{1}} + |y_{2}|_{\omega_{2}})}}{|k_{0}|^{n_{0} + n_{1} + 1-\theta}}\;.
\end{equation}
Consider now the main term in (\ref{eq:S2diff0}). The contribution associated with $\omega = -$ satisfies the bound (\ref{eq:Rn0}), using the exponential decay of the edge modes. Next, consider the main contribution to (\ref{eq:S2diff0}) with $\omega = +$. Similarly to what we did in the analysis of the vertex function, we separate the $m=0$ term in the sum from the others. Proceeding as after (\ref{eq:IIA}), the terms with $m\neq 0$ give a contribution that is bounded by $e^{-c(x_{2} + y_{2})} / \| {\bm q} \|^{(1/2)(n_{0} + n_{1} + 1)}$, which in particular satisfies (\ref{eq:Rn0}). This concludes the proof of (\ref{eq:Rn00}), (\ref{eq:Rn0}).
\end{proof}
\subsection{Proof of Theorem \ref{thm:resp}} 
In this section we shall prove Theorem \ref{thm:resp}, which allows to compute the edge transport coefficients. The analysis will rely on the asymptotics of the two-point function, Eq. (\ref{eq:2ptmain}), on the relations (\ref{eq:relzeta}), and on the discussion of Section \ref{sec:JJ} about the structure of the current-current correlations functions, recall Eqs. (\ref{eq:chi0WI}), (\ref{eq:chi1WI}).

\subsubsection{Susceptibility}\label{sec:susce}
Let us start from the susceptibility, $\chi^{\beta, L}_{0}(\eta, \theta)$. To avoid carrying out the error term $\varepsilon^{\beta,L}_{0}(\eta, \theta)$ in (\ref{eq:wick}), we will suppose for the moment that $\eta = \eta_{\beta} \in (2\pi / \beta) \mathbb{N}$; considering a general $\eta>0$ will simply introduce an additive error term bounded as in (\ref{eq:epsest}), which we will take into account at the very end. From (\ref{eq:wick}), using Wick's rule:
\begin{equation}
\chi^{\beta, L}_{0}(\eta, \theta)  =-\theta \sum_{{\vec x},{\vec y}} \mu(\theta\vec  x)\phi_{\ell}(\theta y_{1}, y_{2})\int_{-\frac{\beta}{2}}^{\frac{\beta}{2}}  d s\,e^{-i\eta s} \Tr \big(S^{\beta,L}_{2}(s,{\vec  x}; 0,{\vec y}) S^{\beta,L}_{2}(0,{\vec y}; s,{\vec x})\big)\;,
\end{equation}
where the trace is over the internal degrees of freedom. We shall use that $\chi^{\beta, L}_{0}(\eta, \theta) = K_{00}^{\beta,L}(\eta, \theta)$, recall Eq. (\ref{eq:chiK}), and we shall now define the decomposition (\ref{eq:singreg}). To this end, the following technical lemma will be useful.
\begin{lemma}\label{lem:contK} Let $F(\vec {\bm x}; \vec {\bm y})$ be a $\beta$-periodic function in the imaginary times and satisfying the cylindric boundary conditions in $\vec x$, $\vec y$, such that:
\begin{equation}\label{eq:Festlem}
\begin{split}
F(\vec {\bm x}; \vec {\bm y}) &= \sum_{m \in \mathbb{Z}} F_{m}({\bm x} - {\bm y}; x_{2}, y_{2}) e^{-i m\alpha y_{1}} \\
\big|F_{m}({\bm x} - {\bm y}; x_{2}, y_{2})\big| &\leq \frac{C e^{-c|m|} e^{-c|x_{2} - y_{2}|}}{1 + \| {\bm x} - {\bm y} \|^{2+\xi}}\qquad \text{with $\xi>0$.}
\end{split}
\end{equation}
Consider the function:
\begin{equation}
g(\eta, \theta) = \theta \sum_{{\vec x},{\vec y}} \mu(\theta\vec  x)\phi_{\ell}(\theta y_{1}, y_{2})
\int_{-\frac{\beta}{2}}^{\frac{\beta}{2}}  d s\,e^{-i\eta s} F((s,\vec x); (0, \vec y))\;.
\end{equation}
Then, for $\alpha>0$:
\begin{equation}
\big|g(\eta, \theta)(\eta,\theta) - g(\eta, \theta)(\eta,\theta')\big|\le C (|\theta|^{\alpha} + |\theta'|^{\alpha})\;.
\end{equation}
\end{lemma}
We postpone the proof of this lemma to Appendix \ref{app:contK}. Let us define $K^{\text{sing}}_{00}(\eta, \theta)$ in (\ref{eq:singreg}) as 
\begin{equation}\label{eq:Ksing00}
K^{\text{sing}}_{00}(\eta, \theta) := - \theta\sum_{{\vec x},{\vec y}} \mu(\theta\vec  x)\phi_{\ell}(\theta y_{1}, y_{2})\int_{-\frac{\beta}{2}}^{\frac{\beta}{2}}  d s\,e^{-i\eta s} \Tr \big(S^{\mr{s}}_{2}(s,{\vec  x}; 0,{\vec y}) S^{\mr{s}}_{2}(0,{\vec y}; s,{\vec x})\big)\;,
\end{equation}
with $S^{\mr{s}}_{2}$ given by (\ref{eq:S2s}). Using the result about the structure of the subleading contributions to the two point function, Eq. (\ref{eq:Rest2pt}), together with Lemma \ref{lem:contK}, we see that $K^{\text{reg}}_{00}(\eta, \theta)$ satisfies the estimate (\ref{eq:contK}). Furthermore, from the information about the two-point function provided by Theorem \ref{thm:2pt}, and from the compact support of the test functions, we have:
\begin{equation}
|K_{\mu\nu}^{\beta,L}(\eta,\theta)|\le C_{\ell} |\log \theta|\;,
\end{equation}
uniformly in $\beta, \eta, L_{1}, L_{2}$. Thus, we can apply the strategy outlined in Section \ref{sec:JJ}, to compute the susceptibility only relying on $K^{\text{sing}}_{00}$.

We rewrite (\ref{eq:Ksing00}) using the explicit expression of $S^{\mr{s}}_{2}$, Eq. (\ref{eq:S2s}). We have:
\begin{equation}\label{eq:Ksing002}
\begin{split}
&K^{\text{sing}}_{00}(\eta, \theta) \\
&\quad= - \theta\sum_{{\vec x},{\vec y} } \mu(\theta \vec x) \phi_{\ell}(\theta y_{1}, y_{2}) | Z_{+}(\vec x) |^{2}  | Z_{+}(\vec y) |^{2} \\
&\quad\qquad \cdot \int_{-\beta/2}^{\beta/2}  d s\, e^{-i\eta s}  \check{g}_{+;\mr{s}}\big((x_{1}, s); (y_{1},0)\big) \check{g}_{+;\mr{s}}\big((y_{1},0); (x_{1}, s)\big)\\
&\qquad+ \mr{O}_{\eta,\theta}(e^{-c L_{2}})\;,
\end{split}
\end{equation}
where $| Z_{\omega}(\vec x) |^{2} = \sum_{\sigma} |Z_{\omega,\sigma}(\vec x)|^{2}$. The exponentially small error term takes into account the contribution due to the case in which at least one chirality is $\omega = -$. This corresponds to an edge state localized at $x_{2} = L_{2}$; the contribution is exponentially small because the test functions are localized at $x_{2} = 0$. The bound for this error term might be non-uniform in $\eta, \theta$ (via a logarithmic divergence); this is not an issue, since the limit $L_{2} \to \infty$ is taken before $(\eta, \theta) \to (0,0)$. 
\paragraph{Analysis of the singular part.} Let us consider the integral in (\ref{eq:Ksing002}). We have:
\begin{equation}\label{eq:Fouriersingularpart}
\begin{split}
&\int_{-\beta/2}^{\beta/2}  d s\, e^{-i\eta s} \check{g}_{+;\mr{s}}\big((x_{1}, s); (y_{1},0)\big) \check{g}_{+;\mr{s}}\big((y_{1},0); (x_{1}, s)\big) \\
&\qquad = \frac{1}{L_{1}} \sum_{p \in \msc D_{L}} e^{ip(x_{1} - y_{1})} \frac{1}{\beta L_{1}} \sum_{\bm{q} \in \msc D_{\beta,L}}  g_{+;\mr{s}}(q_{0} + \eta, q_1 + p) g_{+;\mr{s}}(q_{0}, q_{1}) \\
&\qquad \equiv \frac{1}{L_{1}} \sum_{p \in \msc D_{L}} e^{ip(x_{1} - y_{1})} \mc{B}_{+}^{\beta, L_{1}}(\eta, p)\;,
\end{split}
\end{equation}
where $\mc{B}_{\omega}^{\beta, L_{1}}(\eta,p)$ is the regularized relativistic bubble diagram at external momentum $(\eta, p)$, with chiralities $\omega$; we refer to Appendix \ref{app:bubble} for its computation. Plugging (\ref{eq:Fouriersingularpart}) into (\ref{eq:Ksing002}), we get:
\begin{equation}\label{eq:sing2}
\begin{split}
\chi^{\text{sing}}_{0}(\eta, \theta) =& -\frac{\theta}{L_{1}}\sum_{p} \sum_{x_{2},y_{2}}\mathcal{F}\big(\mu_{\theta}(\cdot,x_{2}) |Z_{+} (\cdot,x_{2})|^{2}\big)(-p) \\
&\quad \cdot \mathcal{F}\big(\phi_{\theta,\ell}(\cdot,y_{2}) | Z_{+}(\cdot,y_{2})|^{2}  \big)(p) \,\mc{B}_{+}^{\beta, L_{1}}(\eta, p)\\
&+\mr{O}_{\eta,\theta}(e^{-c L_{2}})\;,
\end{split}
\end{equation}
where $\mathcal{F}\colon f(\cdot)\mapsto \hat f(\cdot)= \sum_x e^{-i(\cdot) x}f(x)$ denotes the Fourier transform and:
\begin{equation}\label{eq:conv}
\begin{split}
\mathcal{F}\big(\mu_{\theta}(\cdot) | Z_{\omega}(\cdot,x_{2})|^{2}\big)(-p) &= \frac{1}{L_{1}} \sum_{q} \hat \mu_{\theta}(-p-q,x_2) \mathcal{F}\big(|Z_{\omega}(\cdot, x_{2}) |^2\big) (q) \\
\mathcal{F}\big( \phi_{\theta,\ell}(\cdot) | Z_{\omega}(\cdot,y_{2}) |^{2}  \big)(p) &= \frac{1}{L_{1}}\sum_{q} \hat \phi_{\theta,\ell}(p - q,y_2) \mathcal{F}\big(|Z_{\omega}(\cdot, y_{2})|^2\big)(q)\;.
 \end{split}
\end{equation}
Now, using  (\ref{eq:Zfou}):
\begin{equation}\label{eq:ZZconv}
\begin{split}
|Z_{\omega}(\vec x)|^2 &= \sum_{n \in \mathbb{Z}} \sum_{\sigma} \Big(\sum_{m\in \mathbb Z}Z_{m,\omega,\sigma}(x_2) \overline{Z_{m-n,\omega,\sigma}(x_2)} \Big) e^{-in\alpha x_{1}} \\
&\equiv \sum_{n \in \mathbb{Z}} (Z_{\omega}\star\overline{Z_{\omega}})(n;x_{2}) e^{-in\alpha x_{1}}\;,
\end{split}
\end{equation}
we can rewrite (\ref{eq:conv}) as:
\begin{equation}
\begin{split}
\mathcal{F}\big(\mu_{\theta}(\cdot) | Z_{\omega}(\cdot,x_{2})|^{2}\big)(-p) &= \sum_{n\in \mathbb{Z}} \hat \mu_{\theta}(-p-n\alpha,x_2) (Z_{\omega}\star\overline{Z_{\omega}})(n;x_{2}) \\
\mathcal{F}( \phi_{\theta,\ell}(\cdot) | Z_{\omega}(\cdot,y_{2}) |^{2}  )(p) &= \sum_{n\in   \mathbb{Z}} \hat \phi_{\theta,\ell}( p-n\alpha,y_{2}) (Z_{\omega}\star\overline{Z_{\omega}})(n;y_{2})\;.
 \end{split}
\end{equation}
Observe that, from (\ref{eq:bdZm}), the following estimate holds:
\begin{equation}\label{eq:convest}
|(Z_{\omega}\star\overline{Z_{\omega}})(n;x_{2})| \le Ce^{-\frac{c}{2}|n|} e^{-c|x_{2}|_{\omega}}\;.
\end{equation}
Plugging these expressions into (\ref{eq:sing2}), we get:
\begin{equation}
\begin{split}
\chi^{\text{sing}}_{0}(\eta, \theta)  &= - \frac{\theta}{L_{1}} \sum_{p} \sum_{n,m}  \sum_{x_{2},y_{2}} (Z_{+}\star\overline{Z_{+}})(m;y_{2}) (Z_{+}\star\overline{Z_{+}})(n;x_{2})\\
&\qquad \cdot \hat \mu_{\theta}(-p-n\alpha,x_2) \hat \phi_{\theta,\ell}( p-m\alpha,y_{2})   \mc{B}_{+}^{\beta, L_{1}}(\eta, p) \\
&\quad +\mr{O}_{\eta,\theta}(e^{-c L_{2}}) \\
&= -\frac{\theta}{L_{1}} \sum_{p} \sum_{n,m} \sum_{x_{2},y_{2}} (Z_{+}\star\overline{Z_{+}})(m;y_{2}) (Z_{+}\star\overline{Z_{+}})(n;x_{2}) \\
&\qquad \cdot \hat \mu_{\theta}(-p,x_2) \hat \phi_{\theta,\ell}\big(p-(m+n)\alpha,y_{2}\big)   \mc{B}_{+}^{\beta, L_{1}}(\eta, p-n\alpha) \\
&\quad + \mr{O}_{\eta,\theta}(e^{-c L_{2}})\;.
\end{split}
\end{equation}
\paragraph{Cutting off large frequencies.} Let:
\begin{equation}
\begin{split}
\chi^{\text{sing}}_{0; N}(\eta, \theta) &:=  -\frac{\theta}{L_{1}} \sum_{p} \sum_{\substack{n,m \\ |n| \le N,\, |m| \le N}} \sum_{x_{2},y_{2}} (Z_{+}\star\overline{Z_{+}})(m;y_{2}) (Z_{+}\star\overline{Z_{+}})(n;x_{2}) \\
&\qquad \cdot  \hat \mu_{\theta}(-p,x_2) \hat \phi_{\theta,\ell}\big(p-(m+n)\alpha,y_{2}\big)  \mc{B}_{+}^{\beta, L_{1}}(\eta, p-n\alpha)\;.
\end{split}
\end{equation}
This function approximates $\chi^{\text{sing}}_{0}(\eta, \theta)$, up to small errors, uniformly in all parameters. In fact, from the estimate (\ref{eq:convest}) for $Z_{n,\omega,\sigma}(x_{2})$, from the boundedness of $\hat \phi$, and from the boundedness of $\mc{B}_{\omega}^{\beta, L_{1}}$, we have:
\begin{equation}\label{eq:eN}
\begin{split}
\big| \chi^{\text{sing}}_{0}(\eta, \theta) - \chi^{\text{sing}}_{0; N}(\eta, \theta) \big| &\le \frac{C}{L_{1}} \sum_{p}  | \hat \mu_{\theta}(-p, x_2) | e^{-cN} + C_{\eta,\theta}e^{-cL_{2}} \\
&\le Ke^{-cN} + C_{\eta,\theta}e^{-cL_{2}}\;,
\end{split}
\end{equation}
where the last term takes into account the exponentially small error in (\ref{eq:Ksing002}). 

\paragraph{Estimating the contribution of $m+n\neq 0$.} Let us now focus on $\chi^{\text{sing}}_{0; N}(\eta, \theta)$. We claim that the contribution with $n+m\neq 0$ gives a subleading contribution, for $\theta$ small. In fact:
\begin{equation}\label{eq:nneqm}
\begin{split}
&\Big|\frac{\theta}{L_{1}} \sum_{p} \sum_{\substack{n\neq -m \\ |n| \le N,\, |m| \le N}} \sum_{x_{2},y_{2}}(Z_{+}\star\overline{Z_{+}})(m;y_{2}) (Z_{+}\star\overline{Z_{+}})(n;x_{2}) \\
&\qquad \cdot \hat \mu_{\theta}(-p,x_2) \hat \phi_{\theta,\ell}\big(p-(m+n)\alpha,y_{2}\big)   \mc{B}_{+}^{\beta, L_{1}}(\eta, p-n\alpha)\Big| \\
&\le \frac{C}{\theta L_{1}} \sum_{p} \sum_{n\neq -m} \sum_{x_{2},y_{2}} e^{-\frac{c}{2}|m|} e^{-cy_{2}} e^{-\frac{c}{2}|n|} e^{-cx_{2}} \\
&\qquad \cdot \frac{C_{r}}{1 + (|p|_{\mathbb{T}}/\theta)^{r}} \frac{C_{r}}{1 + \big(\big|\big(p - (m+n)\alpha\big)\big|_{\mathbb{T}}/\theta\big)^{r}}\;,\qquad \text{for all $r\in \mathbb{N}$}
\end{split}
\end{equation}
where the last two factors bound the functions $\hat \mu_{\theta}$ and $\hat \phi_{\theta,\ell}$, recall (\ref{eq:mutheta}). Thus, using that:
\begin{equation}
\begin{split}
&\frac{1}{1 + (|p|_{\mathbb{T}}/\theta)^{r}} \frac{1}{1 + \big(\big|\big(p - (m+n)\alpha\big)\big|_{\mathbb{T}}/\theta\big)^{r}} \\
& \leq  \frac{C_{r}}{1 + (|p|_{\mathbb{T}}/\theta)^{r/2}} \frac{1}{1 + (|p|_{\mathbb{T}}/\theta)^{r/2}}\frac{1}{1 + \big(\big|\big(p - (m+n)\alpha\big)\big|_{\mathbb{T}}/\theta\big)^{r/2}} \\
&= \frac{C_{r}}{1 + (|p|_{\mathbb{T}}/\theta)^{r/2}} \frac{\big(\big|p - p + (m+n)\alpha\big|_{\mathbb{T}}/\theta\big)^{r/2}}{1 + (|p|_{\mathbb{T}}/\theta)^{r/2}}\\&\qquad \cdot \frac{1}{1 + \big(\big|\big(p - (m+n)\alpha\big)\big|_{\mathbb{T}}/\theta\big)^{r/2}} \frac{1}{\big(\big|(m+n)\alpha\big|_{\mathbb{T}}/\theta\big)^{r/2}} \\
&\leq \frac{K_{r}}{1 + (|p|_{\mathbb{T}}/\theta)^{r/2}} \frac{\big(\big|p \big|_{\mathbb{T}}/\theta\big)^{r/2} + \big(\big|- p + (m+n)\alpha\big|_{\mathbb{T}}/\theta\big)^{r/2}}{1 + (|p|_{\mathbb{T}}/\theta)^{r/2}}\\&\quad \cdot \frac{1}{1 + \big(\big|\big(p - (m+n)\alpha\big)\big|_{\mathbb{T}}/\theta\big)^{r/2}} \frac{1}{\big(\big|(m+n)\alpha\big|_{\mathbb{T}}/\theta\big)^{r/2}} \\
&\leq \frac{2K_{r}}{1 + (|p|_{\mathbb{T}}/\theta)^{r/2}} \frac{1}{\big(\big|(m+n)\alpha\big|_{\mathbb{T}}/\theta\big)^{r/2}}\;,
\end{split}
\end{equation}
and that $|(n+m)\alpha|_{\mathbb{T}} \ge c / |n+m|^{\tau} \geq  C N^{-\tau}$ for $0<|n+m| \le 2N$, we get
\begin{equation}
(\ref{eq:nneqm}) \le C_{N} \theta^{r/2}\;,
\end{equation}
uniformly in all the other parameters. 
\paragraph{Estimating the contribution of $m+n = 0$, $n\neq 0$.} We are now left with the contribution to $\chi^{\text{sing}}_{0; N}(\eta, \theta)$ coming from the modes $m,n$ such that $m+n = 0$. It is:
\begin{equation}
\begin{split}
\chi^{\text{main}}_{0; N}(\eta, \theta) &= -\frac{\theta}{L_{1}} \sum_{p} \sum_{n: |n| \le N} \sum_{x_{2},y_{2}} (Z_{+}\star\overline{Z_{+}})(-n;y_{2}) (Z_{+}\star\overline{Z_{+}})(n;x_{2}) \\
&\qquad \cdot\hat \mu_{\theta}(-p,x_2) \hat \phi_{\theta,\ell}(p,y_{2})   \mc{B}_{+}^{\beta, L_{1}}(\eta, p-n\alpha)
\end{split}
\end{equation}
which we further rewrite as:
\begin{equation}\label{eq:chimain}
\begin{split}
\chi^{\text{main}}_{0; N}(\eta, \theta) &= -\frac{\theta}{L_{1}} \sum_{p} \sum_{x_{2},y_{2}}(Z_{+}\star\overline{Z_{+}})(0;y_{2}) (Z_{+}\star\overline{Z_{+}})(0;x_{2}) \\
&\qquad \cdot\hat \mu_{\theta}(-p,x_2) \hat \phi_{\theta,\ell}(p,y_{2})   \mc{B}_{+}^{\beta, L_{1}}(\eta, p)\\
&\quad + R_{0;N}(\eta,\theta)
\end{split}
\end{equation}
where $R_{0;N}(\eta,\theta)$ takes into account the contribution of the modes $n\neq 0$, $|n| \le N$. We rewrite this last term as, after the change of variables $p/\theta \to p$:
\begin{equation}\label{eq:R0N}
\begin{split}
&R_{0;N}(\eta,\theta) \\
&= -\frac{1}{\theta L_{1}} \sum_{p \in \frac{2\pi}{\theta L_{1}} (\mathbb{Z} / L_{1}\mathbb{Z})} \sum_{n: 0<|n| \le N} \sum_{x_{2},y_{2}}(Z_{+}\star\overline{Z_{+}})(-n;y_{2}) (Z_{+}\star\overline{Z_{+}})(n;x_{2}) \\
&\qquad \cdot \theta\hat \mu_{\theta}(-\theta p,\theta x_2) \,\theta \hat \phi_{\theta,\ell}(\theta p,y_{2})   \mc{B}_{+}^{\beta, L_{1}}(\eta, \theta p-n\alpha)\;.
\end{split}
\end{equation}
In order to bound this term, let us recall the estimates, valid for $p\in (2\pi / L_{1}\theta) (\mathbb{Z} / L_{1}\mathbb{Z})$:
\begin{equation}\label{eq:estmuphi}
\begin{split}
\theta \big|\hat \mu_{\theta}(\theta p,\theta x_2)\big| &\le \frac{C_{r}}{1+(|p|_{\mathbb{T}_{\theta^{-1}}})^{r}}\frac{1}{1+|\theta x_{2}|^{r}} \\
\theta \big|\hat \phi_{\theta,\ell}(-\theta p,y_{2})\big| &\le \frac{C_{r}}{1+(|p|_{\mathbb{T}_{\theta^{-1}}})^{r}} \frac{1}{1+|y_{2} / \ell|^{r}}\;,
\end{split}
\end{equation}
with $\mathbb{T}_{\theta^{-1}}$ the torus of side $2\pi/\theta$. Observe that $|n\alpha|_{\mathbb{T}} \ge c / N^{\tau} = c_{N}$ for $0<|n| \le N$; thus, we can further separate the sum over $p$ as:
\begin{equation}
\sum_{p \in \frac{2\pi}{\theta L_{1}} (\mathbb{Z} / L_{1}\mathbb{Z})} (\cdots) = \sum_{\substack{p \in \frac{2\pi}{\theta L_{1}} (\mathbb{Z} / L_{1}\mathbb{Z}) \\ 2\theta |p|_{\mathbb{T}_{\theta^{-1}}} \le c_{N}}} (\cdots) +  \sum_{\substack{p \in \frac{2\pi}{\theta L_{1}} (\mathbb{Z} / L_{1}\mathbb{Z}) \\ 2\theta |p|_{\mathbb{T}_{\theta^{-1}}} > c_{N}}} (\cdots)\;;
\end{equation}
correspondingly, we write:
\begin{equation}\label{eq:R><}
R_{0;N}(\eta,\theta) = R^{\le }_{0;N}(\eta,\theta) + R^{>}_{0;N}(\eta,\theta)\;.
\end{equation}
Thanks to the estimates (\ref{eq:estmuphi}) for the decay of the test functions in Fourier space, $R^{>}_{0;N}(\eta,\theta)$ satisfies:
\begin{equation}\label{eq:R<}
| R^{>}_{0;N}(\eta,\theta) | \le C_{N,r} \theta^{r}\;,\qquad \text{for all $r\in \mathbb{N}$.}
\end{equation}
Next, consider $R^{\leq}_{0;N}$. Here, since $\theta | p |_{\mathbb{T}_{\theta^{-1}}} = | \theta p |_{\mathbb{T}} < c_{N}/2$ and $|n  \alpha |_{\mathbb{T}}>c_{N}$, we will use that the function $(\eta, \theta)\mapsto \mc{B}_{+}^{\beta, L_{1}}(\eta, \theta p-n\alpha)$ is continuous in $\theta p$ at $\theta p = 0$, which can be checked from the explicit expression obtained in Appendix \ref{app:bubble}, and reported below, see (\ref{eq:binfty}). We claim that:
\begin{equation}\label{eq:Rele}
\Big| R^{\le }_{0;N}(\eta,\theta) - R^{\le }_{0;N}(\eta,\theta') \Big| \leq C_{N} (|\theta|^{\alpha} + |\theta'|^{\alpha})
\end{equation}
for some $\alpha > 0$. To see this, we approximate the sum over the momenta as an integral:
\begin{equation}\label{eq:rmin}
\begin{split}
&R^{\le }_{0;N}(\eta,\theta) \\
&= -\int_{\mathbb{T}_{\theta^{-1}}:\, 2\theta |p|_{\mathbb{T}_{\theta^{-1}}} \le c_{N}} \frac{dp}{(2\pi)} \sum_{n: 0<|n| \le N} \sum_{x_{2},y_{2}}(Z_{+}\star\overline{Z_{+}})(-n;y_{2}) (Z_{+}\star\overline{Z_{+}})(n;x_{2}) \\
&\qquad \cdot \theta\hat \mu_{\theta}(-\theta p,\theta x_2) \,\theta \hat \phi_{\theta,\ell}(\theta p,y_{2})    \mc{B}_{+}^{\infty}(\eta, \theta p-n\alpha) + \mr{o}(1)\;,
\end{split}
\end{equation}
where the error terms vanish as $L_{1} \to \infty$ (to control the errors coming from the integral approximation). Then, using the exponential decay of $Z_{m,\omega,\sigma}(x_2)$, we can localize the test functions at $x_{2} = y_{2} = 0$, up to errors that vanish as $\theta \to 0$ and $\ell \to \infty$, which are included in the new $\mr{o}(1)$ terms:
\begin{equation}
\begin{split}
R^{\le }_{0;N}(\eta,\theta) &= -\int_{\mathbb{T}_{\theta^{-1}}:\, 2\theta |p|_{\mathbb{T}_{\theta^{-1}}} \le c_{N}} \frac{dp}{(2\pi)} \sum_{n: 0<|n| \le N} \sum_{x_{2},y_{2}}(Z_{+}\star\overline{Z_{+}})(-n;y_{2}) (Z_{+}\star\overline{Z_{+}})(n;x_{2}) \\
&\qquad \cdot \theta\hat \mu_{\theta}(-\theta p, 0) \,\theta \hat \phi_{\theta,\ell}(\theta p, 0)   \mc{B}_{+}^{\infty}(\eta, \theta p-n\alpha) + \mr{o}(1)\;.
\end{split}
\end{equation}
Next, we rewrite:
\begin{equation}\label{eq:Rleq}
\begin{split}
R^{\le }_{0;N}(\eta,\theta) &= -\int_{\mathbb{T}_{\theta^{-1}}:\, 2\theta |p|_{\mathbb{T}_{\theta^{-1}}} \le c_{N}} \frac{dp}{(2\pi)}\, \sum_{n: 0<|n| \le N} \sum_{x_{2},y_{2}} (Z_{+}\star\overline{Z_{+}})(-n;y_{2}) (Z_{+}\star\overline{Z_{+}})(n;x_{2}) \\
&\qquad \cdot \theta\hat \mu_{\theta}(-\theta p,0) \,\theta \hat \phi_{\theta,\ell}(\theta p,0)   \mc{B}_{+}^{\infty}(\eta,-n\alpha) \\
&\quad + \int_{\mathbb{T}_{\theta^{-1}}:\, 2\theta |p|_{\mathbb{T}_{\theta^{-1}}} \le c_{N}} \frac{dp}{(2\pi)} \sum_{n: 0<|n| \le N} \sum_{x_{2},y_{2}}(Z_{+}\star\overline{Z_{+}})(-n;y_{2}) (Z_{+}\star\overline{Z_{+}})(n;x_{2}) \\
&\qquad \cdot \theta\hat \mu_{\theta}(-\theta p,0) \,\theta \hat \phi_{\theta,\ell}(\theta p,0)   \Big( \mc{B}_{+}^{\infty}(\eta, \theta p-n\alpha) - \mc{B}_{+}^{\infty}(\eta,-n\alpha)\Big) + \mr{o}(1)\;.
\end{split}
\end{equation}
Consider the first term in the right-hand side of (\ref{eq:Rleq}). We have, using the decay properties of the test functions, and recalling the definition of (\ref{eq:fourier}) of periodized, rescaled test function:
\begin{equation}\label{eq:rminmain}
\begin{split}
&\int_{\mathbb{T}_{\theta^{-1}}:\, 2\theta |p|_{\mathbb{T}_{\theta^{-1}}} \le c_{N}} \frac{dp}{(2\pi)} \sum_{n: 0<|n| \le N} \sum_{x_{2},y_{2}}(Z_{+}\star\overline{Z_{+}})(-n;y_{2}) (Z_{+}\star\overline{Z_{+}})(n;x_{2}) \\
&\qquad \cdot \theta\hat \mu_{\theta}(-\theta p,0) \,\theta \hat \phi_{\theta,\ell}(\theta p,0)    \mc{B}_{+}^{\infty}(\eta,-n\alpha) \\
&= \int \frac{dp}{(2\pi)} \sum_{n: 0<|n| \le N} \sum_{x_{2},y_{2}}(Z_{+}\star\overline{Z_{+}})(-n;y_{2}) (Z_{+}\star\overline{Z_{+}})(n;x_{2}) \\
&\qquad \cdot \hat \mu_{\infty}(-p,0) \hat \phi_{\infty,}(p,0)   \mc{B}_{+}^{\infty}(\eta,-n\alpha) + \mr{o}_{N}(1)\;,
\end{split}
\end{equation}
where the main term in the right-hand side is now independent of $\theta$ and the new error $\mr{o}_{N}(1)$ can be made arbitrarily small for $\theta$ small enough. Consider now the second term in the right-hand side of (\ref{eq:Rleq}). Using that $|n\alpha|_{\mathbb{T}} > c_{N}$, we have, for $2\theta |p| \le c_{N}$:
\begin{equation}
\Big| \mc{B}_{+}^{\infty}(\eta, \theta p-n\alpha) - \mc{B}_{+}^{\infty}(\eta,-n\alpha) \Big| \leq C_{N} \theta |p|_{\mathbb{T}_{\theta}^{-1}}\;.
\end{equation}
Plugging this bound in the second term in the right-hand side of (\ref{eq:Rleq}), and using the decay properties of the test functions (\ref{eq:estmuphi}), we get:
\begin{equation}\label{eq:rmin2}
\begin{split}
&\Big|\int_{\mathbb{T}_{\theta^{-1}}:\, 2\theta |p|_{\mathbb{T}_{\theta^{-1}}} \le c_{N}} \frac{dp}{(2\pi)} \sum_{n: 0<|n| \le N} \sum_{x_{2},y_{2}} (Z_{+}\star\overline{Z_{+}})(-n;y_{2}) (Z_{+}\star\overline{Z_{+}})(n;x_{2}) \\
&\qquad \cdot \theta\hat \mu_{\theta}(-\theta p,0) \,\theta \hat \phi_{\theta,\ell}(\theta p,0)   \Big( \mc{B}_{+}^{\infty}(\eta, \theta p-n\alpha) - \mc{B}_{+}^{\infty}(\eta,-n\alpha)\Big) \Big| \leq C \theta\;.
\end{split}
\end{equation}
Thus, putting together (\ref{eq:rmin})-(\ref{eq:rmin2}), and using that the main term in (\ref{eq:rminmain}) is independent of $\theta$, the bound (\ref{eq:Rele}) follows. 

Combining (\ref{eq:R<}) and (\ref{eq:Rele}), the function $R_{0;N}(\eta,\theta)$ satisfies the bound (\ref{eq:contK}), with a constant $C_{N}$; hence, $R_{0;N}(\eta,\theta)$ can be considered together with $K^{\text{reg}}_{00}(\eta,\theta)$.
\paragraph{Conclusion.} Let $\widetilde{\chi}^{\text{main}}_{0}(\eta,\theta)$ be the first term in the right-hand side of (\ref{eq:chimain}). As $\beta, L_{1}\to \infty$:
\begin{equation}\label{eq:sumtoint}
\begin{split}
\widetilde{\chi}^{\text{main}}_{0}(\eta,\theta) &= - \int \frac{d p}{2\pi} \sum_{x_{2},y_{2}}  (Z_{+}\star\overline{Z_{+}})(0;y_{2}) (Z_{+}\star\overline{Z_{+}})(0;x_{2}) \\
&\qquad \cdot \hat \mu_{\infty}( -p,\theta x_2) \hat \phi_{\infty}( p,y_{2})   \mc{B}_{+}^{\infty}(\eta, \theta p) + \mr{o}_{\beta,L}(1)\;,
\end{split}
\end{equation}
where $ \mc{B}_{+}^{\infty}(\eta, \theta p) = \lim_{\beta, L_{1}\to \infty}  \mc{B}_{+}^{\beta,L_{1}}(
-\eta, \theta p)$, and where $ \mr{o}_{\beta,L}(1)$ takes into account errors vanishing as $\beta, L_{1}\to \infty$, arising from the approximation of the sum as an integral. As discussed in Appendix \ref{app:bubble}:
\begin{equation}\label{eq:binfty}
\mc{B}_{+}^{\infty}(\eta, \theta p) =  \frac{1}{4\pi \big|v_{1,+}(\lambda)  v_{0,+}(\lambda) \big|}\frac{iv_{0,+}(\lambda)\eta-v_{1,+}(\lambda)\theta p}{i v_{0,+}(\lambda)\eta +v_{1,+}(\lambda)\theta p} + r_{+}(\eta,\theta p)\;,
\end{equation}
where $r_{+}(\eta,\theta p)$ vanishes continuously as $(\eta, \theta) \to (0,0)$. Now, by (\ref{eq:chi0WI}) we have:
\begin{equation}\label{eq:chi0afterWI}
\begin{split}
\chi^{\beta,L}_{0}(\eta, \theta) &= - \int \frac{d p}{2\pi} \sum_{x_{2},y_{2}} (Z_{+}\star\overline{Z_{+}})(0;y_{2}) (Z_{+}\star\overline{Z_{+}})(0;x_{2}) \\
&\qquad \cdot \hat \mu_{\infty}(- p,\theta x_2) \hat \phi_{\infty}( p,y_{2}) \big(  \mc{B}_{+}^{\infty}(\eta, \theta p) - \mc{B}_{+}^{\infty}(\eta, \eta^{2} p)\big) \\
&\quad + \mf{e}^{\beta,L}_{0}(\eta,\theta,N)\;,
\end{split}
\end{equation}
where $\mf{e}^{\beta,L}_{0}(\eta,\theta,N)$ collects the error term $E_{00}^{\beta,L}(\eta,\theta)$ (updated to  take into account also $R^{\leq }_{0;N}$, Eqs. (\ref{eq:R0N})-(\ref{eq:R><})) plus: the error term coming from the Wick rotation at finite temperature, Eq. (\ref{eq:epsest}), which arises if $\eta >0$ is not in $(2\pi/\beta)\mathbb{N}$; the error term (\ref{eq:eN}); the error term (\ref{eq:R<}); the error term coming from the integral approximation (\ref{eq:sumtoint}). All in all, the net error term goes to zero in the following order of limits:
\begin{equation}\label{eq:errdef}
\lim_{\ell \to \infty} \lim_{N\to \infty} \lim_{(\eta,\theta) \to 0} \lim_{\substack{\beta, L_{i}\to \infty \\ \kappa L_{i} \ge \beta}}\mf{e}^{\beta,L}_{0}(\eta,\theta,N) = 0\;,
\end{equation}
where the relative order of $\eta$ and $\theta$ is irrelevant. 

In the following, we shall denote by $\mf{e}^{\beta,L}_{0,a}(\eta,\theta)$, with $a=1,2,3,4$, further error terms that satisfy (\ref{eq:errdef}). Consider the main term in (\ref{eq:chi0afterWI}). Observe that, from (\ref{eq:binfty}):
\begin{equation}
\mc{B}_{+}^{\infty}(\eta, \eta^{2} p) = \frac{1}{4\pi \big|v_{1,+}(\lambda)  v_{0,+}(\lambda) \big|} + \tilde r(\eta, \eta p)\;,
\end{equation}
where $\tilde r(\eta, \eta p)$ vanishes continuously as $(\eta, \theta) \to (0,0)$ for all fixed $p$. Then:
\begin{equation}
\begin{split}
&\mc{B}_{+}^{\infty}(\eta, \theta p) - \mc{B}_{+}^{\infty}(\eta, \eta^{2} p) \\
&\quad = \frac{1}{2\pi \big|v_{1,+}(\lambda)  v_{0,+}(\lambda) \big|}\frac{-v_{1,+}(\lambda)\theta p}{i v_{0,+}(\lambda)\eta +v_{1,+}(\lambda)\theta p} + r(\eta,\theta p) - \tilde r(\eta, \eta p)\;,
\end{split}
\end{equation}
we find:
\begin{equation}
\begin{split}
\chi^{\beta,L}_{0}(\eta, \theta) &= \int \frac{d p}{2\pi} \sum_{x_{2},y_{2}} (Z_{+}\star\overline{Z_{+}})(0;y_{2}) (Z_{+}\star\overline{Z_{+}})(0;x_{2}) \\
&\qquad \cdot \hat \mu_{\infty}( -p,\theta x_2) \hat \phi_{\infty}( p,y_{2}) \\
&\qquad\cdot\frac{1}{2\pi \big|v_{1,+}(\lambda)  v_{0,+}(\lambda) \big|}\frac{v_{1,+}(\lambda)\theta p}{i v_{0,+}(\lambda)\eta +v_{1,+}(\lambda)\theta p} \\
&\quad + \mf{e}^{\beta,L}_{0,2}(\eta,\theta,N)\;.
\end{split}
\end{equation}
In order to further simplify the expression, we will rely on the relations implied by the vertex Ward identity, Eqs. (\ref{eq:vertrel}). Using that:
\begin{equation}
\begin{split}
&\sum_{x_{2},y_{2}} (Z_{+}\star\overline{Z_{+}})(0;y_{2}) (Z_{+}\star\overline{Z_{+}})(0;x_{2})  \hat \mu(- p,\theta x_2) \hat \phi( p,y_{2}) \\
&\quad = v_{0,+}(\lambda)^{2} \hat \mu( p, 0) \hat \phi( -p, 0) + \mr{o}_{\theta,\ell}(1)\;,
\end{split}
\end{equation}
where $\mr{o}_{\theta,\ell}(1)$ vanishes as $\theta \to 0$ and $\ell\to \infty$,  we get:
\begin{equation}\label{eq:chi0aaa}
\begin{split}
\chi^{\beta,L}_{0}(\eta, \theta) &=  \int \frac{dp}{(2\pi)^2} \hat \mu_{\infty}(- p, 0) \hat \phi_{\infty}( p,0)\frac{\text{sgn}(v_{1,+})\,\theta p}{i \eta + \mf{v}(\lambda) \theta p} + \mf{e}^{\beta,L}_{0,3}(\eta,\theta,N)\;,
\end{split}
\end{equation}
where $\mf{v}(\lambda) := \lim_{\beta\to \infty} \lim_{L_{1},L_{2}\to \infty} v_{1,+} / v_{0,+}$, where the $L_{1}\to \infty$ limit is taken over sequences specified in Remark \ref{rem:lim} item (ii). The multiscale method used in this paper also allows to prove the existence of the limits; we refer the reader to {\it e.g.} \cite[Section 4.4]{GM} for a proof of this statement for the quasi-periodic Ising model, which can be adapted to the present case. In particular, for $\eta \ll \theta$:
\begin{equation}
\chi^{\beta,L}_{0}(\eta, \theta) = \frac{1}{2\pi |\mf{v}(\lambda)|} \langle \mu, \phi \rangle_{\text{edge}} + \mf{e}^{\beta,L}_{0,4}(\eta,\theta,N)\;,
\end{equation}
which proves the first of (\ref{eq:respf}).

\subsubsection{Edge conductance} 

To conclude the proof of Theorem \ref{thm:resp}, let us consider the edge conductance and prove the second of (\ref{eq:respf}). Recall the expression of the edge conductance after Wick rotation, for $\eta \in (2\pi / \beta) \mathbb{N}$:
\begin{equation}
\chi_{1}^{\beta,L}(\eta, \theta) = \theta\sum_{\vec x,\vec y } \mu(\theta\vec  x)  \phi_{\ell}(\theta y_{1}, y_{2}) \int_{-\beta/2}^{\beta/2}  d s\, e^{-i\eta s} \langle \timord \gamma_{s}(n_{\vec x})\;; j_{1,\vec y} \rangle_{\beta,\mu,L}\;,
\end{equation}
where the current operator is:
\begin{equation}\label{eq:currchi}
 j_{1,\vec y}:=  j_{\vec y,\vec y+\vec e_1}+\frac{1}{2}\big( j_{\vec y,\vec y+\vec e_1+\vec e_2}+  j_{\vec y,\vec y+\vec e_1-\vec e_2}+  j_{\vec y+\vec e_2,\vec y+\vec e_1}+  j_{\vec y-\vec e_2,\vec y+\vec e_1}\big)\;,
\end{equation}
and the bond current is:
\begin{equation}\label{eq:bondcu}
j_{\vec y,\vec z} = i\sum_{\sigma,\zeta} \big( H_{\zeta,\sigma}(\vec y,\vec z)  a^*_{\vec y ,\zeta} a_{\vec z ,\sigma}-H_{\sigma,\zeta}(\vec z,\vec y) a^*_{\vec z,\sigma} a_{\vec y ,\zeta}\big)\;.
\end{equation}
We will compute the edge conductance by studying all the contributions associated with the bond currents, and then we will put everything together. To this end, let us introduce the notation, for $\vec a,\vec b\in \Lambda_{L}$:
\begin{equation}
j_{y}^{ab} := i\sum_{\sigma,\zeta} H^{ab}_{\zeta,\sigma}(y_{2}) a^*_{\vec y + \vec a ,\zeta} a_{\vec y + \vec b,\sigma}\;,
\end{equation}
where:
\begin{equation}\label{eq:Hab}
H^{ab}_{\zeta,\sigma}(y_{2}) := H_{\zeta,\sigma}(\vec y + \vec a, \vec y + \vec b)\;;
\end{equation}
by translation-invariance in the $y_{1}$ direction, the right-hand side of (\ref{eq:Hab}) only depends on $\vec a, \vec b, y_{2}$. It is clear that the edge current can be expressed as a combination of the operators $j_{y}^{ab}$.

Thus, consider (dropping the $\beta, \mu, L$ labels in the state):
\begin{equation}
\begin{split}
&\chi_{ab}^{\beta,L}(\eta, \theta) \\
&= \theta\sum_{\vec x,\vec y } \mu(\theta\vec  x)  \phi_{\ell}(\theta y_{1}, y_{2}) \int_{-\frac{\beta}{2}}^{\frac{\beta}{2}}  d s\, e^{-i\eta s} \langle \timord \gamma_{s}(n_{\vec x})\;; j^{ab}_{\vec y} \rangle \\
&= i\theta\sum_{\vec x,\vec y }\sum_{\sigma,\zeta}\mu(\theta\vec  x)    \phi_{\ell}(\theta y_{1}, y_{2}) H^{ab}_{\zeta,\sigma}(y_{2}) \\
&\qquad\cdot\int_{-\frac{\beta}{2}}^{\frac{\beta}{2}}  d s\, e^{-i\eta s} \langle \timord \gamma_{s}(n_{\vec x})\;; a^*_{\vec y + \vec a ,\zeta} a_{\vec y + \vec b,\sigma} \rangle\;.
\end{split}
\end{equation}
We can evaluate the average by Wick's rule. We get:
\begin{equation}
\begin{split}
&\chi_{ab}^{\beta,L}(\eta, \theta) \\
&= -i\theta\sum_{\vec x,\vec y }\sum_{\sigma,\zeta,\rho}  \mu(\theta\vec  x)   \phi_{\ell}(\theta y_{1}, y_{2}) H^{ab}_{\zeta,\sigma}(y_{2}) \\
&\qquad \cdot \int_{-\frac{\beta}{2}}^{\frac{\beta}{2}}  d s\, e^{-i\eta s} S_{2;\rho,\zeta}\big((\vec x, s); (\vec y + \vec a, 0)\big) S_{2;\sigma,\rho}\big((\vec y + \vec b, 0); (\vec x, s)\big)\;.
\end{split}
\end{equation}
As for the susceptibility, we define the singular part of $\chi_{ab}^{\beta,L}(\eta, \theta)$ as:
\begin{equation}
\begin{split}
&\chi_{ab}^{\text{sing}}(\eta, \theta) \\
&= -i\theta\sum_{\vec x,\vec y } \sum_{\sigma,\zeta,\rho}\mu(\theta\vec  x)   \phi_{\ell}(\theta y_{1}, y_{2}) H^{ab}_{\zeta,\sigma}(y_{2}) \\
&\qquad \cdot \int_{-\frac{\beta}{2}}^{\frac{\beta}{2}}  d s\, e^{-i\eta s} S^{\mr{s}}_{2;\rho,\zeta}\big((\vec x, s); (\vec y + \vec a, 0)\big) S^{\mr{s}}_{2;\sigma,\rho}\big((\vec y + \vec b, 0); (\vec x, s)\big)\;.
\end{split}
\end{equation}
By Theorem \ref{thm:2pt}:
\begin{equation}
\begin{split}
&S^{\mr{s}}_{2;\rho,\zeta}\big((\vec x, s); (\vec y + \vec a, 0)\big) \\
&= \sum_{\omega} e^{i k_{F}^{\omega}(\lambda)(x_{1} - y_{1} - a_{1})}Z_{\omega,\rho}(\vec x) \check{g}_{\mr{s};\omega}\big(( x_{1},s); (y_{1} +a_{1}, 0)\big) \overline{Z_{\omega,\zeta}(y_{1} + a_{1})} \\
&S^{\mr{s}}_{2;\sigma,\rho}\big(( y_{1} +  b_{1}, 0); ( x_{1}, s)\big) \\
&= \sum_{\omega} e^{i k_{F}^{\omega}(\lambda)(y_{1} + b_{1} - x_{1})}Z_{\omega,\sigma}(\vec y + \vec b) \check{g}_{\mr{s};\omega}\big( y_{1}+ b_{1},0); ( x_{1}, s)\big) \overline{Z_{\omega,\rho}(\vec x)}\;.
\end{split}
\end{equation}
Now, for the purpose of determining the singular part of the response function, we can safely drop $a_{1}$ and $b_{1}$ in the relativistic propagator. In fact:
\begin{equation}\label{eq:diffprop}
\big| \check{g}_{\mr{s};\omega}\big(( y_{1} +  b_{1},0); ( x_{1}, s)\big) - \check{g}_{\mr{s};\omega}\big(( y_{1} ,0); (x_{1}, s)\big)\big| \le \frac{C_{b}}{1 + \| ( x_{1}, s) - ( y_{1}, 0) \|^{1+\alpha}}\;,
\end{equation}
with $\alpha > 0$. Since in our application $|\vec b| \le 2$, the above bound proves that the difference between the propagators has improved decay. We use this observation to rewrite:
\begin{equation}\label{eq:chi1sing}
\begin{split}
\chi_{ab}^{\text{sing}}(\eta, \theta) &= -i\theta\sum_{\vec x,\vec y }\sum_{\sigma,\zeta,\rho} \mu(\theta\vec  x)    \phi_{\ell}(\theta y_{1}, y_{2}) e^{i k_{F}^{+}(\lambda)(b_{1} - a_{1})}H^{ab}_{\zeta,\sigma}(y_{2}) \\
&\qquad \cdot Z_{+,\rho}(\vec x) \overline{Z_{+,\zeta}(\vec y + \vec a)} Z_{+,\sigma}(\vec y + \vec b) \overline{Z_{+,\rho}(\vec x)} \\
&\qquad \cdot \int_{-\frac{\beta}{2}}^{\frac{\beta}{2}}  d s\, e^{-i\eta s} \check{g}_{\mr{s};+}\big(( x_{1},s); ( y_{1}, 0)\big) \check{g}_{\mr{s};+}\big(( y_{1},0); ( x_{1}, s)\big) \\
&\quad + \mf{e}_{1,1}(\eta, \theta) + \mf{e}_{1,2}(\eta, \theta)
\end{split}
\end{equation}
where $\mf{e}_{1,1}(\eta, \theta)$ collects the contributions coming from the chiralities associated with the edge mode as $x_{2} = L_{2}$, and $\mf{e}_{1,2}(\eta, \theta)$ the contribution coming from the difference of the propagators (\ref{eq:diffprop}). Therefore,
\begin{equation}
| \mf{e}_{1,1}(\eta, \theta) | \le C_{\theta,\ell} e^{-cL_{2}}\;,
\end{equation}
while $\mf{e}_{1,2}(\eta, \theta)$ satisfies the bound (\ref{eq:contK}), and hence it can be absorbed into the regular part. Now, let $Z_{+,\zeta}(\vec y + \vec a) \equiv Z^{a}_{+,\zeta}(\vec y)$, where:
\begin{equation}\label{eq:Za}
\begin{split}
Z^{a}_{+,\zeta}(\vec y) &= \sum_{n} e^{-i n \alpha y_{1}} e^{-i n \alpha a_{1}} Z_{+,n,\zeta}(y_{2} + a_{2}) \\
&\equiv \sum_{n} e^{-i n \alpha y_{1}} Z^{a}_{+,n,\zeta}(y_{2})\;.
\end{split}
\end{equation}
Let us denote by $\widetilde{\chi}^{\text{sing}}_{ab}(\eta,\theta)$ the main contribution to the right-hand side of (\ref{eq:chi1sing}). Using the above notation, we have:
\begin{equation}
\begin{split}
\widetilde \chi_{ab}^{\text{sing}}(\eta, \theta) &= -i\theta\sum_{\vec x,\vec y } \sum_{\sigma,\zeta,\rho}  \mu(\theta\vec  x)  \phi_{\ell}(\theta y_{1}, y_{2}) e^{i k_{F}^{+}(\lambda)(b_{1} - a_{1})}H^{ab}_{\zeta,\sigma}(y_{2}) \\
&\qquad \cdot Z_{+,\rho}(\vec x) \overline{Z^{a}_{+,\zeta}(\vec y)} Z^{b}_{+,\sigma}(\vec y) \overline{Z_{+,\rho}(\vec x)} \\
&\qquad \cdot \int_{-\frac{\beta}{2}}^{\frac{\beta}{2}}  d s\, e^{-i\eta s} \check{g}_{\mr{s};+}\big(( x_{1},s); (y_{1}, 0)\big) \check{g}_{\mr{s};+}\big(( y_{1},0); ( x_{1}, s)\big)\;,
\end{split}
\end{equation}
which we rewrite as, using (\ref{eq:Fouriersingularpart}):
\begin{equation}
\begin{split}
\widetilde \chi_{ab}^{\text{sing}}(\eta, \theta) &= -i\theta\sum_{\vec x,\vec y } \sum_{\sigma,\zeta,\rho} \mu(\theta\vec  x)    \phi_{\ell}(\theta y_{1}, y_{2}) e^{i k_{F}^{+}(\lambda)(b_{1} - a_{1})}H^{ab}_{\zeta,\sigma}(y_{2}) \\
&\qquad \cdot Z_{+,\rho}(\vec x) \overline{Z^{a}_{+,\zeta}(\vec y)} Z^{b}_{+,\sigma}(\vec y) \overline{Z_{+,\rho}(\vec x)} \\
&\qquad \cdot \frac{1}{L_{1}} \sum_{p} e^{ip (x_{1} - y_{1})} \mc{B}^{\beta,L_{1}}_{+}(\eta,p)\;.
\end{split}
\end{equation}
That is:
\begin{equation}
\begin{split}
\widetilde \chi_{ab}^{\text{sing}}(\eta, \theta) &=  -\frac{\theta}{L_{1}}\sum_{p} \sum_{x_{2},y_{2}}\sum_{\sigma,\zeta,\rho}\mathcal{F}\big(\mu_{\theta}(\cdot,x_{2}) |Z_{+,\rho} (\cdot,x_{2})|^{2}\big)(-p) \\
&\qquad \cdot \mathcal{F}\big(\phi^{ab}_{\theta,\ell;\zeta,\sigma}(\cdot,y_{2}) Z^{b}_{+,\sigma}(\cdot,y_{2}) \overline{Z^{a}_{+,\zeta}(\cdot, y_{2})}\big)(p) \mc{B}_{+}^{\beta, L_{1}}(\eta, p)\;,
\end{split}
\end{equation}
where we set:
\begin{equation}\label{eq:phiab}
\phi^{ab}_{\theta,\ell;\zeta,\sigma}(\vec y) := i\phi_{\theta,\ell}(\vec y) e^{i k_{F}^{+}(\lambda)(b_{1} - a_{1})} H^{ab}_{\zeta,\sigma}(y_{2})\;.
\end{equation}
The next step is to expand the oscillatory factors into their Fourier series, and proceed as after (\ref{eq:ZZconv}) to isolate the main singular term from other contributions that are either small or regular in $\eta, \theta$; the only difference in the present setting is the $a,b$ decorations of the Fourier coefficients, and the redefinition (\ref{eq:phiab}) of the test function. Both are completely irrelevant for the purpose of the discussion following (\ref{eq:sing2}) till (\ref{eq:sumtoint}). Thus, we get:
\begin{equation}
\begin{split}
\widetilde \chi_{ab}^{\text{sing}}(\eta, \theta) &=  - \int \frac{d p}{2\pi} \sum_{x_{2},y_{2}} \sum_{\sigma, \zeta,\rho} (Z^{b}_{+,\sigma}\star\overline{Z^{a}_{+,\zeta}})(0;y_{2}) (Z_{+,\rho}\star\overline{Z_{+,\rho}})(0;x_{2}) \\
&\qquad \cdot \hat \mu_{\infty}( -p,\theta x_2) \hat \phi^{ab}_{\infty; \zeta,\sigma}( p,y_{2})   \mc{B}_{+}^{\infty}(\eta, \theta p) + \mr{o}(1)\;,
\end{split}
\end{equation}
where:
\begin{equation}
\hat \phi^{ab}_{\infty;\zeta,\sigma}( -p,y_{2}) = i \hat \phi_{\infty}(-p,y_{2}) e^{i k_{F}^{+}(\lambda)(b_{1} - a_{1})} H^{ab}_{\zeta,\sigma}(y_{2})\;.
\end{equation}
Up to vanishing error terms as $\theta\to 0$, $\ell \to \infty$, we can localize the test functions on the boundary. We get:
\begin{equation}
\begin{split}
\widetilde \chi_{ab}^{\text{sing}}(\eta, \theta) = & - \Big(\int \frac{d p}{2\pi}\,  \hat \mu_{\infty}( -p, 0)\hat \phi_{\infty}(p,0)  \mc{B}_{+}^{\infty}(\eta, \theta p)\Big)\\
&\quad\cdot i v_{0,+}(\lambda) \sum_{\sigma,\zeta}\sum_{y_{2}} (Z^{b}_{+,\sigma}*\overline{Z^{a}_{+,\zeta}})(0;y_{2}) e^{i k_{F}^{+}(\lambda)(b_{1} - a_{1})} H^{ab}_{\zeta,\sigma}(y_{2})  \\
&+ \mr{o}(1)\;,
\end{split}
\end{equation}
where we used the first of the relations (\ref{eq:vertrel}), as for the susceptibility. Let us now consider the last factor. We rewrite it as:
\begin{equation}
\begin{split}
&i\sum_{\sigma,\zeta} \sum_{y_{2}} (Z^{b}_{+,\sigma}*\overline{Z^{a}_{+,\zeta}})(0;y_{2}) e^{i k_{F}^{+}(\lambda)(b_{1} - a_{1})} H^{ab}_{\zeta,\sigma}(y_{2}) \\
&\quad =i \sum_{n} e^{i (k_{F}^{+} - n\alpha)(b_{1} - a_{1})}\sum_{y_{2}} \big\langle Z_{+,n}(y_{2} + a_{2}),  H_{\zeta,\sigma}(\vec y + \vec a; \vec y + \vec b) Z_{+,n}(y_{2} + b_{2})\big\rangle \;,
\end{split}
\end{equation}
recall (\ref{eq:Hab}) and (\ref{eq:Za}), and where the scalar product is over the internal degrees of freedom as in (\ref{eq:zetamu}). We now have to add up all the possible choices of $(\vec a, \vec b)$, and subtracting the case $\vec a \leftrightarrow \vec b$, because of the definition of bond current (\ref{eq:bondcu}). From (\ref{eq:currchi}), the possible choices for $(\vec a, \vec b)$, associated with the bond currents, are:
\begin{equation}\label{eq:abcases}
(0, \vec e_{1})\;,\quad (0, \vec e_{1} + \vec e_{2})\;,\quad (0, \vec e_{1} - \vec e_{2})\;,\quad (\vec e_{2}, \vec e_{1})\;,\quad (-\vec e_{2}, \vec e_{1})\;,
\end{equation}
minus the cases $\vec a\leftrightarrow \vec b$. Except for the first case in (\ref{eq:abcases}), the other bond currents come with a factor $1/2$. However, when summing over $y_{2}$, second case and the fourth case, and the third case and the fifth case, give the same outcome, using the Dirichlet boundary conditions; recall Eq. (\ref{eq:jsumx2}).

Therefore, summing over all bond currents we get:
\begin{equation}\label{eq:partH}
\begin{split}
&i\sum_{n} e^{i (k_{F}^{+}(\lambda) - n\alpha)} \sum_{y_{2}} \big\langle Z_{+,n}(y_{2}),  H(\vec y; \vec y + \vec e_{1}) Z_{+,n}(y_{2})\big\rangle \\
& + i\sum_{n}  e^{i (k_{F}^{+}(\lambda) - n\alpha)} \sum_{y_{2}} \big\langle Z_{+,n}(y_{2}),  H(\vec y; \vec y + \vec e_{1} + \vec e_{2}) Z_{+,n}(y_{2} + 1)\big\rangle \\ 
& + i\sum_{n} e^{i (k_{F}^{+}(\lambda) - n\alpha)} \sum_{y_{2}} \big\langle Z_{+,n}(y_{2}),  H(\vec y; \vec y + \vec e_{1} - \vec e_{2}) Z_{+,n}(y_{2} - 1)\big\rangle + \text{c.c.;}
\end{split}
\end{equation}
adding the complex conjugate is the same as subtracting the cases $\vec a\leftrightarrow \vec b$. 
Recalling (\ref{eq:currentkernel}) and (\ref{eq:currentkernelzero}), we rewrite the first term as:
\begin{equation}
\begin{split}
&i\sum_{n} e^{i (k_{F}^{+}(\lambda) - n\alpha)} \sum_{y_{2}} \big\langle Z_{+,n}(y_{2}),  H(\vec y; \vec y + \vec e_{1}) Z_{+,n}(y_{2})\big\rangle + \text{c.c.} \\
&= \sum_{n} \big\langle Z_{+,n}(y_{2}), \partial_{k_{1}} \hat H\big(k_{F}^{+}(\lambda) - n\alpha; y_{2}, y_{2}\big) Z_{+,n}(y_{2}) \big\rangle\;, 
\end{split}
\end{equation}
 Repeating the computation for the other two cases in (\ref{eq:partH}), we see that:
\begin{equation}
(\ref{eq:partH}) = \sum_{n} \big\langle Z_{+,n}, \partial_{k_{1}} \hat H\big(k_{F}^{+}(\lambda) - n\alpha\big) Z_{+,n} \big\rangle\;,
\end{equation}
where now the scalar product is over the internal degrees of freedom and over the second spatial variable. These quantity is precisely the coefficient $\zeta_{1,+}$, appearing in the analysis of the vertex function, Eq. (\ref{eq:zetamu}). By the second of (\ref{eq:vertrel}), we know that $\zeta_{1,+} = v_{1,+}$. Coming back to the edge conductance, we obtained the following representation for its singular contribution:
\begin{equation}
\widetilde \chi^{\text{sing}}_{1}(\eta, \theta) = v_{0,+}(\lambda) v_{1,+}(\lambda) \Big(\int \frac{d p}{2\pi}\,  \hat \mu_{\infty}( -p, 0)\hat \phi_{\infty}(p,0)  \mc{B}_{+}^{\infty}(\eta, \theta p)\Big) + \mr{o}(1)\;.
\end{equation}
By the consequence of the current-current Ward identity (\ref{eq:chi1WI}), we obtain the analogue of (\ref{eq:chi0afterWI}), (\ref{eq:chi0aaa}) namely:
\begin{equation}
\chi^{\beta,L}_{1}(\eta, \theta) = \int \frac{dp}{(2\pi)^2} \hat \mu_{\infty}( -p, 0) \hat \phi_{\infty}( p,0)\frac{|\mf{v}(\lambda)|\theta p}{i \eta + \mf{v}(\lambda) \theta p} + \tilde{\mf{e}}^{\beta,L}_{1}(\eta,\theta,N)\;.
\end{equation}
In particular, for $\eta \ll \theta$:
\begin{equation}
\chi^{\beta,L}_{1}(\eta, \theta) =  \frac{\text{sgn}(\mf{v}(\lambda))}{2\pi} \langle \mu, \varphi \rangle_{\text{edge}} + \mf{e}_{1}^{\beta,L}(\eta,\theta,N)\;.
\end{equation}
This proves the second of (\ref{eq:respf}), and concludes the proof of Theorem \ref{thm:resp}.  \qed

\appendix

\section{Decay estimate away from the Fermi level}\label{app:CT}

In this appendix we shall prove the bound (\ref{eq:bdGb}) for the propagator associated with energies away from the Fermi level $\mu$. Recall:
\begin{equation}\label{eq:CT1}
\begin{split}
G^{(\text{b})}(\bm{k},x_2, y_2) &= \Big(\big(1-\chi(k_{0}) \chi\big(\hat H(k_{1}) - \mu\big)\big) G(\bm{k})\Big)(x_{2},y_{2}) \\
&= (1-\chi(k_{0})) G(\bm{k},x_2, y_2) + \chi(k_{0}) \Big(\big(1 - \chi\big(\hat H(k_{1}) - \mu\big)\big) G(\bm{k})\Big)(x_{2},y_{2}) \\
&\equiv G^{(\text{b}); 1}(\bm{k},x_2, y_2) + G^{(\text{b}); 2}(\bm{k},x_2, y_2)\;.
\end{split}
\end{equation}
For the first term in (\ref{eq:CT1}), the function $(1-\chi(k_{0}))$ implies that $|k_{0}| \geq C\delta$. For a finite ranged Hamiltonian $H$, the Combes-Thomas bound states that, for $z \notin \sigma(H)$, see {\it e.g.} \cite{AW}:
\begin{equation}
\Big\| \frac{1}{H - z}(x_{2}, y_{2}) \Big\| \leq \frac{C}{|\text{dist}(z,\sigma(H))|} e^{-c|x_{2} -y_{2}|}\;,
\end{equation}
where the constant $c$ depends on the distance from $z$ to the spectrum of $H$. Thus, in our case, for $|k_{0}|\geq C\delta$:
\begin{equation}\label{eq:Gb1}
\Big\| G^{(\text{b}); 1}(\bm{k},x_2, y_2)\Big\| \leq \frac{C_{\delta}}{1+|k_{0}|} e^{-c|x_{2}- y_{2}|}\;.
\end{equation}
Consider now the term $G^{(\text{b}); 2}_{\sigma,\zeta}(\bm{k},x_2, y_2)$. By the approximate linearity of the edge modes close to $k_{F}^{\omega}$, if $|k_{1} - k_{F}^{\omega}| \geq \kappa \delta / v_{\omega}$ for all $\omega$ and for a suitable $\kappa>0$ then:
\begin{equation}
G^{(\text{b}); 2}(\bm{k},x_2, y_2) = \frac{\chi(k_{0})}{ik_{0} + \hat H(k_{1})- \mu}(x_{2}, y_{2})
\end{equation}
and $\text{dist}(\mu, \sigma(\hat H(k_{1}))) \geq c\delta$. Therefore, we can apply the Combes-Thomas bound and we get:
\begin{equation}\label{eq:Gb21}
\Big\|G^{(\text{b}); 2}(\bm{k},x_2, y_2)\Big\| \leq \frac{C_{\delta}}{|k_{0}| + 1} e^{-c|x_{2} - y_{2}|}\;.
\end{equation}
Finally, suppose that $|k_{1} - k_{F}^{\omega}| < \delta / v_{\omega}$ for some $\omega$. Let us introduce the short-hand notation:
\begin{equation}
\chi_{>\delta}(\hat H(k_{1})) := 1 - \chi\big(\hat H(k_{1}) - \mu\big)\;.
\end{equation}
Observe that the function $\chi_{>\delta}(\lambda)$ is equal to $1$ for $|\lambda - \mu| \geq 2\delta$, it vanishes for $|\lambda| \leq \delta$ and it is between $0$ and $1$ for $\delta \leq |\lambda - \mu| \leq 2\delta$. We write:
\begin{equation}\label{eq:A7}
\begin{split}
G^{(\text{b}); 2}(\bm{k}) = \chi(k_{0}) \chi_{>\delta}(\hat H(k_{1})) P G_{\sigma,\zeta}(\bm{k}) + \chi(k_{0}) \chi_{>\delta}(\hat H(k_{1})) P^{\perp} G_{\sigma,\zeta}(\bm{k})
\end{split}
\end{equation}
wher $P = \mathbbm{1}(|\hat H(k_{1}) - \mu| \leq 2\delta)$. Observe that, by the assumptions on the Hamiltonian (Assumption \ref{ass:H3}):
\begin{equation}
P = \sum_{e: |\varepsilon_{e}(k_{1}) - \mu| \leq 2\delta} P_{e}\;,
\end{equation}
with $P_{e}$ the projector over the edge mode of $\hat H(k_{1})$ labelled by $e$. Also,
\begin{equation}\label{eq:PPperp}
\chi_{>\delta}(\hat H(k_{1})) P^{\perp} = P^{\perp}\;,\qquad P^{\perp} \frac{1}{ik_{0} + \hat H(k_{1}) - \mu} = P^{\perp} \frac{1}{ik_{0} + P^{\perp}\hat H(k_{1}) - \mu}\;.
\end{equation}
Consider the first term in the right-hand side of (\ref{eq:A7}). We have:
\begin{equation}
\chi(k_{0}) \chi_{>\delta}(\hat H(k_{1})) P G_{\sigma,\zeta}(\bm{k}) = \chi(k_{0})\sum_{e: |\varepsilon_{e}(k_{1}) - \mu| \leq 2\delta} \frac{\chi_{>\delta}(\varepsilon_{e}(k_{1}))}{ik_{0} + \varepsilon_{e}(k_{1}) - \mu} P_{e}(x_{2};y_{2})\;.
\end{equation}
By the exponential decay of the edge modes, we easily get:
\begin{equation}\label{eq:Gb3}
\Big\| \chi(k_{0}) \chi_{>\delta}(\hat H(k_{1})) P G_{\sigma,\zeta}(\bm{k}) \Big\| \leq C_{\delta} \chi(k_{0}) e^{-|x_{2} - y_{2}|}\;.
\end{equation}
Consider now the second term in the right-hand side of (\ref{eq:A7}). Here we have:
\begin{equation}
\chi(k_{0}) \chi_{>\delta}(\hat H(k_{1})) P^{\perp} G_{\sigma,\zeta}(\bm{k}) = \chi(k_{0}) P^{\perp} \frac{1}{ik_{0} + P^{\perp}\hat H(k_{1}) - \mu}\;,
\end{equation}
and 
\begin{equation}
P^{\perp}\hat H(k_{1}) = \hat H(k_{1}) - \sum_{e: |\varepsilon_{e}(k_{1}) - \mu| \leq 2\delta} \varepsilon_{e}(k_{1}) P_{e} =: \widetilde{H}(k_{1})\;.
\end{equation}
By Assumption \ref{ass:H3}, we observe that $\text{dist}(\mu, \sigma(\widetilde{H}(k_{1}))) \geq 2\delta$. The Hamiltonian $\widetilde{H}(k_{1})$ is not finite-ranged, due to the projectors of the edge modes. However, one can check that, for $|\alpha|$ sufficiently small:
\begin{equation}
\| e^{\alpha \hat x_{2}} \widetilde{H}(k_{1}) e^{-\alpha \hat x_{2}} - \widetilde{H}(k_{1})  \| \leq C|\alpha|\;.
\end{equation}
This bound is all it is needed to prove the Combes-Thomas estimate for the resolvent of $\widetilde{H}(k_{1})$; it follows from the short-range of $\hat H(k_{1})$, and from the exponential decay of the edge modes. Thus, from the exponential decay of $P^{\perp}(x_{2}, y_{2}) = \mathbbm{1}(x_{2},y_{2}) - P(x_{2},y_{2})$, and from the Combes-Thomas bound for the resolvent of $\widetilde{H}(k_{1})$, we obtain:
\begin{equation}\label{eq:Gb4}
\Big\| \chi(k_{0}) \Big( P^{\perp} \frac{1}{ik_{0} + P^{\perp}\hat H(k_{1}) - \mu}\Big) (x_{2},y_{2}) \Big\| \leq C_{\delta} \chi(k_{0}) e^{-c|x_{2} - y_{2}|}\;.
\end{equation}
Finally, from (\ref{eq:Gb1}), (\ref{eq:Gb21}), (\ref{eq:Gb3}), (\ref{eq:Gb4}), we get:
\begin{equation}
\Big\| G^{(\text{b})}(\bm{k},x_2, y_2) \Big\| \leq \frac{C_{\delta}}{|k_{0}| + 1} e^{-c|x_{2} - y_{2}|}\;.
\end{equation}
The bound for the derivatives of $G^{(\text{b})}$ can be proved in the same way, using the exponential decay of the derivatives of the edge modes, Assumption \ref{ass:H3}. We omit the details.

%&\quad = \chi(k_{0})\sum_{e: |\varepsilon_{e}(k_{1}) - \mu| \leq 2\delta} \frac{\chi_{>\delta}(\varepsilon_{e}(k_{1}))}{ik_{0} + \varepsilon_{e}(k_{1}) - \mu} P_{e}(x_{2};y_{2}) \\
%&\qquad + \chi(k_{0}) (1 - P) G_{\sigma,\zeta}(\bm{k})\;,

\section{Recursion relation for the two-point function}\label{app:2pt}

In this appendix we shall prove the coupled recursion relations (\ref{eq:recphi2}), (\ref{eq:rec2}). 

\paragraph{Proof of (\ref{eq:recphi2}).} To begin, we write:
\begin{equation}\label{eq:A1}
W_{\phi\phi;n,\sigma,\zeta}^{(h)}({\bm k}; x_{2}, y_{2}) = \beta L_{1} \frac{\partial^{2}}{\partial \phi^{-}_{{\bm k} + n{\bm \alpha}, y_{2}, \zeta}\partial \phi^{+}_{{\bm k}, x_{2}, \sigma}} \log \mathbbm{E}_{h}\Big( e^{V^{(h)}(\psi^{(\leq h)}, \phi)} \Big)\Big|_{\phi=\psi^{(<h)} = 0}\;.
\end{equation}
where $\mathbbm{E}_{h}$ is the Gaussian expectation associated with $\widetilde{P}_{h}(d\psi^{(h)})$, recall (\ref{eq:tildeP2pt}), and $V^{(h)}(\psi^{(\leq h)}, \phi)$ is a short-hand notation for the argument of the exponential in (\ref{eq:tildeP2pt}). We compute:
\begin{equation}\label{eq:1der}
\begin{split}
&\frac{\partial}{\partial \phi^{+}_{{\bm k}, x_{2}, \sigma}}  \log \mathbbm{E}_{h}\Big( e^{V^{(h)}(\psi^{(\leq h)}, \phi)} \Big) \\
&\qquad = \frac{1}{\beta L_{1}} \sum_{m,\omega} W^{(h)}_{\phi\psi;m,\omega,\sigma}({\bm q}({\bm k}), x_{2})   \frac{\mathbbm{E}_{h}\Big( \psi^{(\leq h)-}_{{\bm q}({\bm k}) + m{\bm \alpha}, \omega}  e^{V^{(h)}(\psi^{(\leq h)}, \phi)}\Big)}{\mathbbm{E}_{h}\Big( e^{V^{(h)}(\psi^{(\leq h)}, \phi)} \Big)}\;,
\end{split}
\end{equation}
and:
\begin{equation}\label{eq:A3}
\begin{split}
&\frac{\partial^{2}}{\partial \phi^{-}_{{\bm k} + n{\bm \alpha}, y_{2}, \zeta}\partial \phi^{+}_{{\bm k}, x_{2}, \sigma}} \log \mathbbm{E}_{h}\Big( e^{V^{(h)}(\psi^{(\leq h)}, \phi)} \Big)\Big|_{\psi^{(<h)} = \phi = 0} =  \\
&\quad = \frac{1}{\beta L_{1}} \sum_{m,\omega} W^{(h)}_{\phi\psi;m,\omega,\sigma}({\bm q}({\bm k}), x_{2})   \frac{\partial}{\partial  \phi^{-}_{{\bm k} + n{\bm \alpha}, y_{2}, \zeta}}\frac{\mathbbm{E}_{h}\Big( \psi^{(\leq h)-}_{{\bm q}({\bm k}) + m{\bm \alpha}, \omega}  e^{V^{(h)}(\psi^{(\leq h)}, \phi)}\Big)}{\mathbbm{E}_{h}\Big( e^{V^{(h)}(\psi^{(\leq h)}, \phi)} \Big)}\Big|_{\psi^{(<h)} = \phi = 0} \\
&\quad =  \frac{1}{(\beta L_{1})^{2}} \sum_{m,\omega} W^{(h)}_{\phi\psi;m,\omega,\sigma}({\bm q}({\bm k}), x_{2}) \sum_{m',\omega'} W^{(h)}_{\psi\phi;m',\omega',\zeta}({\bm q}(\bm k) + (n-m'){\bm \alpha}, y_{2}) \\
&\qquad \cdot \frac{\mathbbm{E}_{h}\Big( \psi^{(\leq h)-}_{{\bm q}({\bm k}) + m{\bm \alpha}, \omega} \psi^{(\leq h)+}_{{\bm q}({\bm k}) + (n-m'){\bm \alpha}, \omega'} e^{V^{(h)}(\psi^{(\leq h)}, 0)} \Big)  }{\mathbbm{E}_{h}\Big( e^{V^{(h)}(\psi^{(\leq h)}, 0)} \Big)}\Big|_{\psi^{(<h)} = 0}\;,
\end{split}
\end{equation}
where we used that the average of an odd number of Grassmann variables is zero. Next, using that the distribution of the Gaussian Grassmann integration is:
\begin{equation}
\exp \Big(- \frac{1}{\beta L_{1}} \sum_{{\bm q}, \omega} \psi^{(h)+}_{{\bm q}, \omega} g^{(h)}_{\omega}({\bm q})^{-1} \psi^{(h)-}_{{\bm q}, \omega} \Big)
\end{equation}
we have, by Grassmann integration by parts:
\begin{equation}\label{eq:parts}
\mathbb{E}_{h}(\psi^{(h)-}_{{\bm q}, \omega} A) = \beta L_{1} g^{(h)}_{\omega}({\bm q}) \mathbbm{E}_{h}\Big( \frac{\partial}{\partial \psi^{(h)+}_{{\bm q},\omega}}A \Big)\;,\quad \mathbb{E}_{h}(\psi^{(h)+}_{{\bm q}, \omega} A) = -\beta L_{1} g^{(h)}_{\omega}({\bm q}) \mathbbm{E}_{h}\Big( \frac{\partial}{\partial \psi^{(h)-}_{{\bm q},\omega}}A \Big)\;.
\end{equation}
We use these identities to write:
\begin{equation}
\begin{split}
&\mathbbm{E}_{h}\Big( \psi^{(h)-}_{{\bm q}({\bm k}) + m{\bm \alpha}, \omega} \psi^{(h)+}_{{\bm q}({\bm k}) + (n-m'){\bm \alpha}, \omega'} e^{V^{(h)}(\psi^{(\leq h)}, 0)} \Big) \\
&\qquad = \beta L_{1} g^{(h)}_{\omega}({\bm q}(\bm k) + m{\bm \alpha}) \mathbbm{E}_{h}\Big( \frac{\partial}{\partial \psi^{(h)+}_{{\bm q}(\bm k) + m{\bm \alpha},\omega}} \psi^{(h)+}_{{\bm q}({\bm k}) + (n-m'){\bm \alpha}, \omega'} e^{V^{(h)}(\psi^{(\leq h)}, 0)} \Big) \\
&\qquad = \delta_{\omega,\omega'} \delta_{n-m',m} \beta L_{1} g^{(h)}_{\omega}({\bm q}(\bm k) + m{\bm \alpha}) \mathbbm{E}_{h}\Big(  e^{V^{(h)}(\psi^{(\leq h)}, 0)} \Big) \\
&\qquad\quad -  \beta L_{1} g^{(h)}_{\omega}({\bm q}(\bm k) + m{\bm \alpha}) \mathbbm{E}_{h}\Big(  \psi^{(h)+}_{{\bm q}({\bm k}) + (n-m'){\bm \alpha}, \omega'} \frac{\partial}{\partial \psi^{(h)+}_{{\bm q}(\bm k) + m{\bm \alpha},\omega}} e^{V^{(h)}(\psi^{(\leq h)}, 0)} \Big) \\
&\qquad = \delta_{\omega,\omega'} \delta_{n-m',m} \beta L_{1} g^{(h)}_{\omega}({\bm q}(\bm k) + m{\bm \alpha}) \mathbbm{E}_{h}\Big(  e^{V^{(h)}(\psi^{(\leq h)}, 0)} \Big) \\
&\qquad\quad + (\beta L_{1})^{2} g^{(h)}_{\omega}({\bm q}(\bm k) + m{\bm \alpha}) g^{(h)}_{\omega'}({\bm q}(\bm k) + (n-m'){\bm \alpha}) \\&\qquad\qquad \cdot \mathbbm{E}_{h}\Big(  \frac{\partial^{2}}{ \partial \psi^{(h)-}_{{\bm q}({\bm k}) + (n-m'){\bm \alpha}, \omega'} \partial \psi^{(h)+}_{{\bm q}(\bm k) + m{\bm \alpha},\omega}} e^{V^{(h)}(\psi^{(\leq h)}, 0)} \Big)\;.
\end{split}
\end{equation}
Thus, from this formula we easily get:
\begin{equation}\label{eq:psipsi}
\begin{split}
&\frac{\mathbbm{E}_{h}\Big( \psi^{(\leq h)-}_{{\bm q}({\bm k}) + m{\bm \alpha}, \omega} \psi^{(\leq h)+}_{{\bm q}({\bm k}) + (n-m'){\bm \alpha}, \omega'} e^{V^{(h)}(\psi^{(\leq h)}, 0)} \Big)  }{\mathbbm{E}_{h}\Big( e^{V^{(h)}(\psi^{(\leq h)}, 0)} \Big)}\Big|_{\psi^{(<h)} = 0} \\
&\qquad = \delta_{\omega,\omega'}\delta_{n-m',m}\beta L_{1} g^{(h)}_{\omega}({\bm q}(\bm k) + m{\bm \alpha}) \\
&\qquad\quad + (\beta L_{1})^{2} g^{(h)}_{\omega}({\bm q}(\bm k) + m{\bm \alpha}) g^{(h)}_{\omega'}({\bm q}(\bm k) + (n-m'){\bm \alpha})\\&\qquad\qquad\cdot  \frac{\partial^{2}}{ \partial \psi^{(<h)-}_{{\bm q}({\bm k}) + (n-m'){\bm \alpha}, \omega'} \partial \psi^{(<h)+}_{{\bm q}(\bm k) + m{\bm \alpha},\omega}} \log \mathbbm{E}_{h}\Big( e^{V^{(h)}(\psi^{(\leq h)}, 0)} \Big)\Big|_{\psi^{(<h)} = 0}\;;
\end{split}
\end{equation}
to obtain the expression, we used that differentiating the integrand with respect to $\psi^{(h)}$ or $\psi^{(<h)}$ is equivalent, and the fact that the average of an odd number of Grassmann variables is zero. Next, we observe that, from the definition of effective potential:
\begin{equation}\label{eq:effpot}
\begin{split}
& \frac{\partial^{2}}{ \partial \psi^{(<h)-}_{{\bm q}({\bm k}) + (n-m'){\bm \alpha}, \omega'} \partial \psi^{(<h)+}_{{\bm q}(\bm k) + m{\bm \alpha},\omega}} \log \mathbbm{E}_{h}\Big( e^{V^{(h)}(\psi^{(\leq h)}, 0)} \Big)\Big|_{\psi^{(<h)} = 0} \\
&\qquad = \frac{1}{\beta L_{1}}V^{(h-1)}_{\omega,\omega';n - m' - m}({\bm q}(\bm k) + m{\bm \alpha})\;;
 \end{split}
\end{equation}
plugging this formula in (\ref{eq:psipsi}) we get:
\begin{equation}
\begin{split}
&\frac{\mathbbm{E}_{h}\Big( \psi^{(\leq h)-}_{{\bm q}({\bm k}) + m{\bm \alpha}, \omega} \psi^{(\leq h)+}_{{\bm q}({\bm k}) + (n-m'){\bm \alpha}, \omega'} e^{V^{(h)}(\psi^{(\leq h)}, 0)} \Big)  }{\mathbbm{E}_{h}\Big( e^{V^{(h)}(\psi^{(\leq h)}, 0)} \Big)}\Big|_{\psi^{(<h)} = 0} \\
&\qquad = \delta_{\omega,\omega'}\delta_{n-m',m} \beta L_{1} g^{(h)}_{\omega}({\bm q}(\bm k) + m{\bm \alpha})\\
&\qquad\quad +  \beta L_{1} g^{(h)}_{\omega}({\bm q}(\bm k) + m{\bm \alpha}) g^{(h)}_{\omega'}({\bm q}(\bm k) + (n-m'){\bm \alpha}) V^{(h-1)}_{\omega,\omega';n-m'-m}({\bm q}(\bm k) + m{\bm \alpha})\;.
\end{split}
\end{equation}
Inserting this expression in (\ref{eq:A3}), we find, recalling (\ref{eq:A1}):
\begin{equation}
\begin{split}
&W_{\phi\phi;n,\sigma,\zeta}^{(h)}({\bm k}; x_{2}, y_{2}) \\
&\quad = \sum_{m,\omega} W^{(h)}_{\phi\psi;m,\omega,\sigma}({\bm q}({\bm k}), x_{2}) g^{(h)}_{\omega}({\bm q}({\bm k}) + m{\bm \alpha}) W^{(h)}_{\psi\phi;n-m,\omega,\zeta}({\bm q}({\bm k}) + m{\bm \alpha}, y_{2}) \\
&\qquad + \sum_{\substack{m,\omega \\ m',\omega'}} W^{(h)}_{\phi\psi;m,\omega,\sigma}({\bm q}({\bm k}), x_{2})  g^{(h)}_{\omega}({\bm q}(\bm k) + m{\bm \alpha}) V^{(h-1)}_{\omega,\omega'; n-m'-m}({\bm q(\bm k)} + m{\bm \alpha}) \\&\qquad\qquad\cdot g^{(h)}_{\omega'}({\bm q}(\bm k) + (n-m') {\bm \alpha}) W^{(h)}_{\psi\phi;m',\omega',\zeta}({\bm q}(\bm k) + (n-m'){\bm \alpha}, y_{2})\;.
\end{split}
\end{equation}
This proves (\ref{eq:recphi2}), after the change of variables $n-m'-m\to m'$.
\paragraph{Proof of (\ref{eq:rec2}).} The starting point is the formula:
\begin{equation}
W^{(h-1)}_{\phi\psi;n,\omega,\sigma}({\bm q}({\bm k}), x_{2}) = \beta L_{1} \frac{\partial^{2}}{\partial \psi^{(<h)-}_{{\bm q}({\bm k}) + n{\bm \alpha}, \omega} \partial \phi^{+}_{{\bm k}, x_{2}, \sigma}  }  \log \mathbbm{E}_{h}\Big( e^{V^{(h)}(\psi^{(\leq h)}, \phi)} \Big)\Big|_{\phi=\psi^{(<h)} = 0}\;.
\end{equation}
By (\ref{eq:1der}), we have
\begin{equation}\label{eq:psiphi1}
\begin{split}
& \frac{\partial^{2}}{\partial \psi^{(<h)-}_{{\bm q}({\bm k}) + n{\bm \alpha}, \omega} \partial \phi^{+}_{{\bm k}, x_{2}, \sigma}  }  \log \mathbbm{E}_{h}\Big( e^{V^{(h)}(\psi^{(\leq h)}, \phi)} \Big) \\
&\quad =  \frac{\partial}{\partial \psi^{(<h)-}_{{\bm q}({\bm k}) + n{\bm \alpha}, \omega}} \frac{1}{\beta L_{1}} \sum_{m,\omega'} W^{(h)}_{\phi\psi;m,\omega',\sigma}({\bm q}({\bm k}), x_{2})   \frac{\mathbbm{E}_{h}\Big( \psi^{(\leq h)-}_{{\bm q}({\bm k}) + m{\bm \alpha}, \omega'}  e^{V^{(h)}(\psi^{(\leq h)}, \phi)}\Big)}{\mathbbm{E}_{h}\Big( e^{V^{(h)}(\psi^{(\leq h)}, \phi)} \Big)} \\
&\quad = \frac{1}{\beta L_{1}} W^{(h)}_{\phi\psi;n,\omega,\sigma}({\bm q}({\bm k}), x_{2})\\
&\qquad - \frac{1}{\beta L_{1}} \sum_{m,\omega'} W^{(h)}_{\phi\psi;m,\omega',\sigma}({\bm q}({\bm k}), x_{2})   \frac{\mathbbm{E}_{h}\Big( \psi^{(\leq h)-}_{{\bm q}({\bm k}) + m{\bm \alpha}, \omega'}  \frac{\partial}{\partial \psi^{(<h)-}_{{\bm q}({\bm k}) + n{\bm \alpha}, \omega}} e^{V^{(h)}(\psi^{(\leq h)}, \phi)}\Big)}{\mathbbm{E}_{h}\Big( e^{V^{(h)}(\psi^{(\leq h)}, \phi)} \Big)}\;;
 \end{split}
\end{equation}
then, by (\ref{eq:parts}) we have:
\begin{equation}
\begin{split}
&\mathbbm{E}_{h}\Big( \psi^{(\leq h)-}_{{\bm q}({\bm k}) + m{\bm \alpha}, \omega'}  \frac{\partial}{\partial \psi^{(<h)-}_{{\bm q}({\bm k}) + n{\bm \alpha}, \omega}} e^{V^{(h)}(\psi^{(\leq h)}, \phi)}\Big)\Big|_{\psi^{(<h)} = \phi = 0} \\
&\quad = \beta L_{1} g^{(h)}_{\omega'}({\bm q}(\bm k) + m{\bm \alpha}) \frac{\partial^{2}}{\partial \psi^{(<h)+}_{{\bm q}({\bm k}) + m{\bm \alpha}, \omega'} \partial \psi^{(<h)-}_{{\bm q}({\bm k}) + n{\bm \alpha}, \omega}} \mathbbm{E}_{h}\Big( e^{V^{(h)}(\psi^{(\leq h)}, \phi)}\Big)\Big|_{\psi^{(<h)} = 0}\;.
\end{split}
\end{equation}
Thus, plugging this identity in (\ref{eq:psiphi1}), and recalling (\ref{eq:effpot}), we have:
\begin{equation}
\begin{split}
W^{(h-1)}_{\phi\psi;n,\omega,\sigma}({\bm q}({\bm k}), x_{2}) &= W^{(h)}_{\phi\psi;n,\omega,\sigma}({\bm q}({\bm k}), x_{2}) \\
&\quad + \sum_{m,\omega'} W^{(h)}_{\phi\psi;m,\omega',\sigma}({\bm q}({\bm k}), x_{2}) g^{(h)}_{\omega'}({\bm q}(\bm k) + m{\bm \alpha}) V^{(h-1)}_{\omega',\omega;n-m}({\bm q}(\bm k) + m{\bm \alpha})\;.
\end{split}
\end{equation}
This concludes the proof of (\ref{eq:rec2}).

\section{The anomalous bubble diagram}\label{app:bubble}

In this appendix we shall compute the anomalous bubble diagram $\mc{B}_{+}^{\beta, L_{1}}(\bm{p})$, and recover the expression in Eq. (\ref{eq:binfty}). The computation is well-known, but we reproduce it here for completeness. Let $\bm p:=(\eta, p)$. First, we approximate $\mc{B}_{+}^{\beta, L_{1}}(\bm{p})$ by an integral, corresponding to the $\beta, L_{1} \to \infty$ limit:

\begin{equation}
    \begin{split}
    &\mc{B}_{+}^{\beta, L_{1}}(\bm{p})\\
    &= \frac{1}{\beta L_{1}} \sum_{\bm{q} \in \msc D_{\beta,L}}  g_{+;\mr{s}}(\bm{p}+\bm{q}) g_{+;\mr{s}}(\bm{q})\\
        &=\underbrace{\int\frac{ d  \bm{q}}{(2\pi)^2}\bigg(\frac{\chi( \bm{p}+\bm{q})}{iv_{0,+}(\lambda)(\eta+q_0)+v_{1,+}(\lambda)(p+q_1)} \frac{\chi(\bm{q})}{iv_{0,+}(\lambda)q_0+v_{1,+}(\lambda)q_1} \bigg)}_{\eqqcolon \mc{B}_{+}^{\infty}(\bm{p})}\\
        &\quad+\mr{o}_{\beta, L}(1)\,.
            \end{split}
\end{equation}
Then,
\begin{equation}
\begin{split}
   &\mc{B}_{+}^{\infty}(\bm{p}) \\
   &\quad= \int \frac{ d  \bm q}{(2\pi)^2} \,\frac{\chi( \bm {q})\chi(\bm{p}+\bm{q} )}{iv_{0,+}(\lambda)\eta+v_{1,+}(\lambda)p} \bigg(\frac{1}{iv_{0,+}(\lambda)q_0+v_{1,+}(\lambda)q_1  }\\
   &\qquad-\frac{1}{iv_{0,+}(\lambda)(\eta+q_0)+v_{1,+}(\lambda)(p+q_1)}\bigg)\\
        &  \quad = \int \frac{ d  \bm q}{(2\pi)^2} \,\frac{\chi(\bm{q})\big(\chi(\bm{p}+\bm{q})-\chi (\bm{q}-\bm{p})\big)}{iv_{0,+}(\lambda)\eta+v_{1,+}(\lambda)p}\frac{1}{iv_{0,+}(\lambda)q_0+v_{1,+}(\lambda)q_1 }\;,
        \end{split}
\end{equation}
which we further rewrite as:
\begin{equation}
    \begin{split}
    \mc{B}_{+}^{\infty}(\bm{p})  &= \frac{2}{iv_{0,+}(\lambda)\eta+v_{1,+}(\lambda)p}\int \frac{ d  \bm q}{(2\pi)^2} \, \frac{\chi(\bm{q})\big(\bm p\cdot{\nabla \chi(\bm{q})}\big) }{iv_{0,+}(\lambda)q_0+v_{1,+}(\lambda)q_1}+\mr O(\|\bm p\| )\\ 
        &=\frac{2\eta}{iv_{0,+}(\lambda)\eta+v_{1,+}(\lambda)p} \underbrace{\int \frac{ d  \bm q}{(2\pi)^2} \frac{\chi(\bm{q})\partial_0 \chi(\bm{q})}{iv_{0,+}(\lambda)q_0+v_{1,+}(\lambda)q_1 }}_{\eqqcolon I_0} \\
    &\quad+\frac{2 p}{iv_{0,+}(\lambda)\eta+v_{1,+}(\lambda) p} \underbrace{\int \frac{ d  \bm q}{(2\pi)^2} \frac{ \chi(\bm{q})\partial_1\chi(\bm{q})}{iv_{0,+}(\lambda)q_0+v_{1,+}(\lambda) q_1 }}_{\eqqcolon I_1} +\mr O(\| \bm p\|).
    \end{split}
\end{equation}
We recall that, with a slight abuse of notation, 
\begin{equation}
    \chi(\bm{q})=\chi\Big(\sqrt{\big(v_{0,+}(\lambda)q_{0}\big)^2+\big(v_{1,+}(\lambda)q_{1}\big)^2}\Big)\;.
\end{equation}
Thus, after the change of variable $\kappa_{0} = v_{0,+}(\lambda) q_{0}$ and $\kappa_{1} = v_{1,+}(\lambda) q_{1}$:
\begin{equation}
    \begin{split}
        I_0& = \frac{1}{|{v}_{1,+}(\lambda) |} \int \frac{ d  \bm \kappa}{(2\pi)^2} \frac{\kappa_0 \, \chi(\|\bm \kappa\|)}{ \|\bm \kappa\|}\frac{\chi'(\|\bm \kappa\|)}{i \kappa_0+\kappa_1}\\
       & =\frac{1}{|{v}_{1,+}(\lambda) |}\int \frac{ d  \bm \kappa}{(2\pi)^2} \frac{\kappa_{1} \,\chi(\|\bm \kappa\|)}{\|\bm \kappa\|} \frac{\chi'(\|\bm \kappa\|)}{i \kappa_1+\kappa_0}\\
       &=-\frac{i}{|{v}_{1,+}(\lambda) |}\int \frac{ d  \bm \kappa}{(2\pi)^2} \frac{\kappa_{1}\, \chi(\|\bm \kappa\|)}{\|\bm \kappa\|}\frac{\chi'(\|\bm \kappa\|)}{-i\kappa_0+\kappa_{1}}\\
       &=-\frac{i}{|{v}_{1,+}(\lambda) |}\int \frac{ d  \bm \kappa}{(2\pi)^2} \frac{\kappa_{1}\, \chi(\|\bm \kappa\|)}{\|\bm \kappa\|}\frac{\chi'(\|\bm \kappa\|)}{i\kappa_0+\kappa_{1}}=   -\frac{i}{\mf{v}(\lambda) } I_1\;,
    \end{split}
\end{equation}
that is $I_1=i\mf{v}(\lambda)I_{0}$. The integral in $I_0$ can be computed explicitly:
\begin{equation}
    \begin{split}
       \int \frac{ d  \bm \kappa}{(2\pi)^2} \frac{\kappa_0 \chi(\|\bm \kappa\|)}{\|\bm \kappa\|}\frac{\chi'(\|\bm \kappa\|)}{i \kappa_0+\kappa_1} 
       &=  \int \frac{ d  \theta\, d  \rho }{(2\pi)^2}\,\rho\,\frac{\rho \cos \theta \,\chi(\rho)\chi'(\rho)}{\rho(i \rho\cos \theta +\rho\sin \theta)}\\
        &=  \int  d  \rho \chi(\rho)\chi'(\rho) \int  \frac{ d  \theta}{(2\pi)^2} \frac{\cos \theta}{i \cos \theta +\sin \theta} \\
        &= \restr{\frac{1}{2}\big(\chi \big)^2(\rho)}{0}^{\mkern-6mu \updelta}  \int \frac{ d  \theta}{(2\pi)^2}   {\cos \theta\big(-i \cos \theta +\sin \theta\big)} \\
        &= \frac{i}{2} \int \frac{ d  \theta}{(2\pi)^2} \cos^2 \theta =\frac{i}{8\pi}\;.
    \end{split}
\end{equation}
In conclusion:
\begin{equation}
    \begin{split}
      \mc{B}_{+}^{\infty}(\bm{p})&=\frac{1}{4\pi \big| v_{0,+}(\lambda) v_{1,+}(\lambda)\big|}\frac{iv_{0,+}(\lambda)\eta-v_{1,+}(\lambda)p}{i v_{0,+}(\lambda)\eta +v_{1,+}(\lambda)p} +\mr O(\|\bm{p}\|) .
    \end{split}
\end{equation}

\section{Proof of (\ref{eq:RbdFou})}\label{app:Rest}

In this appendix we shall prove a momentum-space estimates for the remainder term in the expression of the two-point function, Eq. (\ref{eq:RbdFou}), which is used the analysis of the vertex Ward identity of Section \ref{sec:vertexcons}. Suppose that ${\bm k}$ is close to ${\bm k}_{F}^{+}$, and that, for $M\geq 1$:
\begin{equation}\label{eq:k0hyp}
|k - k_{F}^{+}|^{M} \leq |k_{0}| \leq |k - k_{F}^{+}|\;.
\end{equation}
Then, recalling the notation $\hat R_{n,\sigma,\zeta}({\bm k}; x_{2}, y_{2}) = \hat R_{\sigma,\zeta}({\bm k}, {\bm k} + n{\bm \alpha}; x_{2}, y_{2})$, and setting $\hat R_{n,\sigma,\zeta} = \hat R_{1;n,\sigma,\zeta} + \hat R_{2;n,\sigma,\zeta}$ where $\hat R_{i;n,\sigma,\zeta}$ are defined in Section \ref{sec:thm2ptproof}, we claim that, for $L_{2}$ large enough:
\begin{equation}\label{eq:R1momentum}
\begin{split}
\big| \text{d}_{k_{0}}^{n_{0}} \text{d}_{k_{1}}^{n_{1}} \hat R_{n,\sigma,\zeta}({\bm k}; x_{2}, y_{2}) \big| &\leq \frac{C_{n_{0},n_{1}} |\lambda|^{\delta_{n\neq 0}} e^{-\frac{c}{16}|n|} e^{-\tilde c|x_{2} - y_{2}|}}{\| {\bm k} - {\bm k}_{F}^{+} \|^{n_{0} + n_{1} + 1-\theta}}\\&\quad + \sum_{\substack{\omega_{1},\omega_{2} \\ (\omega_{1},\omega_{2}) \neq (+,+)}} \!\!\! C_{n_{0},n_{1}}\frac{e^{-\frac{c}{16}|n|} e^{-c(|x_{2}|_{\omega_{1}} + |y_{2}|_{\omega_{2}})}}{|k_{0}|^{n_{0} + n_{1} + 1-\theta}}\;.
\end{split}
\end{equation}
The bound holds true for $\hat R_{2;n,\sigma,\zeta}({\bm k}; x_{2}, y_{2})$, thanks to (\ref{eq:R2bd}). Let us now prove it for $\hat R_{1;n,\sigma,\zeta}({\bm k}; x_{2}, y_{2})$ defined as in (\ref{eq:R1n}). Let us denote by $\hat R_{1;+,n,\sigma,\zeta}({\bm k}; x_{2}, y_{2})$ the contribution to $\hat R_{1;n,\sigma,\zeta}({\bm k}; x_{2}, y_{2})$ coming from all chiralities equal to $+$. Then, by the exponential decay of the edge modes, and using that $\gamma^{h-1} \geq |k_{0}|$, we easily get:
\begin{equation}
\begin{split}
&\Big| \text{d}_{k_{0}}^{n_{0}} \text{d}_{k_{1}}^{n_{1}} \hat R_{1;n,\sigma,\zeta}({\bm k}; x_{2}, y_{2}) - \text{d}_{k_{0}}^{n_{0}} \text{d}_{k_{1}}^{n_{1}} \hat R_{1;+,n,\sigma,\zeta}({\bm k}; x_{2}, y_{2}) \Big| \\&\qquad \leq   \sum_{\substack{\omega_{1},\omega_{2} \\ (\omega_{1},\omega_{2}) \neq (+,+)}} \!\!\! C_{n_{0},n_{1}}\frac{e^{-\frac{c}{16}|n|} e^{-c(|x_{2}|_{\omega_{1}} + |y_{2}|_{\omega_{2}})}}{|k_{0}|^{n_{0} + n_{1} + 1-\theta}}\;.
\end{split}
\end{equation}
Next, consider $\hat R_{1;+,n,\sigma,\zeta}({\bm k}; x_{2}, y_{2})$. Observe that $\hat R_{1;+,n,\sigma,\zeta}$ is given by a sum over scales $h$ of contributions of the form (\ref{eq:Ah0}) or of the form (\ref{eq:Bh}). We organize the sum over $m$ and over the scales $h$ as:
\begin{equation}\label{eq:organize}
\begin{split}
&\sum_{h, m} (\cdots) = \sum_{h, m} \mathbbm{1}(m=0)(\cdots) \\&\quad  + \sum_{m} \sum_{h:\, \gamma^{h} \geq |k - k_{F}^{+}|}\mathbbm{1}(m\neq 0)(\cdots) + \sum_{m} \sum_{h:\, \gamma^{h} < |k - k_{F}^{+}|} \mathbbm{1}(m\neq 0)(\cdots)\;;
\end{split}
\end{equation}
correspondingly, we write:
\begin{equation}
\begin{split}
&\hat R_{1;+,n,\sigma,\zeta}({\bm k}; x_{2}, y_{2}) \\
&= \hat R^{(a)}_{1;+,n,\sigma,\zeta}({\bm k}; x_{2}, y_{2}) + \hat R^{(b)}_{1;+,n,\sigma,\zeta}({\bm k}; x_{2}, y_{2}) + \hat R^{(c)}_{1;+,n,\sigma,\zeta}({\bm k}; x_{2}, y_{2})\;.
\end{split}
\end{equation}
Consider the term $\hat R^{(a)}_{1;+,n,\sigma,\zeta}({\bm k}; x_{2}, y_{2})$. Here, the scale $h$ of the contributions (\ref{eq:Ah0}), (\ref{eq:Bh}) has to be such that $\delta\gamma^{h-1} \leq \|{\bm k} - {\bm k}_{F}^{+}\| \leq \delta\gamma^{h+1}$. Therefore, from the bounds (\ref{eq:errA1}), (\ref{eq:errB1}) we get:
\begin{equation}\label{eq:R1a}
\sum_{h} \Big|\text{d}_{k_{0}}^{n_{0}} \text{d}_{k_{1}}^{n_{1}}\hat R^{(a);(h)}_{1;+,n,\sigma,\zeta}({\bm k}; x_{2}, y_{2})\Big| \leq C_{n_{0},n_{1}} \frac{|\lambda|^{\delta_{n\neq 0}} e^{-\frac{c}{16}|n|} e^{- \frac{\tilde c}{32}x_{2}} e^{- \frac{\tilde c}{32}y_{2}}}{\|{\bm k} - {\bm k}^{+}_{F}\|^{n_{0} + n_{1}+1-\theta}}\;.
\end{equation}
Consider now the term $\hat R^{(b)}_{1;n,\sigma,\zeta}({\bm k}; x_{2}, y_{2})$. Here, using that ${\bm k} - {\bm k}_{F}^{+} + m\alpha$ has to be in the support of $g_{+}^{(h)}$, together with the Diophantine condition, we get:
\begin{equation}
\Big(m\neq 0\;,\quad | m \alpha |_{\mathbb{T}} \leq 2\delta\gamma^{h+1}\Big)\Rightarrow |m| \geq C\gamma^{-h/\tau}\;.
\end{equation}
Thus, we can use a fraction of the exponential decay in $m$ of the bound for the kernel $W^{(h)}_{\phi\psi;m,+,\sigma}$ to prove that:
\begin{equation}\label{eq:R1b}
\sum_{h} \Big|\text{d}_{k_{0}}^{n_{0}} \text{d}_{k_{1}}^{n_{1}}\hat R^{(b);(h)}_{1;+,n,\sigma,\zeta}({\bm k}; x_{2}, y_{2})\Big| \leq C_{n_{0},n_{1}} |\lambda|^{\delta_{n\neq 0}} e^{-\frac{c}{16}|n|} e^{- \frac{\tilde c}{32}x_{2}} e^{- \frac{\tilde c}{32}y_{2}}\;.
\end{equation}
Finally, consider the last term in (\ref{eq:organize}). Here, we use that the last scale is determined by $k_{0}$, since $\delta\gamma^{h-1} \geq |k_{0}|$. Also, using that $|k - k_{F}^{+} + m\alpha|_{\mathbb{T}} \leq \gamma^{h+1}$, that $\gamma^{h} < |k - k_{F}^{+}|$, and hence that $|m\alpha|_{\mathbb{T}} \leq \gamma |k - k_{F}^{+}| + |k-k_{F}^{+}|$, the Diophantine condition in this regime implies:
\begin{equation}
\Big(m\neq 0\;,\quad | m \alpha |_{\mathbb{T}} \leq 2\gamma |k - k_{F}^{+}|\Big)\Rightarrow |m| \geq C|k-k_{F}^{+}|^{-\frac{1}{\tau}}\;.
\end{equation}
Therefore, using again a fraction of the exponential decay in $m$ of the bound for the kernel $W^{(h)}_{\phi\psi;m,+,\sigma}$, we get:
\begin{equation}\label{eq:R1c}
\begin{split}
\sum_{h} \Big| \text{d}_{k_{0}}^{n_{0}} \text{d}_{k_{1}}^{n_{1}} \hat R^{(c);(h)}_{1;+,n,\sigma,\zeta}({\bm k}; x_{2}, y_{2})\Big| &\leq C_{n_{0},n_{1}} \frac{|\lambda|^{\delta_{n\neq 0}} e^{-\frac{c}{16}|n|} e^{- \frac{\tilde c}{32}x_{2}} e^{- \frac{\tilde c}{32}y_{2}}}{|k_{0}|^{n_{0} + n_{1}+1-\theta}} e^{-\delta |k - k_{F}^{+}|^{-\frac{1}{\tau}}} \\
&\leq C_{n_{0},n_{1}}  \frac{|\lambda|^{\delta_{n\neq 0}} e^{-\frac{c}{16}|n|} e^{- \frac{\tilde c}{32}x_{2}} e^{- \frac{\tilde c}{32}y_{2}}}{\|{\bm k} - {\bm k_{F}^{+}} \|^{1-\theta}}\;,
\end{split}
\end{equation}
where we used that, thanks to the assumption (\ref{eq:k0hyp}):
\begin{equation}
\begin{split}
\frac{e^{-\delta |k - k_{F}^{+}|^{-\frac{1}{\tau}}}}{|k_{0}|^{n_{0} + n_{1} + 1-\theta}} &= \frac{e^{-\delta |k - k_{F}^{+}|^{-\frac{1}{\tau}}}}{\|{\bm k} - {\bm k_{F}^{+}} \|^{n_{0} + n_{1} + 1-\theta}} \frac{\|{\bm k} - {\bm k_{F}^{+}} \|^{n_{0} + n_{1} + 1-\theta}}{|k_{0}|^{n_{0} + n_{1} + 1-\theta}} \\
&\leq K\frac{e^{-\delta |k - k_{F}^{+}|^{-\frac{1}{\tau}}}}{\|{\bm k} - {\bm k_{F}^{+}} \|^{n_{0} + n_{1} + 1-\theta}} \frac{| k - k_{F}^{+} |^{n_{0} + n_{1} + 1-\theta}}{|k - k_{F}^{+}|^{M(n_{0} + n_{1} + 1-\theta)}} \\
&\leq \frac{K_{n_{0}, n_{1}}}{\|{\bm k} - {\bm k_{F}^{+}} \|^{n_{0} + n_{1} + 1-\theta}}\;.
\end{split}
\end{equation}
Thus, Eqs. (\ref{eq:R1a}), (\ref{eq:R1b}), (\ref{eq:R1c}) prove the momentum-space estimate (\ref{eq:R1momentum}).

\section{Proof of Lemma \ref{lem:contK}}\label{app:contK}

Let us rewrite the function $g(\eta,\theta)$ as:
\begin{equation}
g(\eta, \theta) = \frac{\theta}{L_{1}}\sum_{p \in \msc D_{L}} \sum_{m\in \mathbb{Z}} \sum_{x_{2},y_{2}}\hat \mu_{\theta}(p, x_{2}) \hat \phi_{\theta,\ell}(-p+m\alpha,y_{2}) \hat F_{m}(\eta, p; x_{2}, y_{2})
\end{equation}
where
\begin{equation}
\hat F_{m}(\eta, p; x_{2}, y_{2}) = \int_{-\beta/2}^{\beta/2} ds\, e^{-i\eta s} \sum_{x_{1} = 1}^{L_{1}}  e^{-ip_{1}z_{1}} F_{m}((s, z_{1}); x_{2}, y_{2})\;.
\end{equation}
Let us change variable, $p \to \theta p$:
\begin{equation}
g(\eta, \theta) = \frac{\theta}{L_{1}}\sum_{p \in \frac{1}{\theta}\msc D_{L}} \sum_{m\in \mathbb{Z}} \sum_{x_{2},y_{2}}\hat \mu_{\theta}(\theta p, x_{2}) \hat \phi_{\theta,\ell}(-\theta p+m\alpha,y_{2}) \hat F_{m}(\eta, \theta p; x_{2}, y_{2})\;,
\end{equation}
where, by (\ref{eq:estmuphi}), the functions $\hat \mu_{\theta}(\theta q, x_{2})$, $\hat \phi_{\theta,\ell}(\theta q,y_{2})$ decay faster than any power in $|q|_{\mathbb{T}_{\theta^{-1}}}$. Thus, in the sum over $m$ we isolate the term with $m=0$ from the rest:
\begin{equation}
\begin{split}
&g(\eta, \theta) \\
&= \frac{\theta}{L_{1}}\sum_{p \in \frac{1}{\theta}\msc D_{L}} \sum_{x_{2},y_{2}}\hat \mu_{\theta}(\theta p, x_{2}) \hat \phi_{\theta,\ell}(-\theta p,y_{2}) \hat F_{0}(\eta, \theta p; x_{2}, y_{2})\\
&\quad + \frac{\theta}{L_{1}}\sum_{p \in \frac{1}{\theta}\msc D_{L}} \sum_{m\neq 0} \sum_{x_{2},y_{2}}\hat \mu_{\theta}(\theta p, x_{2}) \hat \phi_{\theta,\ell}(-\theta (p + m\alpha/\theta),y_{2}) \hat F_{m}(\eta, \theta p; x_{2}, y_{2})\\
&=: g_{0}(\eta, \theta) + g_{>}(\eta, \theta)\;.
\end{split}
\end{equation}
As a consequence of the assumptions on $F$, we have
\begin{equation}
|\hat F_{m}(\eta, \theta p; x_{2}, y_{2})|\leq Ce^{-c|m|} e^{-c|x_{2} - y_{2}|}
\end{equation}
uniformly in all other parameters; hence, we can estimate:
\begin{equation}
\begin{split}
|g_{>}(\eta, \theta)| &\leq \frac{1}{\theta L_{1}}\sum_{p \in \frac{1}{\theta}\msc D_{L}} \sum_{m\neq 0} \sum_{x_{2},y_{2}} \frac{C_{r}}{1+(|p|_{\mathbb{T}_{\theta^{-1}}})^{r}} \frac{1}{1+(|p + m\alpha/\theta|_{\mathbb{T}_{\theta^{-1}}})^{r}} \\&\quad \cdot \frac{1}{1 + | \theta x_{2}|^{r}} \frac{1}{1+|y_{2} / \ell|^{r}} e^{-c|m|} e^{-c|x_{2} - y_{2}|}\;.
\end{split}
\end{equation}
We then separate the values of $m$ such that $| m\alpha |_{\mathbb{T}} > \sqrt{\theta}$ and $| m\alpha |_{\mathbb{T}} \leq \sqrt{\theta}$; both contributions are small as $\theta \to 0$, for different reasons. For the former, we proceed as follows:
\begin{equation}
\begin{split}
&\frac{1}{1+(|p|_{\mathbb{T}_{\theta^{-1}}})^{r}} \frac{1}{1+(|p + m\alpha/\theta|_{\mathbb{T}_{\theta^{-1}}})^{r}} \\
&\quad \leq \frac{C_{r}}{1+(|p|_{\mathbb{T}_{\theta^{-1}}})^{r/2}} \frac{1}{1+(|p|_{\mathbb{T}_{\theta^{-1}}})^{r/2}} \frac{1}{1+(|p + m\alpha/\theta|_{\mathbb{T}_{\theta^{-1}}})^{r/2}} \\
&\quad = \frac{C_{r}}{1+(|p|_{\mathbb{T}_{\theta^{-1}}})^{r/2}} \frac{1}{1+(|p|_{\mathbb{T}_{\theta^{-1}}})^{r/2}} \frac{|-p + p + m\alpha/\theta|^{r/2}_{\mathbb{T}_{\theta^{-1}}}}{1+(|p + m\alpha/\theta|_{\mathbb{T}_{\theta^{-1}}})^{r/2}}  \frac{1}{1+(|m\alpha/\theta|_{\mathbb{T}_{\theta^{-1}}})^{r/2}} \\
&\quad \leq \frac{K_{r}}{1+(|p|_{\mathbb{T}_{\theta^{-1}}})^{r/2}} \frac{1}{1+(|p|_{\mathbb{T}_{\theta^{-1}}})^{r/2}} \frac{|p + m\alpha/\theta|^{r/2}_{\mathbb{T}_{\theta^{-1}}} + |p|^{r/2}_{\mathbb{T}_{\theta^{-1}}}}{1+(|p + m\alpha/\theta|_{\mathbb{T}_{\theta^{-1}}})^{r/2}}  \frac{1}{1+(|m\alpha/\theta|_{\mathbb{T}_{\theta^{-1}}})^{r/2}} \\
&\quad\leq \frac{2K_{r}}{1+(|p|_{\mathbb{T}_{\theta^{-1}}})^{r/2}}  \frac{1}{1+(|m\alpha/\theta|_{\mathbb{T}_{\theta^{-1}}})^{r/2}}\;, 
\end{split}
\end{equation}
and we use the smallness of the function $e^{-c|m|} | m\alpha/\theta |_{\mathbb{T}_{\theta^{-1}}}^{-r/2}$ as $\theta \ll 1$. For the regime $| m\alpha |_{\mathbb{T}} \leq \sqrt{\theta}$, we use the Diophantine condition; we omit the details. We obtain:
\begin{equation}
|g_{>}(\eta, \theta)| \leq C_{\ell}\theta\;.
\end{equation}
Next, consider the term $g_{0}(\eta, \theta)$. Using that, as a consequence of assumption (\ref{eq:Festlem}),
\begin{equation}
\Big| \hat F_{m}(\eta, \theta p; x_{2}, y_{2}) - \hat F_{m}(\eta, 0; x_{2}, y_{2}) \Big| \leq C |\theta p|^{\xi/2} e^{-c|m|} e^{-c|x_{2} - y_{2}|}\;,
\end{equation}
we have:
\begin{equation}\label{eq:C9}
\begin{split}
g_{0}(\eta, \theta) &= \frac{\theta}{L_{1}}\sum_{p \in \frac{1}{\theta}\msc D_{L}} \sum_{x_{2},y_{2}}\hat \mu_{\theta}(\theta p, x_{2}) \hat \phi_{\theta,\ell}(-\theta p,y_{2}) \hat F_{0}(\eta, 0; x_{2}, y_{2}) + \tilde g_{0}(\eta, \theta) \\
|\tilde g_{0}(\eta, \theta)| &\leq C_{\ell}\theta^{\xi/2}\;.
\end{split}
\end{equation}
Concerning the main term in (\ref{eq:C9}), for $\theta \ll 1$ and as $L_{1}\to \infty$, we can approximate it as an integral:
\begin{equation}\label{eq:gintegral}
\begin{split}
&\frac{\theta}{L_{1}}\sum_{p \in \frac{1}{\theta}\msc D_{L}} \sum_{x_{2},y_{2}}\hat \mu_{\theta}(\theta p, x_{2}) \hat \phi_{\theta,\ell}(-\theta p,y_{2}) \hat F_{0}(\eta, 0; x_{2}, y_{2}) \\
&\quad = \sum_{x_{2}, y_{2}} \int_{\mathbb{R}} \frac{dp}{(2\pi)}\, \hat \mu(p, 0) \hat \varphi(-p, y_{2}/\ell) \hat F_{0}(\eta, 0; x_{2}, y_{2}) + O_{\ell}(\theta^{\alpha})\;.
\end{split}
\end{equation}
for some $\alpha > 0$.

All in all, we rewrote $g(\eta, \theta)$ as a quantity that does not depend on $\theta$, given by the integral in (\ref{eq:gintegral}), plus error terms that vanish with $\theta$ uniformly in all parameters except $\ell$. In particular,
\begin{equation}
| g(\eta, \theta) - g(\eta, \theta') | \leq C_{\ell} (|\theta|^{\alpha} + |\theta'|^{\alpha})
\end{equation}
for some $\alpha>0$. This concludes the proof of Lemma \ref{lem:contK}.\qed

\end{document}